\documentclass{article}

\usepackage{fullpage}
\usepackage{natbib}
\usepackage[utf8]{inputenc} 
\usepackage[T1]{fontenc}
\usepackage{hyperref}       
\usepackage{url}            
\usepackage{booktabs}       
\usepackage{amsfonts}       
\usepackage{nicefrac}       
\usepackage{microtype}      
\usepackage{xcolor}
\usepackage{algorithm}
\usepackage{algorithmicx}
\usepackage[noend]{algpseudocode}
\usepackage{amsmath}
\usepackage{amssymb}
\usepackage{mathtools}
\usepackage{amsthm}
\usepackage{aliascnt}
\usepackage[capitalize,noabbrev]{cleveref}
\usepackage{enumitem}
\usepackage{physics}

\usepackage[mathlines]{lineno}
\usepackage{etoolbox}
\newcommand*\linenomathpatch[1]{
  \cspreto{#1}{\linenomath}
  \cspreto{#1*}{\linenomath}
  \csappto{end#1}{\endlinenomath}
  \csappto{end#1*}{\endlinenomath}
}
\newcommand*\linenomathpatchAMS[1]{
  \cspreto{#1}{\linenomathAMS}
  \cspreto{#1*}{\linenomathAMS}
  \csappto{end#1}{\endlinenomath}
  \csappto{end#1*}{\endlinenomath}
}

\expandafter\ifx\linenomath\linenomathWithnumbers
  \let\linenomathAMS\linenomathWithnumbers
  
  \patchcmd\linenomathAMS{\advance\postdisplaypenalty\linenopenalty}{}{}{}
\else
  \let\linenomathAMS\linenomathNonumbers
\fi

\linenomathpatch{equation}
\linenomathpatchAMS{gather}
\linenomathpatchAMS{multline}
\linenomathpatchAMS{align}
\linenomathpatchAMS{alignat}
\linenomathpatchAMS{flalign}

\theoremstyle{plain}

\newtheorem{theorem}{Theorem}[section]
\crefname{theorem}{Theorem}{Theorems}
\Crefname{theorem}{Theorem}{Theorems}

\newaliascnt{proposition}{theorem}

\aliascntresetthe{proposition}
\crefname{proposition}{Proposition}{Propositions}
\Crefname{proposition}{Proposition}{Propositions}

\newaliascnt{lemma}{theorem}
\newtheorem{lemma}[lemma]{Lemma}
\aliascntresetthe{lemma}
\crefname{lemma}{Lemma}{Lemmas}
\Crefname{lemma}{Lemma}{Lemmas}

\newaliascnt{fact}{theorem}
\newtheorem{fact}[fact]{Fact}
\aliascntresetthe{fact}
\crefname{fact}{Fact}{Facts}
\Crefname{fact}{Fact}{Facts}

\newaliascnt{corollary}{theorem}
\newtheorem{corollary}[corollary]{Corollary}
\aliascntresetthe{corollary}
\crefname{corollary}{Corollary}{Corollaries}
\Crefname{corollary}{Corollary}{Corollaries}

\newaliascnt{claim}{theorem}

\aliascntresetthe{claim}
\crefname{claim}{Claim}{Claims}
\Crefname{claim}{Claim}{Claims}

\theoremstyle{definition}

\newaliascnt{definition}{theorem}
\newtheorem{definition}[definition]{Definition}
\aliascntresetthe{definition}
\crefname{definition}{Definition}{Definitions}
\Crefname{definition}{Definition}{Definitions}

\newaliascnt{assumption}{theorem}

\aliascntresetthe{assumption}
\crefname{assumption}{Assumption}{Assumptions}
\Crefname{assumption}{Assumption}{Assumptions}

\newaliascnt{problem}{theorem}
\newtheorem{problem}[problem]{Problem}
\aliascntresetthe{problem}
\crefname{problem}{Problem}{Problems}
\Crefname{problem}{Problem}{Problems}

\newaliascnt{remark}{theorem}
\newtheorem{remark}[remark]{Remark}
\aliascntresetthe{remark}
\crefname{remark}{Remark}{Remarks}
\Crefname{remark}{Remark}{Remarks}

\usepackage{amsmath,amsfonts,bm}

\def\eqref#1{equation~\ref{#1}}

\def\1{\bm{1}}

\DeclareMathAlphabet{\mathsfit}{\encodingdefault}{\sfdefault}{m}{sl}
\SetMathAlphabet{\mathsfit}{bold}{\encodingdefault}{\sfdefault}{bx}{n}

\def\sR{{\mathbb{R}}}

\newcommand{\E}{\mathbb{E}}

\newcommand{\Var}{\mathrm{Var}}

\usepackage{marginnote}

\allowdisplaybreaks

\algnewcommand{\LeftComment}[1]{\Statex $\triangleright$ #1}
\algnewcommand\algorithmicforeach{\textbf{for each}}
\algdef{S}[FOR]{ForEach}[1]{\algorithmicforeach\ #1\ \algorithmicdo}
\algnewcommand{\IIf}[1]{\State\algorithmicif\ #1\ \algorithmicthen}
\algnewcommand{\ElseIIf}[1]{\algorithmicelse\ #1}
\algnewcommand{\EndIIf}{\unskip\ \algorithmicend\ \algorithmicif}

\newcommand{\poly}{\operatorname{poly}}

\DeclareRobustCommand{\CM}{{\bf CM}}
\DeclareRobustCommand{\CMin}{{\bf CM}^\mathrm{in}}
\DeclareRobustCommand{\CMout}{{\bf CM}^\mathrm{out}}
\DeclareRobustCommand{\CMpos}{{\bf CM}^\mathrm{pos}}
\DeclareRobustCommand{\CMneg}{{\bf CM}^\mathrm{neg}}
\DeclareRobustCommand{\opt}{\mathrm{OPT}}

\DeclareRobustCommand{\neg}{\mathrm{neg}}
\DeclareRobustCommand{\pos}{\mathrm{pos}}
\definecolor{offwhite}{rgb}{0.98, 0.98, 0.98}

\newcommand{\eps}{\varepsilon}
\newcommand{\mcut}{\textsc{Max-Cut}}
\newcommand{\mcsp}{\textsc{Max-CSP}}

\newcommand{\mkcsp}{\textsc{Max-$k$CSP}}

\title{Learning-Augmented Streaming Algorithms for Approximating Boolean Max-CSPs\thanks{All authors contributed equally and are listed in alphabetical order.}}

\author{
  Yinhao Dong\thanks{School of Computer Science and Technology, University of Science and Technology of China. Email: \texttt{yhdong@mail.ustc.edu.cn}.}
  \and
  Pan Peng\thanks{School of Computer Science and Technology, University of Science and Technology of China. Email: \texttt{ppeng@ustc.edu.cn}.}
  \and
  Ali Vakilian\thanks{Rutgers University. Email: \texttt{a.vakilian@rutgers.edu}.}
  \and
  Santhoshini Velusamy\thanks{University of Waterloo. Email: \texttt{santhoshini.velusamy@uwaterloo.ca}.}
  \and
  Xiaoyang Xu\thanks{University of California, Santa Barbara. Email: \texttt{xiaoyangxu@ucsb.edu}.}
}
\date{}

\begin{document}

\maketitle

\begin{abstract}
We study learning-augmented streaming algorithms for estimating the value of Boolean Maximum Constraint Satisfaction Problems (Max-CSPs). Specifically, we consider streaming algorithms equipped with an $\varepsilon$-accurate oracle: for each variable, the oracle outputs a label in $\{-1,1\}$ that agrees with a fixed optimal assignment with probability $1/2+\varepsilon$, and disagrees otherwise. Prior work by Dong, Peng, and Vakilian [ITCS 2025] showed that, for the special case of Max-Cut, such an oracle enables a $(1/2+\Omega(\varepsilon^2))$-approximation using $\poly(1/\varepsilon)$ words of space in insertion-only streams, and $\poly(1/\varepsilon,\log n)$ words of space in dynamic streams. We substantially generalize and strengthen their result, provided that a highly accurate estimate of $\varepsilon$ is available. For any Boolean Max-$2$CSP and any $\eta>0$ independent of $\varepsilon$, we obtain, with high constant probability, a single-pass $(1-\eta)$-approximation using $\poly(1/\varepsilon,1/\eta)$ words of space in insertion-only streams, and $\poly(1/\varepsilon,1/\eta,\log n)$ words of space in dynamic streams, where $n$ is the number of variables. We further extend our techniques to general Boolean Max-$k$CSPs for $k\ge 3$, with slightly worse space complexity.
\end{abstract}

\newpage
\tableofcontents
\newpage
\section{Introduction}
Constraint Satisfaction Problems (CSPs) form a highly expressive and well-studied class of combinatorial optimization problems. 
The goal is to assign values from a finite domain to $n$ variables so as to maximize the number of satisfied constraints, where each constraint depends only on a small subset of variables.
This framework captures classical problems such as \textsc{Max-Cut}, \textsc{Max-Dicut}, and \textsc{Max-$2$AND}. 
Beyond their central role in approximation algorithms and complexity theory, CSP-like finite-domain
labeling formulations also appear throughout machine learning, especially in structured prediction,
conditional random fields, MAP inference, and energy minimization in graphical models and computer
vision~\citep{lafferty2001conditional,taskar2003max,tsochantaridis2004support,boykov2001fast}. 
More recently, CSPs and SAT have also served as benchmarks and targets for learning-based solvers and
algorithm-selection methods~\citep{xu2018towards,selsam2019learning}.

Recently, there has been a flurry of work on CSPs in the \emph{streaming} model (e.g.~\cite{guruswami2017streaming,chou2022linear,chou2024sketching}; see also \cref{sec:relatedwork}).  In this model, the constraints of the instance arrive
sequentially, possibly with deletions in the dynamic setting, and the algorithm
is allowed only limited memory, typically sublinear in the input size. Since
explicitly outputting an assignment for \textsc{Max-Cut} or more general
\textsc{Max-CSP}s already requires $\Omega(n)$ space, the standard streaming
objective is instead to estimate the optimum value of the instance.

A central motivation for studying streaming CSPs is that they admit unconditional, information-theoretic lower bounds (typically via communication complexity) capturing the limits of compressing an instance while preserving its optimum value. Streaming algorithms can still achieve meaningful approximations, and a key theme is understanding the trade-off between approximation quality and space complexity (see recent surveys \cite{sudan2022streaming,singer2025nine}). Given these abundant lower bounds, it is natural to ask whether they can be circumvented under reasonable additional assumptions, that is, \emph{beyond the worst-case} setting.

In many applications, the abundance of data, together with powerful prediction models such as deep neural networks and large language models, makes it possible to obtain reliable hints about the underlying optimization instance. This perspective has motivated the rapidly growing area of
\emph{learning-augmented algorithms}, also known as \emph{algorithms with
predictions}, which studies algorithms that exploit oracle predictions while
retaining provable guarantees.\footnote{See
\url{https://algorithms-with-predictions.github.io/} for an up-to-date
repository of papers in this area.} 
A particularly simple and widely used model is that of \emph{$\eps$-accurate predictions}, where the oracle predicts membership in a fixed optimal solution (or an optimal certificate) with probability only slightly better than random guessing, namely $1/2+\eps$. 
This model was introduced in~\cite{cohen2024max,ghoshal2025constraint} for problems such as \mcut~and \mcsp s, showing that even such mildly biased predictions can be leveraged to obtain provable improvements over prediction-free approximation barriers. 
For example, for \mcut~these results can go beyond the classical $0.878$ threshold achieved by the Goemans--Williamson SDP rounding~\cite{goemans1995improved}, which is known to be optimal under the Unique Games Conjecture~\cite{khot2007optimal}.

The work most closely related to ours is~\cite{dong2025learning}, which studies streaming \mcut~with $\eps$-accurate predictions. They obtain a streaming algorithm that breaks the Kapralov--Krachun lower bound~\cite{kapralov2019optimal} for beating a $1/2$-approximation in $o(n/2^{\poly(1/\eps)})$ space. This demonstrates that even
weak, noisy predictions can fundamentally change the power of streaming
algorithms for CSPs.

In this paper, we study {\em learning-augmented streaming algorithms for Boolean \mcsp s} and prove that they are significantly more powerful than classical streaming algorithms. In particular, we consider an $\eps$-accurate oracle similar to the one defined in \cite{dong2025learning}, where there is a hidden optimal assignment $\boldsymbol{\sigma}^*\in \{-1,+1\}^n$ to the variables $\{x_1,\dots,x_n\}$, and the streaming algorithm has (unlimited) oracle access to a prediction vector $\mathbf{Y}\in \{-1,+1\}^n$, where for every variable $x_i$, $Y_i = \sigma_i^*$ \emph{independently} with probability $1/2+\eps$.

\subsection{Our results}
Our main result shows that weakly accurate predictions dramatically increase the
power of streaming algorithms for Boolean \mkcsp{}s. For any constant
$\eps>0$, access to an $\eps$-accurate oracle is sufficient to obtain an
arbitrarily accurate constant-factor approximation using extremely small space:
for insertion-only streams, where constraints are only added, our algorithms use
space independent of $n$ up to the word size; for dynamic streams, where
constraints may be both inserted and deleted, they use only $\poly(\log n)$
space. 
We assume that the prediction accuracy parameter $\eps$ is known to the algorithm in advance.\footnote{In fact, it suffices to know a highly accurate estimate of $\eps$ for our algorithms to work. Specifically, it is enough to have an estimate $\eps' \in (0,1/2]$ such that $\eps' = (1\pm O(\eta))\cdot \eps$ for \textsc{Max-$2$CSP}s; 
for \textsc{Max-$k$CSP}s the hidden constant may depend on $k$. See Appendix~\ref{sec:estimated-eps} for details. In contrast, the algorithm of \cite{dong2025learning} for \textsc{Max-Cut} requires a constant-factor estimate of $\eps$.}
We now state our main results.

\begin{theorem}\label{thm:generalMax-CSPs}
Let $\eps \in (0, \frac{1}{2}], \delta \in (0,1)$ and $k\in \mathbb{Z}_{>0}$. Let $f:\{-1,1\}^k\rightarrow\{0,1\}$ be a Boolean function of arity $k$. For any $\eta\in (0,1)$, given an $\eps$-accurate oracle $\mathcal{O}$ for a Boolean $\textsc{Max-}k\textsc{CSP}(f)$ instance $\Psi$, there exists a single-pass algorithm that outputs a $(1-\eta)$-approximation to the maximum value of $\Psi$ with probability at least $1-\delta$, using the following amount of space (in words):
\begin{enumerate}
    \item[(1)] if $k=2$, then the algorithm uses $\poly(1/\eps,1/\delta,1/\eta)$  space in insertion-only streams, and $\poly(1/\eps,1/\delta,1/\eta,\log n)$  space in dynamic streams; 
    \item[(2)] if $k\ge 3$, then the algorithm uses $\exp(\poly(k,1/\eps^k,1/\eta,\log(1/\delta)))$  space in insertion-only streams, and $\exp(\poly(k,1/\eps^k,1/\eta,\log(1/\delta))) \poly(\log n)$ space in dynamic streams.
\end{enumerate}
\end{theorem}

\textbf{Comparison to previous learning-augmented algorithms.} Our results generalize the learning-augmented streaming algorithm
of~\cite{dong2025learning} for \textsc{Max-Cut} to arbitrary Boolean
\mcsp{}s, while also improving the approximation--space tradeoff even for the
special case of \textsc{Max-Cut}. Specifically,
\cite{dong2025learning} obtains a
$\frac12+\Omega(\eps^2)$ approximation using $\poly(1/\eps)$ space in
insertion-only streams and $\poly(1/\eps,\log n)$ space in dynamic streams.
In contrast, our algorithm achieves a $(1-\eta)$-approximation for any
$\eta>0$, at the cost of an additional polynomial dependence on $1/\eta$ in 
space.

\textbf{Comparison to previous classical algorithms.}
For one-pass streaming algorithms in the \emph{dynamic} setting without predictions, strong lower bounds are known.
Chou et al.~\cite{chou2024sketching} showed that for every \textsc{Max-CSP}, there exists an explicit constant $\alpha$ such that, for any $\eta>0$:
(i) there is an $(\alpha-\eta)$-approximation streaming algorithm using $\poly(\log n)$ space, whereas
(ii) any $(\alpha+\eta)$-approximation requires $\Omega(\sqrt{n})$ space. 
For Boolean \textsc{Max-$2$CSP}s, Chou et al.~\cite{chou2020optimal} proved
analogous lower bounds even in the insertion-only setting. For example, for
\textsc{Max-Dicut} and \textsc{Max-$2$AND}, the threshold is
$\alpha=4/9$, so any $(4/9+\eta)$-approximation requires
$\Omega(\sqrt n)$ space. For higher arity, the barriers are even stronger:
\cite{chou2022linear} showed that, without predictions, no $o(n)$-space
streaming algorithm can achieve better than a $2^{-(k-1)}$ approximation for
\textsc{Max-$k$AND}. Similarly, \cite{chou2020optimal} showed that no
$o(\sqrt n)$-space algorithm for \textsc{Max-$k$SAT} can surpass a
$1/\sqrt{2}$ approximation. These lower bounds hold for every $k$ and even on
instances with bounded maximum degree.

Our learning-augmented algorithms bypass these worst-case barriers. With only
weakly accurate independent predictions, they achieve approximations arbitrarily
close to $1$ using space that is independent of $n$ in insertion-only streams
and polylogarithmic in $n$ in dynamic streams, without imposing any bounded-degree
assumption.
We provide a technical overview in \cref{sec:techniques}.

\subsection{Related work}\label{sec:relatedwork}

\textbf{Streaming CSPs.}
There has been a long line of work on streaming algorithms for constraint satisfaction problems (CSPs), including:
single-pass lower bounds for \mcut~\citep{kogan2015sketching,kapralov2014streaming,kapralov2017,kapralov2019optimal};
multi-pass lower bounds for \mcut~\citep{assadi2020multi,assadi2021graph,chen2023towards,fei2025multi} and other CSPs~\citep{fei2025dichotomy};
algorithms and lower bounds for approximating \textsc{Max-Dicut}~\citep{guruswami2017streaming,chou2020optimal,saxena2023streaming,saxena2023improved,saxena2025streaming,velusamy2025,azarmehr2025half};
quantum algorithms and lower bounds for \mcut~and \textsc{Max-Dicut}~\citep{kallaugher2022quantum,kallaugher2024exponential,kallaugher2025design};
and results on other specific CSPs, including Unique Games~\citep{guruswami2019streaming}, monarchy-like predicates~\citep{chou2022sketching}, and \textsc{Max-$k$AND}~\cite{singer2023oblivious}.
More broadly, dichotomy theorems and results for streaming and sketching general CSPs have been established in~\citep{chou2021approximability,chou2022linear,chou2024sketching}.

\textbf{Learning-augmented algorithms.}
Learning-augmented algorithms have also been actively studied across several areas, including online algorithms~\citep{mahdian2007allocating,lykouris2021competitive,kumar2018improving,LindermayrMS25}, data structures~\cite{KraskaBCDP18,mitzenmacher2018model,VaidyaKMK21}, graph algorithms~\cite{dinitz2021faster,ghoshal2025constraint,braverman2024learning}, and sublinear-time algorithms~\cite{indyk2019learning,EdenINRSW21,LiLLVW23}.
Our algorithms fall within the framework of learning-augmented streaming algorithms~\cite{hsu2019learning,jiang2020learning,aamand2019frequency,chen2022triangle,aamand2023improved,AamandCGSW25,Dong25CC,Nagawanshi26moment},
which focuses on using predictions to improve space–accuracy trade-offs. 
More recently, learning-augmented (offline) algorithms have been studied for \mcut~and CSPs~\citep{cohen2024max,ghoshal2025constraint,aamand2025improved,attias2025learning} under variants of the prediction model considered here. 
In particular, for the \mcut~problem where the goal is to output a solution rather than merely estimate its value, \cite{cohen2024max} and \cite{aamand2025improved} achieved approximation ratios exceeding $0.878$ using predictions.
Additionally, \cite{ghoshal2025constraint} developed algorithms that yield near-optimal solutions when the average degree exceeds a threshold depending only on the prediction bias, independent of the instance size; they further extended these results to weighted graphs and the \textsc{Max-3Lin} problem. Finally, \cite{attias2025learning} studied learning-augmented Boolean satisfiability, obtaining both exponential-time speedups for {\sc $k$-SAT} and provable approximation improvements for {\sc Max-SAT} variants, including near-optimal guarantees on instances of sufficiently large average degree.

\section{Preliminaries}

\textbf{Notation.\ } 
Let $[n]:= \{1,2,\dots,n\}$ for any $n\in \mathbb{Z}_{>0}$.
Throughout the paper we let $G=(V,E)$ be a directed, unweighted graph with $n$ vertices and $m$ edges. We assume that $G$ is loopless and has no repeated directed edges, while both orientations of a vertex pair may appear. For each vertex $v\in V$, let $\deg(v)$ denote the number of edges incident to $v$, counting both incoming and outgoing edges.

Given a vertex $v\in V$ and disjoint sets $S, T \subseteq V$, let $e(S,v)$ denote the number of edges from $S$ to $v$; that is, $e(S,v) := |\{(u,v)\in E: u\in S\}|$. Similarly, let $e(v,T)$ denote the number of edges from $v$ to $T$; that is, $e(v,T) := |\{(v,u)\in E: u\in T\}|$.
We define $e(S,T)$ to be the number of edges from $S$ to $T$; $e(S, T) := |\{(u,v)\in E: u\in S\wedge v\in T\}| = \sum_{v\in T} e(S, v) = \sum_{v\in S} e(v, T)$.
A \emph{directed cut} (or \emph{dicut}) $(S, T)$ of $G$ is an ordered bipartition of $V$ into two disjoint subsets $S$ and $T$. 
Following the standard convention, we refer to $S$ as the \emph{left} side and $T$ as the \emph{right} side.
The \emph{value} of the dicut $(S, T)$ is $e(S,T)$.
In the descriptions of our algorithms for \textsc{Max-Dicut}, for a vertex set $S \subseteq V$, we define $S^+$ (resp. $S^-$) as the set of vertices with predicted $+$ (resp. $-$) signs by the given $\eps$-accurate oracle $\mathcal{O}$. Note that $S^+$ and $S^-$ form a bipartition of $S$.

\subsection{Problem definitions} 
\begin{problem}[Boolean \textsc{Max-$k$CSP}]\label{problem:kcsp}
Let $k\in \mathbb{Z}_{>0}$. Let $f:\{-1,1\}^k \rightarrow \{0,1\}$ be a Boolean function of arity $k$. Let $x_1,\dots,x_n$ be Boolean \emph{variables}.
Given a Boolean variable $x$, we refer to $x$ and $-x$ as the \emph{literals} associated with $x$. In other words, a literal is either a variable itself or its negation.
A \emph{constraint} $C$ is a triple $(f,\mathbf{j},\mathbf{b})$, where $\mathbf{j}=(j_1,\dots,j_k)\in [n]^k$ has distinct coordinates and $\mathbf{b}=(b_1,\dots,b_k)\in \{-1,1\}^k$. It represents applying the function $f$ to the literals $b_1x_{j_1}, \dots, b_k x_{j_k}$.

An instance $\Psi$ of the Boolean $\textsc{Max-$k$CSP}(f)$ problem consists of
$m$ constraints $C_1,\dots,C_m$ over $n$ variables $x_1,x_2,\dots,x_n$.
For an assignment $\boldsymbol{\sigma}\in \{-1,1\}^n$, we say a constraint $C=(f,\mathbf{j},\mathbf{b})$ is \emph{satisfied} by $\boldsymbol{\sigma}$ if $f(b_1\sigma_{j_1},\dots,b_k\sigma_{j_k})=1$.
The \emph{value} of $\boldsymbol{\sigma}$ on 
$\Psi$, denoted $\mathrm{val}_{\Psi}(\boldsymbol{\sigma})$,
is the number of constraints satisfied by $\boldsymbol{\sigma}$.
The goal of the problem is to compute the maximum  value on $\Psi$ over all assignments, i.e., $\max_{\boldsymbol{\sigma}\in \{-1,1\}^n} \mathrm{val}_{\Psi}(\boldsymbol{\sigma})$.

Unless stated otherwise, throughout the paper \textsc{Max-$k$CSP} refers to Boolean \textsc{Max-$k$CSP}.
Since the function $f$ will be clear from the context, we will often omit $f$ and write a constraint as $(\mathbf{j},\mathbf{b})$.
We assume that all \textsc{Max-$k$CSP} instances are simple and unweighted: no variable appears more than once in a constraint, and no two constraints have the same function $f$ and the same set of literals.
For each variable $x_i$, we further define $\deg(i)$ as the number of constraints in which $x_i$ appears, either positively as $x_i$ or negatively as $-x_i$.

\end{problem}

\begin{problem}[\textsc{Max-$k$AND}]\label{problem:kand}
    \textsc{Max-$k$AND} is the special case of $\textsc{Max-$k$CSP}(f)$ with $f=\mathrm{AND}$.
    Specifically, $\mathrm{AND}(b_1\sigma_{j_1},\dots,b_k\sigma_{j_k})=1$ iff $b_1\sigma_{j_1} =\cdots =b_k\sigma_{j_k}=1$.
\end{problem}

\textbf{Streaming model.\ }
In this work, we consider the streaming model of computation, where the input instance arrives as a sequence that the algorithm scans in one pass, with the goal of outputting a solution for the problem at hand while minimizing space usage.
For \textsc{Max-CSP}s, an instance $\Psi$ is represented as a sequence of constraints, i.e., $\langle C_1,C_2,\dots \rangle$ where each $C_i$ is a constraint in $\Psi$, and the goal is to output an estimate of the maximum value of $\Psi$. 
When the sequence contains only constraint insertions, we refer to it as an \emph{insertion-only} stream. When it contains both constraint insertions and deletions, we refer to it as a \emph{dynamic} stream and represent it as $\langle (C_1,\Delta_1),(C_2,\Delta_2),\dots \rangle$, where for each $i$, $\Delta_i=1$ (resp. $-1$) denotes the  insertion (resp. deletion) of constraint $C_i$.

\textbf{Prediction model.\ }
In this work, we assume that the algorithm is equipped with an oracle $\mathcal{O}$ that provides $\eps$-accurate predictions $\mathbf{Y}\in \{\pm 1\}^n$, where $\eps \in (0,\frac{1}{2}]$. 
We can also assume that each arriving constraint in the stream is annotated with the predicted labels of its associated variables. Specifically, for each variable index $i \in [n]$, $ \Pr[Y_i = \sigma_i^*] = \frac{1}{2} + \eps $ and $ \Pr[Y_i = -\sigma_i^*] = \frac{1}{2} - \eps $, where $ \boldsymbol{\sigma}^* \in \{-1, 1\}^n $ is some fixed but unknown optimal assignment to the variables. We also assume that the predictions are independent.\footnote{In fact, our algorithms only require $4$-wise independence for \textsc{Max-$2$CSP}s.}

\textbf{Problem statement.\ }
In this paper, we consider the problem of estimating the maximum value of a Boolean \textsc{Max-CSP} instance, in the streaming model using $\eps$-accurate predictions.
We say an algorithm achieves an $\alpha$-approximation if it outputs a value $\mathrm{ALG}$ such that $\alpha \cdot \opt \le \mathrm{ALG}\le \opt$, where $\opt$ denotes the optimal (maximum) value.
We assume that the prediction accuracy parameter $\eps$ is known to the algorithm in advance.

\section{Technical overview}
\label{sec:techniques}
In this section, we provide a technical overview, primarily using \textsc{Max-Dicut} to illustrate our ideas. 
The complete algorithms and their analyses for \textsc{Max-Dicut} are provided in Appendix~\ref{sec:maxdicut}.
\begin{problem}[\textsc{Max-Dicut}]
    Given a directed graph $G=(V,E)$, the goal is to output the maximum value of all possible directed cuts, i.e., $\max_{S\subseteq V} e(S,V\setminus S)$.

    Note that \textsc{Max-Dicut} can be viewed as a special case of \textsc{Max-$2$AND}: the variables are indexed by the vertices, and each directed edge corresponds to a constraint with $\mathbf{b}=(1,-1)$ (see \Cref{problem:kand}). 
\end{problem}

Specifically, the input corresponds to a graph stream consisting of edge insertions (and deletions) of a directed graph $G$, arriving one by one in an arbitrary order. The goal is to estimate the value of the maximum directed cut of $G$, which maximizes the number of edges directed from the first set to the second across all bipartitions of the vertices. This problem naturally generalizes the maximum cut problem \textsc{Max-Cut} for undirected graphs.

At a high level, we view the oracle predictions, given as a vector $\mathbf{Y} \in \{-1,+1\}^n$, as a noisy encoding of an optimal cut $(S^*,T^*)$, represented by $\mathbf{x}^* \in \{\pm1\}^n$. We then design
a collection of \emph{linear measurements} of the stream that can be \emph{decoded} to recover the optimal value.

Before presenting our approach in detail, we note that a natural strategy, used in recent work on learning-augmented algorithms for streaming \mcut~\citep{dong2025learning}, is to follow the predictions, i.e., evaluate the cut induced by predictions.
Fix an optimal solution $\mathbf{x}^*\in\{\pm1\}^n$.
Since each endpoint may flip independently with probability $\frac12-\eps$, a cut edge under $\mathbf{x}^*$
may survive (i.e., still appears to be cut) under $\mathbf{Y}$ with probability at most
$\frac12+\poly(\eps)$, and therefore such arguments
inherently yield at best a $\frac12+\poly(\eps)$ approximation to the optimum value.

For \textsc{Max-Dicut}, the direct analogue of the survival-based approach is to evaluate the directed cut
$(S_\mathbf{Y},T_\mathbf{Y})$ induced by predictions, where $S_\mathbf{Y}=\{v:Y_v=+1\}$ and $T_\mathbf{Y}=\{v:Y_v=-1\}$.
However, an optimal edge $(u,v)$ with $u\in S^*$ and $v\in T^*$ contributes to $(S_\mathbf{Y},T_\mathbf{Y})$ only if
$Y_u=+1$ and $Y_v=-1$, which happens with probability $\tfrac14+O(\eps)$.
Hence this paradigm is inherently limited to a $(1/4+O(\eps))$-approximation for \textsc{Max-Dicut}.
By comparison, without predictions the best known one-pass $O(\log n)$-space algorithms achieve a $(4/9-\eta)$-approximation~\citep{chou2020optimal,guruswami2017streaming},
and with $\widetilde{O}(\sqrt n)$ space one can reach roughly a $0.485$-approximation on general graphs~\citep{saxena2023streaming,saxena2023improved}.
Very recently, a $(1/2-\eta)$-approximation in truly sublinear space was obtained for general graphs, matching the $1/2$ barrier~\citep{azarmehr2025half,kapralov2019optimal}.

Our approach departs from this paradigm: rather than estimating the fraction of optimal cut edges that ``survive'' the noise,
we use \emph{all edge outcomes induced by the noise} and explicitly \emph{invert the noise} to recover the hidden cut value.
Fix an optimal bipartition $(S^*,T^*)$ for $\textsc{Max-Dicut}$ and its corresponding labeling
$\mathbf{x}^*\in\{\pm1\}^n$.
The oracle returns $\mathbf{Y}\in\{\pm1\}^n$ such that for each $v\in V$,
\[
\Pr[Y_v=x_v^*]=\tfrac12+\eps,\qquad \Pr[Y_v=-x_v^*]=\tfrac12-\eps,
\]
independently over vertices.
Thus, $\mathbf{Y}$ can be seen as the output of a binary symmetric channel applied to the optimal labels.

For a directed edge $(u,v)$, the pair $(Y_u,Y_v)$ falls into one of three categories that are relevant to \textsc{Max-Dicut}: (1) predicted to go from left to right: $(Y_u,Y_v)=(+1,-1)$, (2) predicted to stay on the same side: $Y_u=Y_v$, (3) predicted to go from right to left: $(Y_u,Y_v)=(-1,+1)$.
Let $X_0,X_1,X_2$ denote the counts of edges in these three categories (restricted to the low-degree part of the graph,
as discussed below).
Crucially, the expectations $(\E[X_0],\E[X_1],\E[X_2])$ form a \emph{linear transformation} of the unknown triple
$(k_0,k_1,k_2)$, where:
$k_0$ is the number of edges going from $S^*$ to $T^*$,
$k_1$ is the number of edges staying inside $S^*$ or inside $T^*$,
and $k_2$ is the number of edges going from $T^*$ to $S^*$.
Since the noise model is i.i.d.\ across vertices, this linear map depends \emph{only} on $\eps$
and has full rank for any $\eps>0$.
Therefore, we can explicitly \emph{invert} this system, via Gaussian elimination, to obtain a linear estimator
\[
    \widetilde{k}_0 = \alpha_0(\eps) \cdot X_0 + \alpha_1(\eps) \cdot X_1 + \alpha_2(\eps)\cdot X_2,
\]
satisfying $\E[\widetilde{k}_0]=k_0$.
Informally, instead of using the predictions to estimate the number of survived edges of an actual cut,
we use the predictions to build an \emph{unbiased decoder} for the \emph{number of optimal cut edges}.
This is the core reason we can surpass the $1/2+\poly(\eps)$-approximation guarantee. 

\textbf{The main obstacle: dependence across edges and variance blow-up.}
While $\widetilde{k}_0$ is unbiased, it is not immediately concentrated:
even though $(Y_v)_{v\in V}$ are independent, the edge indicators are not independent
because edges share endpoints.
If a vertex participates in many edges, a single prediction bit $Y_v$ influences many summands,
and the variance of the edge-type counts can become too large.
To control this, we decompose the graph into low- and high-degree parts. We partition $V$ into low-degree and high-degree vertices: $L=\{v\in V:\deg(v)<\theta m\}$ and $H=V\setminus L$, for a ratio $\theta=\Theta(\poly(\eps,\eta)\cdot \delta)$ chosen in the analysis.
This decomposition has two key properties:

\begin{itemize}
    \item{{\em Low-degree stability:}} In the induced subgraph on $L$, each edge shares endpoints
    with only $O(\theta m)$ other edges. This bounds the covariance structure and yields
    $\Var[X_i]\le O(\theta m^2)$ for $i\in\{0,1,2\}$, which implies concentration of $X_0,X_1,X_2$
    around their expectations via Chebyshev's inequality.
    Plugging these bounds into the decoded estimator $\widetilde{k}_L$ gives $|\widetilde{k}_L-k_L|\le \poly(\eps,\eta)\cdot m$ with high probability, where $k_L$ denotes the number of optimal \textsc{Max-Dicut} edges completely inside $L$.

    \item{{\em Few high-degree vertices:}} Each $v\in H$ has degree at least $\theta m$,
    hence the number of such vertices is bounded by $|H|\le 2/\theta=\poly(1/\eps,1/\delta,1/\eta)$.
    This allows us to explicitly maintain per-vertex statistics for high-degree vertices.
\end{itemize}

\textbf{Decoding the contribution of high-degree vertices using 2-by-2 inversion.}
For every high-degree vertex $v\in H$ we maintain four counters describing its incident edges to $L$,
separated by the predicted label on the low-degree endpoint: $X_{L^+,v}, X_{L^-,v}, X_{v,L^+}, X_{v,L^-}$. These are the numbers of edges from $L^\pm$ to $v$ and from $v$ to $L^\pm$, respectively,
where $L^+=\{u\in L:Y_u=+1\}$ and $L^-=\{u\in L:Y_u=-1\}$.

Now fix a high-degree vertex $v$ and consider the unknown quantities $k_{S^*,v}, k_{T^*,v}, k_{v,S^*}, k_{v,T^*}$,
which count edges between $v$ and the \emph{true} partition $S^*\cap L$ vs.\ $T^*\cap L$.
The expectation of $(X_{L^+,v},X_{L^-,v})$ is a simple linear system in $(k_{S^*,v},k_{T^*,v})$
(and similarly for the outgoing direction).
Because the noise is independent across low-degree endpoints, this system can be inverted explicitly,
yielding unbiased estimators $\widetilde{k}_{S^*,v}$ and $\widetilde{k}_{v,T^*}$.
Moreover, each $X_{\cdot,v}$ is a sum of independent Bernoulli variables over low-degree endpoints, so by concentration and a union bound over $H$, with high probability every $v\in H$ satisfies $|\widetilde{k}_{S^*,v}-k_{S^*,v}| \le \poly(\eps,\eta) \cdot m$ and $|\widetilde{k}_{v,T^*}-k_{v,T^*}| \le \poly(\eps,\eta)\cdot m$.

\begin{figure}[t!]
    \centering
    \includegraphics[width=0.6\linewidth]{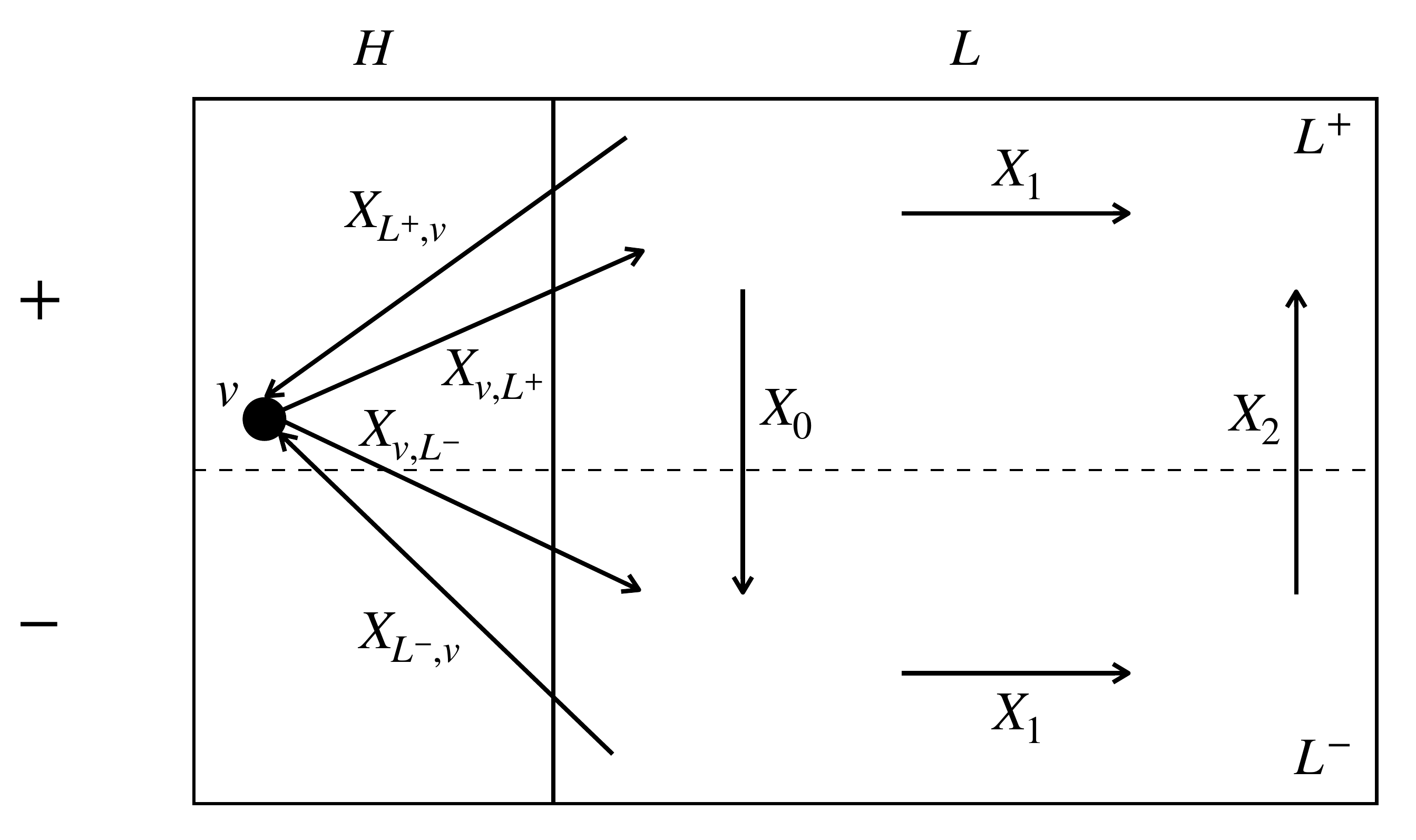}
    \caption{Illustration of the counters maintained in the algorithm for \textsc{Max-Dicut}. The rectangle represents the vertex set $V$, partitioned into high-degree vertices $H$ (left) and low-degree vertices $L$ (right). We also partition $V$ by the predicted labels (indicated by the horizontal dashed line): the upper region contains vertices predicted as $+1$ (denoted by ``$+$''), and the lower region contains vertices predicted as $-1$ (denoted by ``$-$'').}
    \label{fig:dicut-counters}
\end{figure}
\cref{fig:dicut-counters} illustrates the meaning of the maintained counters in the algorithm, including $X_0,X_1,X_2$ and the incidence counters $X_{L^+,v},X_{L^-,v},X_{v,L^+},X_{v,L^-}$ for each $v\in H$.

\textbf{Greedy assignment of high-degree vertices.}
Conditioned on the low-degree part behaving well, we may think of $(S^*\cap L,T^*\cap L)$ as fixed,
and decide independently for each $v\in H$ whether it should join the left side or the right side
in order to maximize the resulting \textsc{Max-Dicut} value contributed by edges between $v$ and $L$.
This contribution is exactly $\max\{k_{S^*,v},k_{v,T^*}\}$ for the true cut, and we approximate it by
$\max\{\widetilde{k}_{S^*,v},\widetilde{k}_{v,T^*}\}$.
Summing over $v\in H$ and adding the decoded low-degree estimate $\widetilde{k}_L$ yields our final estimator.

\textbf{``Non-realizable'' estimator.} A subtle but important point is that the resulting quantity need not correspond to the value of
\emph{any explicit cut computed from the predictions} (or any efficiently constructed cut at all);
it is best viewed as a \emph{decoded estimate} of the optimal value.
Nevertheless, the proof shows that it is a \emph{valid estimator}:
it is (i) close to $\opt$ up to $\poly(\eps,\eta)\cdot m$ additive error, and (ii) never exceeds $\opt$
once we account for the small estimation slack.

Finally, we must account for edges whose endpoints both lie in $H$.
Since $|H|$ is bounded by $\poly(1/\eps,1/\delta,1/\eta)$, the number of such edges becomes negligible compared to $\opt$
under a mild density assumption on $m$, and thus does not affect the $(1-o(1))$ approximation factor.

\textbf{Streaming implementation.} Our streaming implementation follows the same high-level techniques as the \textsc{Max-Cut} algorithm of~\citep{dong2025learning}. 
We implement the described algorithm in the streaming model by approximating the low-/high-degree decomposition and maintaining just enough statistics to run the same decoding at the end. 
Since we cannot afford to identify the high-degree set or store all degrees, we first use reservoir sampling (and $\ell_0$-sampling in the dynamic setting) to extract a small candidate set $\widetilde{H}$ that contains all truly high-degree vertices with high probability, and let $\widetilde{L}=V\setminus \widetilde{H}$. During the stream we maintain (i) a few exact global counters for the prediction-induced edge types and (ii) compact sketches that, for each $v\in \widetilde{H}$, estimate its incidence to the predicted sides. In post-processing, we subtract the estimated contribution of $\widetilde{H}$ from the global counts to recover the corresponding statistics on $\widetilde{L}$, and then apply the same approach of solving linear systems as in the offline algorithm for the $\widetilde{L}$ subgraph to obtain a near-optimal estimation of the \textsc{Max-Dicut} value.

\textbf{Extending to \textsc{Max-$k$AND}.}
For \textsc{Max-$2$AND}, our algorithm closely parallels that for \textsc{Max-Dicut}. A \textsc{Max-$2$AND} instance consists of constraints, which generalize directed edges, and each constraint contains two literals, which play a role analogous to edge directions.
For $k\ge 3$, 
the main challenge is choosing the high-degree set itself. It must be small enough to enumerate, but after enumeration the remaining oracle noise must concentrate uniformly for every high assignment. We grow $H$ with the moving threshold $\theta m/(|H|+1)$ where $\theta$ is a parameter constant dependent on $\eta, \delta$ and $\eps$: the harmonic-series bound keeps $|H|$ bounded in terms of $k,\eps,\eta,\delta$, while the same $(|H|+1)^{-1}$ margin makes every low variable's influence small enough to union-bound over all $2^{|H|}$ decoded values.

\textbf{Extending to general Boolean \textsc{Max-CSP}s.}
There is a standard black-box reduction that transforms any instance of Boolean \textsc{Max-$k$CSP} into an instance of \textsc{Max-$k$AND}, while preserving the optimum value and increasing the number of constraints by at most a factor of $2^k$. Applying this reduction allows us to immediately extend our results to Boolean \textsc{Max-CSP}s.

\section{Our algorithms for \textsc{Max-$2$AND}}\label{sec:max-dicut-2and}

In this section, we present our learning-augmented streaming algorithms for \textsc{Max-$2$AND}. To illustrate the main algorithmic ideas, we first describe an algorithm in the \emph{offline} setting, where the entire input instance is available, in \cref{sec:max-2and-main-text}. 
We then present its streaming implementations in insertion-only and dynamic streams in Sections~\ref{sec:insertion-2and} and~\ref{sec:dynamic-2and}, respectively.

\subsection{Offline algorithm}
\label{sec:max-2and-main-text}
\paragraph{Overview.}
We partition the variables into low-degree variables $L$ and high-degree variables $H$, where $i\in H$ iff $\deg(i)\ge \theta m$ for $\theta:=\Theta(\eps^{4}\delta\eta^2)$.
We then scan the constraints, query the oracle $\mathcal{O}$ on the variables involved, and update a small set of counters. 
These counters are analogous to those in the \textsc{Max-Dicut} case, as described in \cref{sec:techniques}, except that they count the corresponding constraints rather than edges. From these counters, we subsequently recover an estimate of the \textsc{Max-$2$AND} value.
The algorithm is formally given in  Algorithm~\ref{alg:max-2and-offline}.

Specifically, for each constraint whose associated variables are both in $L$, we maintain three counters: (i) $X_0$: the number of such constraints for which both literals are satisfied by the predicted assignment, (ii) $X_1$: the number for which exactly one literal is satisfied, and (iii) $X_2$: the number for which neither literal is satisfied. These statistics later let us decode an accurate estimate of the optimal contribution from constraints whose variables both lie in $L$ (Line~\ref{line:tildekL-2and}).

Furthermore, for each high-degree variable $i\in H$, we maintain four counters: (i) $X_{i,L(+)}^{\pos}$: the number of constraints in which $i$ appears positively and the low-degree literal is satisfied by the predicted assignment; (ii) $X_{i,L(-)}^{\pos}$: $i$ appears positively and the low-degree literal is not satisfied; (iii) $X_{i,L(+)}^{\neg}$: $i$ appears negatively and the low-degree literal is satisfied; and (iv) $X_{i,L(-)}^{\neg}$: $i$ appears negatively and the low-degree literal is not satisfied.
Using these counts, we decode accurate estimates of the corresponding quantities with respect to the unknown optimal assignment $\boldsymbol{\sigma}^*$ (Lines \ref{line:num-true-constraints} and \ref{line:tilde-i-L-pos-neg}).
We then assign each $i\in H$ to maximize its estimated contribution to the value, and combine this with the estimate from constraints whose variables both lie in $L$ to obtain our final \textsc{Max-$2$AND} value estimate (Line~\ref{line:return-max-2and}).

Next, we analyze the approximation guarantee of \cref{alg:max-2and-offline}.

\begin{algorithm}[t!]
\begin{algorithmic}[1]
\Require $\textsc{Max-$2$AND}$ instance $\Psi=(C_1,\dots,C_m)$, an $\eps$-accurate oracle $\mathcal{O} \rightarrow \{-1, 1\}$, $\delta,\eta \in (0,1)$
\Ensure Estimate of the maximum value of $\Psi$

\State $\theta \leftarrow \frac{\eps^{4}\delta\eta^2}{768}$, $L\leftarrow \{i\in [n]: \deg(i) < \theta m\}, H\leftarrow [n] \setminus L$
\State {\bf initialize} $X_0\leftarrow 0, X_1 \leftarrow 0, X_2 \leftarrow 0$.
\State {\bf initialize} $X_{i,L(+)}^{\pos} \leftarrow 0, X_{i,L(-)}^{\pos} \leftarrow 0, X_{i,L(+)}^{\neg} \leftarrow 0, X_{i,L(-)}^{\neg}\leftarrow 0$ for each $i\in H$.
\ForEach{constraint $((j_1,j_2),(b_1,b_2))\in \Psi$}
    \State $Y_{j_1} = \mathcal{O}(j_1), Y_{j_2} = \mathcal{O}(j_2)$
    \If{$j_1\in L$ and $j_2 \in L$}
        \If{$b_1 Y_{j_1} = b_2 Y_{j_2} = 1$}
            $X_0 \leftarrow X_0 +1$
        \EndIf
        \If{$b_1 Y_{j_1} \neq b_2 Y_{j_2}$}
            $X_1 \leftarrow X_1 +1$
        \EndIf
        \If{$b_1 Y_{j_1} = b_2 Y_{j_2} = -1$}
            $X_2 \leftarrow X_2 +1$
        \EndIf
    \EndIf
    
    \If{$j_1\in H$ and $j_2 \in L$}
        \If{$b_1 = 1$}
            \State \textbf{if} $b_2 Y_{j_2} = 1$ \textbf{then} $X_{j_1,L(+)}^{\pos} \leftarrow X_{j_1,L(+)}^{\pos} + 1$ \textbf{else} $X_{j_1,L(-)}^{\pos} \leftarrow X_{j_1,L(-)}^{\pos} + 1$      
        \Else{}
            \State \textbf{if} $b_2 Y_{j_2} = 1$ \textbf{then} $X_{j_1,L(+)}^{\neg} \leftarrow X_{j_1,L(+)}^{\neg} + 1$ \textbf{else} $X_{j_1,L(-)}^{\neg} \leftarrow X_{j_1,L(-)}^{\neg} + 1$ 
        \EndIf
    \EndIf
    
    \If{$j_1\in L$ and $j_2 \in H$}
    
        \If{$b_2 = 1$}
            \State \textbf{if} $b_1 Y_{j_1} = 1$ \textbf{then} $X_{j_2,L(+)}^{\pos} \leftarrow X_{j_2,L(+)}^{\pos} + 1$ \textbf{else} $X_{j_2,L(-)}^{\pos} \leftarrow X_{j_2,L(-)}^{\pos} + 1$      
        \Else{}
            \State \textbf{if} $b_1 Y_{j_1} = 1$ \textbf{then} $X_{j_2,L(+)}^{\neg} \leftarrow X_{j_2,L(+)}^{\neg} + 1$ \textbf{else} $X_{j_2,L(-)}^{\neg} \leftarrow X_{j_2,L(-)}^{\neg} + 1$ 
        \EndIf
    \EndIf
\EndFor
\State $\widetilde{k}_{L} \leftarrow \frac{1}{4\eps^2}[(\frac{1}{2}+\eps)^2 X_0 - (\frac{1}{2}+\eps)(\frac{1}{2}-\eps)X_1 + (\frac{1}{2}-\eps)^2 X_2]$ \phantomsection \label{line:tildekL-2and}

\ForEach{$i\in H$} 
\State $\widetilde{k}_{i,L(+)}^{\pos} \leftarrow (\frac{1}{2}+\frac{1}{4\eps})X_{i,L(+)}^{\pos} + (\frac{1}{2}-\frac{1}{4\eps})X_{i,L(-)}^{\pos}$
\phantomsection \label{line:num-true-constraints}
\State $\widetilde{k}_{i,L(+)}^{\neg} \leftarrow (\frac{1}{2}+\frac{1}{4\eps})X_{i,L(+)}^{\neg} + (\frac{1}{2}-\frac{1}{4\eps})X_{i,L(-)}^{\neg}$  \phantomsection \label{line:tilde-i-L-pos-neg}

\EndFor
\State \textbf{return} $\widetilde{k}_{L} + \sum_{i\in H}\max\{\widetilde{k}_{i,L(+)}^{\pos}, \widetilde{k}_{i,L(+)}^{\neg}\}  - \frac{\eta m}{16}$ 
\phantomsection \label{line:return-max-2and}
\end{algorithmic}
\caption{A learning-augmented \emph{offline} algorithm for estimating the \textsc{Max-$2$AND} value}
\label{alg:max-2and-offline}
\end{algorithm}

\begin{theorem}\label{thm:offlinemax2and}
    Let $\eps \in (0,\frac{1}{2}], \delta,\eta \in (0,1)$, and $\Psi$ be a  \textsc{Max-$2$AND} instance with $m=\Omega(\eps^{-26}\delta^{-10}\eta^{-14})$. Given an $\eps$-accurate oracle $\mathcal{O}$ for $\Psi$,
    \cref{alg:max-2and-offline} provides a $(1-\eta)$-approximation for the maximum value of $\Psi$ with probability at least $1-\delta$.
    \label{thm:max-2and-offline}
\end{theorem}

The proof of Theorem~\ref{thm:offlinemax2and} relies on the following key properties of the  counters maintained by \cref{alg:max-2and-offline}, which are formalized in the lemmas below. 
\subsubsection{Key properties}

Let $X_0,X_1$, and $X_2$ denote the random variables given by the final values of the corresponding counters at the end of the algorithm.
We first prove concentration bounds for these random variables around their expectations.
\begin{lemma}
\label{lem:concentration-low-max-2and}
Let $\eps \in (0,\frac{1}{2}]$ and $\delta,\eta \in (0,1)$.
For each $i\in \{0,1,2\}$, we have $|X_i - \E[X_i]| \le \frac{\eps^{2}\eta m}{8}$ with probability at least $1-\frac{\delta}{6}$.
\end{lemma}
\begin{proof}
    Let $\Psi_L:=\{((j_1,j_2),(b_1,b_2))\in \Psi: j_1,j_2\in L\}$ denote the set of constraints whose associated variables are both in $L$.
    Let $m_L:=|\Psi_L|$.
    For each constraint $((j_1,j_2),(b_1,b_2))\in \Psi_L$, define the following indicator random variables:
    \begin{align*}
    X_{0,j_1 j_2} &=
    \begin{cases}
    1 & \text{if } b_1 Y_{j_1}=b_2 Y_{j_2}=1,\\
    0 & \text{otherwise},
    \end{cases}
    &
    X_{1,j_1 j_2} &=
    \begin{cases}
    1 & \text{if } b_1 Y_{j_1} \neq b_2 Y_{j_2},\\
    0 & \text{otherwise},
    \end{cases}
    &
    X_{2,j_1 j_2} &=
    \begin{cases}
    1 & \text{if } b_1 Y_{j_1}=b_2 Y_{j_2}=-1,\\
    0 & \text{otherwise}.
    \end{cases}
    \end{align*}
    
    Note that for $i\in \{0,1,2\}$, we have $\E[X_{i,j_1 j_2}] = \Pr[X_{i,j_1 j_2} = 1]$ and $X_i=\sum_{((j_1,j_2),(b_1,b_2))\in \Psi_L}X_{i,j_1 j_2}$. Therefore,
    \begin{align*}
        \E[X_i^2] &= \E[ (\sum_{((j_1,j_2),(b_1,b_2))\in \Psi_L} X_{i,j_1j_2})^2]\\
        &= \E[\sum_{((j_1,j_2),(b_1,b_2))\in \Psi_L}\sum_{((j'_1,j'_2),(b'_1,b'_2))\in \Psi_L}X_{i,j_1 j_2} X_{i,j'_1 j'_2}]\\
        &=\sum_{((j_1,j_2),(b_1,b_2))\in \Psi_L}\sum_{((j'_1,j'_2),(b'_1,b'_2))\in \Psi_L}\E[X_{i,j_1 j_2}X_{i,j'_1 j'_2}]\\ 
        &=\sum_{((j_1,j_2),(b_1,b_2))\in \Psi_L}\sum_{((j'_1,j'_2),(b'_1,b'_2))\in \Psi_L} \Pr[X_{i,j_1 j_2} = 1\wedge X_{i,j'_1 j'_2} = 1]\\
        &= \sum_{\substack{((j_1,j_2),(b_1,b_2))\in \Psi_L,\\((j'_1,j'_2),(b'_1,b'_2))\in \Psi_L,\\ \{j_1,j_2\}\cap \{j'_1,j'_2\} = \emptyset}} \Pr[X_{i,j_1 j_2} = 1\wedge X_{i,j'_1 j'_2} = 1] + \sum_{\substack{((j_1,j_2),(b_1,b_2))\in \Psi_L,\\((j'_1,j'_2),(b'_1,b'_2))\in \Psi_L,\\ \{j_1,j_2\}\cap \{j'_1,j'_2\} \neq \emptyset}} \Pr[X_{i,j_1 j_2} = 1\wedge X_{i,j'_1 j'_2} = 1]\\
        &\leq \sum_{\substack{((j_1,j_2),(b_1,b_2))\in \Psi_L,\\((j'_1,j'_2),(b'_1,b'_2))\in \Psi_L,\\ \{j_1,j_2\}\cap \{j'_1,j'_2\} = \emptyset}} \Pr[X_{i,j_1 j_2} = 1] \Pr[X_{i,j'_1 j'_2} = 1] + \sum_{\substack{((j_1,j_2),(b_1,b_2))\in \Psi_L,\\((j'_1,j'_2),(b'_1,b'_2))\in \Psi_L,\\ \{j_1,j_2\}\cap \{j'_1,j'_2\} \neq \emptyset}} 1\\
        &\leq \sum_{((j_1,j_2),(b_1,b_2))\in \Psi_L}\sum_{((j'_1,j'_2),(b'_1,b'_2))\in \Psi_L}\E[X_{i,j_1 j_2}]\E[X_{i,j'_1 j'_2}] + m_L\cdot 2\cdot \theta m \\
        &\leq (\E[\sum_{((j_1,j_2),(b_1,b_2))\in \Psi_L}X_{i,j_1 j_2}])^2 + 2\theta m^2\\
        &= (\E[X_i])^2 + 2\theta m^2,
    \end{align*}
    where the third-to-last step follows from $\deg(i) < \theta m$ for $i\in L$, and the second-to-last step follows from $m_L \le m$.
    
    Then we have $\Var[X_i] = \E[X_i^2] - (\E[X_i])^2 \le 2\theta m^2$.
    By Chebyshev's inequality, 
    \begin{align*}
        \Pr[|X_i - \E[X_i]| > \frac{\eps^{2}\eta  m}{8}] 
        \le \frac{\Var[X_i]}{(\eps^{2}\eta m/8)^2}
        &\le \frac{128\theta m^2}{\eps^{4}\eta^2 m^2} &&\rhd \Var[X_i] \le 2\theta m^2\\
        &= \frac{128\theta}{\eps^{4}\eta^2}
        < \frac{\delta}{6}.  &&\rhd \theta < \frac{\eps^{4} \delta\eta^2}{768}
    \end{align*}
\end{proof}

Recall that $\boldsymbol{\sigma}^*$ denotes an optimal assignment for the \textsc{Max-$2$AND} instance $\Psi$.
Define 
    $k_{L}:= |\{((j_1,j_2),(b_1,b_2))\in \Psi: j_1, j_2\in L,\ b_1 \sigma^*_{j_1} = b_2 \sigma^*_{j_2}=1\}|$
to be the number of constraints in $\Psi$ whose associated variables are both in $L$ (i.e., both low-degree) that are satisfied by $\boldsymbol{\sigma}^*$.
Recall that $\widetilde{k}_L$, computed at Line \ref{line:tildekL-2and} of \Cref{alg:max-2and-offline}, is an estimate of $k_{L}$.
The following lemma shows that $\widetilde{k}_L$ is close to $k_L$ up to an additive error of $\poly(\eta)\cdot m$ with high probability.
\begin{lemma}
\label{lem:concentration-k-L-max-2and}
Let $\delta,\eta \in (0,1)$.
We have $|\widetilde{k}_L - k_L| \le \frac{\eta m}{32}$ with probability at least $1-\frac{\delta}{2}$.
\end{lemma}
\begin{proof}
    Let $k_0:=|\{((j_1,j_2),(b_1,b_2))\in \Psi: j_1, j_2\in L,\ b_1 \sigma^*_{j_1} = b_2 \sigma^*_{j_2}=1\}|$ denote the number of constraints in $\Psi$ whose associated variables are both in $L$ and for which both literals are satisfied by $\boldsymbol{\sigma}^*$. Note that $k_0 = k_L$.
    
    Let $k_1:=|\{((j_1,j_2),(b_1,b_2))\in \Psi: j_1, j_2\in L,\ b_1 \sigma^*_{j_1} \neq b_2 \sigma^*_{j_2}\}|$ denote the number of constraints in $\Psi$ whose associated variables are both in $L$ and for which exactly one literal is satisfied by $\boldsymbol{\sigma}^*$.

    Let $k_2:=|\{((j_1,j_2),(b_1,b_2))\in \Psi: j_1, j_2\in L,\ b_1 \sigma^*_{j_1} = b_2 \sigma^*_{j_2}=-1\}|$ denote the number of constraints in $\Psi$ whose associated variables are both in $L$ and for which neither literal is satisfied by $\boldsymbol{\sigma}^*$.
    
According to \cref{alg:max-2and-offline}, we have
\begin{align*}
    \E[X_0] &= \left(\frac{1}{2}+\eps \right)^2 \cdot k_0 + \left(\frac{1}{2}+\eps \right)\left(\frac{1}{2}-\eps \right)\cdot k_1 + \left(\frac{1}{2}-\eps \right)^2 \cdot k_2,\\
    \E[X_1] &= 2\left(\frac{1}{2}+\eps \right)\left(\frac{1}{2}-\eps \right)\cdot k_0 + \left[ \left(\frac{1}{2}+\eps \right)^2 + \left(\frac{1}{2}-\eps \right)^2 \right]\cdot k_1 + 2\left(\frac{1}{2}+\eps \right)\left(\frac{1}{2}-\eps \right)\cdot k_2,\\
    \E[X_2] &= \left(\frac{1}{2}-\eps \right)^2 \cdot k_0 + \left(\frac{1}{2}+\eps \right)\left(\frac{1}{2}-\eps \right)\cdot k_1 + \left(\frac{1}{2}+\eps \right)^2 \cdot k_2.
\end{align*}
Therefore, by Gaussian elimination, we have 
\begin{align}
    k_L = k_0 = \frac{1}{4\eps^2}\left[\left(\frac{1}{2}+\eps\right)^2 \cdot \E[X_0] - \left(\frac{1}{2}+\eps\right)\left(\frac{1}{2}-\eps\right)\cdot \E[X_1] + \left(\frac{1}{2}-\eps\right)^2 \cdot \E[X_2]\right]
\label{eq:gaussian-opt-max-2and}
\end{align}
By applying union bound on Lemma~\ref{lem:concentration-low-max-2and}, we have that with probability at least $1-\frac{\delta}{6}\cdot 3 = 1-\frac{\delta}{2}$, $|X_i - \E[X_i]| \le \frac{\eps^{2}\eta m}{8}$ for all $i\in \{0,1,2\}$. Therefore,
\begin{align*}
    \widetilde{k}_L =~&\frac{1}{4\eps^2}\left[\left(\frac{1}{2}+\eps\right)^2 X_0 - \left(\frac{1}{2}+\eps\right)\left(\frac{1}{2}-\eps\right)X_1 + \left(\frac{1}{2}-\eps\right)^2 X_2\right]\\
    \le~&\frac{1}{4\eps^2}\left[\left(\frac{1}{2}+\eps\right)^2 \left(\E[X_0]+\frac{\eps^{2}\eta m}{8}\right) - \left(\frac{1}{2}+\eps\right)\left(\frac{1}{2}-\eps\right)\left(\E[X_1]-\frac{\eps^{2}\eta m}{8}\right)  
    + \left(\frac{1}{2}-\eps\right)^2 \left(\E[X_2]+\frac{\eps^{2}\eta m}{8}\right)\right] \\
    =~&k_L +  \frac{\eta m}{32}\left[ \left(\frac{1}{2}+\eps\right)^2 + \left(\frac{1}{2}+\eps\right)\left(\frac{1}{2}-\eps\right) + \left(\frac{1}{2}-\eps\right)^2 \right]\\
    =~&k_L + \frac{\eta m}{32}\left(\frac{3}{4}+\eps^2\right)\\
    \le~&k_L+ \frac{\eta m}{32},
\end{align*}
where the third-to-last step follows from \cref{eq:gaussian-opt-max-2and}.
Similarly, $\widetilde{k}_L \ge k_L - \frac{\eta m}{32}$.
\end{proof}

For each $i\in H$, define 
\begin{enumerate}
    \item $k_{i,L(+)}^{\pos}:= |\{((i,j),(b_i,b_j))\in \Psi: j\in L, b_i=1, b_j\sigma_j^* = 1\}|$ to be the number of constraints involving $x_i$ and some low-degree variable $x_j$ such that the literal of $x_i$ appears positively (i.e., $b_i=1$) and the low-side literal is satisfied by $\boldsymbol{\sigma}^*$ (i.e., $b_j\sigma_j^*=1$). 
    We analogously define:
    \item $k_{i,L(-)}^{\pos}:= |\{((i,j),(b_i,b_j))\in \Psi: j\in L, b_i=1, b_j\sigma_j^* = -1\}|$.
    \item $k_{i,L(+)}^{\neg}:= |\{((i,j),(b_i,b_j))\in \Psi: j\in L, b_i=-1, b_j\sigma_j^* = 1\}|$.
    \item $k_{i,L(-)}^{\neg}:= |\{((i,j),(b_i,b_j))\in \Psi: j\in L, b_i=-1, b_j\sigma_j^* = -1\}|$.
\end{enumerate}

Recall that $\widetilde{k}_{i,L(+)}^{\pos}$ and $\widetilde{k}_{i,L(+)}^{\neg}$, computed at Lines \ref{line:num-true-constraints} and \ref{line:tilde-i-L-pos-neg} of \Cref{alg:max-2and-offline}, are estimates of $k_{i,L(+)}^{\pos}$ and $k_{i,L(+)}^{\neg}$, respectively.
The following lemma shows that $\widetilde{k}_{i,L(+)}^{\pos}$ and $\widetilde{k}_{i,L(+)}^{\neg}$ are close to $k_{i,L(+)}^{\pos}$ and $k_{i,L(+)}^{\neg}$, respectively, up to an additive error of $\poly(\eps,\delta,\eta)\cdot m$ with high probability.

\begin{lemma}
\label{lem:concentration-general-max-2and}
Let $\eps \in (0,\frac{1}{2}]$ and $\delta,\eta \in (0,1)$. For each  $i\in H$,
each of the following  holds with probability at least $1-\frac{\eps^{4} \delta^2\eta^2}{6144}$:
\begin{enumerate}
    \item[(1)] $|\widetilde{k}_{i,L(+)}^{\pos} - k_{i,L(+)}^{\pos}| \le \frac{\eps^{4}\delta\eta^3 m}{49152}$;
    \item[(2)] $|\widetilde{k}_{i,L(+)}^{\neg} - k_{i,L(+)}^{\neg}| \le\frac{\eps^{4}\delta\eta^3 m}{49152}$.
\end{enumerate}
\end{lemma}
\begin{proof}
    We prove the first inequality below; the second inequality follows by the same argument.

    For each $i\in H$, let $X_{i,L(+)}^{\pos},X_{i,L(-)}^{\pos},X_{i,L(+)}^{\neg}$, and $X_{i,L(-)}^{\neg}$ denote the random variables given by the final values of the corresponding counters at the end of the algorithm. We have
\begin{align*}
    \E[X_{i,L(+)}^{\pos}] &= \left(\frac{1}{2}+\eps \right)\cdot k_{i,L(+)}^{\pos} + \left(\frac{1}{2}-\eps \right) \cdot k_{i,L(-)}^{\pos},\\
    \E[X_{i,L(-)}^{\pos}] &= \left(\frac{1}{2}-\eps \right)\cdot k_{i,L(+)}^{\pos} + \left(\frac{1}{2}+\eps \right) \cdot k_{i,L(-)}^{\pos},\\
    \E[X_{i,L(+)}^{\neg}] &= \left(\frac{1}{2}+\eps \right)\cdot k_{i,L(+)}^{\neg} + \left(\frac{1}{2}-\eps \right) \cdot k_{i,L(-)}^{\neg},\\
    \E[X_{i,L(-)}^{\neg}] &= \left(\frac{1}{2}-\eps \right)\cdot k_{i,L(+)}^{\neg} + \left(\frac{1}{2}+\eps \right) \cdot k_{i,L(-)}^{\neg}.
\end{align*}

By solving the above system of linear equations, we obtain
\begin{align*}
    k_{i,L(+)}^{\pos} &= \left( \frac{1}{2}+\frac{1}{4\eps}\right)\cdot \E[X_{i,L(+)}^{\pos}] + \left( \frac{1}{2}-\frac{1}{4\eps}\right)\cdot \E[X_{i,L(-)}^{\pos}], \\
    k_{i,L(+)}^{\neg} &=\left( \frac{1}{2}+\frac{1}{4\eps}\right)\cdot \E[X_{i,L(+)}^{\neg}] + \left( \frac{1}{2}-\frac{1}{4\eps}\right)\cdot \E[X_{i,L(-)}^{\neg}].
\end{align*}

According to Line~\ref{line:num-true-constraints} of \cref{alg:max-2and-offline} and by linearity of expectation, we have
    \begin{align*}
        \E[\widetilde{k}_{i,L(+)}^{\pos}] &= \E\left[\left(\frac{1}{2}+\frac{1}{4\eps}\right)X_{i,L(+)}^{\pos} + \left(\frac{1}{2}-\frac{1}{4\eps}\right)X_{i,L(-)}^{\pos}\right]\\
        &= \left(\frac{1}{2}+\frac{1}{4\eps}\right)\E[X_{i,L(+)}^{\pos}] + \left(\frac{1}{2}-\frac{1}{4\eps}\right)\E[X_{i,L(-)}^{\pos}]\\
        &= {k}_{i,L(+)}^{\pos}.
    \end{align*}

    Next, we calculate $\Var[\widetilde{k}_{i,L(+)}^{\pos}]$.
    Note that 
    \begin{align*}
        \widetilde{k}_{i,L(+)}^{\pos} &= \left(\frac{1}{2}+\frac{1}{4\eps}\right)X_{i,L(+)}^{\pos} + \left(\frac{1}{2}-\frac{1}{4\eps}\right)X_{i,L(-)}^{\pos} \\
        &= \left(\frac{1}{2}+\frac{1}{4\eps}\right)X_{i,L(+)}^{\pos} + \left(\frac{1}{2}-\frac{1}{4\eps}\right)(c_{i,L}^{\pos} - X_{i,L(+)}^{\pos}) \\
        &= \frac{1}{2\eps}X_{i,L(+)}^{\pos} + \left(\frac{1}{2}-\frac{1}{4\eps}\right)c_{i,L}^{\pos},
    \end{align*}
    where $c_{i,L}^{\pos}:= |\{((i,j),(b_i,b_j))\in \Psi: j\in L, b_i=1\}|$ denotes the number of constraints involving $x_i$ and some low-degree variable $x_j$ such that the literal of $x_i$ appears positively.
    Therefore, $\Var[\widetilde{k}_{i,L(+)}^{\pos}]=\frac{1}{4\eps^2}\Var[X_{i,L(+)}^{\pos}]$.

    Given $i\in H$, for each constraint $((i,j),(b_i,b_j))\in \Psi$ with $j\in L$ and $b_i=1$, define an indicator random variable $X_{ij}$  as follows:
    \begin{align*}
        X_{ij} =
    \begin{cases}
    1 & \text{if } b_j Y_j=1,\\
    0 & \text{otherwise}.
    \end{cases}
    \end{align*}
    Note that $X_{i,L(+)}^{\pos}=\sum_{((i,j),(b_i,b_j))\in \Psi: j\in L,b_i=1}X_{ij}$ and $\E[X_{ij}] = \Pr[X_{ij} = 1]$.
    Then $\Var[X_{ij}] = \E[X^2_{ij}] - (\E[X_{ij}])^2 = \E[X_{ij}]\cdot (1-\E[X_{ij}]) =(\frac{1}{2}+\eps)(\frac{1}{2}-\eps)=\frac{1}{4}-\eps^2 \le \frac{1}{4}$. 
    Since $X_{ij}$'s are negatively correlated, we have
    \begin{align*}
        \Var[X_{i,L(+)}^{\pos}] \le \sum_{\substack{((i,j),(b_i,b_j))\in \Psi:\\j\in L,b_i=1}}\Var[X_{ij}]\le \frac{m}{4}.
    \end{align*}

    So, $\Var[\widetilde{k}_{i,L(+)}^{\pos}]=\frac{1}{4\eps^2}\Var[X_{i,L(+)}^{\pos}] \le \frac{m}{16\eps^2}$.
    By Chebyshev's inequality,
    \begin{align*}
        \Pr[|\widetilde{k}_{i,L(+)}^{\pos}-k_{i,L(+)}^{\pos}|
        \ge \frac{\eps^{4}\delta\eta^3 m}{49152}] &= \Pr[|\widetilde{k}_{i,L(+)}^{\pos}-\E[\widetilde{k}_{i,L(+)}^{\pos}]|\ge \frac{\eps^{4}\delta\eta^3 m}{49152}] \\
        &\le \frac{\Var[\widetilde{k}_{i,L(+)}^{\pos}]}{(\eps^{4}\delta\eta^3 m/49152)^2}
        \le \frac{49152^2/16}{\eps^{10}\delta^2\eta^6 m} \le \frac{\eps^{4} \delta^2\eta^2}{6144},
    \end{align*}
    where the last step follows from our assumption that $m =\Omega(\eps^{-26}\delta^{-10}\eta^{-14})$.
\end{proof}

Define 
    $k_{H}:= |\{((j_1,j_2),(b_1,b_2))\in \Psi: j_1, j_2\in H,\ b_1 \sigma^*_{j_1} = b_2 \sigma^*_{j_2}=1\}|$
to be the number of constraints in $\Psi$ whose associated variables are both in $H$ (i.e., both high-degree) that are satisfied by $\boldsymbol{\sigma}^*$.
Note that $|H| \le \frac{2m}{\theta m}=\frac{2}{\theta} = \frac{1536}{\eps^{4} \delta\eta^2}$. Thus we have the following fact.
\begin{fact}
\label{fact:size-H-2and}
It holds that
\(
k_H < |H|^2
\le \frac{1536^2}{\eps^{8}\delta^2\eta^4}.
\)
\end{fact}

\subsubsection{Proof of \cref{thm:max-2and-offline}} 
\begin{proof}
Recall that $k_L$ (resp. $k_H$) is the number of constraints in $\Psi$ whose associated variables are both in $L$ (resp. $H$) that are satisfied by $\boldsymbol{\sigma}^*$.

Furthermore, define 
\begin{align*}
    k_{L,H}:=|\{((j_1,j_2),(b_1,b_2))\in \Psi: |\{j_1,j_2\}\cap L| =|\{j_1,j_2\}\cap H|=1, b_1\sigma^*_{j_1} = b_2\sigma^*_{j_2} = 1\}|
\end{align*}
to be the number of constraints in $\Psi$ whose associated variables are one in $L$ and the other in $H$ that are satisfied by $\boldsymbol{\sigma}^*$.

Therefore, we have $\opt = k_L+k_H+ k_{L,H}$. Let $\mathrm{ALG}$ denote the output of \cref{alg:max-2and-offline}.
By union bound, with probability at least $1-\frac{\delta}{2}-\frac{\eps^{4}\delta^2\eta^2}{6144}\cdot 2 \cdot |H| \ge 1-\delta$, we have
\begin{align*}
        \mathrm{ALG} &= \widetilde{k}_{L} + \sum_{i\in H}\max\{\widetilde{k}_{i,L(+)}^{\pos}, \widetilde{k}_{i,L(+)}^{\neg}\} - \frac{\eta m}{16}\\
        &\le k_L + \frac{\eta m}{32} + \sum_{i\in H}\max\{{k}_{i,L(+)}^{\pos}, {k}_{i,L(+)}^{\neg}\} + |H|\cdot \frac{\eps^{4}\delta\eta^3 m}{49152} - \frac{\eta m}{16} &&\rhd \text{Lemma~\ref{lem:concentration-k-L-max-2and} and Lemma~\ref{lem:concentration-general-max-2and}}\\
        &\le k_L + \sum_{i\in H}\max\{{k}_{i,L(+)}^{\pos}, {k}_{i,L(+)}^{\neg}\} && \rhd |H|\le \frac{1536}{\eps^{4}\delta\eta^2}\\
        &\le \opt,
    \end{align*}
where the last step follows since $k_L + \sum_{i\in H}\max\{{k}_{i,L(+)}^{\pos}, {k}_{i,L(+)}^{\neg}\}$ is a lower bound on the value of a feasible assignment.
Specifically, we fix the restriction of the assignment on $L$ to be $\boldsymbol{\sigma}^*|_{L}$, and then, for each $i\in H$, we set $\sigma_i\in\{-1,1\}$ to maximize the number of satisfied constraints involving variable $x_i$ and a low-degree variable.

Furthermore, 
    \begin{align*}
        \mathrm{ALG} &= \widetilde{k}_{L} + \sum_{i\in H}\max\{\widetilde{k}_{i,L(+)}^{\pos}, \widetilde{k}_{i,L(+)}^{\neg}\} - \frac{\eta m}{16}\\
        &\ge k_L - \frac{\eta m}{32} + \sum_{i\in H}\max\{{k}_{i,L(+)}^{\pos}, {k}_{i,L(+)}^{\neg}\} - |H|\cdot \frac{\eps^{4}\delta\eta^3 m}{49152} - \frac{\eta m}{16} &&\rhd \text{Lemma~\ref{lem:concentration-k-L-max-2and} and Lemma~\ref{lem:concentration-general-max-2and}}\\
        &\ge k_L -\frac{\eta m}{32}+ \sum_{i\in H}\max\{{k}_{i,L(+)}^{\pos}, {k}_{i,L(+)}^{\neg}\} -\frac{\eta m}{32} - \frac{\eta m}{16} &&\rhd |H|\le \frac{1536}{\eps^{4}\delta\eta^2}\\
        &\ge k_L + k_{L,H}  -\frac{\eta m}{8} &&\rhd \sum_{i\in H}\max\{{k}_{i,L(+)}^{\pos}, {k}_{i,L(+)}^{\neg}\} \ge k_{L,H}\\
        &= \opt - k_H  -\frac{\eta m}{8} &&\rhd \opt = k_L+k_H+ k_{L,H}\\
        &\ge 
        \opt - \frac{1536^2}{\eps^{8} \delta^2\eta^4} - \frac{\eta}{2} \cdot \opt &&\rhd \text{Fact~\ref{fact:size-H-2and} and } m\le 4\opt\\
        &\ge (1-\eta)\cdot \opt, &&\rhd \frac{1536^2}{\eps^{8} \delta^2\eta^4} \le \frac{\eta m}{8}\le \frac{\eta}{2}\cdot \opt 
    \end{align*}
    where the fourth step follows from 
    \begin{align*}
        \sum_{i\in H}\max\{{k}_{i,L(+)}^{\pos}, {k}_{i,L(+)}^{\neg}\} &= \sum_{i\in H,\sigma^*_i=1}\max\{{k}_{i,L(+)}^{\pos}, {k}_{i,L(+)}^{\neg}\}+\sum_{i\in H,\sigma^*_i=-1}\max\{{k}_{i,L(+)}^{\pos}, {k}_{i,L(+)}^{\neg}\} \\
        &\ge \sum_{i\in H,\sigma^*_i=1} {k}_{i,L(+)}^{\pos} + \sum_{i\in H,\sigma^*_i=-1} {k}_{i,L(+)}^{\neg}
        = k_{L,H}.
    \end{align*}
    
\end{proof}

\subsection{Our algorithm in insertion-only streams}\label{sec:insertion-2and}
In this section we show how to implement the offline estimator (Algorithm~\ref{alg:max-2and-offline}) for \textsc{Max-$2$AND} in the streaming model. Suppose the constraints arrive one by one in an insertion-only stream, in an arbitrary order.

Recall that Algorithm~\ref{alg:max-2and-offline} (i) partitions the variables into low-degree $L$ and high-degree $H$,
(ii) maintains constraint-type counts within $L$ induced by the predictions and approximately computes (via a small linear system) an estimate of the optimal contribution from $L$,
and (iii) for each $i\in H$, maintains incidence counts for constraints involving $i$ and variables in $L$, which are used to assign $i$ that maximizes $i$’s estimated contribution.
The main challenge in the streaming setting is that we do \emph{not} know $H$ and $L$ in advance, and we cannot afford to maintain degrees for all variables explicitly.

Our streaming implementation is built around three ingredients that together let us ``simulate'' the offline estimator:
\begin{enumerate}
    \item 
    We maintain a uniform sample $F$ of $\poly(1/\eps,1/\delta,1/\eta)$ constraints using reservoir sampling.
    At the end of the stream we set $\widetilde{H}$ to be the set of associated variables of constraints in $F$.
    With high probability, every truly high-degree variable appears in $F$ (hence $H\subseteq \widetilde{H}$),
    while $|\widetilde{H}|=\poly(1/\eps,1/\delta,1/\eta)$ remains small.
    We then define $\widetilde{L}=[n]\setminus \widetilde{H}$.

    \item 
    Since we cannot store per-variable incident counts, we maintain four CountMin sketches
    $\CMpos[+]$, $\CMneg[+]$, $\CMpos[-]$, and $\CMneg[-]$.
    Intuitively, they estimate (for each variable $i$) the numbers of constraints in which $i$ appears positively or negatively and the other literal is satisfied or not satisfied by the predicted assignment.
    When a constraint $((j_1,j_2),(b_1,b_2))$ arrives, we query both variables and update the appropriate sketches:
    we update $\CM^{\pos}[b_2\mathcal{O}(j_2)][j_1]$ if $b_1=1$ and $\CM^{\neg}[b_2\mathcal{O}(j_2)][j_1]$ if $b_1=-1$; symmetrically, we update $\CM^{\pos}[b_1\mathcal{O}(j_1)][j_2]$ if $b_2=1$ and $\CM^{\neg}[b_1\mathcal{O}(j_1)][j_2]$ if $b_2=-1$.
    This suffices to recover (up to a small additive error) the quantities $X_{i,\widetilde{L}(+)}^{\pos},X_{i,\widetilde{L}(+)}^{\neg},X_{i,\widetilde{L}(-)}^{\pos},X_{i,\widetilde{L}(-)}^{\neg}$ for each $i\in \widetilde{H}$ at the end of the stream.

    \item 
    In parallel, we maintain exact counters over the entire stream for (i) $X_{+,+}$: the number of such constraints for which both literals are satisfied by the predicted assignment, (ii) $X_{+,-}$: the number for which exactly one literal is satisfied, and (iii) $X_{-,-}$: the number for which neither literal is satisfied.
    These provide the global constraint-type statistics induced by the predictions.
\end{enumerate}

The algorithm is formally described in \cref{alg:max-2and-insertion}. We provide the following theoretical guarantee on its approximation ratio and space complexity. 

\begin{algorithm}[t!]
\begin{algorithmic}[1]
\Require $\textsc{Max-$2$AND}$ instance $\Psi$ as an insertion-only stream of constraints, an $\eps$-accurate oracle $\mathcal{O} \rightarrow \{-1, 1\}$, $\delta,\eta\in (0,1)$
\Ensure Estimate of the maximum value of $\Psi$
\LeftComment{\textbf{\textcolor{blue}{Pre-processing phase}}}
\State {\bf initialize} $X_{+,+}\leftarrow 0, X_{+,-}\leftarrow 0, X_{-,-}\leftarrow 0, m\leftarrow 0$
\State {\bf initialize} four CountMin sketches $\CMpos[+], \CMneg[+], \CMpos[-]$ and $\CMneg[-]$ with accuracy parameter $\frac{\eps^{10}\delta^5 \eta^6}{128\beta}$ and failure probability $\frac{\eps^{8}\delta^4\eta^4}{64\beta}$, where $\beta = 2^{28}$.

\LeftComment{\textbf{\textcolor{blue}{Streaming phase}}}
\ForEach{constraint $((j_1,j_2),(b_1,b_2))$ in the stream}
    \State $m \leftarrow m+1,\ Y_{j_1} = \mathcal{O}(j_1), Y_{j_2} = \mathcal{O}(j_2)$
    \IIf{$b_1 Y_{j_1} = b_2 Y_{j_2} = 1$} $X_{+,+} \leftarrow X_{+,+} + 1$
    \IIf{$b_1 Y_{j_1} \neq b_2 Y_{j_2}$} $X_{+,-} \leftarrow X_{+,-} + 1$
    \IIf{$b_1 Y_{j_1} = b_2 Y_{j_2} = -1$} $X_{-,-} \leftarrow X_{-,-} + 1$
    \IIf{$b_1 = 1$}
         $\textsc{Update} (\CMpos[b_2 Y_{j_2}][j_1])$
    \ElseIIf $\textsc{Update} (\CMneg[b_2 Y_{j_2}][j_1])$
    
    \IIf{$b_2 = 1$}
         $\textsc{Update} (\CMpos[b_1 Y_{j_1}][j_2])$
    \ElseIIf
        $\textsc{Update} (\CMneg[b_1 Y_{j_1}][j_2])$
    
\EndFor

\State In parallel, uniformly sample $\frac{\beta}{\eps^{8}\delta^3\eta^4}$ constraints $F$ in the stream via reservoir sampling.

\LeftComment{\textbf{\textcolor{blue}{Post-processing phase}}}
\ForEach{$i\in [n]$}
    \State $\widetilde{X}_{i,+}^{\pos} \leftarrow \textsc{Query}(\CMpos[+][i]), \quad \widetilde{X}_{i,+}^{\neg} \leftarrow \textsc{Query}(\CMneg[+][i])$
    \State $\widetilde{X}_{i,-}^{\pos} \leftarrow \textsc{Query}(\CMpos[-][i]),\quad \widetilde{X}_{i,-}^{\neg} \leftarrow \textsc{Query}(\CMneg[-][i])$
\EndFor
\State $\widetilde{H} \leftarrow \bigcup_{((j_1,j_2),(b_1,b_2))\in F} \{j_1,j_2\}, \widetilde{L} \leftarrow [n]\setminus \widetilde{H}$
\State $\widetilde{H}^+ \leftarrow \{i\in \widetilde{H}: Y_i = 1\}, \widetilde{H}^- \leftarrow \widetilde{H} \setminus \widetilde{H}^+$ 
\State $\overline{X}_0 \leftarrow X_{+,+} - \sum_{i\in \widetilde{H}^+} \widetilde{X}^{\pos}_{i,+}-\sum_{i\in \widetilde{H}^-} \widetilde{X}^{\neg}_{i,+}$
\State $\overline{X}_1 \leftarrow X_{+,-} - \sum_{i\in \widetilde{H}^+}(\widetilde{X}^{\pos}_{i,-}+\widetilde{X}^{\neg}_{i,+}) - \sum_{i\in \widetilde{H}^-}(\widetilde{X}^{\pos}_{i,+}+\widetilde{X}^{\neg}_{i,-})$ 
\State $\overline{X}_2 \leftarrow X_{-,-} - \sum_{i\in \widetilde{H}^+} \widetilde{X}^{\neg}_{i,-}-\sum_{i\in \widetilde{H}^-} \widetilde{X}^{\pos}_{i,-}$
\State $\overline{k}_{\widetilde{L}}\leftarrow \frac{1}{4\eps^2}[(\frac{1}{2}+\eps)^2 \overline{X}_0 - (\frac{1}{2}+\eps)(\frac{1}{2}-\eps)\overline{X}_1 + (\frac{1}{2}-\eps)^2 \overline{X}_2]$
\ForEach{$i\in \widetilde{H}$} 
\State $\overline{k}_{i,\widetilde{L}(+)}^{\pos} \leftarrow (\frac{1}{2}+\frac{1}{4\eps})\cdot \widetilde{X}_{i,+}^{\pos} + (\frac{1}{2}-\frac{1}{4\eps})\cdot \widetilde{X}_{i,-}^{\pos},\quad \overline{k}_{i,\widetilde{L}(+)}^{\neg} \leftarrow (\frac{1}{2}+\frac{1}{4\eps})\cdot \widetilde{X}_{i,+}^{\neg} + (\frac{1}{2}-\frac{1}{4\eps})\cdot \widetilde{X}_{i,-}^{\neg}$ 
\EndFor
\State \textbf{return} $\overline{k}_{\widetilde{L}} + \sum_{i\in \widetilde{H}}\max\{\overline{k}_{i,\widetilde{L}(+)}^{\pos}, \overline{k}_{i,\widetilde{L}(+)}^{\neg}\}  -  \frac{\eta m}{16}$ 
\end{algorithmic}
\caption{Estimating the \textsc{Max-$2$AND} value in insertion-only streams}
\label{alg:max-2and-insertion}
\end{algorithm}

\begin{theorem}
\label{thm:max-2and-insertion}
Let $\eps\in (0,\frac{1}{2}]$ and $\delta \in (0,1)$. For any $\eta \in (0,1)$,
given an $\eps$-accurate oracle $\mathcal{O}$ for a \textsc{Max-$2$AND} instance $\Psi$, 
there exists a single-pass $(1-\eta)$-approximation algorithm for estimating the maximum value of $\Psi$ in insertion-only streams. The algorithm uses $O(\frac{1}{\eps^{10}\delta^5\eta^6}\log\frac{1}{\eps \delta\eta})$ words of space. The approximation holds with probability at least $1-\delta$.
\end{theorem}

We prove Theorem~\ref{thm:max-2and-insertion} by establishing several lemmas that justify the correctness of our streaming simulation of Algorithm~\ref{alg:max-2and-offline} using the high-degree threshold $\theta m$ for $\theta:=\frac{\eps^4 \delta \eta^2}{4096}$, i.e., $H:=\{i\in [n]: \deg(i) \ge \theta m\}$.
As we will see, this choice of threshold guarantees that our streaming-constructed partition $(\widetilde L,\widetilde H)$ satisfies the key properties required for the offline analysis, and yields the desired success probability.
We begin by showing that reservoir sampling identifies a superset $\widetilde{H}$ of the high-degree variables $H$ with high probability.
Recall that $\widetilde{H}$ is the set of associated variables  of all constraints in the uniformly sampled constraint set $F$. 

\begin{lemma}
\label{lem:high-deg-reservoir-max-2and}
Let $\delta \in (0,1)$.
Then
 $\widetilde{H}$ contains all high-degree variables $H$ in $\Psi$ (i.e., $H \subseteq \widetilde{H}$) with probability at least $1-\frac{\delta}{8}$.
\end{lemma}
\begin{proof}
    Let $\theta:=\frac{\eps^{4}\delta\eta^2}{4096}$. Consider any $i \in H$. Then, by definition of $H$, we have $\deg(i) \geq \theta m$. 
    For each constraint $C_j$ involving variable $x_i$, define an indicator random variable 
    \begin{equation*}
            X_{j}= \begin{cases}1 & \text { if } C_j\in F, \\ 0 & \text { otherwise. }\end{cases}
    \end{equation*}
    So, $\E[X_j] = \Pr[X_j=1] = \frac{\beta}{\eps^{8}\delta^3\eta^4 m}$. We define $X:= \sum_{j\in [\deg(i)]} X_j$ to denote the number of constraints involving variable $x_i$ that are sampled by the end of the stream. Then, $\E[X] = \sum_{j\in [\deg(i)]} \E[X_j] = \frac{\beta\deg(i)}{\eps^{8}\delta^3 \eta^4m}$. 

    Next, we bound $\Var[X]$. For any $j\in [\deg(i)]$, $\Var[X_j] = \E[X^2_j] - (\E[X_j])^2 = \E[X_j] - (\E[X_j])^2 \le \frac{\beta}{\eps^{8} \delta^3\eta^4 m}$.
    Since $(X_j)_{j\in [\deg(i)]}$ are negatively correlated, $\Var[X] \le \sum_{j\in [\deg(i)]}\Var[X_j] \le \deg(i)\cdot \frac{\beta}{\eps^{8} \delta^3\eta^4 m} = \frac{\beta\deg(i)}{\eps^{8} \delta^3\eta^4 m}$.
    Therefore, by Chebyshev's inequality, 
    \begin{align*}
        \Pr[i\notin \widetilde{H}] = \Pr[X\leq 0] 
        &= \Pr\left[X \leq \E[X]- \frac{\beta\deg(i)}{\eps^{8} \delta^3\eta^4 m}\right]\\
        &\leq \frac{\Var[X]\cdot \eps^{16} \delta^6\eta^8 m^2}{\beta^2\deg^2(i)}\\
        &\leq \frac{\eps^{8} \delta^3\eta^4 m}{\beta \deg(i)} &&\rhd \Var[X] \leq \frac{\beta\deg(i)}{\eps^{8} \delta^3 \eta^4m}\\
        &\leq \frac{4096\eps^{4}\delta^2\eta^2}{\beta}. &&\rhd \deg(i)\geq \theta m
    \end{align*}
    Recall that $|H| \le \frac{2m}{\theta m}=\frac{2}{\theta} = \frac{8192}{\eps^{4} \delta\eta^2}$.
    By union bound, with probability at least $1-\frac{8192}{\eps^{4} \delta\eta^2}\cdot \frac{4096\eps^{4}\delta^2\eta^2}{\beta}  = 1-\frac{\delta}{8}$ (since $\beta = 2^{28}$), $\widetilde{H}$ contains all high-degree variables $H$ in $\Psi$.
\end{proof}

For each $i\in [n]$, define 
\begin{enumerate}
    \item $X_{i,+}^{\pos}:= |\{((i,j),(b_i,b_j))\in \Psi: j\in [n], b_i=1, b_j Y_j = 1\}|$ to be the number of constraints involving $x_i$ such that the literal of $x_i$ appears positively (i.e., $b_i=1$) and the other literal is satisfied by the prediction $\mathbf{Y}$ (i.e., $b_j Y_j=1$). 
    The other quantities are defined analogously:
    \item $X_{i,-}^{\pos}:= |\{((i,j),(b_i,b_j))\in \Psi: j\in [n], b_i=1, b_jY_j = -1\}|$.
    \item $X_{i,+}^{\neg}:= |\{((i,j),(b_i,b_j))\in \Psi: j\in [n], b_i=-1, b_jY_j = 1\}|$.
    \item $X_{i,-}^{\neg}:= |\{((i,j),(b_i,b_j))\in \Psi: j\in [n], b_i=-1, b_jY_j = -1\}|$.
\end{enumerate}

\begin{lemma}
\label{lem:error-bound-insertion-max-2and}
    Let $\eps \in (0,\frac{1}{2}]$ and $\delta,\eta \in (0,1)$.
    For all $i\in \widetilde{H}$, we have
    \begin{enumerate}
        \item[(i)] $\widetilde{X}_{i,+}^{\pos} \geq X_{i,+}^{\pos},\  \widetilde{X}_{i,-}^{\pos} \geq X_{i,-}^{\pos},\  \widetilde{X}_{i,+}^{\neg} \geq X_{i,+}^{\neg},\  \widetilde{X}_{i,-}^{\neg} \geq X_{i,-}^{\neg}$;
        \item[(ii)] With probability at least $1-\frac{\delta}{8}$, all of the following inequalities hold: $\widetilde{X}_{i,+}^{\pos} \leq X_{i,+}^{\pos} + \frac{\eps^{10}\delta^5\eta^6 }{128\beta}m,\  \widetilde{X}_{i,-}^{\pos} \leq X_{i,-}^{\pos}+ \frac{\eps^{10}\delta^5\eta^6 }{128\beta}m,\  \widetilde{X}_{i,+}^{\neg} \leq X_{i,+}^{\neg}+ \frac{\eps^{10}\delta^5\eta^6 }{128\beta}m,\  \widetilde{X}_{i,-}^{\neg} \leq X_{i,-}^{\neg}+\frac{\eps^{10}\delta^5\eta^6 }{128\beta}m $. 
    \end{enumerate} 
\end{lemma}
\begin{proof}
    (i) By \cref{thm:CM-v2}, for all $i\in [n]$, we have $\widetilde{X}_{i,+}^{\pos} \geq X_{i,+}^{\pos},\  \widetilde{X}_{i,-}^{\pos} \geq X_{i,-}^{\pos},\  \widetilde{X}_{i,+}^{\neg} \geq X_{i,+}^{\neg}$, and $\widetilde{X}_{i,-}^{\neg} \geq X_{i,-}^{\neg}$.

    (ii) By \cref{thm:CM-v2}, for any $i\in [n]$, with probability at least $1-\frac{\eps^{8}\delta^4\eta^4}{64\beta}$, each of the following four inequalities holds: $\widetilde{X}_{i,+}^{\pos} \leq X_{i,+}^{\pos} + \frac{\eps^{10}\delta^5\eta^6 }{128\beta}m,\  \widetilde{X}_{i,-}^{\pos} \leq X_{i,-}^{\pos}+ \frac{\eps^{10}\delta^5\eta^6 }{128\beta}m,\  \widetilde{X}_{i,+}^{\neg} \leq X_{i,+}^{\neg}+ \frac{\eps^{10}\delta^5\eta^6 }{128\beta}m$, and $\widetilde{X}_{i,-}^{\neg} \leq X_{i,-}^{\neg}+\frac{\eps^{10}\delta^5\eta^6 }{128\beta}m$.
    Since $|\widetilde{H}|\le \frac{2\beta}{\eps^{8}\delta^3\eta^4}$, by applying the union bound we have that, with probability at least $1-\frac{\eps^{8}\delta^4\eta^4}{64\beta} \cdot \frac{2\beta}{\eps^{8}\delta^3\eta^4} \cdot 4 \ge 1-\frac{\delta}{8}$, for all $i\in \widetilde{H}$, all of the above four inequalities hold.
\end{proof}

\paragraph{Proof of \cref{thm:max-2and-insertion}.}
Given the above lemmas, we are ready to prove \cref{thm:max-2and-insertion}. 
\begin{proof}
We can assume that $m =\Omega(\eps^{-26}\delta^{-10}\eta^{-14})$, as otherwise we can store the whole instance in $O(\poly(1/\eps, 1/\delta, 1/\eta))$ space.

In \cref{alg:max-2and-insertion}, we maintain four CountMin sketches with accuracy parameter $\frac{\eps^{10}\delta^5\eta^6}{128\beta}$ and failure probability $\frac{\eps^{8}\delta^4\eta^4}{64\beta}$, where $\beta=2^{28}$. By \cref{thm:CM-v2}, this takes $O(\frac{1}{\eps^{10}\delta^5\eta^6}\log\frac{1}{\eps^{8} \delta^4\eta^4}) = O( \frac{1}{\eps^{10}\delta^5\eta^6}\log\frac{1}{\eps \delta\eta} )$ words of space. 
    Also, we use reservoir sampling to uniformly sample $O(\frac{1}{\eps^{8}\delta^3\eta^4})$ constraints, which takes $O(\frac{1}{\eps^{8}\delta^3\eta^4})$ words of space. 
    Therefore, the total space complexity of \cref{alg:max-2and-insertion} is $O( \frac{1}{\eps^{10}\delta^5\eta^6}\log\frac{1}{\eps \delta\eta} )$ words.

    Next, we analyze the approximation ratio.
Since $[n] = \widetilde{H}\cup \widetilde{L}$, we have
\begin{enumerate}
    \item $X_{i,+}^{\pos} = X_{i,\widetilde{H}(+)}^{\pos} + X_{i,\widetilde{L}(+)}^{\pos}$, where $X_{i,\widetilde{H}(+)}^{\pos} := |\{((i,j),(b_i,b_j))\in \Psi: j\in \widetilde{H}, b_i=1, b_j Y_j = 1\}|$ and $X_{i,\widetilde{L}(+)}^{\pos} := |\{((i,j),(b_i,b_j))\in \Psi: j\in \widetilde{L}, b_i=1, b_j Y_j = 1\}|$. 
    \item $X_{i,-}^{\pos} = X_{i,\widetilde{H}(-)}^{\pos} + X_{i,\widetilde{L}(-)}^{\pos}$, where $X_{i,\widetilde{H}(-)}^{\pos} := |\{((i,j),(b_i,b_j))\in \Psi: j\in \widetilde{H}, b_i=1, b_j Y_j = -1\}|$ and $X_{i,\widetilde{L}(-)}^{\pos} := |\{((i,j),(b_i,b_j))\in \Psi: j\in \widetilde{L}, b_i=1, b_j Y_j = -1\}|$.
    \item $X_{i,+}^{\neg} = X_{i,\widetilde{H}(+)}^{\neg} + X_{i,\widetilde{L}(+)}^{\neg}$, where $X_{i,\widetilde{H}(+)}^{\neg} := |\{((i,j),(b_i,b_j))\in \Psi: j\in \widetilde{H}, b_i=-1, b_j Y_j = 1\}|$ and $X_{i,\widetilde{L}(+)}^{\neg} := |\{((i,j),(b_i,b_j))\in \Psi: j\in \widetilde{L}, b_i=-1, b_j Y_j = 1\}|$.
    \item $X_{i,-}^{\neg} = X_{i,\widetilde{H}(-)}^{\neg} + X_{i,\widetilde{L}(-)}^{\neg}$, where $X_{i,\widetilde{H}(-)}^{\neg} := |\{((i,j),(b_i,b_j))\in \Psi: j\in \widetilde{H}, b_i=-1, b_j Y_j = -1\}|$ and $X_{i,\widetilde{L}(-)}^{\neg} := |\{((i,j),(b_i,b_j))\in \Psi: j\in \widetilde{L}, b_i=-1, b_j Y_j = -1\}|$.
    
\end{enumerate}
By Lemma~\ref{lem:error-bound-insertion-max-2and}, with probability at least $1-\frac{\delta}{8}$, for all $i\in \widetilde{H}$, all of the following inequalities hold: 
\begin{align*}
    X_{i,\widetilde{H}(+)}^{\pos} + X_{i,\widetilde{L}(+)}^{\pos} \leq \widetilde{X}_{i,+}^{\pos} \leq X_{i,\widetilde{H}(+)}^{\pos} + X_{i,\widetilde{L}(+)}^{\pos} + \frac{\eps^{10}\delta^5\eta^6}{128\beta} m,\\
    X_{i,\widetilde{H}(-)}^{\pos} + X_{i,\widetilde{L}(-)}^{\pos} \leq \widetilde{X}_{i,-}^{\pos} \leq X_{i,\widetilde{H}(-)}^{\pos} + X_{i,\widetilde{L}(-)}^{\pos} + \frac{\eps^{10}\delta^5\eta^6}{128\beta} m,\\
    X_{i,\widetilde{H}(+)}^{\neg} + X_{i,\widetilde{L}(+)}^{\neg} \leq \widetilde{X}_{i,+}^{\neg} \leq X_{i,\widetilde{H}(+)}^{\neg} + X_{i,\widetilde{L}(+)}^{\neg} + \frac{\eps^{10}\delta^5\eta^6}{128\beta} m,\\
    X_{i,\widetilde{H}(-)}^{\neg} + X_{i,\widetilde{L}(-)}^{\neg} \leq \widetilde{X}_{i,-}^{\neg} \leq X_{i,\widetilde{H}(-)}^{\neg} + X_{i,\widetilde{L}(-)}^{\neg} + \frac{\eps^{10}\delta^5\eta^6}{128\beta} m.
\end{align*}
Since $|F|= \frac{\beta}{\eps^{8}\delta^3\eta^4}$ and $\widetilde{H}$ is the set of associated variables of constraints in $F$, we have that $|\widetilde{H}|\le \frac{2\beta}{\eps^{8}\delta^3\eta^4}$. Thus, it holds that $0 \le X_{i,\widetilde{H}(+)}^{\pos}, X_{i,\widetilde{H}(-)}^{\pos}, X_{i,\widetilde{H}(+)}^{\neg}, X_{i,\widetilde{H}(-)}^{\neg}\le \frac{2\beta}{\eps^{8}\delta^3\eta^4}$. Therefore,
\begin{align}
      0\leq \widetilde{X}_{i,+}^{\pos} -  X_{i,\widetilde{L}(+)}^{\pos}\leq \frac{2\beta}{\eps^{8}\delta^3\eta^4}+  \frac{\eps^{10}\delta^5\eta^6}{128\beta} m,\label{ineq:cm_pos+}\\ 
      0 \leq \widetilde{X}_{i,-}^{\pos} - X_{i,\widetilde{L}(-)}^{\pos} \leq \frac{2\beta}{\eps^{8}\delta^3\eta^4} +  \frac{\eps^{10}\delta^5\eta^6}{128\beta} m,\label{ineq:cm_pos-}\\
      0 \leq \widetilde{X}_{i,+}^{\neg} - X_{i,\widetilde{L}(+)}^{\neg} \leq \frac{2\beta}{\eps^{8}\delta^3\eta^4}+  \frac{\eps^{10}\delta^5\eta^6}{128\beta} m,\label{ineq:cm_neg+}\\
      0 \leq \widetilde{X}_{i,-}^{\neg} - X_{i,\widetilde{L}(-)}^{\neg} \leq \frac{2\beta}{\eps^{8}\delta^3\eta^4} + \frac{\eps^{10}\delta^5\eta^6}{128\beta} m. \label{ineq:cm_neg-}
\end{align}

We define
\begin{enumerate}
    \item $\widetilde{X}_0 := |\{((j_1,j_2),(b_1,b_2))\in \Psi: j_1, j_2\in \widetilde{L},\ b_1 Y_{j_1} = b_2 Y_{j_2}=1\}|$ to be the number of constraints in $\Psi$ whose associated variables are both in $\widetilde{L}$ and for which both literals are satisfied by the predictions $\mathbf{Y}$.
    \item $\widetilde{X}_1 := |\{((j_1,j_2),(b_1,b_2))\in \Psi: j_1, j_2\in \widetilde{L},\ b_1 Y_{j_1} \neq b_2 Y_{j_2}\}|$ to be the number of constraints in $\Psi$ whose associated variables are both in $\widetilde{L}$ and for which exactly one literal is satisfied by $\mathbf{Y}$.
    \item $\widetilde{X}_2 := |\{((j_1,j_2),(b_1,b_2))\in \Psi: j_1, j_2\in \widetilde{L},\ b_1 Y_{j_1} = b_2 Y_{j_2}=-1\}|$ to be the number of constraints in $\Psi$ whose associated variables are both in $\widetilde{L}$ and for which neither literal is satisfied by  $\mathbf{Y}$.
\end{enumerate}
Similarly, we define
\begin{enumerate}
    \item $\widetilde{X}_{0,\widetilde{H}} := |\{((j_1,j_2),(b_1,b_2))\in \Psi: j_1, j_2\in \widetilde{H},\ b_1 Y_{j_1} = b_2 Y_{j_2}=1\}|$ to be the number of constraints in $\Psi$ whose associated variables are both in $\widetilde{H}$ and for which both literals are satisfied by $\mathbf{Y}$.
    \item $\widetilde{X}_{1,\widetilde{H}} := |\{((j_1,j_2),(b_1,b_2))\in \Psi: j_1, j_2\in \widetilde{H},\ b_1 Y_{j_1} \neq b_2 Y_{j_2}\}|$ to be the number of constraints in $\Psi$ whose associated variables are both in $\widetilde{H}$ and for which exactly one literal is satisfied by $\mathbf{Y}$.
    \item $\widetilde{X}_{2,\widetilde{H}} := |\{((j_1,j_2),(b_1,b_2))\in \Psi: j_1, j_2\in \widetilde{H},\ b_1 Y_{j_1} = b_2 Y_{j_2}=-1\}|$ to be the number of constraints in $\Psi$ whose associated variables are both in $\widetilde{H}$ and for which neither literal is satisfied by  $\mathbf{Y}$.
\end{enumerate}

We have

\begin{align*}
    \widetilde{X}_0 &=  X_{+,+} - \widetilde{X}_{0,\widetilde{H}} - \sum_{i\in \widetilde{H}^+}X_{i,\widetilde{L}(+)}^{\pos}-\sum_{i\in \widetilde{H}^-}X_{i,\widetilde{L}(+)}^{\neg},\\
    \widetilde{X}_1 &= 
     X_{+,-} - \widetilde{X}_{1,\widetilde{H}} - \sum_{i\in \widetilde{H}^+}(X_{i,\widetilde{L}(-)}^{\pos} + X_{i,\widetilde{L}(+)}^{\neg}) - \sum_{i\in \widetilde{H}^-}(X_{i,\widetilde{L}(+)}^{\pos} + X_{i,\widetilde{L}(-)}^{\neg}), \\
    \widetilde{X}_2 &=  X_{-,-} - \widetilde{X}_{2,\widetilde{H}} -\sum_{i\in \widetilde{H}^+}X_{i,\widetilde{L}(-)}^{\neg}-\sum_{i\in \widetilde{H}^-}X_{i,\widetilde{L}(-)}^{\pos}.
\end{align*}

Recall that
\begin{align*}
\overline{X}_0 &= X_{+,+} - \sum_{i\in \widetilde{H}^+} \widetilde{X}^{\pos}_{i,+}-\sum_{i\in \widetilde{H}^-} \widetilde{X}^{\neg}_{i,+},\\
\overline{X}_1 &=X_{+,-} - \sum_{i\in \widetilde{H}^+}(\widetilde{X}^{\pos}_{i,-}+\widetilde{X}^{\neg}_{i,+}) - \sum_{i\in \widetilde{H}^-}(\widetilde{X}^{\pos}_{i,+}+\widetilde{X}^{\neg}_{i,-}),\\
\overline{X}_2 &= X_{-,-} - \sum_{i\in \widetilde{H}^+} \widetilde{X}^{\neg}_{i,-}-\sum_{i\in \widetilde{H}^-} \widetilde{X}^{\pos}_{i,-}.
\end{align*}
Therefore,
\begin{align*}
    \overline{X}_0 - \widetilde{X}_0 &= \widetilde{X}_{0,\widetilde{H}} + \sum_{i\in \widetilde{H}^+}(X_{i,\widetilde{L}(+)}^{\pos} - \widetilde{X}^{\pos}_{i,+}) + \sum_{i\in \widetilde{H}^-} (X_{i,\widetilde{L}(+)}^{\neg} - \widetilde{X}^{\neg}_{i,+}),\\
    \overline{X}_1 - \widetilde{X}_1 &= \widetilde{X}_{1,\widetilde{H}} + \sum_{i\in \widetilde{H}^+} \left[(X_{i,\widetilde{L}(-)}^{\pos} - \widetilde{X}^{\pos}_{i,-}) +  (X_{i,\widetilde{L}(+)}^{\neg} - \widetilde{X}^{\neg}_{i,+})\right]
     + \sum_{i\in \widetilde{H}^-} \left[(X_{i,\widetilde{L}(+)}^{\pos} - \widetilde{X}^{\pos}_{i,+}) +  (X_{i,\widetilde{L}(-)}^{\neg} - \widetilde{X}^{\neg}_{i,-})\right],\\
    \overline{X}_2 - \widetilde{X}_2 &=\widetilde{X}_{2,\widetilde{H}} + \sum_{i\in \widetilde{H}^+}(X_{i,\widetilde{L}(-)}^{\neg} - \widetilde{X}^{\neg}_{i,-}) + \sum_{i\in \widetilde{H}^-} (X_{i,\widetilde{L}(-)}^{\pos} - \widetilde{X}^{\pos}_{i,-}).
\end{align*}
Since $|\widetilde{H}|\le \frac{2\beta}{\eps^{8}\delta^3\eta^4}$, we have $|\widetilde{H}^+|, |\widetilde{H}^-|\le \frac{2\beta}{\eps^{8}\delta^3\eta^4}$, and $\widetilde{X}_{0,\widetilde{H}},\widetilde{X}_{1,\widetilde{H}},\widetilde{X}_{2,\widetilde{H}}\le \binom{|\widetilde{H}|}{2}\cdot 2^2 \le \frac{8\beta^2}{\eps^{16}\delta^6\eta^8}$.
Moreover, combining inequalities (\ref{ineq:cm_pos+})--(\ref{ineq:cm_neg-}), we have
\begin{align*}
    \widetilde{X}_0-\left(\frac{16\beta^2}{\eps^{16}\delta^6\eta^8}+\frac{\eps^{2}\delta^2\eta^2 m}{32} \right)\le \overline{X}_0  \le \widetilde{X}_0+\frac{8\beta^2}{\eps^{16}\delta^6\eta^8},\\
    \widetilde{X}_1-\left(\frac{32\beta^2}{\eps^{16}\delta^6\eta^8}+\frac{\eps^{2}\delta^2\eta^2 m}{16} \right) \le \overline{X}_1  \le \widetilde{X}_1+\frac{8\beta^2}{\eps^{16}\delta^6\eta^8},\\
    \widetilde{X}_2-\left(\frac{16\beta^2}{\eps^{16}\delta^6\eta^8}+\frac{\eps^{2}\delta^2\eta^2 m}{32} \right) \le \overline{X}_2   \le \widetilde{X}_2+\frac{8\beta^2}{\eps^{16}\delta^6\eta^8}.
\end{align*}

Let $\widetilde{k}_{\widetilde{L}} := \frac{1}{4\eps^2}[(\frac{1}{2}+\eps)^2 \widetilde{X}_0 - (\frac{1}{2}+\eps)(\frac{1}{2}-\eps)\widetilde{X}_1 + (\frac{1}{2}-\eps)^2 \widetilde{X}_2]$. 
Now we are ready to show that $\overline{k}_{\widetilde{L}}$ is close to $\widetilde{k}_{\widetilde{L}}$ up to an additive error of $\poly(\eta)\cdot m$ with probability at least $1-\frac{\delta}{8}$.
We have
\begin{align*}
    \overline{k}_{\widetilde{L}} 
    =~&\frac{1}{4\eps^2}\left[\left(\frac{1}{2}+\eps\right)^2 \overline{X}_0 - \left(\frac{1}{2}+\eps\right)\left(\frac{1}{2}-\eps\right)\overline{X}_1 + \left(\frac{1}{2}-\eps\right)^2 \overline{X}_2\right]\\
    \ge~&\frac{1}{4\eps^2}\left[\left(\frac{1}{2}+\eps\right)^2 \left(\widetilde{X}_0-\left(\frac{16\beta^2}{\eps^{16}\delta^6\eta^8}+\frac{\eps^{2}\delta^2\eta^2 m}{32} \right)\right) - \left(\frac{1}{2}+\eps\right)\left(\frac{1}{2}-\eps\right)\left(\widetilde{X}_1+\frac{8\beta^2}{\eps^{16}\delta^6\eta^8}\right) \right.\\
    &+ \left.\left(\frac{1}{2}-\eps\right)^2 \left(\widetilde{X}_2-\left(\frac{16\beta^2}{\eps^{16}\delta^6\eta^8}+\frac{\eps^{2}\delta^2\eta^2 m}{32} \right)\right)\right]\\
    =~&\widetilde{k}_{\widetilde{L}} - \frac{1}{4\eps^2}\left[\left(\frac{1}{2}+\eps\right)^2 \left(\frac{16\beta^2}{\eps^{16}\delta^6\eta^8}+\frac{\eps^{2}\delta^2\eta^2 m}{32} \right) + \left(\frac{1}{2}+\eps\right)\left(\frac{1}{2}-\eps\right)\frac{8\beta^2}{\eps^{16}\delta^6\eta^8} \right.\\
    &+ \left.\left(\frac{1}{2}-\eps\right)^2 \left(\frac{16\beta^2}{\eps^{16}\delta^6\eta^8}+\frac{\eps^{2}\delta^2\eta^2 m}{32} \right)\right]\\
    =~&\widetilde{k}_{\widetilde{L}} - \left(\frac{5\beta^{2}}{2\eps^{18}\delta^{6}\eta^8} + \frac{6\beta^{2}}{\eps^{16}\delta^{6}\eta^8} + \frac{\delta^{2}\eta^2m}{256} + \frac{\eps^{2}\delta^2\eta^2 m}{64}\right)\\
    \ge~&\widetilde{k}_{\widetilde{L}} -  \frac{\eta m}{32},
\end{align*}
where the last step follows since $m = \Omega(\eps^{-26}\delta^{-10}\eta^{-14})$. On the other hand,
\begin{align*}
\overline{k}_{\widetilde{L}} =~&\frac{1}{4\eps^2}\left[\left(\frac{1}{2}+\eps\right)^2 \overline{X}_0 - \left(\frac{1}{2}+\eps\right)\left(\frac{1}{2}-\eps\right)\overline{X}_1 + \left(\frac{1}{2}-\eps\right)^2 \overline{X}_2\right]\\
\le~&\frac{1}{4\eps^2}\left[\left(\frac{1}{2}+\eps\right)^2 \left(\widetilde{X}_0+\frac{8\beta^2}{\eps^{16}\delta^6\eta^8}\right) - \left(\frac{1}{2}+\eps\right)\left(\frac{1}{2}-\eps\right)\left(\widetilde{X}_1-\left(\frac{32\beta^2}{\eps^{16}\delta^6\eta^8}+\frac{\eps^{2}\delta^2\eta^2 m}{16} \right)\right)\right.\\
&+ \left.\left(\frac{1}{2}-\eps\right)^2 \left(\widetilde{X}_2+\frac{8\beta^2}{\eps^{16}\delta^6\eta^8}\right)\right]\\
=~&\widetilde{k}_{\widetilde{L}} +\frac{1}{4\eps^2}\left[\left(\frac{1}{2}+\eps\right)^2 \frac{8\beta^2}{\eps^{16}\delta^6\eta^8} + \left(\frac{1}{2}+\eps\right)\left(\frac{1}{2}-\eps\right)\left(\frac{32\beta^2}{\eps^{16}\delta^6\eta^8}+\frac{\eps^{2}\delta^2\eta^2 m}{16} \right) + \left(\frac{1}{2}-\eps\right)^2 \frac{8\beta^2}{\eps^{16}\delta^6\eta^8}\right]\\
=~&\widetilde{k}_{\widetilde{L}} + \left(\frac{3\beta^{2}}{\eps^{18}\delta^{6}\eta^8} - \frac{4\beta^{2}}{\eps^{16}\delta^{6}\eta^8}+\frac{ \delta^{2}\eta^2m}{256}-\frac{\eps^{2}\delta^2\eta^2m}{64}\right)\\
\le~&\widetilde{k}_{\widetilde{L}} + \frac{\eta m}{64},
\end{align*}
where the last step follows since $m = \Omega(\eps^{-26}\delta^{-10}\eta^{-14})$.

For each $i\in \widetilde{H}$, let $\widetilde{k}_{i,\widetilde{L}(+)}^{\pos} := (\frac{1}{2}+\frac{1}{4\eps})\cdot X_{i,\widetilde{L}(+)}^{\pos} + (\frac{1}{2}-\frac{1}{4\eps})\cdot X_{i,\widetilde{L}(-)}^{\pos}$ and $\widetilde{k}_{i,\widetilde{L}(+)}^{\neg} := (\frac{1}{2}+\frac{1}{4\eps})\cdot X_{i,\widetilde{L}(+)}^{\neg} + (\frac{1}{2}-\frac{1}{4\eps})\cdot X_{i,\widetilde{L}(-)}^{\neg}$.
We have
\begin{align*}
    \overline{k}_{i,\widetilde{L}(+)}^{\pos} &= \left(\frac{1}{2}+\frac{1}{4\eps}\right)\cdot \widetilde{X}_{i,+}^{\pos} + \left(\frac{1}{2}-\frac{1}{4\eps}\right)\cdot \widetilde{X}_{i,-}^{\pos}\\
    &\ge \left(\frac{1}{2}+\frac{1}{4\eps}\right)\cdot X_{i,\widetilde{L}(+)}^{\pos} + \left(\frac{1}{2}-\frac{1}{4\eps}\right) \left( X_{i,\widetilde{L}(-)}^{\pos} + \frac{2\beta}{\eps^{8}\delta^3\eta^4} + \frac{\eps^{10}\delta^5 \eta^6m}{128\beta} \right)\\
    &= \widetilde{k}_{i,\widetilde{L}(+)}^{\pos} - \left(\frac{1}{4\eps}-\frac{1}{2}\right) \left(\frac{2\beta}{\eps^{8}\delta^3\eta^4} + \frac{\eps^{10}\delta^5 \eta^6m}{128\beta} \right)\\
    &= \widetilde{k}_{i,\widetilde{L}(+)}^{\pos} -\left(\frac{\beta}{2\eps^{9}\delta^3\eta^4} - \frac{\beta}{\eps^{8}\delta^3\eta^4} + \frac{\eps^{9}\delta^5\eta^6 m}{512\beta} - \frac{\eps^{10}\delta^5\eta^6 m}{256\beta}\right)\\
    &\ge \widetilde{k}_{i,\widetilde{L}(+)}^{\pos} - \frac{\eps^8 \delta^3 \eta^5 m}{128\beta},
\end{align*}
where the last step follows since $m = \Omega(\eps^{-26}\delta^{-10}\eta^{-14})$. On the other hand,
\begin{align*}
    \overline{k}_{i,\widetilde{L}(+)}^{\pos} &= \left(\frac{1}{2}+\frac{1}{4\eps}\right)\cdot \widetilde{X}_{i,+}^{\pos} + \left(\frac{1}{2}-\frac{1}{4\eps}\right)\cdot \widetilde{X}_{i,-}^{\pos}\\
    &\le \left(\frac{1}{2}+\frac{1}{4\eps}\right)\left( X_{i,\widetilde{L}(+)}^{\pos} + \frac{2\beta}{\eps^{8}\delta^3\eta^4} + \frac{\eps^{10}\delta^5\eta^6 m}{128\beta} \right) + \left(\frac{1}{2}-\frac{1}{4\eps}\right)\cdot X_{i,\widetilde{L}(-)}^{\pos}\\
    &= \widetilde{k}_{i,\widetilde{L}(+)}^{\pos} + \left(\frac{1}{2}+\frac{1}{4\eps}\right)\left(\frac{2\beta}{\eps^{8}\delta^3\eta^4} + \frac{\eps^{10}\delta^5\eta^6 m}{128\beta} \right)\\
    &= \widetilde{k}_{i,\widetilde{L}(+)}^{\pos} +\left(\frac{\beta}{2\eps^{9}\delta^3\eta^4} + \frac{\beta}{\eps^{8}\delta^3\eta^4} + \frac{\eps^{9}\delta^5\eta^6 m}{512\beta} + \frac{\eps^{10}\delta^5\eta^6 m}{256\beta}\right)\\
    &\le \widetilde{k}_{i,\widetilde{L}(+)}^{\pos} +\frac{\eps^8 \delta^3 \eta^5 m}{128\beta},
\end{align*}
where the last step follows since $m = \Omega(\eps^{-26}\delta^{-10}\eta^{-14})$. Therefore, $|\overline{k}_{i,\widetilde{L}(+)}^{\pos} - \widetilde{k}_{i,\widetilde{L}(+)}^{\pos}| \le \frac{\eps^8 \delta^3 \eta^5 m}{128\beta}$.
Similarly, $|\overline{k}_{i,\widetilde{L}(+)}^{\neg} - \widetilde{k}_{i,\widetilde{L}(+)}^{\neg}| \le \frac{\eps^8 \delta^3 \eta^5 m}{128\beta}$.

Let $\widetilde{\mathrm{ALG}}:= \widetilde{k}_{\widetilde{L}} + \sum_{i\in \widetilde{H}}\max\{\widetilde{k}_{i,\widetilde{L}(+)}^{\pos}, \widetilde{k}_{i,\widetilde{L}(+)}^{\neg}\}  - \frac{\eta m}{32}$ denote the output of the algorithm obtained by running \cref{alg:max-2and-offline} with $L$ and $H$ replaced by $\widetilde L$ and $\widetilde H$, respectively.
Analogous to the proof of \cref{thm:max-2and-offline}, we can show that $(1-\frac{\eta}{2})\cdot \opt\le \widetilde{\mathrm{ALG}} \le \opt$ with probability at least $1-\frac{3\delta}{4}$. The argument relies primarily on two properties of the partition $(\widetilde{L},\widetilde{H})$: (i) every variable in $\widetilde{L}$ has degree at most $\theta m$ for $\theta:=\frac{\eps^4 \delta \eta^2}{4096}$ (since $\widetilde{L} \subseteq L$), and (ii) $|\widetilde{H}|\leq \frac{2\beta}{\eps^8 \delta^3 \eta^4}$.

Let $\mathrm{ALG}$ denote the output of \cref{alg:max-2and-insertion}. By applying the union bound on Lemma~\ref{lem:high-deg-reservoir-max-2and} and Lemma~\ref{lem:error-bound-insertion-max-2and}, with probability at least $1-\frac{\delta}{8}-\frac{\delta}{8}-\frac{3\delta}{4} = 1-\delta$, we have
\begin{align*}
    \mathrm{ALG} &= \overline{k}_{\widetilde{L}} + \sum_{i\in \widetilde{H}}\max\{\overline{k}_{i,\widetilde{L}(+)}^{\pos}, \overline{k}_{i,\widetilde{L}(+)}^{\neg}\} - \frac{\eta m}{16}\\
    &\le \widetilde{k}_{\widetilde{L}} + \frac{\eta m}{64} + \sum_{i\in \widetilde{H}}\max\{\widetilde{k}_{i,\widetilde{L}(+)}^{\pos}, \widetilde{k}_{i,\widetilde{L}(+)}^{\neg}\} + |\widetilde{H}|\cdot \frac{\eps^8 \delta^3 \eta^5 m}{128\beta} - \frac{\eta m}{16}\\
    &\le \widetilde{\mathrm{ALG}} + \frac{\eta m}{32}+\frac{\eta m}{64}+\frac{\eta m}{64} - \frac{\eta m}{16} &&\rhd |\widetilde{H}|\le \frac{2\beta}{\eps^8 \delta^3 \eta^4}\\
    &\le \opt,
\end{align*}
and
\begin{align*}
    \mathrm{ALG} &= \overline{k}_{\widetilde{L}} + \sum_{i\in \widetilde{H}}\max\{\overline{k}_{i,\widetilde{L}(+)}^{\pos}, \overline{k}_{i,\widetilde{L}(+)}^{\neg}\} - \frac{\eta m}{16}\\
    &\ge \widetilde{k}_{\widetilde{L}} - \frac{\eta m}{32} + \sum_{i\in \widetilde{H}}\max\{\widetilde{k}_{i,\widetilde{L}(+)}^{\pos}, \widetilde{k}_{i,\widetilde{L}(+)}^{\neg}\} - |\widetilde{H}|\cdot \frac{\eps^8 \delta^3 \eta^5 m}{128\beta} - \frac{\eta m}{16}\\
    &\ge \widetilde{\mathrm{ALG}} -\frac{\eta m}{64} - \frac{\eta m}{16} &&\rhd |\widetilde{H}|\le \frac{2\beta}{\eps^8\delta^3\eta^4}\\
    &\ge \widetilde{\mathrm{ALG}} - \frac{\eta m}{8} \\
    &\ge \qty(1-\frac{\eta}{2})\cdot \opt - \frac{\eta}{2} \cdot\opt &&\rhd m\le 4\opt \\
    &= (1-\eta)\cdot \opt.
\end{align*}
\end{proof}

\subsection{Extension to dynamic streams}
\label{sec:dynamic-2and}
Our algorithm for \textsc{Max-$2$AND} in dynamic streams can be easily obtained by modifying our algorithm in insertion-only streams (i.e., \cref{alg:max-2and-insertion}). Specifically, we use $\ell_0$-sampling~\citep{jowhari2011tight} 
to compute a small set $\widetilde{H}$ at the end of the stream such that $\widetilde{H} \supseteq H$.

The algorithm is formally given in Algorithm~\ref{alg:max-2and-dynamic}. 
We have the following theoretical guarantee stating its approximation ratio and space complexity. We defer the proof to Appendix~\ref{sec:proof-thm:max-2and-dynamic}.

\begin{algorithm}[t!]
\begin{algorithmic}[1]
\Require $\textsc{Max-$2$AND}$ instance $\Psi$ as a dynamic stream of constraints, an $\eps$-accurate oracle $\mathcal{O} \rightarrow \{-1, 1\}$, $\delta,\eta\in (0,1)$
\Ensure Estimate of the maximum value of $\Psi$
\LeftComment{\textbf{\textcolor{blue}{Pre-processing phase}}}
\State {\bf initialize} $X_{+,+}\leftarrow 0, X_{+,-}\leftarrow 0, X_{-,-}\leftarrow 0, m\leftarrow 0$
\State {\bf initialize} four CountMin sketches $\CMpos[+], \CMneg[+], \CMpos[-]$ and $\CMneg[-]$ with accuracy parameter $\frac{\eps^{10}\delta^5 \eta^6}{128\beta}$ and failure probability $\frac{\eps^{8}\delta^4\eta^4}{64\beta}$, where $\beta = 2^{29}\cdot 25$.

\LeftComment{\textbf{\textcolor{blue}{Streaming phase}}}
\ForEach{item $((\mathbf{j}=(j_1,j_2),\mathbf{b}=(b_1,b_2)), \Delta_{(\mathbf{j},\mathbf{b})})$ in the dynamic stream}
    \State $Y_{j_1} = \mathcal{O}(j_1), Y_{j_2} = \mathcal{O}(j_2)$
    \If{$\Delta_{(\mathbf{j},\mathbf{b})}=1$}
        \State $m \leftarrow m+1$
        \IIf{$b_1 Y_{j_1} = b_2 Y_{j_2} = 1$} $X_{+,+} \leftarrow X_{+,+} + 1$
        \IIf{$b_1 Y_{j_1} \neq b_2 Y_{j_2}$} $X_{+,-} \leftarrow X_{+,-} + 1$
        \IIf{$b_1 Y_{j_1} = b_2 Y_{j_2} = -1$} $X_{-,-} \leftarrow X_{-,-} + 1$
    \Else{}
        \State $m \leftarrow m-1$
        \IIf{$b_1 Y_{j_1} = b_2 Y_{j_2} = 1$} $X_{+,+} \leftarrow X_{+,+} - 1$
        \IIf{$b_1 Y_{j_1} \neq b_2 Y_{j_2}$} $X_{+,-} \leftarrow X_{+,-} - 1$
        \IIf{$b_1 Y_{j_1} = b_2 Y_{j_2} = -1$} $X_{-,-} \leftarrow X_{-,-} - 1$
    \EndIf
    \IIf{$b_1 = 1$}
          $\textsc{Update} (\CMpos[b_2 Y_{j_2}][j_1],\Delta_{(\mathbf{j},\mathbf{b})})$
    \ElseIIf
        $\textsc{Update} (\CMneg[b_2 Y_{j_2}][j_1],\Delta_{(\mathbf{j},\mathbf{b})})$

    \IIf{$b_2 = 1$}
          $\textsc{Update} (\CMpos[b_1 Y_{j_1}][j_2],\Delta_{(\mathbf{j},\mathbf{b})})$
    \ElseIIf
         $\textsc{Update} (\CMneg[b_1 Y_{j_1}][j_2],\Delta_{(\mathbf{j},\mathbf{b})})$
\EndFor

\State In parallel, maintain $\frac{\beta}{\eps^{8}\delta^3\eta^4}$ $\ell_0$-samplers (with failure probability $\frac{\eps^8 \delta^4 \eta^4}{32\beta}$) to obtain $\frac{\beta}{\eps^{8}\delta^3\eta^4}$ different constraints $F$.\phantomsection \label{line:ell-0-sampling-2and}

\LeftComment{\textbf{\textcolor{blue}{Post-processing phase}}}
\ForEach{$i\in [n]$}
    \State $\widetilde{X}_{i,+}^{\pos} \leftarrow \textsc{Query}(\CMpos[+][i]), \quad \widetilde{X}_{i,+}^{\neg} \leftarrow \textsc{Query}(\CMneg[+][i])$
    \State $\widetilde{X}_{i,-}^{\pos} \leftarrow \textsc{Query}(\CMpos[-][i]),\quad \widetilde{X}_{i,-}^{\neg} \leftarrow \textsc{Query}(\CMneg[-][i])$
\EndFor
\State $\widetilde{H} \leftarrow \bigcup_{((j_1,j_2),(b_1,b_2))\in F} \{j_1,j_2\}, \widetilde{L} \leftarrow [n]\setminus \widetilde{H}$
\State $\widetilde{H}^+ \leftarrow \{i\in \widetilde{H}: Y_i = 1\}, \widetilde{H}^- \leftarrow \widetilde{H} \setminus \widetilde{H}^+$ 
\State $\overline{X}_0 \leftarrow X_{+,+} - \sum_{i\in \widetilde{H}^+} \widetilde{X}^{\pos}_{i,+}-\sum_{i\in \widetilde{H}^-} \widetilde{X}^{\neg}_{i,+}$
\State $\overline{X}_1 \leftarrow X_{+,-} - \sum_{i\in \widetilde{H}^+}(\widetilde{X}^{\pos}_{i,-}+\widetilde{X}^{\neg}_{i,+}) - \sum_{i\in \widetilde{H}^-}(\widetilde{X}^{\pos}_{i,+}+\widetilde{X}^{\neg}_{i,-})$ 
\State $\overline{X}_2 \leftarrow X_{-,-} - \sum_{i\in \widetilde{H}^+} \widetilde{X}^{\neg}_{i,-}-\sum_{i\in \widetilde{H}^-} \widetilde{X}^{\pos}_{i,-}$
\State $\overline{k}_{\widetilde{L}}\leftarrow \frac{1}{4\eps^2}[(\frac{1}{2}+\eps)^2 \overline{X}_0 - (\frac{1}{2}+\eps)(\frac{1}{2}-\eps)\overline{X}_1 + (\frac{1}{2}-\eps)^2 \overline{X}_2]$
\ForEach{$i\in \widetilde{H}$} 
\State $\overline{k}_{i,\widetilde{L}(+)}^{\pos} \leftarrow (\frac{1}{2}+\frac{1}{4\eps})\cdot \widetilde{X}_{i,+}^{\pos} + (\frac{1}{2}-\frac{1}{4\eps})\cdot \widetilde{X}_{i,-}^{\pos},\quad \overline{k}_{i,\widetilde{L}(+)}^{\neg} \leftarrow (\frac{1}{2}+\frac{1}{4\eps})\cdot \widetilde{X}_{i,+}^{\neg} + (\frac{1}{2}-\frac{1}{4\eps})\cdot \widetilde{X}_{i,-}^{\neg}$ 
\EndFor
\State \textbf{return} $\overline{k}_{\widetilde{L}} + \sum_{i\in \widetilde{H}}\max\{\overline{k}_{i,\widetilde{L}(+)}^{\pos}, \overline{k}_{i,\widetilde{L}(+)}^{\neg}\}  -  \frac{\eta m}{16}$ 
\end{algorithmic}
\caption{Estimating the \textsc{Max-$2$AND} value in dynamic streams}
\label{alg:max-2and-dynamic}
\end{algorithm}

\begin{theorem}
\label{thm:max-2and-dynamic}
Let $\eps\in (0,\frac{1}{2}]$ and $\delta \in (0,1)$. For any $\eta \in (0,1)$,
given an $\eps$-accurate oracle $\mathcal{O}$ for a \textsc{Max-$2$AND} instance $\Psi$, 
there exists a single-pass $(1-\eta)$-approximation algorithm for estimating the maximum value of $\Psi$ in dynamic streams. The algorithm uses $O(\frac{\log^2 n}{\eps^{10}\delta^5\eta^6}\log\frac{1}{\eps \delta\eta})$ bits of space. The approximation holds with probability at least $1-\delta$.
\end{theorem}

\section{Our algorithms for \textsc{Max-$k$AND} with $k\ge 3$}\label{sec:max-dicut-kand}
In this section, we extend our algorithms for $\textsc{Max-$2$AND}$ to $\textsc{Max-$k$AND}$ with $k\ge 3$.
\cref{sec:bounded-max-kand} gives a bounded-degree warm-up, and \cref{sec:unbounded-max-kand} gives the unbounded-degree algorithm used in the main theorem.
Here, the notion of variable degree is analogous to that of vertex degree in graphs: 
for each $i\in [n]$, we (slightly abusing graph terminology) let $\deg(i)$ denote the number of constraints in which $x_i$ appears, either positively as $x_i$ or negatively as $-x_i$.

\subsection{Bounded-degree instances}
\label{sec:bounded-max-kand}
In this section, we present the streaming algorithm for \textsc{Max-$k$AND} with $k\ge 3$ when all variables have bounded degree.
Specifically, we maintain $k+1$ constraint-type counters $X_0,X_1,\dots,X_k$ during the stream. For each $i\in \{0,\dots,k\}$, $X_i$ records the number of constraints for which exactly $i$ literals are \emph{not} satisfied by the predicted assignment.
At the end of the stream, we use these counters to decode an estimate of the optimal \textsc{Max-$k$AND} value by solving a linear system.
Refer to~\cref{alg:bounded-max-kand} for a formal description. 

Note that although the algorithm is presented for insertion-only streams, it extends directly to dynamic streams: when a constraint deletion arrives, we simply decrement the corresponding counters.
In addition, we only require the bounded-degree assumption to hold at the end of the stream (i.e., for the final instance).
\begin{algorithm}[t!]
\begin{algorithmic}[1]
\Require $\textsc{Max-$k$AND}$ instance $\Psi$ with maximum degree $\Delta < \frac{\eps^{2k} \delta\eta^2 m}{4k(k+1)}$ as a stream of constraints, an $\eps$-accurate oracle $\mathcal{O} \rightarrow \{-1, 1\}$
\Ensure Estimate of the maximum value of $\Psi$
\LeftComment{\textbf{\textcolor{blue}{Pre-processing phase}}}
\State {\bf initialize} $X_0,X_1,\dots,X_k \leftarrow 0, m\leftarrow 0$.
\LeftComment{\textbf{\textcolor{blue}{Streaming phase}}}
\ForEach{constraint $((j_1,\dots,j_k),(b_1,\dots,b_k))$ in the stream}
    \State $m \leftarrow m +1$
    \State $Y_{j_1} = \mathcal{O}(j_1), \dots, Y_{j_k} = \mathcal{O}(j_k)$
    \State $i \leftarrow |\{t\in [k]: b_t Y_{j_t} \neq 1\}|$
    \State $X_i \leftarrow X_i + 1$
    
\EndFor
\LeftComment{\textbf{\textcolor{blue}{Post-processing phase}}}
\State \textbf{return} $\frac{1}{(2\eps)^k}\sum_{i=0}^k (\frac{1}{2}+\eps)^{k-i} (-\frac{1}{2}+\eps)^i  X_i - \frac{\eta m}{2^{k+1}}$
\end{algorithmic}
\caption{A learning-augmented streaming algorithm for estimating the \textsc{Max-$k$AND} value for $k\ge 3$ in bounded-degree instances}
\label{alg:bounded-max-kand}
\end{algorithm}

We have the following performance guarantee regarding our streaming algorithms for bounded-degree \textsc{Max-$k$AND} instances with $k\geq 3$.

\begin{theorem}
\label{thm:bounded-max-kand}
Let $\eps \in (0, \frac{1}{2}], \delta \in (0,1)$, and $k\ge 3$. For any $\eta\in (0,1)$, given an $\eps$-accurate oracle $\mathcal{O}$ for a \textsc{Max-$k$AND} instance $\Psi$, if the maximum variable degree $\Delta$ satisfies $\Delta < \frac{\eps^{2k} \delta\eta^2 m}{4k(k+1)}$, then there exists a single-pass algorithm that, with probability at least $1-\delta$, outputs a $(1-\eta)$-approximation to the maximum value of $\Psi$ in both the insertion-only and dynamic streaming models.  
The algorithm uses $O(k)$ words of space. 
\end{theorem}

For each $i\in \{0,\dots,k\}$, let $X_i$ denote the random variable given by the final value of the corresponding counter at the end of the algorithm.
We first prove concentration bounds for $X_i$'s around their expectations.
\begin{lemma}
\label{lem:concentration-max-kand}
Let $\eps \in (0,\frac{1}{2}]$ and $\delta \in (0,1)$.
For each $i\in \{0,\dots,k\}$, we have $|X_i - \E[X_i]| \le \frac{\eps^k \eta m}{2}$ with probability at least $1-\frac{\delta}{k+1}$.
\end{lemma}
\begin{proof}
    For each constraint $(\mathbf{j}=(j_1,\dots,j_k),\mathbf{b}=(b_1,\dots,b_k))\in \Psi$, define the following indicator random variables:
    \begin{align*}
    X_{i,\mathbf{j}} =
    \begin{cases}
    1 & \text{if } |\{t\in [k]: b_t Y_{j_t} \neq 1\}|=i,\\
    0 & \text{otherwise},
    \end{cases}
    \end{align*}
    for $i\in \{0,\dots,k\}$.
    Note that we have $\E[X_{i,\mathbf{j}}] = \Pr[X_{i,\mathbf{j}} = 1]$ and $X_i=\sum_{(\mathbf{j},\mathbf{b})\in \Psi}X_{i,\mathbf{j}}$. Therefore,
    \begin{align*}
        \E[X_i^2] &= \E\qty[ \qty(\sum_{(\mathbf{j},\mathbf{b})\in \Psi} X_{i,\mathbf{j}})^2]
        = \E\qty[\sum_{(\mathbf{j},\mathbf{b})\in \Psi}\sum_{(\mathbf{j}',\mathbf{b}')\in \Psi}X_{i,\mathbf{j}} X_{i,\mathbf{j}'}]\\
        &=\sum_{(\mathbf{j},\mathbf{b})\in \Psi}\sum_{(\mathbf{j}',\mathbf{b}')\in \Psi}\E[X_{i,\mathbf{j}}X_{i,\mathbf{j}'}]\\ 
        &=\sum_{(\mathbf{j},\mathbf{b})\in \Psi}\sum_{(\mathbf{j}',\mathbf{b}')\in \Psi} \Pr[X_{i,\mathbf{j}} = 1\wedge X_{i,\mathbf{j}'} = 1]\\
        &= \sum_{\substack{(\mathbf{j},\mathbf{b}),(\mathbf{j}',\mathbf{b}')\in \Psi:\\ \{j_1,\dots,j_k\}\cap \{j'_1,\dots,j'_k\} = \emptyset}} \Pr[X_{i,\mathbf{j}} = 1\wedge X_{i,\mathbf{j}'} = 1] + \sum_{\substack{(\mathbf{j},\mathbf{b}),(\mathbf{j}',\mathbf{b}')\in \Psi:\\ \{j_1,\dots,j_k\}\cap \{j'_1,\dots,j'_k\} \neq \emptyset}} \Pr[X_{i,\mathbf{j}} = 1\wedge X_{i,\mathbf{j}'} = 1]\\
        &\leq \sum_{\substack{(\mathbf{j},\mathbf{b}),(\mathbf{j}',\mathbf{b}')\in \Psi:\\ \{j_1,\dots,j_k\}\cap \{j'_1,\dots,j'_k\} = \emptyset}} \Pr[X_{i,\mathbf{j}} = 1] \Pr[X_{i,\mathbf{j}'} = 1] + \sum_{\substack{(\mathbf{j},\mathbf{b}),(\mathbf{j}',\mathbf{b}')\in \Psi:\\ \{j_1,\dots,j_k\}\cap \{j'_1,\dots,j'_k\} \neq \emptyset}} 1\\
        &\leq \sum_{(\mathbf{j},\mathbf{b})\in \Psi}\sum_{(\mathbf{j}',\mathbf{b}')\in \Psi}\E[X_{i,\mathbf{j}}]\E[X_{i,\mathbf{j}'}] + m\cdot k\cdot \Delta \\
        &\leq \qty(\E\qty[\sum_{(\mathbf{j},\mathbf{b})\in \Psi}X_{i,\mathbf{j}}])^2 + k\Delta m\\
        &= (\E[X_i])^2 + k\Delta m.
    \end{align*}
    
    Then we have $\Var[X_i] = \E[X_i^2] - (\E[X_i])^2 \le k\Delta m$.
    By Chebyshev's inequality, 
    \begin{align*}
        \Pr[|X_i - \E[X_i]| > \frac{\eps^k \eta m}{2}] 
        \le \frac{\Var[X_i]}{(\eps^{k}\eta m/2)^2}
        &\le \frac{4k\Delta m}{\eps^{2k}\eta^2 m^2} &&\rhd \Var[X_i] \le k\Delta m\\
        &= \frac{4k\Delta}{\eps^{2k}\eta^2 m}
        < \frac{\delta}{k+1}.  &&\rhd \Delta < \frac{\eps^{2k} \delta\eta^2 m}{4k(k+1)}
    \end{align*}
\end{proof}
\paragraph{Proof of \cref{thm:bounded-max-kand}.}
Now we are ready to prove \cref{thm:bounded-max-kand}. 

\begin{proof}
    We first analyze the space complexity.
    In \cref{alg:bounded-max-kand}, we mainly maintain $k+1$ counters: $X_0, \dots, X_{k}$, corresponding to the constraints categorized by the predicted signs of their associated variables. Each counter requires $O(\log n)$ bits of space, since $0\le X_i \le m < 2^k n^k$ for $i\in \{0,\dots,k\}$.
    Therefore, the total space complexity of the algorithm is $O(k)$ words. 

    Next, we show that \cref{alg:bounded-max-kand} with high probability outputs a $(1-\eta)$-approximation of the \textsc{Max-$k$AND} value when all variables have bounded degrees.
    
Let $\boldsymbol{\sigma}^*\in \{\pm 1\}^n$ denote an optimal assignment for the \textsc{Max-$k$AND} instance $\Psi$.
For $i\in \{0,\dots,k\}$,
let $k_i:= |\{(\mathbf{j}=(j_1,\dots,j_k),\mathbf{b}=(b_1,\dots,b_k))\in \Psi:  |\{t\in [k]: b_t \sigma^*_{j_t} \neq 1\}|=i\}|$ denote the number of constraints in $\Psi$ for which exactly $i$ of the $k$ literals are not satisfied by $\boldsymbol{\sigma}^*$. Note that $k_0 = \opt$.
According to \cref{alg:bounded-max-kand}, for each $i\in \{0,\dots,k\}$, we have
\begin{align*}
    \E[X_i] = \sum_{j=0}^k k_j \sum_{s=\max\{0,i-(k-j)\}}^{\min\{i,j\}}\binom{j}{s}\binom{k-j}{i-s}\left(\frac{1}{2}+\eps\right)^{k-(i+j-2s)} \left(\frac{1}{2}-\eps\right)^{i+j-2s}.
\end{align*}
Therefore, by Gaussian elimination, we have 
\begin{align}
    \opt = k_0 = \frac{1}{(2\eps)^k}\sum_{i=0}^k \left(\frac{1}{2}+\eps\right)^{k-i} \left(-\frac{1}{2}+\eps\right)^i  \E[X_i].
\label{eq:gaussian-opt-kand}
\end{align}

Let $\mathrm{ALG}$ denote the output of \cref{alg:bounded-max-kand}.
By applying union bound on Lemma~\ref{lem:concentration-max-kand}, we have that with probability at least $1-\frac{\delta}{k+1}\cdot (k+1)= 1-\delta$, $|X_i - \E[X_i]| \le \frac{\eps^k \eta m}{2}$ for all $i\in \{0,\dots,k\}$. Therefore, 
\begin{align*}
    \mathrm{ALG} &= \frac{1}{(2\eps)^k}\sum_{i=0}^k \left(\frac{1}{2}+\eps\right)^{k-i} \left(-\frac{1}{2}+\eps\right)^i  X_i - \frac{\eta m}{2^{k+1}}\\
    &= \frac{1}{(2\eps)^k}\sum_{i=0}^k \left(\frac{1}{2}+\eps\right)^{k-i} \left(-\frac{1}{2}+\eps\right)^i  \left(\E[X_i]\pm\frac{\eps^k \eta m}{2}\right) - \frac{\eta m}{2^{k+1}}\\
    &\le \opt + \frac{\eps^{k}\eta m}{2\cdot (2\eps)^k}\sum_{i=0}^k \left|\left(\frac{1}{2}+\eps\right)^{k-i} \left(-\frac{1}{2}+\eps\right)^i\right| - \frac{\eta m}{2^{k+1}} && \rhd \text{\cref{eq:gaussian-opt-kand}}\\
    &\le \opt + \frac{\eta m}{2^{k+1}} - \frac{\eta m}{2^{k+1}} && \rhd \sum_{i=0}^k \left|\left(\frac{1}{2}+\eps\right)^{k-i} \left(-\frac{1}{2}+\eps\right)^i\right| \le 1\\
    &= \opt.
\end{align*}

On the other hand,

\begin{align*}
    \mathrm{ALG} &= \frac{1}{(2\eps)^k}\sum_{i=0}^k \left(\frac{1}{2}+\eps\right)^{k-i} \left(-\frac{1}{2}+\eps\right)^i  X_i - \frac{\eta m}{2^{k+1}}\\
    &= \frac{1}{(2\eps)^k}\sum_{i=0}^k \left(\frac{1}{2}+\eps\right)^{k-i} \left(-\frac{1}{2}+\eps\right)^i  \left(\E[X_i]\pm\frac{\eps^k \eta m}{2}\right) - \frac{\eta m}{2^{k+1}}\\
    &\ge \opt - \frac{\eps^{k}\eta m}{2\cdot (2\eps)^k}\sum_{i=0}^k \left|\left(\frac{1}{2}+\eps\right)^{k-i} \left(-\frac{1}{2}+\eps\right)^i\right| - \frac{\eta m}{2^{k+1}} && \rhd \text{\cref{eq:gaussian-opt-kand}}\\
    &\ge \opt - \frac{\eta m}{2^{k+1}} - \frac{\eta m}{2^{k+1}} && \rhd \sum_{i=0}^k \left|\left(\frac{1}{2}+\eps\right)^{k-i} \left(-\frac{1}{2}+\eps\right)^i\right| \le 1\\
    &= \opt - \frac{\eta m}{2^k}\\
    &\ge (1-\eta)\cdot \opt. &&\rhd \opt \ge \frac{m}{2^k}
\end{align*}

So, with probability at least $1-\delta$, \cref{alg:bounded-max-kand} is a $(1-\eta)$-approximation.
\end{proof}

The bounded-degree algorithm above is a warm-up. The following subsection gives the arbitrary-degree \textsc{Max-$k$AND} algorithm, still under the simple, unweighted convention from \cref{problem:kcsp}.

\subsection{Unbounded-degree instances}
\label{sec:unbounded-max-kand}
We now remove the bounded-degree assumption for \textsc{Max-$k$AND}.  
The resulting space dependence on $k$, $1/\eps$, $1/\eta$, and $\log(1/\delta)$ is
exponential.
However, the space is independent of $n$ in insertion-only streams and
polylogarithmic in $n$ in dynamic streams.

We use the same high-degree and low-degree notation as the preceding
sections.  Whenever $H\subseteq[n]$ is fixed, let $L=[n]\setminus H$.  For an
assignment $\boldsymbol{\sigma}\in\{\pm1\}^n$, consider the decoded value
\begin{align}
\sum_{(\mathbf{j},\mathbf{b})\in\Psi}
\prod_{\substack{t\in[k]\\j_t\in H}}
\frac{1+b_t\sigma_{j_t}}{2}
\prod_{\substack{t\in[k]\\j_t\in L}}
\left(\frac12+\frac{b_tY_{j_t}}{4\eps}\right).
\label{eq:unbounded-max-kand-decoder}
\end{align}
The first product evaluates the literals indexed by $H$ exactly under
$\boldsymbol{\sigma}$.  Each factor in the second product is an unbiased
estimate of the corresponding literal under the fixed optimal assignment,
because
\begin{align}
\E\left[\frac12+\frac{b_tY_{j_t}}{4\eps}\right]
=
\frac12+\frac{b_t\sigma_{j_t}^*}{2}
=
\frac{1+b_t\sigma_{j_t}^*}{2}.
\label{eq:unbounded-max-kand-one-literal}
\end{align}
The displayed value depends on $\boldsymbol{\sigma}$ only through its
coordinates indexed by $H$.

\subsubsection{Offline algorithm}
\label{sec:offline-unbounded-max-kand}

The high-degree set is grown adaptively.  When it currently contains $|H|$
variables, another variable is added only if its degree exceeds
$\theta m/(|H|+1)$.  The harmonic-series bound keeps $H$ small, while the
factor $(|H|+1)^{-1}$ makes the decoded values concentrate simultaneously over
all assignments to the variables indexed by $H$.

\begin{algorithm}[t!]
\caption{Offline learning-augmented estimator for unbounded-degree
\textsc{Max-$k$AND}}
\label{alg:offline-unbounded-max-kand}
\begin{algorithmic}[1]
\Require \textsc{Max-$k$AND} instance $\Psi=(C_1,\dots,C_m)$, an
$\eps$-accurate oracle $\mathcal{O}\rightarrow\{-1,1\}$, and
$\delta,\eta\in(0,1)$
\Ensure Estimate of the maximum value of $\Psi$
\If{$m=0$}
    \State \textbf{return} $0$
\EndIf
\State $\theta\leftarrow \frac{\eps^{2k}\eta^2}{4k\log(4/\delta)}$
\State Compute $\deg(i)$ for every variable appearing in $\Psi$.
\State $H\leftarrow\emptyset$.
\State Repeatedly add any $i\notin H$ satisfying
$\deg(i)>\frac{\theta m}{|H|+1}$, until no such variable remains.
\State $L\leftarrow[n]\setminus H$.
\State Query $Y_i=\mathcal{O}(i)$ for every $i\in L$ with $\deg(i)>0$.
\State \textbf{return}
\[
\max_{\boldsymbol{\sigma}\in\{\pm1\}^n}
\sum_{(\mathbf{j},\mathbf{b})\in\Psi}
\prod_{\substack{t\in[k]\\j_t\in H}}
\frac{1+b_t\sigma_{j_t}}{2}
\prod_{\substack{t\in[k]\\j_t\in L}}
\left(\frac12+\frac{b_tY_{j_t}}{4\eps}\right)
-
\frac{\eta m}{2^{k+1}}.
\]
\Statex The maximized expression depends only on the coordinates indexed by
$H$, so the maximum is evaluated over $2^{|H|}$ assignments.
\end{algorithmic}
\end{algorithm}

\begin{theorem}
\label{thm:offline-unbounded-max-kand}
Let $\eps\in(0,\frac12]$, $\delta,\eta\in(0,1)$, and $k\ge3$.  Consider an
instance $\Psi$ of \textsc{Max-$k$AND} and an $\eps$-accurate oracle
$\mathcal{O}$ for $\Psi$.  Then \cref{alg:offline-unbounded-max-kand} outputs
a $(1-\eta)$-approximation to the maximum value of $\Psi$ with probability at
least $1-\delta$.  Its running time is
\(
km\cdot
\exp(\exp(O(k^2\eps^{-2k}\eta^{-2}\log(4/\delta)))).
\)
\end{theorem}

\begin{lemma}
\label{lem:offline-unbounded-max-kand-size-H}
Assume $m>0$, and let $H$ be the set computed by
\cref{alg:offline-unbounded-max-kand}.  Then
\[
|H|\le
\exp({4k^2\eps^{-2k}\eta^{-2}\log(4/\delta)})-1.
\]
\end{lemma}

\begin{proof}
The claim is immediate if $H=\emptyset$.  Otherwise, let $i_s$ be the $s$-th
variable added to $H$, for $s\in[|H|]$.  Immediately before $i_s$ is added,
the set contains $s-1$ variables, and therefore
\(
\deg(i_s)>\frac{\theta m}{s}
\).
Every constraint contains exactly $k$ distinct variables, so
\[
km=\sum_{i\in[n]}\deg(i)
\ge
\sum_{s=1}^{|H|}\deg(i_s)
>
\theta m\sum_{s=1}^{|H|}\frac1s.
\]
Since $m>0$, the standard harmonic bound gives
\(
\log(|H|+1)
\le
\sum_{s=1}^{|H|}\frac1s
<
\frac{k}{\theta}
\).
Substituting the value of $\theta$ and exponentiating proves the claim.
\end{proof}

\begin{lemma}
\label{lem:unbounded-max-kand-expectation}
Fix $H\subseteq[n]$, let $L=[n]\setminus H$, and fix an assignment
$\boldsymbol{\sigma}\in\{\pm1\}^n$.  Then the expectation of
(\ref{eq:unbounded-max-kand-decoder}) is
\[
\sum_{(\mathbf{j},\mathbf{b})\in\Psi}
\prod_{\substack{t\in[k]\\j_t\in H}}
\frac{1+b_t\sigma_{j_t}}{2}
\prod_{\substack{t\in[k]\\j_t\in L}}
\frac{1+b_t\sigma_{j_t}^*}{2},
\]
which is at most $\opt$.  Equality holds when
$\boldsymbol{\sigma}=\boldsymbol{\sigma}^*$.  Consequently, the maximum over
$\boldsymbol{\sigma}$ of the expected decoded value is $\opt$.
\end{lemma}

\begin{proof}
The coordinates of every constraint are distinct, and the oracle labels are
independent.  Hence \cref{eq:unbounded-max-kand-one-literal} allows the
expectation to pass through the product over the variables indexed by $L$,
which gives the displayed identity.  The resulting sum is the value of the
assignment that uses $\boldsymbol{\sigma}$ on the variables indexed by $H$ and
$\boldsymbol{\sigma}^*$ on the variables indexed by $L$, so it is at most
$\opt$.  Taking $\boldsymbol{\sigma}=\boldsymbol{\sigma}^*$ gives equality.
\end{proof}

\begin{lemma}
\label{lem:unbounded-max-kand-concentration}
Let $\eps\in(0,\frac12]$, $\eta>0$, and $\theta>0$.  Let
$H\subseteq[n]$ and $L=[n]\setminus H$.  Suppose that
\(
\deg(i)\le\theta m/(|H|+1)
\) for every $i\in L$.
Then the probability that, for some assignment
$\boldsymbol{\sigma}\in\{\pm1\}^n$, the decoded value in
(\ref{eq:unbounded-max-kand-decoder}) differs from its expectation by at least
$\eta m/2^{k+1}$ is at most
\[
2^{|H|+1}
\exp\left(
-\eta^2\eps^{2k}(|H|+1)/(2k\theta)
\right).
\]
\end{lemma}

\begin{proof}
Fix $\boldsymbol{\sigma}$.  The decoded value is a function of the independent
oracle labels $(Y_i)_{i\in L}$.  Changing $Y_i$ affects only constraints
containing $x_i$.  In each such constraint, the factor containing $Y_i$
changes by $1/(2\eps)$, while every other decoded factor has absolute value at
most $(2\eps)^{-1}$ and every exact factor has absolute value at most $1$.
Thus, in the notation of \cref{thm:bounded-differences}, one may take
\(
c_i=(2\eps)^{-k}\deg(i)
\).
The residual degree bound and $\sum_{i\in L}\deg(i)\le km$ imply
\[
\sum_{i\in L}c_i^2
\le
(2\eps)^{-2k}
\frac{\theta m}{|H|+1}
\sum_{i\in L}\deg(i)
\le
\frac{k\theta m^2}{(2\eps)^{2k}(|H|+1)}.
\]
Applying \cref{thm:bounded-differences} with deviation
$\eta m/2^{k+1}$ gives, for the fixed assignment, failure probability at most
\(
2\exp\left(
-\frac{\eta^2\eps^{2k}(|H|+1)}{2k\theta}
\right)
\).
The decoded value depends only on the coordinates of
$\boldsymbol{\sigma}$ indexed by $H$.  A union bound over the $2^{|H|}$
possible choices of those coordinates proves the lemma.
\end{proof}

\paragraph{Proof of \cref{thm:offline-unbounded-max-kand}.}
Now we are ready to prove \cref{thm:offline-unbounded-max-kand}.
\begin{proof}
If $m=0$, the algorithm returns the exact value.  Assume $m>0$.
\Cref{lem:offline-unbounded-max-kand-size-H} gives the stated number of
assignments to enumerate.  Computing all degrees takes $O(km)$ time, and each
assignment is evaluated by one scan of the $m$ constraints, giving the stated
running time.

At termination, every $i\in L$ satisfies
$\deg(i)\le\theta m/(|H|+1)$.  By
\cref{lem:unbounded-max-kand-concentration} and the chosen value of $\theta$,
the probability that any decoded value differs from its expectation by more
than $\eta m/2^{k+1}$ is at most
\begin{align*}
2^{|H|+1}
\exp\left(-2(|H|+1)\log(4/\delta)\right)
&=
\left(\frac{\delta^2}{8}\right)^{|H|+1}
\le\delta.
\end{align*}
Condition on the complementary event, and let $\mathrm{ALG}$ denote the value
returned by the algorithm.  By
\cref{lem:unbounded-max-kand-expectation},
\[
\mathrm{ALG}
\le
\opt+\frac{\eta m}{2^{k+1}}-\frac{\eta m}{2^{k+1}}
=
\opt.
\]
Taking $\boldsymbol{\sigma}=\boldsymbol{\sigma}^*$ in the maximum gives
\[
\mathrm{ALG}
\ge
\opt-\frac{\eta m}{2^{k+1}}-\frac{\eta m}{2^{k+1}}
=
\opt-\frac{\eta m}{2^k}.
\]
A uniformly random assignment satisfies each \textsc{Max-$k$AND} constraint
with probability $2^{-k}$, so $\opt\ge m/2^k$.  Therefore
$\mathrm{ALG}\ge(1-\eta)\opt$.
\end{proof}

\subsubsection{Our algorithm in insertion-only streams}
\label{sec:streaming-unbounded-max-kand}

We next implement the offline estimator in one pass.  A reservoir sample $F$
provides candidate high-degree variables.  A CountMin sketch of the degree
vector, written using the frequency notation $f_i=\deg(i)$ from the CountMin
preliminaries, supplies estimates $\widetilde f_i$ for variables appearing in
$F$.  The set $\widetilde H$ is then the longest degree-ordered prefix
satisfying the adaptive threshold, and
$\widetilde L=[n]\setminus\widetilde H$.

A second CountMin sketch reconstructs the decoded value after $\widetilde H$
is known.  For every subset of positions in an arriving constraint, its key
records the number $s$ of selected positions, the number $i$ of unselected
literals not satisfied by the prediction, and the selected signed variables
in increasing variable-index order.  Thus all reconstruction counts are
coordinates of a single sketch.

For a fixed assignment $\boldsymbol{\sigma}$, the value reconstructed from
this sketch is
\begin{equation}
\begin{aligned}
&\sum_{s=0}^k
\sum_{\substack{j_1,\dots,j_s\in\widetilde H\\j_1<\cdots<j_s}}
\sum_{(b_1,\dots,b_s)\in\{\pm1\}^s}
\left[
\prod_{\ell=1}^s
b_\ell\left(
\frac{\sigma_{j_\ell}}2-\frac{Y_{j_\ell}}{4\eps}
\right)
\right]\\
&\qquad\cdot
\sum_{i=0}^{k-s}
\left(\frac12+\frac1{4\eps}\right)^{k-s-i}
\left(\frac12-\frac1{4\eps}\right)^i\\
&\qquad\cdot
\textsc{Query}\left(
\CM\left[(s,i,((j_1,b_1),\dots,(j_s,b_s)))\right]
\right).
\end{aligned}
\label{eq:streaming-unbounded-max-kand-reconstruction}
\end{equation}
For $s=0$, the signed-variable tuple and the product are empty, and the
product equals $1$.

The reservoir capacity is at least
\(
\lceil(8/\theta)\exp(8k/\theta)
\log((64/\delta)(k/\theta)\exp(8k/\theta))\rceil
\).  The degree sketch has accuracy parameter
\(
1/((8k/\theta)\exp(8k/\theta))
\)
and failure probability equal to $\delta/(16k)$ divided by the reservoir
capacity.  A union bound over all degree queries for variables appearing in
$F$ therefore contributes at most $\delta/16$.
The reconstruction sketch $\CM$ has accuracy parameter
\(
1/((8/\eta)(2/\eps)^k(k+1)(1+2\exp(8k/\theta))^k)
\)
and failure probability
\(
1/((16/\delta)(k+1)(1+2\exp(8k/\theta))^k)
\).
The reservoir, the two sketches, and the oracle labels use independent
randomness.

\begin{algorithm}[t!]
\caption{Insertion-only streaming estimator for unbounded-degree
\textsc{Max-$k$AND}}
\label{alg:streaming-unbounded-max-kand}
\begin{algorithmic}[1]
\Require \textsc{Max-$k$AND} instance $\Psi$ as an insertion-only stream of
constraints, an $\eps$-accurate oracle $\mathcal{O}\rightarrow\{-1,1\}$, and
$\delta,\eta\in(0,1)$
\Ensure Estimate of the maximum value of $\Psi$
\LeftComment{\textbf{\textcolor{blue}{Pre-processing phase}}}
\State $m\leftarrow0$ and
$\theta\leftarrow\frac{\eps^{2k}\eta^2}{64k\log(16/\delta)}$.
\State Maintain a reservoir sample $F$ with the capacity specified above; if
fewer constraints arrive, store all of them.
\State \textbf{initialize} the degree sketch and the reconstruction sketch
$\CM$ with the parameters specified above.

\LeftComment{\textbf{\textcolor{blue}{Streaming phase}}}
\ForEach{constraint $((j_1,\dots,j_k),(b_1,\dots,b_k))$ in the stream}
    \State $m\leftarrow m+1$ and query
    $Y_{j_t}=\mathcal{O}(j_t)$ for every $t\in[k]$.
    \State Update the degree sketch once at every coordinate $j_t$.
    \ForEach{$s\in\{0,\dots,k\}$ and every
    $1\le t_1<\cdots<t_s\le k$}
        \State $i\leftarrow
        \left|\left\{t\in[k]\setminus\{t_1,\dots,t_s\}:
        b_tY_{j_t}=-1\right\}\right|$.
        \State Sort $(j_{t_1},b_{t_1}),\dots,(j_{t_s},b_{t_s})$ by variable
        index and update $\CM$ at the key consisting of $s$, $i$, and this
        ordered tuple.
    \EndFor
    \State Update the reservoir sample $F$.
\EndFor

\LeftComment{\textbf{\textcolor{blue}{Post-processing phase}}}
\If{$m=0$}
    \State \textbf{return} $0$
\EndIf
\State For every variable $i$ appearing in a constraint in $F$, let
$\widetilde f_i$ be the estimate returned by the degree sketch.
\State Sort these variables in non-increasing order of $\widetilde f_i$.
\State Let $\widetilde H$ be the longest prefix such that the $s$-th variable
in the ordering has estimate greater than $\theta m/(2s)$ for every
$s\in[|\widetilde H|]$; set
$\widetilde L\leftarrow[n]\setminus\widetilde H$.
\State Query $Y_i=\mathcal{O}(i)$ for every $i\in\widetilde H$.
\State Evaluate (\ref{eq:streaming-unbounded-max-kand-reconstruction}) for all
$2^{|\widetilde H|}$ assignments to the variables indexed by
$\widetilde H$.
\State \textbf{return} the maximum reconstructed value minus
$\frac{\eta m}{2^{k+1}}$.
\end{algorithmic}
\end{algorithm}

\begin{theorem}
\label{thm:streaming-unbounded-max-kand}
Let $\eps\in(0,\frac12]$, $\delta\in(0,1)$, and $k\ge3$.  For any
$\eta\in(0,1)$, given an $\eps$-accurate oracle $\mathcal{O}$ for a
\textsc{Max-$k$AND} instance $\Psi$,
\cref{alg:streaming-unbounded-max-kand} outputs a $(1-\eta)$-approximation to
the maximum value of $\Psi$ in insertion-only streams with probability at
least $1-\delta$.  The algorithm uses
\(
\exp(O(k^3\eps^{-2k}\eta^{-2}\log(16/\delta)))
\)
words of space.
\end{theorem}

\begin{lemma}
\label{lem:streaming-unbounded-high-set}
Assume the final stream contains $m>0$ constraints, and let $\theta$ be the
value used in \cref{alg:streaming-unbounded-max-kand}.  With probability at
least $1-\delta/4$, the set $\widetilde H$ computed by the algorithm satisfies
\(
|\widetilde H|<\exp(8k/\theta)
\)
and
\(
\deg(i)\le\theta m/(|\widetilde H|+1)
\)
for every \(i\notin\widetilde H\).
\end{lemma}

\begin{proof}
Consider a variable $i$ with
\(
\deg(i)>m/((8/\theta)\exp(8k/\theta))
\).
If the stream is shorter than the reservoir capacity, every constraint is
stored.  Otherwise, the probability that no sampled constraint contains $i$
is at most
\begin{align*}
\left(1-\frac{\deg(i)}m\right)^{|F|}
\le
\exp\left(-|F|\frac{\deg(i)}m\right)
\le
\frac{\theta\delta}{64k\exp(8k/\theta)}.
\end{align*}
Because the total variable degree is $km$, at most
$(8k/\theta)\exp(8k/\theta)$ variables exceed the displayed degree threshold.
A union bound therefore shows that, with probability at least $1-\delta/8$,
every such variable appears in a sampled constraint.

The reservoir contains at most its stated capacity, so the algorithm queries
degree estimates for at most $k$ times that many variables.  By
\cref{thm:CM-v2} and the chosen failure probability, with probability at
least $1-\delta/16$ all queried variables simultaneously satisfy
\[
\deg(i)\le\widetilde f_i
\le
\deg(i)+m/((8/\theta)\exp(8k/\theta)).
\]
Condition on both events.  Their joint probability is at least
$1-3\delta/16\ge1-\delta/4$.

Let $i_s$ be the $s$-th variable in the non-increasing ordering used to form
$\widetilde H$.  Suppose for contradiction that
$|\widetilde H|\ge\exp(8k/\theta)$.  For every positive integer $s$
satisfying $s<\exp(8k/\theta)$, the prefix rule and the sketch error give
\[
\deg(i_s)
>
\frac{\theta m}{2s}
-
\frac{\theta m}{8\exp(8k/\theta)}
\ge
\frac{3\theta m}{8s}.
\]
Consequently,
\begin{align*}
km
\ge
\sum_{1\le s<\exp(8k/\theta)}\deg(i_s)
>
\frac{3\theta m}{8}
\sum_{1\le s<\exp(8k/\theta)}\frac1s
\ge
\frac{3\theta m}{8}
\log\left(\exp(8k/\theta)\right)
\ge3km,
\end{align*}
which is impossible.  Hence $|\widetilde H|<\exp(8k/\theta)$.

It remains to prove the residual degree bound.  If
$i\notin\widetilde H$ appears in a sampled constraint, maximality of the
prefix and the ordering by $\widetilde f_i$ imply
\[
\deg(i)\le\widetilde f_i
\le
\frac{\theta m}{2(|\widetilde H|+1)}
\le
\frac{\theta m}{|\widetilde H|+1}.
\]
If $i$ does not appear in a sampled constraint, the first conditioned event
gives
\(
\deg(i)\le\theta m/(8\exp(8k/\theta))
\).
Since $|\widetilde H|<\exp(8k/\theta)$, we have
$|\widetilde H|+1<\exp(8k/\theta)+1\le2\exp(8k/\theta)$, and the desired
bound follows.
\end{proof}

\begin{lemma}
\label{lem:streaming-unbounded-reconstruction}
Fix a set $\widetilde H$ satisfying
$|\widetilde H|<\exp(8k/\theta)$, and fix all oracle labels.  With probability
at least $1-\delta/4$ over the hash functions of $\CM$, simultaneously for
every assignment to the variables indexed by $\widetilde H$, the reconstructed
value in (\ref{eq:streaming-unbounded-max-kand-reconstruction}) differs from
(\ref{eq:unbounded-max-kand-decoder}), with $H=\widetilde H$ and
$L=\widetilde L$, by at most
\(
\frac{\eta m}{2^{k+2}}.
\)
\end{lemma}

\begin{proof}
Fix an assignment $\boldsymbol{\sigma}$.  For every literal whose variable is
indexed by $\widetilde H$, the exact factor can be written as
\begin{align}
\frac{1+b_t\sigma_{j_t}}2
=
\left(\frac12+\frac{b_tY_{j_t}}{4\eps}\right)
+
b_t\left(\frac{\sigma_{j_t}}2-\frac{Y_{j_t}}{4\eps}\right).
\label{eq:unbounded-max-kand-expansion}
\end{align}
Expand the product by selecting any number $s$ of the high-variable positions
from which the second term in
\cref{eq:unbounded-max-kand-expansion} is taken.  The remaining $k-s$ factors
are determined by their predicted signs: if exactly $i$ of them have
$b_tY_{j_t}=-1$, their product is
\(
\left(1/2+1/(4\eps)\right)^{k-s-i}
\left(1/2-1/(4\eps)\right)^i
\).
During the stream, $\CM$ receives exactly one update at the corresponding key
for this selected tuple of signed variables.  Therefore, if every CountMin
query in (\ref{eq:streaming-unbounded-max-kand-reconstruction}) were replaced
by its exact frequency, the displayed expansion would be summed once for every
constraint and would equal (\ref{eq:unbounded-max-kand-decoder}).

Each constraint updates exactly $2^k$ reconstruction keys, so the frequency
vector represented by $\CM$ has total mass $2^k m$.  By
\cref{thm:CM-v2} and the chosen accuracy parameter, every fixed query has
additive error at most
\[
\frac{\eta\eps^k m}
{8(k+1)(1+2\exp(8k/\theta))^k}
\]
with its prescribed failure probability.

The queried keys are independent of the assignment.  Their total number is at
most
\begin{align*}
\sum_{s=0}^k
(k-s+1)\binom{|\widetilde H|}{s}2^s
\le
(k+1)\sum_{s=0}^k(2|\widetilde H|)^s
\le
(k+1)(1+2\exp(8k/\theta))^k.
\end{align*}
A union bound and the chosen failure probability show that every query has the
preceding error bound simultaneously with probability at least
$1-\delta/16$, which is at least $1-\delta/4$.

Every factor in the coefficient of a query has absolute value at most
$(2\eps)^{-1}$.  There are exactly $k$ factors after combining the $s$
selected terms with the $k-s$ unselected terms.  On the simultaneous good
event, the reconstruction error for every assignment is therefore at most
\begin{align*}
(k+1)(1+2\exp(8k/\theta))^k
\cdot
\frac{\eta\eps^k m}
{8(k+1)(1+2\exp(8k/\theta))^k}
\cdot
(2\eps)^{-k}
=
\frac{\eta m}{2^{k+3}}
\le
\frac{\eta m}{2^{k+2}}.
\end{align*}
\end{proof}

\paragraph{Proof of \cref{thm:streaming-unbounded-max-kand}.}
Now we are ready to prove \cref{thm:streaming-unbounded-max-kand}.
\begin{proof}
If $m=0$, the algorithm returns the exact value.  Assume $m>0$.  By
\cref{lem:streaming-unbounded-high-set}, with probability at least
$1-\delta/4$ the computed set $\widetilde H$ satisfies the stated cardinality
bound and
\(
\deg(i)\le\theta m/(|\widetilde H|+1)
\)
for every \(i\in\widetilde L\).
Condition on the sampling and degree-sketch randomness that determines such a
set.  This randomness is independent of the oracle, so the labels indexed by
$\widetilde L$ remain mutually independent and $\eps$-accurate.

Apply \cref{lem:unbounded-max-kand-concentration} with
$H=\widetilde H$, $L=\widetilde L$, and $\eta/2$ in place of $\eta$.  Since
$\theta=\eps^{2k}\eta^2/(64k\log(16/\delta))$, the conditional failure
probability is at most
\begin{align*}
2^{|\widetilde H|+1}
\exp\left(-8(|\widetilde H|+1)\log(16/\delta)\right)
\le
\frac{\delta}{2}.
\end{align*}
Thus, except with conditional probability at most $\delta/2$, every decoded
value differs from its expectation by at most $\eta m/2^{k+2}$.

After conditioning also on the oracle labels,
\cref{lem:streaming-unbounded-reconstruction} contributes failure probability
at most $\delta/4$.  Combining the events, with probability at least
$1-\delta$, the sketch-computed value for every assignment differs from the
corresponding expected decoded value by at most $\eta m/2^{k+1}$.  By
\cref{lem:unbounded-max-kand-expectation}, the maximum expected decoded value
is $\opt$.  Therefore the returned value $\mathrm{ALG}$ satisfies
\[
\mathrm{ALG}
\le
\opt+\frac{\eta m}{2^{k+1}}-\frac{\eta m}{2^{k+1}}
=
\opt.
\]
Taking $\boldsymbol{\sigma}=\boldsymbol{\sigma}^*$ gives
\[
\mathrm{ALG}
\ge
\opt-\frac{\eta m}{2^{k+1}}-\frac{\eta m}{2^{k+1}}
=
\opt-\frac{\eta m}{2^k}
\ge
(1-\eta)\opt.
\]

It remains to bound the space.  The reservoir capacity and the inverse
accuracy of the degree sketch are bounded by
\(
\exp\left(
O\left(k^2\eps^{-2k}\eta^{-2}\log(16/\delta)\right)
\right)
\).
The inverse accuracy of $\CM$ is
\(
(8/\eta)\cdot(2/\eps)^k(k+1)(1+2\exp(8k/\theta))^k
\).
Substituting the value of $\theta$, and absorbing
all polynomial and logarithmic factors into the displayed exponential, gives
the stated bound.
\end{proof}

\subsubsection{Our algorithm in dynamic streams}
\label{sec:dynamic-unbounded-max-kand}

The dynamic algorithm uses the same value of $\theta$, the same sample size,
and the same degree and reconstruction sketches as
\cref{alg:streaming-unbounded-max-kand}.  The reservoir is replaced by that
many independent $\ell_0$-samplers on the active constraint vector.  In
post-processing, every variable appearing among their outputs is treated as a
sampled variable; repeated outputs remain independent draws.  The active
constraint count $m$, the degree sketch, and $\CM$ are updated linearly using
the stream-update sign $\Delta_i$.  Post-processing is otherwise identical to
the insertion-only algorithm.

\begin{theorem}
\label{thm:dynamic-streaming-unbounded-max-kand}
Let $\eps\in(0,\frac12]$, $\delta\in(0,1)$, and $k\ge3$.  For any
$\eta\in(0,1)$, given an $\eps$-accurate oracle $\mathcal{O}$ for a
\textsc{Max-$k$AND} instance $\Psi$, there is a single-pass dynamic streaming
algorithm that outputs a $(1-\eta)$-approximation to the maximum value of
$\Psi$ with probability at least $1-\delta$.  The algorithm uses
\(
\exp(
O(k^3\eps^{-2k}\eta^{-2}\log(16/\delta))
)
\log n
\)
words of space.
\end{theorem}

\begin{proof}
Maintain the number of independent $\ell_0$-samplers specified above, with
individual failure probabilities chosen so that their total failure
probability is at most $\delta/16$.  Maintain the degree sketch and $\CM$ with
the parameters in \cref{alg:streaming-unbounded-max-kand}, applying every
update with the sign $\Delta_i$ of the corresponding dynamic-stream update.

If the final active instance is empty, returning $0$ is exact.  Assume $m>0$.
With probability at least $1-\delta/16$, every $\ell_0$-sampler succeeds.
Conditioned on this event, the samplers give independent uniform samples, with
replacement, from the final active constraint set.  The sampling argument in
\cref{lem:streaming-unbounded-high-set} then contributes failure probability
at most $\delta/8$, and the degree-sketch event contributes at most
$\delta/16$.  Hence, with probability at least $1-\delta/4$, the computed set
$\widetilde H$ satisfies
\(
|\widetilde H|<\exp(8k/\theta)
\)
and
\(
\deg(i)\le\theta m/(|\widetilde H|+1)
\) for every \(i\in\widetilde L\).

Condition on such a fixed $\widetilde H$.  The samplers and the degree sketch
are independent of the oracle labels, so the labels indexed by
$\widetilde L$ remain mutually independent and $\eps$-accurate.  Applying
\cref{lem:unbounded-max-kand-concentration} with $\eta/2$ in place of $\eta$
gives failure probability at most $\delta/2$ for simultaneous additive error
$\eta m/2^{k+2}$ over all assignments.

After conditioning on the oracle labels, the final frequency vector represented
by $\CM$ is nonnegative.  The sketch is linear in the stream updates, so its
final table is exactly the CountMin table of this final vector.  Therefore
\cref{lem:streaming-unbounded-reconstruction} applies and contributes failure
probability at most $\delta/4$.  With total probability at least $1-\delta$,
every sketch-computed value differs from the corresponding expected decoded
value by at most $\eta m/2^{k+1}$.  The same comparison as in the
insertion-only proof gives
\[
(1-\eta)\opt\le\mathrm{ALG}\le\opt.
\]

The active-constraint universe has size at most $2^k n^k$.  By
\cref{thm:l0-sampler}, each sampler uses $O(k^2\log^2 n)$ bits times the
logarithm of its inverse failure probability.  The number of samplers and all
dependence on $k$, $1/\eps$, $1/\eta$, and $1/\delta$ are absorbed by the
displayed exponential factor.  Dividing the resulting bit bound by the
$\Theta(\log n)$ word size gives the stated word bound.  The CountMin sketches
require no larger asymptotic space.
\end{proof}

\section{Extension to \textsc{Max-$k$CSP}s: Proof of \cref{thm:generalMax-CSPs}}
\label{sec:proof-thm-generalMax-CSPs}
\begin{proof}[Proof of \cref{thm:generalMax-CSPs}]
        The proof follows from a standard black-box reduction to the \textsc{Max-$k$AND} problem. In particular, we reduce any instance $\Psi$ of $\textsc{Max-$k$CSP}(f)$ to an instance $\widetilde{\Psi}$ of \textsc{Max-$k$AND} such that the maximum values of $\Psi$ and $\widetilde{\Psi}$ are equal, while increasing the number of constraints by at most a factor $2^k$. Given such a reduction, \cref{thm:generalMax-CSPs} follows from \cref{thm:max-2and-insertion}, \cref{thm:max-2and-dynamic}, \cref{thm:streaming-unbounded-max-kand}, and \cref{thm:dynamic-streaming-unbounded-max-kand}; for $k\ge 3$, the $2^k$ blow-up is absorbed into the displayed $k$-dependence, and the dynamic-space factor is written as $\poly(\log n)$.

        We now explain how the reduction works. 
The instance $\widetilde{\Psi}$ is over the same set of $n$ variables $x_1,\dots,x_n$. 
        Consider any constraint $(f,\mathbf{j},\mathbf{b})$ in $\Psi$. For every $\mathbf{a}\in \{-1,1\}^k$ such that $f(\mathbf{a})=1$, we add the constraint $(\textsc{AND},\mathbf{j},\mathbf{b}\circ \mathbf{a})$ to $\widetilde{\Psi}$, where $\mathbf{b}\circ \mathbf{a} = (b_1a_1,\dots,b_ka_k)$. 
        
We now claim that for every $\boldsymbol{\sigma}\in \{-1,1\}^n$, 
\[\mathrm{val}_{\Psi}(\boldsymbol{\sigma}) = \mathrm{val}_{\widetilde{\Psi}}(\boldsymbol{\sigma}).\]

        Suppose that $(f,\mathbf{j},\mathbf{b})$ is satisfied by $\boldsymbol{\sigma}$. Then, $f(b_1\sigma_{j_1},\dots,b_k\sigma_{j_k}) =1 $ which implies that there exists an $\mathbf{a}\in \{-1,1\}^k$ such that $f(\mathbf{a})=1$ and $b_i\sigma_{j_i}=a_i, \forall i\in [k]$. Equivalently $\sigma_{j_i}=b_ia_i,\forall i\in [k]$. Thus, $\textsc{AND}(\mathbf{j},\mathbf{b}\circ \mathbf{a})=1$ and by definition, no other constraint of the form $(\textsc{AND},\mathbf{j},\mathbf{b}\circ \mathbf{a'})$ where $\mathbf{a'} \ne \mathbf{a}$ can be satisfied by $\boldsymbol{\sigma}$.

        On the other hand, if $(f,\mathbf{j},\mathbf{b})$ is not satisfied by $\boldsymbol{\sigma}$, then there exists an $\mathbf{a}\in \{-1,1\}^k$ such that $f(\mathbf{a})=0$ and $b_i\sigma_{j_i}=a_i, \forall i\in [k]$. Equivalently $\sigma_{j_i}=b_ia_i,\forall i\in [k]$. Thus, $\textsc{AND}(\mathbf{j},\mathbf{b}\circ \mathbf{a})=1$ and by definition, no constraint in  $\{(\textsc{AND},\mathbf{j},\mathbf{b}\circ \mathbf{a'}):f(\mathbf{a'})=1\}$ can be satisfied by $\boldsymbol{\sigma}$.

        Thus, if $(f,\mathbf{j},\mathbf{b})$ is satisfied by $\boldsymbol{\sigma}$, then exactly one of the constraints in $\{(\textsc{AND},\mathbf{j},\mathbf{b}\circ \mathbf{a}):f(\mathbf{a})=1\}$ is satisfied by $\boldsymbol{\sigma}$. On the other hand, if $(f,\mathbf{j},\mathbf{b})$ is not satisfied by $\boldsymbol{\sigma}$, then $\boldsymbol{\sigma}$ also does not satisfy any constraint in $\{(\textsc{AND},\mathbf{j},\mathbf{b}\circ \mathbf{a}):f(\mathbf{a})=1\}$. It follows that the maximum values of $\Psi$ and $\widetilde{\Psi}$ are equal. Since $|\{\mathbf{a}:f(\mathbf{a})=1\}|\le 2^k$, the bounds on the number of constraints and the maximum variable degree in $\widetilde{\Psi}$ follow immediately.
    \end{proof}
    
\section{Conclusion}
We have exhibited that learning-augmented streaming algorithms with oracle access to $\eps$-accurate noisy predictions can obtain highly-accurate estimates for the values of Boolean \mcsp s. There are several interesting open questions that arise from our work which we state below.
\begin{enumerate}
    \item Can we relax the assumption that the streaming algorithm has to know $\eps$ \emph{with high accuracy}? For instance, could one obtain similar results assuming only a constant-factor approximation to $\eps$ is known to the algorithm? Or can we prove lower bounds showing that this is impossible?
    \item Are there suitable generalizations of the prediction models for \mcsp s defined on arbitrary $q$-ary alphabets? 
 
\end{enumerate}

\section*{Acknowledgments}
This work began as part of a discussion between Pan Peng, Ali Vakilian, and Santhoshini Velusamy when they were attending the Workshop on High-Dimensional and Complex Data Algorithms in Venice, Italy in May 2025. The work was conducted in part while Yinhao Dong and Xiaoyang Xu were visiting the Toyota Technological Institute at Chicago (TTIC) in summer 2025, and while Yinhao Dong, Pan Peng, and Ali Vakilian were attending the Sublinear Algorithms Reunion at the Simons Institute in summer 2025.

Y.D. is supported in part by the National Center for Interdisciplinary Studies and Talents Cultivation, USTC. P.P. is supported in part by NSFC Grant 62272431 and Quantum Science and Technology -- National Science and Technology Major Project (Grant No. 2021ZD0302901). A.V. is supported in part by NSF CAREER Award CCF-2542925. S.V. was supported in part by NSF award CCF-2348475 when she was affiliated with TTIC. S.V. is also supported in part by NSERC Discovery grant RGPIN-2026-04689 and ECR supplement DGECR-2026-00421. X.X. was supported by NSF award CCF-2348475 when visiting TTIC.

We thank Yury Makarychev and Daniel Lokshtanov for insightful discussions and valuable feedback.

\section*{AI Disclosure}
We used AI tools, including GPT-5.4 Pro and GPT-5.5 Pro, throughout this project. Specifically, GPT-5.5 Pro assisted with the generalization of the algorithms for unbounded-degree \textsc{Max-$k$AND} with $k\ge 3$ (\cref{sec:unbounded-max-kand}) based on the authors' ideas, as well as with the calculations in Appendix~\ref{sec:estimated-eps-max-kand}. GPT-5.4 Pro assisted with the calculations in Appendix~\ref{sec:estimated-eps-max-dicut}. All other research contributions were made solely by the authors.

AI models were also used to assist with drafting, proofreading, and language editing throughout the writing process. The authors substantially revised and rewrote the AI-generated content and carefully verified all mathematical statements and proofs, making corrections and improvements as needed. The authors verified the AI-assisted results and take full responsibility for the correctness of all claims in this paper.
\bibliographystyle{alpha}
\bibliography{main}

@inproceedings{aamand2025improved,
  author       = {Anders Aamand and
                  Justin Y. Chen and
                  Siddharth Gollapudi and
                  Sandeep Silwal and
                  Hao Wu},
  title        = {Improved Approximations for Hard Graph Problems using Predictions},
  booktitle    = {Forty-second International Conference on Machine Learning (ICML)},
  year         = {2025}
}

@inproceedings{braverman2024learning,
  author       = {Vladimir Braverman and
                  Prathamesh Dharangutte and
                  Vihan Shah and
                  Chen Wang},
  title        = {Learning-Augmented Maximum Independent Set},
  booktitle    = {Approximation, Randomization, and Combinatorial Optimization. Algorithms
                  and Techniques ({APPROX/RANDOM})},
  series       = {LIPIcs},
  volume       = {317},
  pages        = {24:1--24:18},
  
  year         = {2024},
  bibsource    = {dblp computer science bibliography, https://dblp.org}
}

@inproceedings{ghoshal2025constraint,
  author       = {Suprovat Ghoshal and
                  Konstantin Makarychev and
                  Yury Makarychev},
  title        = {Constraint Satisfaction Problems with Advice},
  booktitle    = {Proceedings of the 2025 Annual {ACM-SIAM} Symposium on Discrete Algorithms (SODA)},
  pages        = {1202--1221},
  publisher    = {{SIAM}},
  year         = {2025},
  bibsource    = {dblp computer science bibliography, https://dblp.org}
}

@article{chou2021approximability,
  title={Approximability of all Boolean CSPs with linear sketches},
  author={Chou, Chi-Ning and Golovnev, Alexander and Sudan, Madhu and Velusamy, Santhoshini},
    journal      = {CoRR},
  volume       = {abs/2102.12351},
  year         = {2021},
  
  eprinttype    = {arXiv},
  eprint       = {2102.12351}
}

@article{chou2024sketching,
  title={Sketching approximability of all finite csps},
  author={Chou, Chi-Ning and Golovnev, Alexander and Sudan, Madhu and Velusamy, Santhoshini},
  journal={J. {ACM}},
  volume={71},
  number={2},
  pages={15:1--15:74},
  year={2024}
}

@inproceedings{chou2022linear,
  title={Linear space streaming lower bounds for approximating CSPs},
  author={Chou, Chi-Ning and Golovnev, Alexander and Sudan, Madhu and Velingker, Ameya and Velusamy, Santhoshini},
  booktitle={Proceedings of the 54th Annual ACM SIGACT Symposium on Theory of Computing ({STOC})},
  pages={275--288},
  year={2022}
}

@inproceedings{hsu2019learning,
  title={Learning-Based Frequency Estimation Algorithms},
  author={Hsu, Chen-Yu and Indyk, Piotr and Katabi, Dina and Vakilian, Ali},
  booktitle    = {7th International Conference on Learning Representations ({ICLR})},
  year={2019}
}

@inproceedings{chen2022triangle,
  author       = {Justin Y. Chen and
                  Talya Eden and
                  Piotr Indyk and
                  Honghao Lin and
                  Shyam Narayanan and
                  Ronitt Rubinfeld and
                  Sandeep Silwal and
                  Tal Wagner and
                  David P. Woodruff and
                  Michael Zhang},
  title        = {Triangle and Four Cycle Counting with Predictions in Graph Streams},
  booktitle    = {the Tenth International Conference on Learning Representations ({ICLR})},
  year         = {2022}
}

@inproceedings{cohen2024max,
  title={Learning-Augmented Approximation Algorithms for Maximum Cut and Related Problems},
  author={Cohen{-}Addad, Vincent and d'Orsi, Tommaso and Gupta, Anupam and Lee, Euiwoong and Panigrahi, Debmalya},
  booktitle    = {Advances in Neural Information Processing Systems 37: Annual Conference
                  on Neural Information Processing Systems (NeurIPS)},
  year={2024}
}

@article{aamand2019frequency,
  title={({Learned}) Frequency estimation algorithms under Zipfian distribution},
  author={Aamand, Anders and Indyk, Piotr and Vakilian, Ali},
  journal      = {CoRR},
  volume       = {abs/1908.05198},
  year         = {2019},
  
  eprinttype    = {arXiv},
  eprint       = {1908.05198}
}

@inproceedings{aamand2023improved,
  title={Improved frequency estimation algorithms with and without predictions},
  author={Aamand, Anders and Chen, Justin and Nguyen, Huy and Silwal, Sandeep and Vakilian, Ali},
  booktitle    = {Advances in Neural Information Processing Systems 36: Annual Conference
                  on Neural Information Processing Systems (NeurIPS)},
  year         = {2023},
}

@inproceedings{jiang2020learning,
  title={Learning-augmented data stream algorithms},
  author={Jiang, Tanqiu and Li, Yi and Lin, Honghao and Ruan, Yisong and Woodruff, David P.},
  booktitle    = {8th International Conference on Learning Representations ({ICLR})},
  year={2020}
}

@inproceedings{indyk2019learning,
  title={Learning-based low-rank approximations},
  author={Indyk, Piotr and Vakilian, Ali and Yuan, Yang},
  booktitle    = {Advances in Neural Information Processing Systems 32: Annual Conference
                  on Neural Information Processing Systems (NeurIPS)},
  pages        = {7400--7410},
  year={2019}
}

@inproceedings{dinitz2021faster,
  title={Faster matchings via learned duals},
  author={Dinitz, Michael and Im, Sungjin and Lavastida, Thomas and Moseley, Benjamin and Vassilvitskii, Sergei},
  booktitle    = {Advances in Neural Information Processing Systems 34: Annual Conference
                  on Neural Information Processing Systems (NeurIPS)},
  pages        = {10393--10406},
  year         = {2021},
}

@inproceedings{mitzenmacher2018model,
  title={A model for learned bloom filters and optimizing by sandwiching},
  author={Mitzenmacher, Michael},
  booktitle    = {Advances in Neural Information Processing Systems 31: Annual Conference
                  on Neural Information Processing Systems (NeurIPS)},
  pages        = {462--471},
  year={2018}
}

@inproceedings{mahdian2007allocating,
  title={Allocating online advertisement space with unreliable estimates},
  author={Mahdian, Mohammad and Nazerzadeh, Hamid and Saberi, Amin},
  booktitle={Proceedings of the 8th ACM Conference on Electronic Commerce (EC)},
  pages={288--294},
  year={2007}
}

@inproceedings{singer2023oblivious,
  title={Oblivious Algorithms for the {Max-$k$AND} Problem},
  author={Singer, Noah G.},
  booktitle    = {Approximation, Randomization, and Combinatorial Optimization. Algorithms
                  and Techniques ({APPROX/RANDOM})},
  series       = {LIPIcs},
  volume       = {275},
  pages        = {15:1--15:19},
  
  year         = {2023}
}

@inproceedings{chou2022sketching,
  title={Sketching Approximability of (Weak) Monarchy Predicates},
  author={Chou, Chi-Ning and Golovnev, Alexander and Shahrasbi, Amirbehshad and Sudan, Madhu and Velusamy, Santhoshini},
  booktitle    = {Approximation, Randomization, and Combinatorial Optimization. Algorithms
                  and Techniques ({APPROX/RANDOM})},
  series       = {LIPIcs},
  volume       = {245},
  pages        = {35:1--35:17},
  
  year         = {2022},
}

@inproceedings{guruswami2019streaming,
  title={Streaming Hardness of Unique Games},
  author={Guruswami, Venkatesan and Tao, Runzhou},
  booktitle    = {Approximation, Randomization, and Combinatorial Optimization. Algorithms
                  and Techniques ({APPROX/RANDOM})},
  series       = {LIPIcs},
  volume       = {145},
  pages        = {5:1--5:12},
  
  year         = {2019}
}

@inproceedings{kallaugher2025design,
  title={How to Design a Quantum Streaming Algorithm Without Knowing Anything About Quantum Computing},
  author={Kallaugher, John and Parekh, Ojas and Voronova, Nadezhda},
  booktitle={2025 Symposium on Simplicity in Algorithms (SOSA)},
  pages={9--45},
  year={2025}
}

@inproceedings{kallaugher2024exponential,
  title={Exponential quantum space advantage for approximating maximum directed cut in the streaming model},
  author={Kallaugher, John and Parekh, Ojas and Voronova, Nadezhda},
  booktitle={Proceedings of the 56th Annual ACM Symposium on Theory of Computing ({STOC})},
  pages={1805--1815},
  year={2024}
}

@inproceedings{kallaugher2022quantum,
  title={The quantum and classical streaming complexity of quantum and classical Max-Cut},
  author={Kallaugher, John and Parekh, Ojas},
  booktitle={63rd IEEE Annual Symposium on Foundations of Computer Science (FOCS)},
  pages={498--506},
  year={2022}
}

@inproceedings{saxena2025streaming,
  title={Streaming Algorithms via Local Algorithms for Maximum Directed Cut},
  author={Saxena, Raghuvansh R. and Singer, Noah G. and Sudan, Madhu and Velusamy, Santhoshini},
  booktitle={Proceedings of the 2025 Annual ACM-SIAM Symposium on Discrete Algorithms (SODA)},
  pages={3392--3408},
  year={2025}
}

@inproceedings{saxena2023improved,
  title={Improved streaming algorithms for Maximum Directed Cut via smoothed snapshots},
  author={Saxena, Raghuvansh R. and Singer, Noah G. and Sudan, Madhu and Velusamy, Santhoshini},
  booktitle={64th IEEE Annual Symposium on Foundations of Computer Science (FOCS)},
  pages={855--870},
  year={2023}
}

@inproceedings{saxena2023streaming,
  title={Streaming complexity of CSPs with randomly ordered constraints},
  author={Saxena, Raghuvansh R. and Singer, Noah and Sudan, Madhu and Velusamy, Santhoshini},
  booktitle={Proceedings of the 2023 Annual ACM-SIAM Symposium on Discrete Algorithms (SODA)},
  pages={4083--4103},
  year={2023}
}

@inproceedings{chou2020optimal,
  title={Optimal streaming approximations for all boolean max-2csps and max-ksat},
  author={Chou, Chi-Ning and Golovnev, Alexander and Velusamy, Santhoshini},
  booktitle={61st {IEEE} Annual Symposium on Foundations of Computer Science (FOCS)},
  pages={330--341},
  year={2020}
}

@inproceedings{guruswami2017streaming,
  title={Streaming complexity of approximating max 2csp and max acyclic subgraph},
  author={Guruswami, Venkatesan and Velingker, Ameya and Velusamy, Santhoshini},
  booktitle    = {Approximation, Randomization, and Combinatorial Optimization. Algorithms
                  and Techniques ({APPROX/RANDOM})},
  series       = {LIPIcs},
  volume       = {81},
  pages        = {8:1--8:19},
  
  year         = {2017}
}

@inproceedings{fei2025dichotomy,
  title={A dichotomy theorem for multi-pass streaming CSPs},
  author={Fei, Yumou and Minzer, Dor and Wang, Shuo},
  booktitle    = {Proceedings of the 58th Annual {ACM} Symposium on Theory of Computing ({STOC})},
  year         = {2026}
}

@inproceedings{fei2025multi,
  title={Multi-Pass Streaming Lower Bounds for Approximating Max-Cut},
  author={Fei, Yumou and Minzer, Dor and Wang, Shuo},
  booktitle    = {66th {IEEE} Annual Symposium on Foundations of Computer Science ({FOCS})},
  year={2025}
}

@inproceedings{chen2023towards,
  title={Towards multi-pass streaming lower bounds for optimal approximation of max-cut},
  author={Chen, Lijie and Kol, Gillat and Paramonov, Dmitry and Saxena, Raghuvansh R. and Song, Zhao and Yu, Huacheng},
  booktitle={Proceedings of the 2023 Annual ACM-SIAM Symposium on Discrete Algorithms (SODA)},
  pages={878--924},
  year={2023}
}

@inproceedings{assadi2021graph,
  title={Graph streaming lower bounds for parameter estimation and property testing via a streaming XOR lemma},
  author={Assadi, Sepehr and N, Vishvajeet},
  booktitle={Proceedings of the 53rd Annual ACM  Symposium on Theory of Computing ({STOC})},
  pages={612--625},
  year={2021}
}

@inproceedings{assadi2020multi,
  title={Multi-pass graph streaming lower bounds for cycle counting, max-cut, matching size, and other problems},
  author={Assadi, Sepehr and Kol, Gillat and Saxena, Raghuvansh R. and Yu, Huacheng},
  booktitle={ 61st {IEEE} Annual Symposium on Foundations of Computer Science (FOCS)},
  pages={354--364},
  year={2020}
}

@inproceedings{kapralov2014streaming,
  title={Streaming lower bounds for approximating MAX-CUT},
  author={Kapralov, Michael and Khanna, Sanjeev and Sudan, Madhu},
  booktitle={Proceedings of the 2015 Annual ACM-SIAM Symposium on Discrete Algorithms (SODA)},
  pages={1263--1282},
  year={2015}
}

@inproceedings{kapralov2017,
  title={{$(1+ \Omega(1))$}-Approximation to MAX-CUT Requires Linear Space},
  author={Kapralov, Michael and Khanna, Sanjeev and Sudan, Madhu and Velingker, Ameya},
  booktitle={Proceedings of the 2017 Annual ACM-SIAM Symposium on Discrete Algorithms (SODA)},
  pages={1703--1722},
  year={2017}
}

@inproceedings{sudan2022streaming,
  title={Streaming and Sketching Complexity of CSPs: A Survey (Invited Talk)},
  author={Sudan, Madhu},
  booktitle={49th International Colloquium on Automata, Languages, and Programming (ICALP)},
  series       = {LIPIcs},
  volume       = {229},
  pages        = {5:1--5:20},
  
  year         = {2022},
  
}

@inproceedings{kogan2015sketching,
  title={Sketching cuts in graphs and hypergraphs},
  author={Kogan, Dmitry and Krauthgamer, Robert},
  booktitle={Proceedings of the 2015 Conference on Innovations in Theoretical Computer Science (ITCS)},
  pages={367--376},
  year={2015}
}

@article{singer2025nine,
  title={Nine lower bound conjectures on streaming approximation algorithms for CSPs},
  author={Singer, Noah G.},
  journal={ACM SIGACT News},
  volume={56},
  number={4},
  pages={25--36},
  year={2025},
  publisher={ACM New York, NY, USA}
}

@inproceedings{dong2025learning,
  author       = {Yinhao Dong and
                  Pan Peng and
                  Ali Vakilian},
  title        = {Learning-Augmented Streaming Algorithms for Approximating {MAX-CUT}},
  booktitle    = {16th Innovations in Theoretical Computer Science Conference ({ITCS})},
  series       = {LIPIcs},
  volume       = {325},
  pages        = {44:1--44:24},
  year         = {2025}
}

@article{cormode2005improved,
  title={An improved data stream summary: the count-min sketch and its applications},
  author={Cormode, Graham and Muthukrishnan, Shan},
  journal      = {J. Algorithms},
  volume       = {55},
  number       = {1},
  pages        = {58--75},
  year         = {2005}
}

@article{vitter1985random,
  author       = {Jeffrey Scott Vitter},
  title        = {Random Sampling with a Reservoir},
  journal      = {{ACM} Trans. Math. Softw.},
  volume       = {11},
  number       = {1},
  pages        = {37--57},
  year         = {1985},
}

@inproceedings{jowhari2011tight,
  title={Tight bounds for lp samplers, finding duplicates in streams, and related problems},
  author={Jowhari, Hossein and Sa{\u{g}}lam, Mert and Tardos, G{\'a}bor},
  booktitle={Proceedings of the 30th ACM SIGMOD-SIGACT-SIGART symposium on Principles of database systems (PODS)},
  pages={49--58},
  year={2011}
}

@inproceedings{kapralov2019optimal,
  author       = {Michael Kapralov and
                  Dmitry Krachun},
  
  title        = {An optimal space lower bound for approximating {MAX-CUT}},
  booktitle    = {Proceedings of the 51st Annual {ACM} {SIGACT} Symposium on Theory
                  of Computing ({STOC})},
  pages        = {277--288},
  
  year         = {2019}
}

@inproceedings{azarmehr2025half,
  title={Half-Approximating Maximum Dicut in the Streaming Setting},
  author={Azarmehr, Amir and Behnezhad, Soheil and Ferrante, Shane and Saneian, Mohammad},
  booktitle    = {Proceedings of the 58th Annual {ACM} Symposium on Theory of Computing ({STOC})},
  year         = {2026}
}

@inproceedings{velusamy2025,
      title={Near-optimal streaming approximation for Max-DICUT in sublinear space using two passes}, 
      author={Santhoshini Velusamy},
      booktitle    = {Approximation, Randomization, and Combinatorial Optimization. Algorithms
                  and Techniques ({APPROX/RANDOM})},
  
  year         = {2026}
}

@article{goemans1995improved,
  title={Improved approximation algorithms for maximum cut and satisfiability problems using semidefinite programming},
  author={Goemans, Michel X. and Williamson, David P.},
  journal={J. {ACM}},
  volume={42},
  number={6},
  pages={1115--1145},
  year={1995}
}

@article{khot2007optimal,
  title={Optimal inapproximability results for MAX-CUT and other 2-variable CSPs?},
  author={Khot, Subhash and Kindler, Guy and Mossel, Elchanan and O’Donnell, Ryan},
  journal={{SIAM} J. Comput.},
  volume={37},
  number={1},
  pages={319--357},
  year={2007},
  publisher={SIAM}
}

@article{lykouris2021competitive,
  author       = {Thodoris Lykouris and
                  Sergei Vassilvitskii},
  title        = {Competitive Caching with Machine Learned Advice},
  journal      = {J. {ACM}},
  volume       = {68},
  number       = {4},
  pages        = {24:1--24:25},
  year         = {2021}
}

@inproceedings{kumar2018improving,
  title={Improving online algorithms via ML predictions},
  author={Kumar, Ravi and Purohit, Manish and Svitkina, Zoya},
  booktitle    = {Advances in Neural Information Processing Systems 31: Annual Conference
                  on Neural Information Processing Systems (NeurIPS)},
  pages        = {9684--9693},
  year         = {2018},
}

@inproceedings{VaidyaKMK21,
  author       = {Kapil Vaidya and
                  Eric Knorr and
                  Michael Mitzenmacher and
                  Tim Kraska},
  title        = {Partitioned Learned Bloom Filters},
  booktitle    = {9th International Conference on Learning Representations (ICLR)},
  year         = {2021}
}

@inproceedings{KraskaBCDP18,
  author       = {Tim Kraska and
  Alex Beutel and
  Ed H. Chi and
  Jeffrey Dean and
  Neoklis Polyzotis},
  title        = {The Case for Learned Index Structures},
  booktitle = {Proceedings of the 2018 International Conference on Management of Data (SIGMOD)},
  pages        = {489--504},
  year         = {2018}
}

@article{LindermayrMS25,
  author       = {Alexander Lindermayr and
                  Nicole Megow and
                  Bertrand Simon},
  title        = {Boosting Double Coverage for k-Server via Imperfect Predictions},
  journal      = {Algorithmica},
  volume       = {87},
  number       = {11},
  pages        = {1477--1517},
  year         = {2025}
}

@inproceedings{EdenINRSW21,
  author       = {Talya Eden and
                  Piotr Indyk and
                  Shyam Narayanan and
                  Ronitt Rubinfeld and
                  Sandeep Silwal and
                  Tal Wagner},
  title        = {Learning-Based Support Estimation in Sublinear Time},
  booktitle    = {9th International Conference on Learning Representations ({ICLR})},
  
  year         = {2021},
}

@inproceedings{LiLLVW23,
  author       = {Yi Li and
                  Honghao Lin and
                  Simin Liu and
                  Ali Vakilian and
                  David P. Woodruff},
  title        = {Learning the Positions in CountSketch},
  booktitle    = {11th International Conference on Learning Representations
                  ({ICLR})},
  
  year         = {2023}
}

@inproceedings{AamandCGSW25,
  author       = {Anders Aamand and
  Justin Y. Chen and
  Siddharth Gollapudi and
  Sandeep Silwal and
  Hao Wu},
  title        = {Learning-Augmented Frequent Directions},
  booktitle    = {the Thirteenth International Conference on Learning Representations ({ICLR})},
  
  year         = {2025},
}

@article{attias2025learning,
  title={Learning-Augmented Algorithms for Boolean Satisfiability},
  author={Attias, Idan and Gao, Xing and Reyzin, Lev},
  journal      = {CoRR},
      volume       = {abs/2505.06146},
      year         = {2025},
      
      eprinttype    = {arXiv},
      eprint       = {2505.06146}
}

@inproceedings{lafferty2001conditional,
  author       = {John D. Lafferty and
                  Andrew McCallum and
                  Fernando C. N. Pereira},
  
  title        = {Conditional Random Fields: Probabilistic Models for Segmenting and
                  Labeling Sequence Data},
  booktitle    = {Proceedings of the Eighteenth International Conference on Machine
                  Learning (ICML)},
  pages        = {282--289},
  
  year         = {2001}
}

@inproceedings{taskar2003max,
  author       = {Benjamin Taskar and
                  Carlos Guestrin and
                  Daphne Koller},
  title        = {Max-Margin Markov Networks},
  booktitle    = {Advances in Neural Information Processing Systems 16: Neural Information
                  Processing Systems ({NIPS})},
  pages        = {25--32},
  
  year         = {2003}
}

@inproceedings{tsochantaridis2004support,
author       = {Ioannis Tsochantaridis and
                  Thomas Hofmann and
                  Thorsten Joachims and
                  Yasemin Altun},
  
  title        = {Support vector machine learning for interdependent and structured
                  output spaces},
  booktitle    = {Machine Learning, Proceedings of the Twenty-first International Conference (ICML)},
  series       = {{ACM} International Conference Proceeding Series},
  year         = {2004}
}

@article{boykov2001fast,
  author       = {Yuri Boykov and
                  Olga Veksler and
                  Ramin Zabih},
  title        = {Fast Approximate Energy Minimization via Graph Cuts},
  journal      = {{IEEE} Trans. Pattern Anal. Mach. Intell.},
  volume       = {23},
  number       = {11},
  pages        = {1222--1239},
  year         = {2001}
}

@inproceedings{xu2018towards,
  author       = {Hong Xu and
                  Sven Koenig and
                  T. K. Satish Kumar},
  title        = {Towards Effective Deep Learning for Constraint Satisfaction Problems},
  booktitle    = {Principles and Practice of Constraint Programming - 24th International
                  Conference ({CP})},
  series       = {Lecture Notes in Computer Science},
  pages        = {588--597},
  year         = {2018}
}

@inproceedings{selsam2019learning,
  author       = {Daniel Selsam and
                  Matthew Lamm and
                  Benedikt B{\"{u}}nz and
                  Percy Liang and
                  Leonardo de Moura and
                  David L. Dill},
  title        = {Learning a {SAT} Solver from Single-Bit Supervision},
  booktitle    = {7th International Conference on Learning Representations ({ICLR})},
  year         = {2019}
}

@inproceedings{Dong25CC,
  author       = {Yinhao Dong and
                  Shan Jiang and
                  Shi Li and
                  Pan Peng},
  title        = {Learning-Augmented Streaming Algorithms for Correlation Clustering},
  booktitle    = {Advances in Neural Information Processing Systems 38: Annual Conference
                  on Neural Information Processing Systems (NeurIPS)},
  year         = {2025},
  bibsource    = {dblp computer science bibliography, https://dblp.org}
}

@inproceedings{Nagawanshi26moment,
  author       = {Soham Nagawanshi and
                  Shalini Panthangi and
                  Chen Wang and
                  David P. Woodruff and
                  Samson Zhou},
  title        = {Learning-Augmented Moment Estimation on Time-Decay Models},
  booktitle    = {14th International Conference on Learning Representations (ICLR)},
  year         = {2026},
  bibsource    = {dblp computer science bibliography, https://dblp.org}
}

\appendix
\section{Additional technical preliminaries}
In this paper, we employ several well-known techniques in the streaming model, as described below.
\paragraph{CountMin sketch.}
Given a stream of $m$ integers from $[n]$, let $f_i$ denote the number of occurrences of $i$ in the stream, for any $i\in [n]$.
In CountMin sketch, there are $k$ uniform and independent random hash functions $h_1,h_2,\dots,h_k:[n] \rightarrow [w]$ and an array $C$ of size $k \times w$. The algorithm maintains $C$, such that $C[\ell,s] = \sum_{j:h_\ell(j)=s}f_j$ at the end of the stream. Upon each query of an arbitrary $i\in [n]$, the algorithm returns $\widetilde{f}_i = \min_{\ell \in [k]} C[\ell,h_\ell(i)]$.

\begin{theorem}[CountMin sketch~\citep{cormode2005improved}]
\label{thm:CM-v2}
Let $\mathbf{f}\in \mathbb{R}_{\ge 0}^n$ be a frequency vector with total mass $m=\|\mathbf{f}\|_1$ maintained under streaming updates, and let $\widetilde f_i$ denote the CountMin estimate for coordinate $i\in[n]$.
For any parameters $\eps,\delta\in(0,1)$, there is a CountMin sketch that uses
$O\!\left(\frac{1}{\eps}\log\frac{1}{\delta}\right)$ words of space and supports point queries such that for every $i\in[n]$,
\[
\widetilde f_i \ge f_i
\qquad\text{and}\qquad
\Pr\!\left[\widetilde f_i \le f_i + \eps m\right]\ge 1-\delta .
\]
\end{theorem}

\paragraph{Reservoir sampling.}
The algorithm samples $k$ elements from a stream, such that at any time $m \geq k$, the sample consists of a uniformly random subset of size $k$ of the elements seen so far.
The space complexity of the reservoir sampling algorithm is $O(k(\log n + \log m))$ bits~\citep{vitter1985random}.

\begin{definition}[$\ell_0$-sampler]
\label{def:l0-sampler}
Let $\mathbf{x}\in \sR^n$ be a non-zero vector and $\delta\in (0,1)$. 
An $\ell_0$-sampler for $\mathbf{x}$ returns FAIL with probability at most $\delta$ and otherwise returns some index $i$ such that $x_i \neq 0$ with probability $\frac{1}{|\operatorname{supp}(\mathbf{x})|}$ where $\operatorname{supp}(\mathbf{x})=\{i\in [n]: x_i\neq 0\}$ is the support of $\mathbf{x}$.
\end{definition}

The following theorem states that $\ell_0$-samplers can be maintained using a single pass in dynamic streams.
\begin{theorem}[Theorem~2~in~\citep{jowhari2011tight}]
\label{thm:l0-sampler}
Let $\delta \in (0,1)$.
There exists a single-pass streaming algorithm for maintaining an $\ell_0$-sampler for a non-zero vector $\mathbf{x}\in \sR^n$ (with failure probability $\delta$) in the dynamic model using $O(\log^2 n \log \delta^{-1})$ bits of space.
\end{theorem}
We also use the following standard concentration inequalities.

\begin{theorem}[Chebyshev's inequality]
Let $X$ be a random variable with expectation $\E[X]$ and variance $\Var[X]$. Then for any $t>0$,
\[
\Pr[|X-\E[X]|\ge t]\le \frac{\Var[X]}{t^2}.
\]
\end{theorem}

\begin{theorem}[Bounded differences inequality]
\label{thm:bounded-differences}
Let $X_1,\dots,X_n \in \mathcal{X}$ be independent random variables. Let $f: \mathcal{X}^n \to \mathbb{R}$ be a function, and 
let $Z=f(X_1,\dots,X_n)$. Suppose that for every $i\in [n]$, changing only the $i$-th coordinate can change the value of $f$ by at most $c_i$. Then for every $t>0$,
\[
\Pr[|Z-\E [Z]|\ge t]
\le
2\exp\!\left(-\frac{2t^2}{\sum_{i=1}^n c_i^2}\right).
\]
\end{theorem}

\section{Our algorithms for \textsc{Max-Dicut}}\label{sec:maxdicut}
In this section, we present our learning-augmented streaming algorithm for \textsc{Max-Dicut}, which serves as a warm-up for our algorithms for \textsc{Max-$k$CSP}s.

\subsection{Offline algorithm}
\label{sec:max-dicut-main-text}
To illustrate the algorithmic ideas, we first describe the algorithm in the \emph{offline} setting, where the entire input graph is available.

We partition the vertices $V$ into low-degree vertices $L$ and high-degree vertices $H$, where $v\in H$ iff $\deg(v)\ge \theta m$ for $\theta:=\Theta(\eps^{4}\delta\eta^2)$.
We then scan the edges, query the oracle $\mathcal{O}$ on each endpoint, and update a small set of counters from which we then recover an estimate of the \textsc{Max-Dicut} value.

Specifically, for each edge with both endpoints in $L$, we maintain three counters:
(i) $X_0$: the number of edges with $(Y_u,Y_v)=(+1,-1)$ (predicted to contribute to the dicut), (ii) $X_1$: the number of edges with $Y_u=Y_v$ (predicted to stay on the same side), and (iii) $X_2$: the number of edges with $(Y_u,Y_v)=(-1,+1)$ (predicted to go in the reverse direction). These statistics later let us decode an accurate estimate $\widetilde{k}_L$ of the optimal dicut contribution coming from edges inside $L$, namely,
    $\widetilde{k}_L := \alpha_0(\eps) \cdot X_0 + \alpha_1(\eps) \cdot X_1 + \alpha_2(\eps)\cdot X_2$.

Furthermore, for each high-degree vertex $v\in H$, we maintain four incidence counters: (i) $X_{L^+,v}$: the numbers of edges from $L^+$ to $v$, (ii) $X_{L^-,v}$: the numbers of edges from $L^-$ to $v$, (iii) $X_{v,L^+}$: the numbers of edges from $v$ to $L^+$, and (iv) $X_{v,L^-}$: the numbers of edges from $v$ to $L^-$, where $L^+$ (resp.\ $L^-$) is the set of vertices in $L$ predicted to be $+1$ (resp.\ $-1$).
Using these counts, we decode accurate estimates of the corresponding edge quantities with respect to the unknown optimal bipartition $(S^*,T^*)$, namely, the number of edges between $v$ and $S^*\cap L$ or $T^*\cap L$ in each direction, denoted by $\widetilde{k}_{S^*,v}$ and $\widetilde{k}_{v,T^*}$, respectively.
We then assign each $v\in H$ to the side that maximizes its estimated contribution to the dicut, and combine this with the estimate from edges inside $L$ to obtain our final \textsc{Max-Dicut} estimate $\widetilde{k}_{L} + \sum_{v\in H}\max\{\widetilde{k}_{S^*,v}, \widetilde{k}_{v,T^*}\}  - \Theta(\eta) m$.

The algorithm is formally given in Algorithm~\ref{alg:general-offline}.
The following theorem states its  approximation guarantee.

\begin{algorithm}[t!]
\begin{algorithmic}[1]
\Require Directed graph $G=(V,E)$, an $\eps$-accurate oracle $\mathcal{O} \rightarrow \{-1, 1\}$, $\delta,\eta \in (0,1)$
\Ensure Estimate of the \textsc{Max-Dicut} value of $G$

\State \phantomsection \label{line:threshold} $m \leftarrow |E|$, $\theta \leftarrow \frac{\eps^{4} \delta \eta^2}{768}$,  $L\leftarrow \{v\in V: \deg(v) < \theta m\}, H\leftarrow V \setminus L$
\State {\bf initialize} $X_0\leftarrow 0, X_1 \leftarrow 0, X_2 \leftarrow 0$.
\State {\bf initialize} $X_{v,L^+} \leftarrow 0, X_{v,L^-} \leftarrow 0, X_{L^+,v} \leftarrow 0, X_{L^-,v}\leftarrow 0$ for each $v\in H$.
\ForEach{directed edge $(u,v)\in E$}
    \State $Y_u = \mathcal{O}(u), Y_v = \mathcal{O}(v)$
    \If{$u\in L$ and $v \in L$}
        \If{$Y_u = 1$ and $Y_v = -1$}
            $X_0 \leftarrow X_0 +1$
        \EndIf
        \If{$Y_u = Y_v$}
            $X_1 \leftarrow X_1 +1$
        \EndIf
        \If{$Y_u = -1$ and $Y_v = 1$}
            $X_2 \leftarrow X_2 +1$
        \EndIf
    \EndIf
    
    \If{$u\in H$ and $v \in L$}
        \State \textbf{if} $Y_v = 1$ \textbf{then} $X_{u,L^+} \leftarrow X_{u,L^+} + 1$ \textbf{else} $X_{u,L^-} \leftarrow X_{u,L^-} + 1$        
    \EndIf
    
    \If{$u\in L$ and $v \in H$}
        \State \textbf{if} $Y_u = 1$ \textbf{then} $X_{L^+,v} \leftarrow X_{L^+,v} + 1$ \textbf{else} $X_{L^-,v} \leftarrow X_{L^-,v} + 1$ 
    \EndIf 
\EndFor
\State \phantomsection \label{line:tildekL} $\widetilde{k}_{L} \leftarrow \frac{1}{4\eps^2}[(\frac{1}{2}+\eps)^2 X_0 - (\frac{1}{2}+\eps)(\frac{1}{2}-\eps)X_1 + (\frac{1}{2}-\eps)^2 X_2]$ 
\ForEach{$v\in H$} 
\State $\widetilde{k}_{S^*,v} \leftarrow (\frac{1}{2}+\frac{1}{4\eps})X_{L^+,v} + (\frac{1}{2}-\frac{1}{4\eps})X_{L^-,v}$, $\widetilde{k}_{v,T^*} \leftarrow (\frac{1}{2}-\frac{1}{4\eps})X_{v,L^+} + (\frac{1}{2}+\frac{1}{4\eps})X_{v,L^-}$  \phantomsection \label{line:tilde-v-S*}
\label{line:num-in-edges}

\EndFor
\State \textbf{return} $\widetilde{k}_{L} + \sum_{v\in H}\max\{\widetilde{k}_{S^*,v}, \widetilde{k}_{v,T^*}\}  - \frac{\eta m}{16}$ 
\end{algorithmic}
\caption{A learning-augmented \emph{offline} algorithm for \textsc{Max-Dicut} estimation in general graphs}
\label{alg:general-offline}
\end{algorithm}

\begin{theorem}
    Let $\eps \in (0,\frac{1}{2}]$, $\delta, \eta \in (0,1)$, and $G$ be a directed graph with $m=\Omega(\eps^{-26}\delta^{-10}\eta^{-14})$. Given an $\eps$-accurate oracle $\mathcal{O}$ for the \textsc{Max-Dicut} of $G$,
    \cref{alg:general-offline} provides a $(1-\eta)$-approximation for the \textsc{Max-Dicut} value of $G$ with probability at least $1-\delta$.
    \label{thm:general-offline}
\end{theorem}

Next, we prove \cref{thm:general-offline}. We first state several key properties of the maintained counters in \cref{alg:general-offline}.
\subsubsection{Key properties}
\label{sec:dicut-offlinekey}

Let $X_0,X_1$, and $X_2$ denote the random variables given by the final values of the corresponding counters at the end of the algorithm.
We first prove concentration bounds for these random variables around their expectations.
\begin{lemma}
\label{lem:concentration-low}
Let $\eps \in (0,\frac{1}{2}]$ and $\delta,\eta \in (0,1)$.
For each $i\in \{0,1,2\}$, we have $|X_i - \E[X_i]| \le \frac{\eps^{2}\eta m}{8}$ with probability at least $1-\frac{\delta}{6}$.
\end{lemma}
\begin{proof}
    Let $E_L:=\{(u,v)\in E: u,v\in L\}$ denote the set of directed edges with both endpoints in $L$.
    Let $m_L:=|E_L|$.
    For each edge $(u,v)\in E_L$, define the following indicator random variables:
    \begin{align*}
    X_{0,uv} &=
    \begin{cases}
    1 & \text{if } Y_u=1 \wedge Y_v=-1,\\
    0 & \text{otherwise},
    \end{cases}
    &
    X_{1,uv} &=
    \begin{cases}
    1 & \text{if } Y_u=Y_v,\\
    0 & \text{otherwise},
    \end{cases}
    &
    X_{2,uv} &=
    \begin{cases}
    1 & \text{if } Y_u=-1 \wedge Y_v=1,\\
    0 & \text{otherwise}.
    \end{cases}
    \end{align*}
    
    Note that for $i\in \{0,1,2\}$, we have $\E[X_{i,uv}] = \Pr[X_{i,uv} = 1]$ and $X_i=\sum_{(u,v)\in E_L}X_{i,uv}$. Therefore,
    \begin{align*}
        \E\qty[X_i^2] &= \E\qty[ \qty(\sum_{(u,v)\in E_L} X_{i,uv})^2]\\
        &= \E\qty[\sum_{(u,v)\in E_L}\sum_{(w,z)\in E_L}X_{i,uv} X_{i,wz}]\\
        &=\sum_{(u,v)\in E_L}\sum_{(w,z)\in E_L}\E[X_{i,uv}X_{i,wz}]\\ 
        &=\sum_{(u,v)\in E_L}\sum_{(w,z)\in E_L} \Pr[X_{i,uv} = 1\wedge X_{i,wz} = 1]\\
        &= \sum_{\substack{(u,v), (w,z)\in E_L: \\ \{u,v\}\cap \{w,z\} = \emptyset}} \Pr[X_{i,uv} = 1\wedge X_{i,wz} = 1] + \sum_{\substack{(u,v), (w,z)\in E_L: \\\{u,v\}\cap \{w,z\} \neq \emptyset}} \Pr[X_{i,uv} = 1\wedge X_{i,wz} = 1]\\
        &\leq \sum_{\substack{(u,v), (w,z)\in E_L: \\\{u,v\}\cap \{w,z\} = \emptyset}} \Pr[X_{i,uv} = 1] \Pr[X_{i,wz} = 1] + \sum_{\substack{(u,v), (w,z)\in E_L: \\\{u,v\}\cap \{w,z\} \neq \emptyset}} 1\\
        &\leq \sum_{(u,v)\in E_L}\sum_{(w,z)\in E_L}\E[X_{i,uv}]\E[X_{i,wz}] + m_L\cdot 2\cdot \theta m \\
        &\leq (\E[\sum_{(u,v)\in E_L}X_{i,uv}])^2 + 2\theta m^2\\
        &= (\E[X_i])^2 + 2\theta m^2,
    \end{align*}
    where the third-to-last step follows from $\deg(v) < \theta m$ for $v\in L$, and the second-to-last step follows from $m_L \le m$.
    
    Then we have $\Var[X_i] = \E[X_i^2] - (\E[X_i])^2 \le 2\theta m^2$.
    By Chebyshev's inequality, 
    \begin{align*}
        \Pr[|X_i - \E[X_i]| > \frac{\eps^{2}\eta m}{8}] 
        \le \frac{\Var[X_i]}{(\eps^{2}\eta m/8)^2}
        &\le \frac{128\theta m^2}{\eps^{4}\eta^2 m^2} &&\rhd \Var[X_i] \le 2\theta m^2\\
        &= \frac{128\theta}{\eps^{4}\eta^2}
        < \frac{\delta}{6}.  &&\rhd \theta < \frac{\eps^{4} \delta \eta^2}{768}
    \end{align*}
\end{proof}

Recall that $(S^*,T^*)$ denotes the optimal bipartition corresponding to some fixed \textsc{Max-Dicut} of $G$.
Define 
    $k_{L}:= |\{(u,v)\in E: u\in S^*\cap L, v\in T^*\cap L\}|$
to be the number of directed edges with both endpoints in $L$ (i.e., both low-degree) that contribute to the \textsc{Max-Dicut} value.

Recall that $\widetilde{k}_L$, computed at Line \ref{line:tildekL} of \Cref{alg:general-offline}, is an estimate of $k_{L}$.
The following lemma shows that $\widetilde{k}_L$ is close to $k_L$ up to an additive error of $\poly(\eta)\cdot m$ with high probability.
\begin{lemma}
\label{lem:concentration-k-L}
Let $\delta,\eta \in (0,1)$.
We have $|\widetilde{k}_L - k_L| \le \frac{\eta m}{32}$ with probability at least $1-\frac{\delta}{2}$.
\end{lemma}
\begin{proof}
    Let $k_0$ denote the number of edges from $S^* \cap L$ to $T^* \cap L$. Let $k_1$ denote the number of edges whose endpoints both lie in $S^*\cap L$ or both lie in $T^*\cap L$. Let $k_2$ denote the number of edges from $T^* \cap L$ to $S^* \cap L$. Note that $k_0 = k_L$.
According to \cref{alg:general-offline}, we have
\begin{align*}
    \E[X_0] &= \left(\frac{1}{2}+\eps \right)^2 \cdot k_0 + \left(\frac{1}{2}+\eps \right)\left(\frac{1}{2}-\eps \right)\cdot k_1 + \left(\frac{1}{2}-\eps \right)^2 \cdot k_2,\\
    \E[X_1] &= 2\left(\frac{1}{2}+\eps \right)\left(\frac{1}{2}-\eps \right)\cdot k_0 + \left[ \left(\frac{1}{2}+\eps \right)^2 + \left(\frac{1}{2}-\eps \right)^2 \right]\cdot k_1 + 2\left(\frac{1}{2}+\eps \right)\left(\frac{1}{2}-\eps \right)\cdot k_2,\\
    \E[X_2] &= \left(\frac{1}{2}-\eps \right)^2 \cdot k_0 + \left(\frac{1}{2}+\eps \right)\left(\frac{1}{2}-\eps \right)\cdot k_1 + \left(\frac{1}{2}+\eps \right)^2 \cdot k_2.
\end{align*}
Therefore, by Gaussian elimination, we have 
\begin{align}
    k_L = k_0 = \frac{1}{4\eps^2}\left[\left(\frac{1}{2}+\eps\right)^2 \cdot \E[X_0] - \left(\frac{1}{2}+\eps\right)\left(\frac{1}{2}-\eps\right)\cdot \E[X_1] + \left(\frac{1}{2}-\eps\right)^2 \cdot \E[X_2]\right]
\label{eq:gaussian-opt}
\end{align}
By applying union bound on Lemma~\ref{lem:concentration-low}, we have that with probability at least $1-\frac{\delta}{6}\cdot 3 = 1-\frac{\delta}{2}$, $|X_i - \E[X_i]| \le \frac{\eps^{2}\eta m}{8}$ for all $i\in \{0,1,2\}$. Therefore,
\begin{align*}
    \widetilde{k}_L =~&\frac{1}{4\eps^2}\left[\left(\frac{1}{2}+\eps\right)^2 X_0 - \left(\frac{1}{2}+\eps\right)\left(\frac{1}{2}-\eps\right)X_1 + \left(\frac{1}{2}-\eps\right)^2 X_2\right]\\
    \le~&\frac{1}{4\eps^2}\left[\left(\frac{1}{2}+\eps\right)^2 \left(\E[X_0]+\frac{\eps^{2}\eta m}{8}\right) - \left(\frac{1}{2}+\eps\right)\left(\frac{1}{2}-\eps\right)\left(\E[X_1]-\frac{\eps^{2}\eta m}{8}\right) 
    + \left(\frac{1}{2}-\eps\right)^2 \left(\E[X_2]+\frac{\eps^{2}\eta m}{8}\right)\right] \\
    =~&k_L +  \frac{\eta m}{32}\left[ \left(\frac{1}{2}+\eps\right)^2 + \left(\frac{1}{2}+\eps\right)\left(\frac{1}{2}-\eps\right) + \left(\frac{1}{2}-\eps\right)^2 \right]\\
    =~&k_L + \frac{\eta m}{32}\left(\frac{3}{4}+\eps^2\right)\\
    \le~&k_L+ \frac{\eta m}{32},
\end{align*}
where the third-to-last step follows from \cref{eq:gaussian-opt}.
Similarly, $\widetilde{k}_L \ge k_L - \frac{\eta m}{32}$.

\end{proof}

For each high-degree vertex $v\in H$, define 
\begin{enumerate}
    \item $k_{S^*,v}:= |\{(u,v)\in E: u\in S^*\cap L\}|$ to be the number of directed edges from $S^*\cap L$ to $v$.
    \item $k_{T^*,v}:= |\{(u,v)\in E: u\in T^*\cap L\}|$ to be the number of directed edges from $T^*\cap L$ to $v$.
    \item $k_{v,S^*}:= |\{(v,u)\in E:  u\in S^*\cap L\}|$ to be the number of directed edges from $v$ to $S^*\cap L$.
    \item $k_{v,T^*}:= |\{(v,u)\in E: u\in T^*\cap L\}|$ to be the number of directed edges from $v$ to $T^*\cap L$.
\end{enumerate}

Recall that $\widetilde{k}_{S^*,v}$ and $\widetilde{k}_{v,T^*}$, computed at Line \ref{line:tilde-v-S*} of \Cref{alg:general-offline}, are estimates of $k_{S^*,v}$ and $k_{v,T^*}$, respectively.
The following lemma shows that $\widetilde{k}_{S^*,v}$ and $\widetilde{k}_{v,T^*}$ are close to $k_{S^*,v}$ and $k_{v,T^*}$, respectively, up to an additive error of $\poly(\eps,\delta,\eta)\cdot m$ with high probability.
\begin{lemma}
\label{lem:concentration-general}
Let $\eps \in (0,\frac{1}{2}]$ and $\delta,\eta \in (0,1)$. For each vertex $v\in H$,
each of the following inequalities holds with probability at least $1-\frac{\eps^{4}\delta^2 \eta^2}{6144}$: 
\begin{enumerate}
    \item[(1)] $|\widetilde{k}_{S^*,v} - k_{S^*,v}| \le \frac{\eps^{4}\delta \eta^3 m}{49152}$.
    \item[(2)] $|\widetilde{k}_{v,T^*} - k_{v,T^*}| \le \frac{\eps^{4}\delta \eta^3 m}{49152}$.
\end{enumerate}

\end{lemma}
\begin{proof}
    We prove the first inequality below; the second inequality follows by the same argument.

    For each vertex $v\in H$, let $X_{L^+,v},X_{L^-,v},X_{v,L^+}$, and $X_{v,L^-}$ denote the random variables given by the final values of the corresponding counters at the end of the algorithm. We have
\begin{align*}
    \E[X_{L^+,v}] &= \left(\frac{1}{2}+\eps \right)\cdot k_{S^*,v} + \left(\frac{1}{2}-\eps \right) \cdot k_{T^*,v},\\
    \E[X_{L^-,v}] &= \left(\frac{1}{2}-\eps \right)\cdot k_{S^*,v} + \left(\frac{1}{2}+\eps \right) \cdot k_{T^*,v},\\
    \E[X_{v,L^+}] &= \left(\frac{1}{2}+\eps \right)\cdot k_{v,S^*} + \left(\frac{1}{2}-\eps \right) \cdot k_{v,T^*},\\
    \E[X_{v,L^-}] &= \left(\frac{1}{2}-\eps \right)\cdot k_{v,S^*} + \left(\frac{1}{2}+\eps \right) \cdot k_{v,T^*}.
\end{align*}

By solving the above system of linear equations, we obtain
\begin{align*}
    k_{S^*,v} &= \left( \frac{1}{2}+\frac{1}{4\eps}\right)\cdot \E[X_{L^+,v}] + \left( \frac{1}{2}-\frac{1}{4\eps}\right)\cdot \E[X_{L^-,v}], \\
    k_{v,T^*} &=\left( \frac{1}{2}-\frac{1}{4\eps}\right)\cdot \E[X_{v,L^+}] + \left( \frac{1}{2}+\frac{1}{4\eps}\right)\cdot \E[X_{v,L^-}].
\end{align*}

According to Line~\ref{line:num-in-edges} of \cref{alg:general-offline} and by linearity of expectation, we have
    \begin{align*}
        \E[\widetilde{k}_{S^*,v}] &= \E\left[\left(\frac{1}{2}+\frac{1}{4\eps}\right)X_{L^+,v} + \left(\frac{1}{2}-\frac{1}{4\eps}\right)X_{L^-,v}\right]\\
        &= \left(\frac{1}{2}+\frac{1}{4\eps}\right)\E[X_{L^+,v}] + \left(\frac{1}{2}-\frac{1}{4\eps}\right)\E[X_{L^-,v}]\\
        &= {k}_{S^*,v}.
    \end{align*}

    Next, we calculate $\Var[\widetilde{k}_{S^*,v}]$.
    Note that 
    \begin{align*}
        \widetilde{k}_{S^*,v} &= \left(\frac{1}{2}+\frac{1}{4\eps}\right)X_{L^+,v} + \left(\frac{1}{2}-\frac{1}{4\eps}\right)X_{L^-,v} \\
        &= \left(\frac{1}{2}+\frac{1}{4\eps}\right)X_{L^+,v} + \left(\frac{1}{2}-\frac{1}{4\eps}\right)(e(L,v) - X_{L^+,v}) \\
        &= \frac{1}{2\eps}X_{L^+,v} + \left(\frac{1}{2}-\frac{1}{4\eps}\right)e(L,v),
    \end{align*}
    where $e(L,v)$ is the number of edges from $L$ to $v$.
    Therefore, $\Var[\widetilde{k}_{S^*,v}]=\frac{1}{4\eps^2}\Var[X_{L^+,v}]$.

    Given a high-degree vertex $v\in H$, for each incoming edge $(u,v)\in E$ with $u\in L$, define an indicator random variable $X_{uv}$  as follows:
    \begin{align*}
        X_{uv} =
    \begin{cases}
    1 & \text{if } Y_u=1,\\
    0 & \text{otherwise}.
    \end{cases}
    \end{align*}
    Note that $X_{L^+,v}=\sum_{(u,v)\in E: u\in L}X_{uv}$ and $\E[X_{uv}] = \Pr[X_{uv} = 1]$.
    Then $\Var[X_{uv}] = \E[X^2_{uv}] - (\E[X_{uv}])^2 = \E[X_{uv}]\cdot (1-\E[X_{uv}]) =(\frac{1}{2}+\eps)(\frac{1}{2}-\eps)=\frac{1}{4}-\eps^2 \le \frac{1}{4}$. 
    Since $X_{uv}$'s are independent, we have
    \begin{align*}
        \Var[X_{L^+,v}] = \sum_{(u,v)\in E: u\in L}\Var[X_{uv}]\le \frac{m}{4}.
    \end{align*}

    So, $\Var[\widetilde{k}_{S^*,v}]=\frac{1}{4\eps^2}\Var[X_{L^+,v}] \le \frac{m}{16\eps^2}$.
    By Chebyshev's inequality,
    \begin{align*}
        \Pr[|\widetilde{k}_{S^*,v}-k_{S^*,v}|
        \ge \frac{\eps^{4}\delta \eta^3 m}{49152}] &= \Pr[|\widetilde{k}_{S^*,v}-\E[\widetilde{k}_{S^*,v}]|
        \ge \frac{\eps^{4}\delta \eta^3 m}{49152}] \\
        &\le \frac{\Var[\widetilde{k}_{S^*,v}]}{(\eps^{4}\delta \eta^3 m/49152)^2}
        \le \frac{49152^2 / 16}{\eps^{10}\delta^2\eta^6 m} \le \frac{\eps^{4}\delta^2 \eta^2}{6144},
    \end{align*}
    where the last step follows from our assumption that $m =\Omega(\eps^{-26}\delta^{-10}\eta^{-14})$.
\end{proof}

Define 
    $k_{H}:= |\{(u,v)\in E: u\in S^*\cap H, v\in T^*\cap H\}|$
to be the number of directed edges with both endpoints in $H$ (i.e., both high-degree) that contribute to the \textsc{Max-Dicut} value.
Note that $|H| \le \frac{2m}{\theta m} = \frac{2}{\theta} = \frac{1536}{\eps^{4} \delta\eta^2}$. Thus, we have the following observation. 

\begin{fact}
\label{fact:size-H}
It holds that $k_{H} < |H|^2 \le \frac{1536^2}{\eps^{8} \delta^2\eta^4}$.
\end{fact}

\subsubsection{Proof of \cref{thm:general-offline}}

\begin{proof}
Recall that $k_L$ (resp. $k_H$) is the number of directed edges with both endpoints in $L$ (resp. $H$) that contribute to the \textsc{Max-Dicut} value.

Furthermore, define 
\begin{align*}
    k_{L,H}:=|\{(u,v)\in E: u\in S^*\cap L, v\in T^*\cap H\}|
\end{align*}
to be the number of directed edges from $L$ to $H$ that contribute to the \textsc{Max-Dicut} value, and 
\begin{align*}
    k_{H,L}:=|\{(u,v)\in E: u\in S^*\cap H, v\in T^*\cap L\}|
\end{align*}
to be the number of directed edges from $H$ to $L$ that contribute to the \textsc{Max-Dicut} value.

Therefore, we have $\opt = k_L+k_H+ k_{L,H} + k_{H,L}$. Let $\mathrm{ALG}$ denote the output of \cref{alg:general-offline}.
By union bound, with probability at least $1-\frac{\delta}{2}-\frac{\eps^{4}\delta^2 \eta^2}{6144}\cdot 2 \cdot |H| \ge 1-\delta$, we have
\begin{align*}
        \mathrm{ALG} &= \widetilde{k}_{L} + \sum_{v\in H}\max\{\widetilde{k}_{S^*,v}, \widetilde{k}_{v,T^*}\} - \frac{\eta m}{16}\\
        &\le k_L + \frac{\eta m}{32} + \sum_{v\in H}\max\{{k}_{S^*,v}, {k}_{v,T^*}\} + |H|\cdot \frac{\eps^{4}\delta \eta^3 m}{49152} - \frac{\eta m}{16} &&\rhd \text{Lemma~\ref{lem:concentration-k-L} and Lemma~\ref{lem:concentration-general}}\\
        &\le k_L + \sum_{v\in H}\max\{{k}_{S^*,v}, {k}_{v,T^*}\} && \rhd |H|\le \frac{1536}{\eps^{4} \delta\eta^2}\\
        &\le \opt,
    \end{align*}
where the last step follows since $k_L + \sum_{v\in H}\max\{{k}_{S^*,v}, {k}_{v,T^*}\}$ is a lower bound on the value of a feasible dicut of $G$.
Specifically, we fix the restriction on $L$  to be $(S^*\cap L, T^*\cap L)$, and then assign each vertex $v \in H$ to the side that maximizes the resulting dicut value.

Furthermore, 
    \begin{align*}
        \mathrm{ALG} &= \widetilde{k}_{L} + \sum_{v\in H}\max\{\widetilde{k}_{S^*,v}, \widetilde{k}_{v,T^*}\} - \frac{\eta m}{16}\\
        &\ge k_L - \frac{\eta m}{32} + \sum_{v\in H}\max\{{k}_{S^*,v}, {k}_{v,T^*}\} - |H|\cdot \frac{\eps^{4}\delta \eta^3 m}{49152} - \frac{\eta m}{16} &&\rhd \text{Lemma~\ref{lem:concentration-k-L} and Lemma~\ref{lem:concentration-general}}\\
        &\ge k_L -\frac{\eta m}{32}+ \sum_{v\in H}\max\{{k}_{S^*,v}, {k}_{v,T^*}\} -\frac{\eta m}{32} - \frac{\eta m}{16} &&\rhd |H|\le \frac{1536}{\eps^{4} \delta\eta^2}\\
        &\ge k_L + k_{L,H} + k_{H,L} -\frac{\eta m}{8} &&\rhd \sum_{v\in H}\max\{{k}_{S^*,v}, {k}_{v,T^*}\} \ge k_{L,H} + k_{H,L}\\
        &= \opt - k_H  -\frac{\eta m}{8} &&\rhd \opt = k_L+k_H+ k_{L,H} + k_{H,L}\\
        &\ge 
        \opt - \frac{1536^2}{\eps^{8} \delta^2\eta^4} - \frac{\eta}{2} \cdot \opt &&\rhd \text{Fact~\ref{fact:size-H} and } m\le 4\opt\\
        &\ge (1-\eta)\cdot \opt, &&\rhd \frac{1536^2}{\eps^{8} \delta^2\eta^4} \le \frac{\eta m}{8}\le \frac{\eta}{2}\cdot \opt 
    \end{align*}
    where the fourth step follows from 
    \begin{align*}
        \sum_{v\in H}\max\{{k}_{S^*,v}, {k}_{v,T^*}\} &= \sum_{v\in H\cap S^*}\max\{{k}_{S^*,v}, {k}_{v,T^*}\}+\sum_{v\in H\cap T^*}\max\{{k}_{S^*,v}, {k}_{v,T^*}\} \\
        &\ge \sum_{v\in H\cap S^*} {k}_{v,T^*} + \sum_{v\in H\cap T^*} {k}_{S^*,v}\\
        &= k_{H,L} + k_{L,H}.
    \end{align*}
\end{proof}

\begin{remark}
\label{rmk:estimate-eps}
    In fact, it suffices to know a highly accurate estimate of $\eps$ for \cref{alg:general-offline} to work. Specifically, it is enough to have an estimate $\eps'\in (0,1/2]$ such that $\eps'=(1\pm O(\eta))\cdot \eps$. 
    We first update the degree threshold in Line~\ref{line:threshold} of \cref{alg:general-offline} to
    $\theta m$, where $\theta := \frac{\eps'^4 \delta \eta^2}{3072}$.
    Then we replace $\eps$ by $\eps'$ in the estimators. Namely, in Line~\ref{line:tildekL},  we set $\widetilde{k}_{L} := \frac{1}{4\eps'^2}[(\frac{1}{2}+\eps')^2 X_0 - (\frac{1}{2}+\eps')(\frac{1}{2}-\eps')X_1 + (\frac{1}{2}-\eps')^2 X_2]$; in Line~\ref{line:tilde-v-S*}, for each $v\in H$, we set $\widetilde{k}_{S^*,v} := (\frac{1}{2}+\frac{1}{4\eps'})X_{L^+,v} + (\frac{1}{2}-\frac{1}{4\eps'})X_{L^-,v}$ and $\widetilde{k}_{v,T^*} := (\frac{1}{2}-\frac{1}{4\eps'})X_{v,L^+} + (\frac{1}{2}+\frac{1}{4\eps'})X_{v,L^-}$.
    The full modified algorithm and its analysis are given in Appendix~\ref{sec:estimated-eps}.
\end{remark}

\subsection{Our algorithm in insertion-only streams}
\label{sec:omitted-code-insertion-dicut}
In this section we show how to implement the offline estimator (Algorithm~\ref{alg:general-offline}) in the streaming model. 
Suppose the edges of $G$ arrive one by one in an insertion-only stream, in an arbitrary order.

Recall that Algorithm~\ref{alg:general-offline} (i) partitions the vertices into low-degree $L$ and high-degree $H$,
(ii) maintains edge-type counts within $G[L]$ induced by the predictions and approximately computes (via a small linear system) an estimate of the optimal dicut contribution from $L$,
and (iii) for each $v\in H$ maintains incidence counts between $v$ and $L^+,L^-$ (in both directions) to optimize $v$'s contribution in the dicut estimation.
The main challenge in the streaming setting is that we do \emph{not} know $H$ and $L$ in advance, and we cannot afford to maintain degrees for all vertices explicitly.

Our streaming implementation is built around three ingredients that together let us ``simulate'' the offline estimator:
\begin{enumerate}
    \item 
    We maintain a uniform sample $F$ of $\poly(1/\eps,1/\delta,1/\eta)$ edges using reservoir sampling.
    At the end of the stream we set $\widetilde{H}$ to be the set of endpoints of edges in $F$.
    With high probability, every truly high-degree vertex appears in $F$ (hence $H\subseteq \widetilde{H}$, where $H$ is the set of high degree vertices defined in Algorithm~\ref{alg:general-offline}),
    while $|\widetilde{H}|=\poly(1/\eps,1/\delta,1/\eta)$ remains small.
    We then define $\widetilde{L}=V\setminus \widetilde{H}$.

    \item 
    Since we cannot store per-vertex incident counts, we maintain four CountMin sketches
    $\CMin[V^+]$, $\CMout[V^+]$, $\CMin[V^-]$, and $\CMout[V^-]$.
    Intuitively, they estimate (for each vertex $v$) the number of incoming/outgoing edges between $v$ and the predicted sides
    $V^+=\{u: \mathcal{O}(u) = +1\}$ and $V^-=\{u: \mathcal{O}(u) = -1\}$.
    When an edge $(u,v)$ arrives, we query both endpoints and update the appropriate sketches:
    we increment the estimate of edges from $V^{\mathcal{O}(u)}$ into $v$, and the estimate of edges from $u$ into $V^{\mathcal{O}(v)}$.
    This suffices to recover (up to a small additive error) the quantities $e(\widetilde{L}^+,v),e(v,\widetilde{L}^+),e(\widetilde{L}^-,v),e(v,\widetilde{L}^-)$ for each $v\in \widetilde{H}$ at the end of the stream.

    \item 
    In parallel, we keep exact counters for $e(V^+,V^-),e(V^-,V^+),e(V^+,V^+),e(V^-,V^-)$ over the entire stream.
    These provide the global edge-type statistics induced by the predictions.
\end{enumerate}

The algorithm is formally given in \cref{alg:insertion}. We have the following theoretical guarantee stating its approximation ratio and space complexity.

\begin{theorem}
\label{thm:insertion}
Let $\eps\in (0,\frac{1}{2}]$ and $\delta \in (0,1)$. For any $\eta \in (0,1)$,
given an $\eps$-accurate oracle $\mathcal{O}$ for the \textsc{Max-Dicut} of $G$, 
there exists a single-pass $(1-\eta)$-approximation algorithm for estimating the \textsc{Max-Dicut} value of $G$ in insertion-only streams. The algorithm uses $O(\frac{1}{\eps^{10}\delta^5\eta^{6}}\log\frac{1}{\eps \delta\eta})$ 
words of space. The approximation holds with probability at least $1-\delta$.
\end{theorem}

\begin{algorithm}[t!]
\begin{algorithmic}[1]
\Require Directed graph $G$ as an insertion-only stream of edges, an $\eps$-accurate oracle $\mathcal{O} \rightarrow \{-1, 1\}$, $\delta,\eta\in (0,1)$
\Ensure Estimate of the \textsc{Max-Dicut} value of $G$
\LeftComment{\textbf{\textcolor{blue}{Pre-processing phase}}}
\State {\bf initialize} $e(V^+,V^+)\leftarrow 0, e(V^-,V^-)\leftarrow 0, e(V^+,V^-)\leftarrow 0, e(V^-,V^+)\leftarrow 0, m\leftarrow 0$
\State {\bf initialize} four CountMin sketches $\CMin[V^+], \CMout[V^+], \CMin[V^-]$ and $\CMout[V^-]$ with accuracy parameter $\frac{\eps^{10}\delta^5\eta^{6}}{128\beta}$ and failure probability $\frac{\eps^{8}\delta^4\eta^4}{64\beta}$, where $\beta=2^{28}$.

\LeftComment{\textbf{\textcolor{blue}{Streaming phase}}}
\ForEach{edge $(u,v)$ in the stream}
    \State $m \leftarrow m+1$
    \State $Y_u = \mathcal{O}(u), Y_v = \mathcal{O}(v)$
    \IIf{$Y_u = 1$ and $Y_v = -1$} $e(V^+,V^-) \leftarrow e(V^+,V^-) + 1$
    \IIf{$Y_u = Y_v = 1$} $e(V^+,V^+) \leftarrow e(V^+,V^+) + 1$
    \IIf{$Y_u = Y_v = -1$} $e(V^-,V^-) \leftarrow e(V^-,V^-) + 1$
    \IIf{$Y_u = -1$ and $Y_v = 1$} $e(V^-,V^+) \leftarrow e(V^-,V^+) + 1$
    \State $\textsc{Update}(\CMin[V^{Y_u}][v]), \textsc{Update}(\CMout[V^{Y_v}][u])$
\EndFor

\State In parallel, uniformly sample $\frac{\beta}{\eps^{8}\delta^3\eta^4}$ edges $F$ in the stream via reservoir sampling.

\LeftComment{\textbf{\textcolor{blue}{Post-processing phase}}}
\ForEach{$v\in V$}
    \State $\widetilde{e}(V^+,v) \leftarrow \textsc{Query}(\CMin[V^+][v]), \quad \widetilde{e}(v,V^+) \leftarrow \textsc{Query}(\CMout[V^+][v])$
    \State $\widetilde{e}(V^-,v) \leftarrow \textsc{Query}(\CMin[V^-][v]),\quad \widetilde{e}(v,V^-) \leftarrow \textsc{Query}(\CMout[V^-][v])$
\EndFor
\State $\widetilde{H} \leftarrow \bigcup_{(u,v)\in F} \{u,v\}, \widetilde{L} \leftarrow V\setminus \widetilde{H}$
\State $\widetilde{H}^+ \leftarrow \{v\in \widetilde{H}: Y_v = 1\}, \widetilde{H}^- \leftarrow \widetilde{H} \setminus \widetilde{H}^+$
\State $\overline{X}_0 \leftarrow e(V^+,V^-) - \sum_{v\in \widetilde{H}^+} \widetilde{e}(v,V^-)-\sum_{v\in \widetilde{H}^-} \widetilde{e}(V^+,v)$
\State $\overline{X}_1 \leftarrow e(V^+,V^+) + e(V^-,V^-) - \sum_{v\in \widetilde{H}^+}(\widetilde{e}(v,V^+)+\widetilde{e}(V^+,v)) - \sum_{v\in \widetilde{H}^-}(\widetilde{e}(v,V^-)+\widetilde{e}(V^-,v))$ 
\State $\overline{X}_2 \leftarrow e(V^-,V^+) - \sum_{v\in \widetilde{H}^+} \widetilde{e}(V^-,v)-\sum_{v\in \widetilde{H}^-} \widetilde{e}(v,V^+)$
\State $\overline{k}_{\widetilde{L}}\leftarrow \frac{1}{4\eps^2}[(\frac{1}{2}+\eps)^2 \overline{X}_0 - (\frac{1}{2}+\eps)(\frac{1}{2}-\eps)\overline{X}_1 + (\frac{1}{2}-\eps)^2 \overline{X}_2]$
\ForEach{$v\in \widetilde{H}$} 
\State $\overline{k}_{S^*,v} \leftarrow (\frac{1}{2}+\frac{1}{4\eps})\cdot \widetilde{e}(V^+,v) + (\frac{1}{2}-\frac{1}{4\eps})\cdot \widetilde{e}(V^-,v)$
\State $\overline{k}_{v,T^*} \leftarrow (\frac{1}{2}-\frac{1}{4\eps})\cdot \widetilde{e}(v,V^+) + (\frac{1}{2}+\frac{1}{4\eps})\cdot \widetilde{e}(v,V^-)$ 
\EndFor
\State \textbf{return} $\overline{k}_{\widetilde{L}} + \sum_{v\in \widetilde{H}}\max\{\overline{k}_{S^*,v}, \overline{k}_{v,T^*}\}  -  \frac{\eta m}{16}$ 
\end{algorithmic}
\caption{Estimating the \textsc{Max-Dicut} value in insertion-only streams} 
\label{alg:insertion}
\end{algorithm}

We prove Theorem~\ref{thm:insertion} by establishing several lemmas that justify the correctness of our streaming simulation of Algorithm~\ref{alg:general-offline} using the high-degree threshold $\theta m$ for $\theta:=\frac{\eps^4 \delta \eta^2}{4096}$, i.e., $H:=\{v\in V: \deg(v) \ge \theta m\}$.
As we will see, this choice of threshold guarantees that our streaming-constructed partition $(\widetilde L,\widetilde H)$ satisfies the key properties required for the offline analysis, and yields the desired success probability.
We begin by showing that reservoir sampling identifies a superset $\widetilde{H}$ of the high-degree vertices $H$ with high probability.
Recall that $\widetilde{H}$ is the set of endpoints  of all edges in the uniformly sampled edge set $F$. 

\begin{lemma}
\label{lem:high-deg-reservoir}
Let $\delta \in (0,1)$.
Then
 $\widetilde{H}$ contains all high-degree vertices $H$ in $G$ (i.e., $H \subseteq \widetilde{H}$) with probability at least $1-\frac{\delta}{8}$.
\end{lemma}
\begin{proof}
    Let $\theta:=\frac{\eps^{4}\delta\eta^2}{4096}$. Consider a vertex $v \in H$, so $\deg(v) \geq \theta m$. 
    For each directed edge $e_i$ incident to $v$, define an indicator random variable 
    \begin{equation*}
            X_{i}= \begin{cases}1 & \text { if } e_i\in F, \\ 0 & \text { otherwise. }\end{cases}
    \end{equation*}
    So, $\E[X_i] = \Pr[X_i=1] = \frac{\beta}{\eps^{8}\delta^3\eta^4m}$. We define $X:= \sum_{i\in [\deg(v)]} X_i$ to denote the number of edges incident to $v$ that are sampled by the end of the stream. Then, $\E[X] = \sum_{i\in [\deg(v)]} \E[X_i] = \frac{\beta\deg(v)}{\eps^{8}\delta^3 \eta^4m}$. 

    Next, we bound $\Var[X]$. For any $i\in [\deg(v)]$, $\Var[X_i] = \E[X^2_i] - (\E[X_i])^2 = \E[X_i] - (\E[X_i])^2 \le \frac{\beta}{\eps^{8} \delta^3\eta^4 m}$.
    Since $(X_i)_{i\in [\deg(v)]}$ are negatively correlated, $\Var[X] \le \sum_{i\in [\deg(v)]}\Var[X_i] \le \deg(v)\cdot \frac{\beta}{\eps^{8} \delta^3\eta^4 m} = \frac{\beta\deg(v)}{\eps^{8} \delta^3\eta^4 m}$.
    Therefore, by Chebyshev's inequality, 
    \begin{align*}
        \Pr[v\notin \widetilde{H}] = \Pr[X\leq 0] 
        &= \Pr\left[X \leq \E[X]- \frac{\beta\deg(v)}{\eps^{8} \delta^3\eta^4 m}\right]\\
        &\leq \frac{\Var[X]\cdot \eps^{16} \delta^6 \eta^ 8 m^2}{ \beta^2\deg^2(v)}\\
        &\leq \frac{\eps^{8} \delta^3 \eta^4 m}{\beta \deg(v)} &&\rhd \Var[X] \leq \frac{\beta\deg(v)}{\eps^{8} \delta^3 \eta^4 m}\\
        &\leq \frac{4096\eps^{4}\delta^2\eta^2}{\beta}. &&\rhd \deg(v)\geq \theta m
    \end{align*}
    Recall that $|H| \le \frac{2m}{\theta m}=\frac{2}{\theta} = \frac{8192}{\eps^{4} \delta\eta^2}$.
    By union bound, with probability at least $1-\frac{8192}{\eps^{4} \delta\eta^2}\cdot \frac{4096\eps^{4}\delta^2\eta^2}{\beta}  = 1-\frac{\delta}{8}$ (since $\beta = 2^{28}$), $\widetilde{H}$ contains all high-degree vertices $H$ in $G$.
\end{proof}

\begin{lemma}
\label{lem:error-bound-insertion}
    Let $\eps \in (0,\frac{1}{2}]$ and $\delta,\eta \in (0,1)$.
    For all vertices $v\in \widetilde{H}$, we have
    \begin{enumerate}
        \item[(i)] $\widetilde{e}(V^+,v) \geq e(V^+,v),\  \widetilde{e}(v,V^+) \geq e(v,V^+),\  \widetilde{e}(V^-,v) \geq e(V^-,v),\  \widetilde{e}(v,V^-) \geq e(v,V^-)$;
        \item[(ii)] With probability at least $1-\frac{\delta}{8}$, all of the following inequalities hold: $\widetilde{e}(V^+,v) \leq e(V^+,v) + \frac{\eps^{10}\delta^5\eta^{6}}{128\beta} m,\  \widetilde{e}(v,V^+) \leq e(v,V^+)+ \frac{\eps^{10}\delta^5\eta^{6}}{128\beta} m,\  \widetilde{e}(V^-,v) \leq e(V^-,v)+ \frac{\eps^{10}\delta^5\eta^{6}}{128\beta} m,\  \widetilde{e}(v,V^-) \leq e(v,V^-)+\frac{\eps^{10}\delta^5\eta^{6}}{128\beta} m $. 
    \end{enumerate} 
\end{lemma}
\begin{proof}
    (i) By \cref{thm:CM-v2}, for all $v\in V$, we have $\widetilde{e}(V^+,v) \geq e(V^+,v),  \widetilde{e}(v,V^+) \geq e(v,V^+),  \widetilde{e}(V^-,v) \geq e(V^-,v)$, and  $\widetilde{e}(v,V^-) \geq e(v,V^-)$.

    (ii) By \cref{thm:CM-v2}, for any $v\in V$, with probability at least $1-\frac{\eps^{8}\delta^4\eta^4}{64\beta}$, each of the following four inequalities holds: $\widetilde{e}(V^+,v) \leq e(V^+,v) + \frac{\eps^{10}\delta^5\eta^{6}}{128\beta} m,\  \widetilde{e}(v,V^+) \leq e(v,V^+)+ \frac{\eps^{10}\delta^5\eta^{6}}{128\beta} m,\ \widetilde{e}(V^-,v) \leq e(V^-,v)+\frac{\eps^{10}\delta^5\eta^{6}}{128\beta} m$, and  $\widetilde{e}(v,V^-) \leq e(v,V^-)+ \frac{\eps^{10}\delta^5\eta^{6}}{128\beta} m$.
    Since $|\widetilde{H}|\le \frac{2\beta}{\eps^{8}\delta^3\eta^4}$, by applying the union bound we have that, with probability at least $1-\frac{\eps^{8}\delta^4\eta^4}{64\beta} \cdot \frac{2\beta}{\eps^{8}\delta^3\eta^4} \cdot 4 \ge 1-\frac{\delta}{8}$, for all $v\in \widetilde{H}$, all of the above four inequalities hold.
\end{proof}

\paragraph{Proof of \cref{thm:insertion}.}
Given the above lemmas, we are ready to prove \cref{thm:insertion}. 
\begin{proof}

We can assume that $m =\Omega(\eps^{-26}\delta^{-10}\eta^{-14})$, as otherwise we can store the whole graph in $O(\poly(1/\eps, 1/\delta, 1/\eta))$ space.
    
In \cref{alg:insertion}, we maintain four CountMin sketches with accuracy parameter $\frac{\eps^{10}\delta^5\eta^{6}}{128\beta}$ 
and failure probability $\frac{\eps^{8}\delta^4\eta^4}{64\beta}$, where $\beta=2^{28}$. By \cref{thm:CM-v2}, this takes $O(\frac{1}{\eps^{10}\delta^5\eta^{6}}\log\frac{1}{\eps^{8} \delta^4\eta^4}) = O( \frac{1}{\eps^{10}\delta^5\eta^{6}}\log\frac{1}{\eps \delta\eta} )$ words of space. 
    Also, we use reservoir sampling to uniformly sample $O(\frac{1}{\eps^{8}\delta^3\eta^4})$ edges, which takes $O(\frac{1}{\eps^{8}\delta^3\eta^4})$ words of space. 
    Therefore, the total space complexity of \cref{alg:insertion} is $O( \frac{1}{\eps^{10}\delta^5\eta^{6}}\log\frac{1}{\eps \delta\eta} )$ words.

    Next, we analyze the approximation ratio.
Since $V^+ = \widetilde{H}^+ \cup \widetilde{L}^+$ and $V^- = \widetilde{H}^- \cup \widetilde{L}^-$, for any $v \in V$, we have
\begin{align*}
    e(V^+,v) &= e(\widetilde{H}^+,v)+ e(\widetilde{L}^+,v), &e(v,V^+) &= e(v,\widetilde{H}^+) + e(v,\widetilde{L}^+),\\
    e(V^-,v) &= e(\widetilde{H}^-,v)+ e(\widetilde{L}^-,v), &e(v,V^-) &= e(v,\widetilde{H}^-) + e(v,\widetilde{L}^-).
\end{align*}
By Lemma~\ref{lem:error-bound-insertion}, with probability at least $1-\frac{\delta}{8}$, for all vertices $v\in \widetilde{H}$, all of the following inequalities hold: 
\begin{align*}
    e(\widetilde{H}^+,v)+ e(\widetilde{L}^+,v) \leq \widetilde{e}(V^+,v) \leq e(\widetilde{H}^+,v)+ e(\widetilde{L}^+,v) + \frac{\eps^{10}\delta^5\eta^{6}}{128\beta} m,\\
    e(v,\widetilde{H}^+) + e(v,\widetilde{L}^+) \leq \widetilde{e}(v,V^+) \leq e(v,\widetilde{H}^+) + e(v,\widetilde{L}^+) + \frac{\eps^{10}\delta^5\eta^{6}}{128\beta} m,\\
    e(\widetilde{H}^-,v)+ e(\widetilde{L}^-,v) \leq \widetilde{e}(V^-,v) \leq e(\widetilde{H}^-,v)+ e(\widetilde{L}^-,v) + \frac{\eps^{10}\delta^5\eta^{6}}{128\beta} m,\\
    e(v,\widetilde{H}^-) + e(v,\widetilde{L}^-) \leq \widetilde{e}(v,V^-) \leq e(v,\widetilde{H}^-) + e(v,\widetilde{L}^-) + \frac{\eps^{10}\delta^5\eta^{6}}{128\beta} m.
\end{align*}
Since $|F|=\frac{\beta}{\eps^{8}\delta^3\eta^4}$ and $\widetilde{H}$ is the set of endpoints of edges in $F$, we have that $|\widetilde{H}|\le \frac{2\beta}{\eps^{8}\delta^3\eta^4}$. Thus, it holds that $0 \le e(\widetilde{H}^+,v), e(v,\widetilde{H}^+), e(\widetilde{H}^-,v), e(v,\widetilde{H}^-)\le \frac{2\beta}{\eps^{8}\delta^3\eta^4}$. Therefore,
\begin{align}
      0\leq \widetilde{e}(V^+,v) -  e(\widetilde{L}^+,v)\leq \frac{2\beta}{\eps^{8}\delta^3\eta^4}+  \frac{\eps^{10}\delta^5\eta^{6}}{128\beta} m,\label{ineq:cm_plus_in}\\ 
      0 \leq \widetilde{e}(v,V^+) - e(v,\widetilde{L}^+) \leq \frac{2\beta}{\eps^{8}\delta^3\eta^4} +  \frac{\eps^{10}\delta^5\eta^{6}}{128\beta} m,\label{ineq:cm_plus_out}\\
      0 \leq \widetilde{e}(V^-,v) - e(\widetilde{L}^-,v) \leq \frac{2\beta}{\eps^{8}\delta^3\eta^4}+  \frac{\eps^{10}\delta^5\eta^{6}}{128\beta} m,\label{ineq:cm_minus_in}\\
      0 \leq \widetilde{e}(v,V^-) - e(v,\widetilde{L}^-) \leq \frac{2\beta}{\eps^{8}\delta^3\eta^4} + \frac{\eps^{10}\delta^5\eta^{6}}{128\beta} m. \label{ineq:cm_minus_out}
\end{align}

For each edge with both endpoints in $\widetilde{L}$,
let $\widetilde{X}_0,\widetilde{X}_1$, and $\widetilde{X}_2$ denote, respectively,  the numbers of edges with (i) $(Y_u,Y_v)=(+1,-1)$, (ii) $Y_u = Y_v$, and (iii) $(Y_u,Y_v)=(-1,+1)$. We have
\begin{align*}
    \widetilde{X}_0 
    =~&e(\widetilde{L}^+,\widetilde{L}^-) = e(V^+,V^-) - e(\widetilde{H}^+,\widetilde{H}^-) - \sum_{v\in \widetilde{H}^+}e(v,\widetilde{L}^-)-\sum_{v\in \widetilde{H}^-}e(\widetilde{L}^+,v),\\
    \widetilde{X}_1 
    =~&e(\widetilde{L}^+,\widetilde{L}^+) + e(\widetilde{L}^-,\widetilde{L}^-)\\
    =~&e(V^+,V^+) + e(V^-,V^-) - e(\widetilde{H}^+,\widetilde{H}^+) - e(\widetilde{H}^-,\widetilde{H}^-) \\ 
    &-\sum_{v\in \widetilde{H}^+}(e(v,\widetilde{L}^+) + e(\widetilde{L}^+,v)) 
    - \sum_{v\in \widetilde{H}^-}(e(v,\widetilde{L}^-) + e(\widetilde{L}^-,v)), \\
    \widetilde{X}_2 
    =~&e(\widetilde{L}^-,\widetilde{L}^+) = e(V^-,V^+) - e(\widetilde{H}^-,\widetilde{H}^+) -\sum_{v\in \widetilde{H}^+}e(\widetilde{L}^-,v)-\sum_{v\in \widetilde{H}^-}e(v,\widetilde{L}^+).
\end{align*}

Recall that
\begin{align*}
\overline{X}_0 &= e(V^+,V^-) - \sum_{v\in \widetilde{H}^+} \widetilde{e}(v,V^-)-\sum_{v\in \widetilde{H}^-} \widetilde{e}(V^+,v),\\
\overline{X}_1 &=e(V^+,V^+) + e(V^-,V^-) - \sum_{v\in \widetilde{H}^+}(\widetilde{e}(v,V^+)+\widetilde{e}(V^+,v)) - \sum_{v\in \widetilde{H}^-}(\widetilde{e}(v,V^-)+\widetilde{e}(V^-,v)),\\
\overline{X}_2 &= e(V^-,V^+) - \sum_{v\in \widetilde{H}^+} \widetilde{e}(V^-,v)-\sum_{v\in \widetilde{H}^-} \widetilde{e}(v,V^+).
\end{align*}
Therefore,
\begin{align*}
    \overline{X}_0 - \widetilde{X}_0 &= e(\widetilde{H}^+,\widetilde{H}^-) + \sum_{v\in \widetilde{H}^+}(e(v,\widetilde{L}^-) - \widetilde{e}(v,V^-)) + \sum_{v\in \widetilde{H}^-} (e(\widetilde{L}^+,v) - \widetilde{e}(V^+,v)),\\
    \overline{X}_1 - \widetilde{X}_1 &= e(\widetilde{H}^+,\widetilde{H}^+) + e(\widetilde{H}^-,\widetilde{H}^-) + \sum_{v\in \widetilde{H}^+} \left[(e(v,\widetilde{L}^+) - \widetilde{e}(v,V^+)) +  (e(\widetilde{L}^+,v) - \widetilde{e}(V^+,v))\right]\\
     &\quad+ \sum_{v\in \widetilde{H}^-} \left[(e(v,\widetilde{L}^-) - \widetilde{e}(v,V^-)) +  (e(\widetilde{L}^-,v) - \widetilde{e}(V^-,v))\right],\\
    \overline{X}_2 - \widetilde{X}_2 &=e(\widetilde{H}^-,\widetilde{H}^+) + \sum_{v\in \widetilde{H}^+}(e(\widetilde{L}^-,v) - \widetilde{e}(V^-,v)) + \sum_{v\in \widetilde{H}^-} (e(v,\widetilde{L}^+) - \widetilde{e}(v,V^+)).
\end{align*}
Since $|\widetilde{H}|\le \frac{2\beta}{\eps^{8}\delta^3\eta^4}$, we have $|\widetilde{H}^+|, |\widetilde{H}^-|\le \frac{2\beta}{\eps^{8}\delta^3\eta^4}$, and $e(\widetilde{H}^+,\widetilde{H}^-),e(\widetilde{H}^-,\widetilde{H}^+),e(\widetilde{H}^+,\widetilde{H}^+), e(\widetilde{H}^-,\widetilde{H}^-)\le \frac{4\beta^2}{\eps^{16}\delta^6 \eta^8}$.
Moreover, combining inequalities (\ref{ineq:cm_plus_in})--(\ref{ineq:cm_minus_out}), we have
\begin{align*}
    \widetilde{X}_0-\left(\frac{8\beta^2}{\eps^{16}\delta^6\eta^8}+\frac {\eps^{2}\delta^2 \eta^2 m}{32} \right)\le \overline{X}_0  \le \widetilde{X}_0+\frac{4\beta^2}{\eps^{16}\delta^6 \eta^8},\\
    \widetilde{X}_1-\left(\frac{16\beta^2}{\eps^{16}\delta^6\eta^8}+\frac {\eps^{2}\delta^2 \eta^2 m}{16}  \right) \le \overline{X}_1  \le \widetilde{X}_1+\frac{4\beta^2}{\eps^{16}\delta^6 \eta^8},\\
    \widetilde{X}_2-\left(\frac{8\beta^2}{\eps^{16}\delta^6\eta^8}+\frac {\eps^{2}\delta^2 \eta^2 m}{32} \right) \le \overline{X}_2   \le \widetilde{X}_2+\frac{4\beta^2}{\eps^{16}\delta^6 \eta^8}.
\end{align*}

Let $\widetilde{k}_{\widetilde{L}} := \frac{1}{4\eps^2}[(\frac{1}{2}+\eps)^2 \widetilde{X}_0 - (\frac{1}{2}+\eps)(\frac{1}{2}-\eps)\widetilde{X}_1 + (\frac{1}{2}-\eps)^2 \widetilde{X}_2]$. 
Now we are ready to show that $\overline{k}_{\widetilde{L}}$ is close to $\widetilde{k}_{\widetilde{L}}$ up to an additive error of $\poly(\eta)\cdot m$ with probability at least $1-\frac{\delta}{8}$.
We have
\begin{align*}
    \overline{k}_{\widetilde{L}} =~&\frac{1}{4\eps^2}\left[\left(\frac{1}{2}+\eps\right)^2 \overline{X}_0 - \left(\frac{1}{2}+\eps\right)\left(\frac{1}{2}-\eps\right)\overline{X}_1 + \left(\frac{1}{2}-\eps\right)^2 \overline{X}_2\right]\\
    \ge~&\frac{1}{4\eps^2}\left[\left(\frac{1}{2}+\eps\right)^2 \left(\widetilde{X}_0-\left(\frac{8\beta^2}{\eps^{16}\delta^6\eta^8}+\frac {\eps^{2}\delta^2\eta^2 m}{32}  \right)\right) - \left(\frac{1}{2}+\eps\right)\left(\frac{1}{2}-\eps\right)\left(\widetilde{X}_1+\frac{4\beta^2}{\eps^{16}\delta^6\eta^8}\right)\right. \\
    &+ \left.\left(\frac{1}{2}-\eps\right)^2 \left(\widetilde{X}_2-\left(\frac{8\beta^2}{\eps^{16}\delta^6\eta^8}+\frac {\eps^{2}\delta^2\eta^2 m}{32}  \right)\right)\right]\\
    =~&\widetilde{k}_{\widetilde{L}} - \frac{1}{4\eps^2}\left[\left(\frac{1}{2}+\eps\right)^2 \left(\frac{8\beta^2}{\eps^{16}\delta^6\eta^8}+\frac {\eps^2\delta^2\eta^2 m}{32}  \right) + \left(\frac{1}{2}+\eps\right)\left(\frac{1}{2}-\eps\right)\frac{4\beta^2}{\eps^{16}\delta^6\eta^8} \right.\\
    &+ \left.\left(\frac{1}{2}-\eps\right)^2 \left(\frac{8\beta^2}{\eps^{16}\delta^6\eta^8}+\frac {\eps^2\delta^2\eta^2 m}{32}  \right)\right]\\
    =~&\widetilde{k}_{\widetilde{L}} - \left(\frac{5\beta^{2}}{4\eps^{18}\delta^{6}\eta^8} + \frac{3\beta^{2}}{\eps^{16}\delta^{6}\eta^8} + \frac{\delta^2 \eta^2 m}{256} + \frac{\eps^{2}\delta^2 \eta^2m}{64}\right) \\
    \ge~&\widetilde{k}_{\widetilde{L}} - \frac{\eta m}{32},
\end{align*}
where the last step follows since $m = \Omega(\eps^{-26}\delta^{-10}\eta^{-14})$. On the other hand,
\begin{align*}
\overline{k}_{\widetilde{L}} =~&\frac{1}{4\eps^2}\left[\left(\frac{1}{2}+\eps\right)^2 \overline{X}_0 - \left(\frac{1}{2}+\eps\right)\left(\frac{1}{2}-\eps\right)\overline{X}_1 + \left(\frac{1}{2}-\eps\right)^2 \overline{X}_2\right]\\
\le~&\frac{1}{4\eps^2}\left[\left(\frac{1}{2}+\eps\right)^2 \left(\widetilde{X}_0+\frac{4\beta^2}{\eps^{16}\delta^6 \eta^8}\right) - \left(\frac{1}{2}+\eps\right)\left(\frac{1}{2}-\eps\right)\left(\widetilde{X}_1-\left(\frac{16\beta^2}{\eps^{16}\delta^6\eta^8}+\frac {\eps^{2}\delta^2\eta^2 m}{16}  \right)\right)\right.\\
&+\left.\left(\frac{1}{2}-\eps\right)^2 \left(\widetilde{X}_2+\frac{4\beta^2}{\eps^{16}\delta^6 \eta^8}\right)\right]\\
=~&\widetilde{k}_{\widetilde{L}} +\frac{1}{4\eps^2}\left[\left(\frac{1}{2}+\eps\right)^2 \frac{4\beta^2}{\eps^{16}\delta^6\eta^8} + \left(\frac{1}{2}+\eps\right)\left(\frac{1}{2}-\eps\right)\left(\frac{16\beta^2}{\eps^{16}\delta^6\eta^8}+\frac {\eps^{2}\delta^2\eta^2 m}{16}  \right) + \left(\frac{1}{2}-\eps\right)^2 \frac{4\beta^2}{\eps^{16}\delta^6\eta^8}\right]\\
=~&\widetilde{k}_{\widetilde{L}} + \left(\frac{3\beta^{2}}{2\eps^{18}\delta^{6}\eta^8} - \frac{2\beta^{2}}{\eps^{16}\delta^{6}\eta^8}+\frac{ \delta^{2}\eta^2m}{256}-\frac{\eps^{2}\delta^2\eta^2m}{64}\right)\\
\le~&\widetilde{k}_{\widetilde{L}} + \frac{\eta m}{64},
\end{align*}
where the last step follows since $m = \Omega(\eps^{-26}\delta^{-10}\eta^{-14})$. 

For each $v\in \widetilde{H}$, let $\widetilde{k}_{S^*,v} := (\frac{1}{2}+\frac{1}{4\eps})\cdot e(\widetilde{L}^+,v) + (\frac{1}{2}-\frac{1}{4\eps})\cdot e(\widetilde{L}^-,v)$ and $\widetilde{k}_{v,T^*} := (\frac{1}{2}-\frac{1}{4\eps})\cdot e(v,\widetilde{L}^+) + (\frac{1}{2}+\frac{1}{4\eps})\cdot e(v,\widetilde{L}^-)$.
We have
\begin{align*}
    \overline{k}_{S^*,v} &= \left(\frac{1}{2}+\frac{1}{4\eps}\right)\cdot \widetilde{e}(V^+,v) + \left(\frac{1}{2}-\frac{1}{4\eps}\right)\cdot \widetilde{e}(V^-,v)\\
    &\ge \left(\frac{1}{2}+\frac{1}{4\eps}\right)\cdot e(\widetilde{L}^+,v) + \left(\frac{1}{2}-\frac{1}{4\eps}\right) \left( e(\widetilde{L}^-,v) + \frac{2\beta}{\eps^{8}\delta^3\eta^4}+  \frac{\eps^{10}\delta^5\eta^{6}m}{128\beta}  \right)\\
    &= \widetilde{k}_{S^*,v} - \left(\frac{1}{4\eps}-\frac{1}{2}\right) \left(\frac{2\beta}{\eps^{8}\delta^3\eta^4}+  \frac{\eps^{10}\delta^5\eta^{6}m}{128\beta}  \right)\\
    &= \widetilde{k}_{S^*,v} -\left(\frac{\beta}{2\eps^{9}\delta^3\eta^4} - \frac{\beta}{\eps^{8}\delta^3\eta^4} + \frac{\eps^{9}\delta^5\eta^{6} m}{512\beta} - \frac{\eps^{10}\delta^5\eta^{6} m}{256\beta}\right)\\
    &\ge \widetilde{k}_{S^*,v} - \frac{\eps^8 \delta^3 \eta^5 m}{128\beta},
\end{align*}
where the last step follows since $m = \Omega(\eps^{-26}\delta^{-10}\eta^{-14})$. On the other hand,
\begin{align*}
    \overline{k}_{S^*,v} &= \left(\frac{1}{2}+\frac{1}{4\eps}\right)\cdot \widetilde{e}(V^+,v) + \left(\frac{1}{2}-\frac{1}{4\eps}\right)\cdot \widetilde{e}(V^-,v)\\
    &\le \left(\frac{1}{2}+\frac{1}{4\eps}\right)\left( e(\widetilde{L}^+,v) + \frac{2\beta}{\eps^{8}\delta^3\eta^4}+  \frac{\eps^{10}\delta^5\eta^{6}m}{128\beta}  \right) + \left(\frac{1}{2}-\frac{1}{4\eps}\right)\cdot e(\widetilde{L}^-,v)\\
    &= \widetilde{k}_{S^*,v} + \left(\frac{1}{2}+\frac{1}{4\eps}\right)\left(\frac{2\beta}{\eps^{8}\delta^3\eta^4}+  \frac{\eps^{10}\delta^5\eta^{6}m}{128\beta}  \right)\\
    &= \widetilde{k}_{S^*,v} +\left(\frac{\beta}{2\eps^{9}\delta^3\eta^4} + \frac{\beta}{\eps^{8}\delta^3\eta^4} + \frac{\eps^{9}\delta^5 \eta^6 m}{512\beta} + \frac{\eps^{10}\delta^5\eta^6 m}{256\beta}\right)\\
    &\le \widetilde{k}_{S^*,v} +\frac{\eps^8 \delta^3 \eta^5 m}{128\beta},
\end{align*}
where the last step follows since $m = \Omega(\eps^{-26}\delta^{-10}\eta^{-14})$. Therefore, $|\overline{k}_{S^*,v} - \widetilde{k}_{S^*,v}| \le \frac{\eps^8 \delta^3 \eta^5 m}{128\beta}$.
Similarly, $|\overline{k}_{v,T^*} - \widetilde{k}_{v,T^*}| \le \frac{\eps^8 \delta^3 \eta^5 m}{128\beta}$.

Let $\widetilde{\mathrm{ALG}}:= \widetilde{k}_{\widetilde{L}} + \sum_{v\in \widetilde{H}}\max\{\widetilde{k}_{S^*,v}, \widetilde{k}_{v,T^*}\}  - \frac{\eta m}{32}$ denote the output of the algorithm obtained by running \cref{alg:general-offline} with $L$ and $H$ replaced by $\widetilde L$ and $\widetilde H$, respectively.
Analogous to the proof of \cref{thm:general-offline}, we can show that $(1-\frac{\eta}2)\cdot \opt\le \widetilde{\mathrm{ALG}} \le \opt$ with probability at least $1-\frac{3\delta}{4}$. The  argument relies primarily on two properties of the partition $(\widetilde{L},\widetilde{H})$: (i) every vertex in $\widetilde{L}$ has degree at most $\theta m$ for $\theta:=\frac{\eps^4 \delta \eta^2}{4096}$ (since $\widetilde{L} \subseteq L$), and (ii) $|\widetilde{H}|\leq \frac{2\beta}{\eps^8 \delta^3 \eta^4}$.

Let $\mathrm{ALG}$ denote the output of \cref{alg:insertion}. By applying the union bound on Lemma~\ref{lem:high-deg-reservoir} and Lemma~\ref{lem:error-bound-insertion}, with probability at least $1-\frac{\delta}{8}-\frac{\delta}{8}-\frac{3\delta}{4} = 1-\delta$, we have
\begin{align*}
    \mathrm{ALG} &= \overline{k}_{\widetilde{L}} + \sum_{v\in \widetilde{H}}\max\{\overline{k}_{S^*,v}, \overline{k}_{v,T^*}\} -  \frac{\eta m}{16}\\
    &\le \widetilde{k}_{\widetilde{L}} + \frac{\eta m}{64} + \sum_{v\in \widetilde{H}}\max\{\widetilde{k}_{S^*,v}, \widetilde{k}_{v,T^*}\} + |\widetilde{H}|\cdot \frac{\eps^8 \delta^3 \eta^5 m}{128\beta} - \frac{\eta m}{16}\\
    &\le \widetilde{\mathrm{ALG}} + \frac{\eta m}{32} + \frac{\eta m}{64}+\frac{\eta m}{64} - \frac{\eta m}{16} &&\rhd |\widetilde{H}|\le \frac{2\beta}{\eps^{8}\delta^3\eta^4}\\
    &\le \opt,
\end{align*}
and
\begin{align*}
    \mathrm{ALG} &= \overline{k}_{\widetilde{L}} + \sum_{v\in \widetilde{H}}\max\{\overline{k}_{S^*,v}, \overline{k}_{v,T^*}\} - \frac{\eta m}{16}\\
    &\ge \widetilde{k}_{\widetilde{L}} - \frac{\eta m}{32} + \sum_{v\in \widetilde{H}}\max\{\widetilde{k}_{S^*,v}, \widetilde{k}_{v,T^*}\} - |\widetilde{H}|\cdot \frac{\eps^8 \delta^3 \eta^5 m}{128\beta} - \frac{\eta m}{16} \\
    &\ge \widetilde{\mathrm{ALG}} -\frac{\eta m}{64} - \frac{\eta m}{16} &&\rhd |\widetilde{H}|\le \frac{2\beta}{\eps^{8}\delta^3\eta^4}\\
    &\geq \widetilde{\mathrm{ALG}} - \frac{\eta m}{8} \\
    &\ge \left(1-\frac\eta2\right)\cdot \opt - \frac{\eta}{2} \cdot\opt && \rhd m \le 4\cdot \opt \\
    &= (1-\eta)\cdot \opt.
\end{align*}
\end{proof}

\subsection{Extension to dynamic streams}
\label{sec:dynamic-max-dicut}
In this subsection, we present our algorithm for \textsc{Max-Dicut} in dynamic streams. Similar to \cite{dong2025learning}'s method, our algorithm in dynamic streams can be easily obtained by modifying our algorithm in insertion-only streams (i.e., \cref{alg:insertion}). Specifically, we use $\ell_0$-sampling~\citep{jowhari2011tight} 
to compute a small set $\widetilde{H}$ at the end of the stream such that $\widetilde{H} \supseteq H$.
The algorithm is formally given in Algorithm~\ref{alg:dynamic-max-dicut}.

We have the following theoretical guarantee stating its approximation ratio and space complexity.

\begin{theorem}
\label{thm:dynamic-max-dicut}
Let $\eps\in (0,\frac{1}{2}]$ and $\delta \in (0,1)$. For any $\eta \in (0,1)$,
given an $\eps$-accurate oracle $\mathcal{O}$ for the \textsc{Max-Dicut} of $G$, 
there exists a single-pass $(1 - \eta)$-approximation  algorithm for estimating the \textsc{Max-Dicut} value of $G$ in dynamic streams. The algorithm uses $O(\frac{\log^2 n}{\eps^{10}\delta^5\eta^{6}}\log\frac{1}{\eps \delta\eta})$ bits of space. 
The approximation holds with probability at least $1-\delta$. 
\end{theorem}

\begin{algorithm}[h!]
\begin{algorithmic}[1]
\Require Directed graph $G$ as a dynamic stream of edges, an $\eps$-accurate oracle $\mathcal{O} \rightarrow \{-1, 1\}$, $\delta,\eta\in (0,1)$
\Ensure Estimate of the \textsc{Max-Dicut} value of $G$
\LeftComment{\textbf{\textcolor{blue}{Pre-processing phase}}}
\State {\bf initialize} $e(V^+,V^+)\leftarrow 0, e(V^-,V^-)\leftarrow 0, e(V^+,V^-)\leftarrow 0, e(V^-,V^+)\leftarrow 0, m\leftarrow 0$
\State {\bf initialize} four CountMin sketches $\CMin[V^+], \CMout[V^+], \CMin[V^-]$ and $\CMout[V^-]$ with accuracy parameter $\frac{\eps^{10}\delta^5\eta^{6}}{128\beta}$ and failure probability $\frac{\eps^{8}\delta^4\eta^4}{64\beta}$, where $\beta=2^{29}\cdot 25$.

\LeftComment{\textbf{\textcolor{blue}{Streaming phase}}}
\ForEach{item $((u,v), \Delta_{(u,v)})$ in the dynamic stream}
    \State $Y_u = \mathcal{O}(u), Y_v = \mathcal{O}(v)$
    \If{$\Delta_{(u,v)}=1$}
        \State $m \leftarrow m+1$
        \IIf{$Y_u = 1$ and $Y_v = -1$} $e(V^+,V^-) \leftarrow e(V^+,V^-) + 1$
        \IIf{$Y_u = Y_v = 1$} $e(V^+,V^+) \leftarrow e(V^+,V^+) + 1$
        \IIf{$Y_u = Y_v = -1$} $e(V^-,V^-) \leftarrow e(V^-,V^-) + 1$
        \IIf{$Y_u = -1$ and $Y_v = 1$} $e(V^-,V^+) \leftarrow e(V^-,V^+) + 1$
    \Else{}
        \State $m \leftarrow m-1$
        \IIf{$Y_u = 1$ and $Y_v = -1$} $e(V^+,V^-) \leftarrow e(V^+,V^-) - 1$
        \IIf{$Y_u = Y_v = 1$} $e(V^+,V^+) \leftarrow e(V^+,V^+) - 1$
        \IIf{$Y_u = Y_v = -1$} $e(V^-,V^-) \leftarrow e(V^-,V^-) - 1$
        \IIf{$Y_u = -1$ and $Y_v = 1$} $e(V^-,V^+) \leftarrow e(V^-,V^+) - 1$
    \EndIf
    \State $\textsc{Update}(\CMin[V^{Y_u}][v],\Delta_{(u,v)}), \textsc{Update}(\CMout[V^{Y_v}][u],\Delta_{(u,v)})$
    
\EndFor

\State In parallel, maintain $\frac{\beta}{\eps^{8}\delta^3\eta^4}$ $\ell_0$-samplers (with failure probability $\frac{\eps^8 \delta^4 \eta^4}{32\beta}$) to obtain $\frac{\beta}{\eps^{8}\delta^3\eta^4}$ different edges $F$.\phantomsection \label{line:ell-0-sampling-dicut}

\LeftComment{\textbf{\textcolor{blue}{Post-processing phase}}}
\ForEach{$v\in V$}
    \State $\widetilde{e}(V^+,v) \leftarrow \textsc{Query}(\CMin[V^+][v]), \quad \widetilde{e}(v,V^+) \leftarrow \textsc{Query}(\CMout[V^+][v])$
    \State $\widetilde{e}(V^-,v) \leftarrow \textsc{Query}(\CMin[V^-][v]),\quad \widetilde{e}(v,V^-) \leftarrow \textsc{Query}(\CMout[V^-][v])$
\EndFor
\State $\widetilde{H} \leftarrow \bigcup_{(u,v)\in F} \{u,v\}, \widetilde{L} \leftarrow V\setminus \widetilde{H}$
\State $\widetilde{H}^+ \leftarrow \{v\in \widetilde{H}: Y_v = 1\}, \widetilde{H}^- \leftarrow \widetilde{H} \setminus \widetilde{H}^+$
\State $\overline{X}_0 \leftarrow e(V^+,V^-) - \sum_{v\in \widetilde{H}^+} \widetilde{e}(v,V^-)-\sum_{v\in \widetilde{H}^-} \widetilde{e}(V^+,v)$
\State $\overline{X}_1 \leftarrow e(V^+,V^+) + e(V^-,V^-) - \sum_{v\in \widetilde{H}^+}(\widetilde{e}(v,V^+)+\widetilde{e}(V^+,v)) - \sum_{v\in \widetilde{H}^-}(\widetilde{e}(v,V^-)+\widetilde{e}(V^-,v))$ 
\State $\overline{X}_2 \leftarrow e(V^-,V^+) - \sum_{v\in \widetilde{H}^+} \widetilde{e}(V^-,v)-\sum_{v\in \widetilde{H}^-} \widetilde{e}(v,V^+)$
\State $\overline{k}_{\widetilde{L}}\leftarrow \frac{1}{4\eps^2}[(\frac{1}{2}+\eps)^2 \overline{X}_0 - (\frac{1}{2}+\eps)(\frac{1}{2}-\eps)\overline{X}_1 + (\frac{1}{2}-\eps)^2 \overline{X}_2]$
\ForEach{$v\in \widetilde{H}$} 
\State $\overline{k}_{S^*,v} \leftarrow (\frac{1}{2}+\frac{1}{4\eps})\cdot \widetilde{e}(V^+,v) + (\frac{1}{2}-\frac{1}{4\eps})\cdot \widetilde{e}(V^-,v)$
\State $\overline{k}_{v,T^*} \leftarrow (\frac{1}{2}-\frac{1}{4\eps})\cdot \widetilde{e}(v,V^+) + (\frac{1}{2}+\frac{1}{4\eps})\cdot \widetilde{e}(v,V^-)$ 
\EndFor
\State \textbf{return} $\overline{k}_{\widetilde{L}} + \sum_{v\in \widetilde{H}}\max\{\overline{k}_{S^*,v}, \overline{k}_{v,T^*}\}  -  \frac{\eta m}{16}$ 
\end{algorithmic}
\caption{Estimating the \textsc{Max-Dicut} value in dynamic streams} 
\label{alg:dynamic-max-dicut}
\end{algorithm}

In the following, we prove \cref{thm:dynamic-max-dicut}.
The proof follows a similar structure to that of \cref{thm:insertion}. First, we show that $\ell_0$-sampling in dynamic streams can effectively replace reservoir sampling in insertion-only streams, allowing us to retrieve a set $\widetilde{H} \supseteq H$ with high probability by the end of the stream.
\begin{lemma}
\label{lem:ell-0-sampling-dicut}
Let $\eps \in (0,\frac{1}{2}]$ and $\delta,\eta \in (0,1)$. With probability at least $1-\frac{\delta}{16}$, Line~\ref{line:ell-0-sampling-dicut} of \cref{alg:dynamic-max-dicut} returns a set $F$ which contains $\frac{\beta}{\eps^8\delta^3\eta^4}$ different edges. 
\end{lemma}
\begin{proof}
    The error is introduced by the $\ell_0$-samplers and the probability that sampling $\frac{\beta}{\eps^8\delta^3\eta^4}$ uniform random edges in parallel does not yield $\frac{\beta}{\eps^8\delta^3\eta^4}$ different edges.

    Since we maintain $\frac{\beta}{\eps^8\delta^3\eta^4}$ $\ell_0$-samplers, by \cref{thm:l0-sampler} and applying the union bound, the failure probability introduced by these $\ell_0$-samplers is at most $\frac{\beta}{\eps^8\delta^3\eta^4} \cdot \frac{\eps^8 \delta^4 \eta^4}{32\beta}=\frac{\delta}{32}$.

    Suppose that all the $\ell_0$-samplers succeed, i.e., $|F| = \frac{\beta}{\eps^8\delta^3\eta^4}$.
    Now we bound the probability that there exist two identical edges in $F$. Consider two fixed edges $e_1,e_2 \in F$. Then $\Pr[e_1 = e_2] = \frac{1}{m}$. By union bound, the probability that there exists at least one pair of identical edges is at most $\binom{\beta / \eps^8\delta^3\eta^4}{2}\cdot \frac{1}{m} \leq \frac{\beta^2}{\eps^{16}\delta^6 \eta^8 m}$.

    Therefore, the probability that $F$ contains $\frac{\beta}{\eps^8\delta^3\eta^4}$ different edges is at least $1-\frac{\delta}{32} - \frac{\beta^2}{\eps^{16}\delta^6 \eta^8 m} \geq 1-\frac{\delta}{32}-\frac{\delta}{32}= 1-\frac{\delta}{16}$ (since we can assume $m=\Omega(\eps^{-26}\delta^{-10}\eta^{-14})$). 
\end{proof}

\begin{lemma}
\label{lem:high-deg-reservoir-dynamic-dicut}
Let $\delta \in (0,1)$.
Then
 $\widetilde{H}$ contains all high-degree vertices $H$ in $G$ (i.e., $H \subseteq \widetilde{H}$) with probability at least $1-\frac{\delta}{16}$.
\end{lemma}
\begin{proof}
    Let $\theta:=\frac{\eps^{4}\delta\eta^2}{4096}$. Consider a vertex $v \in H$, so $\deg(v) \geq \theta m$. 
    For each directed edge $e_i$ incident to $v$, define an indicator random variable 
    \begin{equation*}
            X_{i}= \begin{cases}1 & \text { if } e_i\in F, \\ 0 & \text { otherwise. }\end{cases}
    \end{equation*}
    So, $\E[X_i] = \Pr[X_i=1] = 1-(1-\frac{1}{m})^{\beta/\eps^8\delta^3\eta^4} \le \frac{\beta}{\eps^8\delta^3\eta^4m}$. 
    On the other hand, $\E[X_i] = 1-(1-~\frac{1}{m})^{\beta/\eps^8\delta^3\eta^4} \ge 1-\exp(-\frac{\beta}{\eps^8\delta^3\eta^4m}) \ge \frac{\beta}{\eps^8\delta^3\eta^4m} - \frac{\beta^2}{2\eps^{16}\delta^6\eta^8m^2}\ge \frac{\beta}{5\eps^{8}\delta^3\eta^4m}$ since $m=\Omega(\eps^{-26}\delta^{-10}\eta^{-14})$.
    We define $X:= \sum_{i\in [\deg(v)]} X_i$ to denote the number of edges incident to $v$ that are sampled by the end of the stream. Then, $\E[X] = \sum_{i\in [\deg(v)]} \E[X_i] \in [\frac{\beta\deg(v)}{5\eps^{8}\delta^3 \eta^4m},\frac{\beta\deg(v)}{\eps^{8}\delta^3 \eta^4m}]$. 

    Next, we bound $\Var[X]$. For any $i\in [\deg(v)]$, $\Var[X_i] = \E[X^2_i] - (\E[X_i])^2 = \E[X_i] - (\E[X_i])^2 \le \frac{\beta}{\eps^8\delta^3\eta^4m}$.
    Since $(X_i)_{i\in [\deg(v)]}$ are negatively correlated, $\Var[X] \le \sum_{i\in [\deg(v)]}\Var[X_i] \le \deg(v)\cdot \frac{\beta}{\eps^{8} \delta^3\eta^4 m} = \frac{\beta\deg(v)}{\eps^{8} \delta^3\eta^4 m}$.
    Therefore, by Chebyshev's inequality, 
    \begin{align*}
        \Pr[v\notin \widetilde{H}] = \Pr[X\leq 0] 
        &\le \Pr\left[X \leq \E[X]- \frac{\beta\deg(v)}{5\eps^{8} \delta^3\eta^4 m}\right]\\
        &\leq \frac{\Var[X]\cdot 25\eps^{16} \delta^6 \eta^ 8 m^2}{ \beta^2\deg^2(v)}\\
        &\leq \frac{25\eps^{8} \delta^3 \eta^4 m}{\beta \deg(v)} &&\rhd \Var[X] \leq \frac{\beta\deg(v)}{\eps^{8} \delta^3 \eta^4 m}\\
        &\leq \frac{25\cdot 4096\eps^{4}\delta^2\eta^2}{\beta}. &&\rhd \deg(v)\geq \theta m
    \end{align*}
    Recall that $|H| \le \frac{2m}{\theta m}=\frac{2}{\theta} = \frac{8192}{\eps^{4} \delta\eta^2}$.
    By union bound, with probability at least $1-\frac{8192}{\eps^{4} \delta\eta^2}\cdot \frac{25\cdot 4096\eps^{4}\delta^2\eta^2}{\beta}  = 1-\frac{\delta}{16}$ (since $\beta = 2^{29}\cdot 25$), $\widetilde{H}$ contains all high-degree vertices $H$ in $G$.
\end{proof}

Next, we show that for all vertices $v \in \widetilde{H}$, our estimates of the number of edges from $v$ to $V^+$, from $v$ to $V^-$, from $V^+$ to $v$, and from $V^-$ to $v$ have additive error at most $\text{poly}(\eps, \delta,\eta) \cdot m$, with high probability. 
\begin{lemma}
\label{lem:error-bound-dynamic-dicut}
    Let $\eps \in (0,\frac{1}{2}]$ and $\delta,\eta \in (0,1)$.
    For all vertices $v\in \widetilde{H}$, we have
    \begin{enumerate}
        \item[(i)] $\widetilde{e}(V^+,v) \geq e(V^+,v),\  \widetilde{e}(v,V^+) \geq e(v,V^+),\  \widetilde{e}(V^-,v) \geq e(V^-,v),\  \widetilde{e}(v,V^-) \geq e(v,V^-)$;
        \item[(ii)] With probability at least $1-\frac{\delta}{8}$, all of the following inequalities hold: $\widetilde{e}(V^+,v) \leq e(V^+,v) + \frac{\eps^{10}\delta^5\eta^{6}}{128\beta} m,\  \widetilde{e}(v,V^+) \leq e(v,V^+)+ \frac{\eps^{10}\delta^5\eta^{6}}{128\beta} m,\  \widetilde{e}(V^-,v) \leq e(V^-,v)+ \frac{\eps^{10}\delta^5\eta^{6}}{128\beta} m,\  \widetilde{e}(v,V^-) \leq e(v,V^-)+\frac{\eps^{10}\delta^5\eta^{6}}{128\beta} m $. 
    \end{enumerate} 
\end{lemma}
\begin{proof}
    (i) By \cref{thm:CM-v2}, for all $v\in V$, we have $\widetilde{e}(V^+,v) \geq e(V^+,v),  \widetilde{e}(v,V^+) \geq e(v,V^+),  \widetilde{e}(V^-,v) \geq e(V^-,v)$, and  $\widetilde{e}(v,V^-) \geq e(v,V^-)$.

    (ii) By \cref{thm:CM-v2}, for any $v\in V$, with probability at least $1-\frac{\eps^{8}\delta^4\eta^4}{64\beta}$, each of the following four inequalities holds: $\widetilde{e}(V^+,v) \leq e(V^+,v) + \frac{\eps^{10}\delta^5\eta^{6}}{128\beta} m,\  \widetilde{e}(v,V^+) \leq e(v,V^+)+ \frac{\eps^{10}\delta^5\eta^{6}}{128\beta} m,\ \widetilde{e}(V^-,v) \leq e(V^-,v)+\frac{\eps^{10}\delta^5\eta^{6}}{128\beta} m$, and  $\widetilde{e}(v,V^-) \leq e(v,V^-)+ \frac{\eps^{10}\delta^5\eta^{6}}{128\beta} m$.
    Since $|\widetilde{H}|\le \frac{2\beta}{\eps^{8}\delta^3\eta^4}$, by applying the union bound we have that, with probability at least $1-\frac{\eps^{8}\delta^4\eta^4}{64\beta} \cdot \frac{2\beta}{\eps^{8}\delta^3\eta^4} \cdot 4 \ge 1-\frac{\delta}{8}$, for all $v\in \widetilde{H}$, all of the above four inequalities hold.
\end{proof}

\paragraph{Proof of \cref{thm:dynamic-max-dicut}.}
Given the above lemmas, we are ready to prove \cref{thm:dynamic-max-dicut}. 
\begin{proof}
We can assume that $m =\Omega(\eps^{-26}\delta^{-10}\eta^{-14})$, as otherwise we can store the whole graph in $O(\poly(1/\eps, 1/\delta, 1/\eta))$ space.

In \cref{alg:dynamic-max-dicut}, we maintain four CountMin sketches with accuracy parameter $\frac{\eps^{10}\delta^5\eta^{6}}{128\beta}$ and failure probability $\frac{\eps^{8}\delta^4\eta^4}{64\beta}$, where $\beta=2^{29}\cdot 25$. By \cref{thm:CM-v2}, this takes $O(\frac{1}{\eps^{10}\delta^5\eta^{6}}\log\frac{1}{\eps^{8} \delta^4\eta^4}) = O( \frac{1}{\eps^{10}\delta^5\eta^{6}}\log\frac{1}{\eps \delta\eta} )$ words of space. 
    Also, we maintain $\frac{\beta}{\eps^8\delta^3\eta^4}$ $\ell_0$-samplers, each with failure probability $\frac{\eps^8 \delta^4 \eta^4}{32\beta}$. By \cref{thm:l0-sampler}, this takes $O(\frac{\log^2 n}{\eps^8\delta^3\eta^4}\log \frac{1}{\eps^{8}\delta^4\eta^4})=O(\frac{\log^2 n}{\eps^8\delta^3\eta^4}\log \frac{1}{\eps\delta\eta})$ bits of space. 
    Therefore, the total space complexity of \cref{alg:dynamic-max-dicut} is $O( \frac{\log^2 n}{\eps^{10}\delta^5\eta^{6}}\log\frac{1}{\eps \delta\eta} )$ bits.

    Next, we analyze the approximation ratio.
Since $V^+ = \widetilde{H}^+ \cup \widetilde{L}^+$ and $V^- = \widetilde{H}^- \cup \widetilde{L}^-$, for any $v \in V$, we have
\begin{align*}
    e(V^+,v) &= e(\widetilde{H}^+,v)+ e(\widetilde{L}^+,v), &e(v,V^+) &= e(v,\widetilde{H}^+) + e(v,\widetilde{L}^+),\\
    e(V^-,v) &= e(\widetilde{H}^-,v)+ e(\widetilde{L}^-,v), &e(v,V^-) &= e(v,\widetilde{H}^-) + e(v,\widetilde{L}^-).
\end{align*}
By Lemma~\ref{lem:error-bound-dynamic-dicut}, with probability at least $1-\frac{\delta}{8}$, for all vertices $v\in \widetilde{H}$, all of the following inequalities hold: 
\begin{align*}
    e(\widetilde{H}^+,v)+ e(\widetilde{L}^+,v) \leq \widetilde{e}(V^+,v) \leq e(\widetilde{H}^+,v)+ e(\widetilde{L}^+,v) + \frac{\eps^{10}\delta^5\eta^{6}}{128\beta} m,\\
    e(v,\widetilde{H}^+) + e(v,\widetilde{L}^+) \leq \widetilde{e}(v,V^+) \leq e(v,\widetilde{H}^+) + e(v,\widetilde{L}^+) + \frac{\eps^{10}\delta^5\eta^{6}}{128\beta} m,\\
    e(\widetilde{H}^-,v)+ e(\widetilde{L}^-,v) \leq \widetilde{e}(V^-,v) \leq e(\widetilde{H}^-,v)+ e(\widetilde{L}^-,v) + \frac{\eps^{10}\delta^5\eta^{6}}{128\beta} m,\\
    e(v,\widetilde{H}^-) + e(v,\widetilde{L}^-) \leq \widetilde{e}(v,V^-) \leq e(v,\widetilde{H}^-) + e(v,\widetilde{L}^-) + \frac{\eps^{10}\delta^5\eta^{6}}{128\beta} m.
\end{align*}
Since $|F|=\frac{\beta}{\eps^{8}\delta^3\eta^4}$ and $\widetilde{H}$ is the set of endpoints of edges in $F$, we have that $|\widetilde{H}|\le \frac{2\beta}{\eps^{8}\delta^3\eta^4}$. Thus, it holds that $0 \le e(\widetilde{H}^+,v), e(v,\widetilde{H}^+), e(\widetilde{H}^-,v), e(v,\widetilde{H}^-)\le \frac{2\beta}{\eps^{8}\delta^3\eta^4}$. Therefore,
\begin{align}
      0\leq \widetilde{e}(V^+,v) -  e(\widetilde{L}^+,v)\leq \frac{2\beta}{\eps^{8}\delta^3\eta^4}+  \frac{\eps^{10}\delta^5\eta^{6}}{128\beta} m,\label{ineq:cm_plus_in_dynamic}\\ 
      0 \leq \widetilde{e}(v,V^+) - e(v,\widetilde{L}^+) \leq \frac{2\beta}{\eps^{8}\delta^3\eta^4} +  \frac{\eps^{10}\delta^5\eta^{6}}{128\beta} m,\label{ineq:cm_plus_out_dynamic}\\
      0 \leq \widetilde{e}(V^-,v) - e(\widetilde{L}^-,v) \leq \frac{2\beta}{\eps^{8}\delta^3\eta^4}+  \frac{\eps^{10}\delta^5\eta^{6}}{128\beta} m,\label{ineq:cm_minus_in_dynamic}\\
      0 \leq \widetilde{e}(v,V^-) - e(v,\widetilde{L}^-) \leq \frac{2\beta}{\eps^{8}\delta^3\eta^4} + \frac{\eps^{10}\delta^5\eta^{6}}{128\beta} m. \label{ineq:cm_minus_out_dynamic}
\end{align}

For each edge with both endpoints in $\widetilde{L}$,
let $\widetilde{X}_0,\widetilde{X}_1$, and $\widetilde{X}_2$ denote, respectively,  the numbers of edges with (i) $(Y_u,Y_v)=(+1,-1)$, (ii) $Y_u = Y_v$, and (iii) $(Y_u,Y_v)=(-1,+1)$. We have
\begin{align*}
    \widetilde{X}_0 = ~&e(\widetilde{L}^+,\widetilde{L}^-) = e(V^+,V^-) - e(\widetilde{H}^+,\widetilde{H}^-) - \sum_{v\in \widetilde{H}^+}e(v,\widetilde{L}^-)-\sum_{v\in \widetilde{H}^-}e(\widetilde{L}^+,v),\\
    \widetilde{X}_1  =~&e(\widetilde{L}^+,\widetilde{L}^+) + e(\widetilde{L}^-,\widetilde{L}^-)\\
    =~&e(V^+,V^+) + e(V^-,V^-) - e(\widetilde{H}^+,\widetilde{H}^+) - e(\widetilde{H}^-,\widetilde{H}^-) \\
    &- \sum_{v\in \widetilde{H}^+}(e(v,\widetilde{L}^+) + e(\widetilde{L}^+,v)) - \sum_{v\in \widetilde{H}^-}(e(v,\widetilde{L}^-) + e(\widetilde{L}^-,v)), \\
    \widetilde{X}_2 =~& e(\widetilde{L}^-,\widetilde{L}^+) = e(V^-,V^+) - e(\widetilde{H}^-,\widetilde{H}^+) -\sum_{v\in \widetilde{H}^+}e(\widetilde{L}^-,v)-\sum_{v\in \widetilde{H}^-}e(v,\widetilde{L}^+).
\end{align*}

Recall that
\begin{align*}
\overline{X}_0 &= e(V^+,V^-) - \sum_{v\in \widetilde{H}^+} \widetilde{e}(v,V^-)-\sum_{v\in \widetilde{H}^-} \widetilde{e}(V^+,v),\\
\overline{X}_1 &=e(V^+,V^+) + e(V^-,V^-) - \sum_{v\in \widetilde{H}^+}(\widetilde{e}(v,V^+)+\widetilde{e}(V^+,v)) - \sum_{v\in \widetilde{H}^-}(\widetilde{e}(v,V^-)+\widetilde{e}(V^-,v)),\\
\overline{X}_2 &= e(V^-,V^+) - \sum_{v\in \widetilde{H}^+} \widetilde{e}(V^-,v)-\sum_{v\in \widetilde{H}^-} \widetilde{e}(v,V^+).
\end{align*}
Therefore,
\begin{align*}
    \overline{X}_0 - \widetilde{X}_0 &= e(\widetilde{H}^+,\widetilde{H}^-) + \sum_{v\in \widetilde{H}^+}(e(v,\widetilde{L}^-) - \widetilde{e}(v,V^-)) + \sum_{v\in \widetilde{H}^-} (e(\widetilde{L}^+,v) - \widetilde{e}(V^+,v)),\\
    \overline{X}_1 - \widetilde{X}_1 &= e(\widetilde{H}^+,\widetilde{H}^+) + e(\widetilde{H}^-,\widetilde{H}^-) + \sum_{v\in \widetilde{H}^+} \left[(e(v,\widetilde{L}^+) - \widetilde{e}(v,V^+)) +  (e(\widetilde{L}^+,v) - \widetilde{e}(V^+,v))\right]\\
     &\quad+ \sum_{v\in \widetilde{H}^-} \left[(e(v,\widetilde{L}^-) - \widetilde{e}(v,V^-)) +  (e(\widetilde{L}^-,v) - \widetilde{e}(V^-,v))\right],\\
    \overline{X}_2 - \widetilde{X}_2 &=e(\widetilde{H}^-,\widetilde{H}^+) + \sum_{v\in \widetilde{H}^+}(e(\widetilde{L}^-,v) - \widetilde{e}(V^-,v)) + \sum_{v\in \widetilde{H}^-} (e(v,\widetilde{L}^+) - \widetilde{e}(v,V^+)).
\end{align*}
Since $|\widetilde{H}|\le \frac{2\beta}{\eps^{8}\delta^3\eta^4}$, we have $|\widetilde{H}^+|, |\widetilde{H}^-|\le \frac{2\beta}{\eps^{8}\delta^3\eta^4}$, and $e(\widetilde{H}^+,\widetilde{H}^-),e(\widetilde{H}^-,\widetilde{H}^+),e(\widetilde{H}^+,\widetilde{H}^+) + e(\widetilde{H}^-,\widetilde{H}^-)\le \frac{4\beta^2}{\eps^{16}\delta^6 \eta^8}$.
Moreover, combining inequalities (\ref{ineq:cm_plus_in_dynamic})--(\ref{ineq:cm_minus_out_dynamic}), we have
\begin{align*}
    \widetilde{X}_0-\left(\frac{8\beta^2}{\eps^{16}\delta^6\eta^8}+ \frac{\eps^{2}\delta^2 \eta^2 m}{32} \right)\le \overline{X}_0  \le \widetilde{X}_0+\frac{4\beta^2}{\eps^{16}\delta^6 \eta^8},\\
    \widetilde{X}_1-\left(\frac{16\beta^2}{\eps^{16}\delta^6\eta^8}+ \frac{\eps^{2}\delta^2 \eta^2 m}{16} \right) \le \overline{X}_1  \le \widetilde{X}_1+\frac{4\beta^2}{\eps^{16}\delta^6 \eta^8},\\
    \widetilde{X}_2-\left(\frac{8\beta^2}{\eps^{16}\delta^6\eta^8}+\frac{\eps^{2}\delta^2 \eta^2 m}{32} \right) \le \overline{X}_2   \le \widetilde{X}_2+\frac{4\beta^2}{\eps^{16}\delta^6 \eta^8}.
\end{align*}

Let $\widetilde{k}_{\widetilde{L}} := \frac{1}{4\eps^2}[(\frac{1}{2}+\eps)^2 \widetilde{X}_0 - (\frac{1}{2}+\eps)(\frac{1}{2}-\eps)\widetilde{X}_1 + (\frac{1}{2}-\eps)^2 \widetilde{X}_2]$. 
Now we are ready to show that $\overline{k}_{\widetilde{L}}$ is close to $\widetilde{k}_{\widetilde{L}}$ up to an additive error of $\poly(\eta)\cdot m$ with probability at least $1-\frac{\delta}{8}$.
We have
\begin{align*}
    \overline{k}_{\widetilde{L}} =~&\frac{1}{4\eps^2}\left[\left(\frac{1}{2}+\eps\right)^2 \overline{X}_0 - \left(\frac{1}{2}+\eps\right)\left(\frac{1}{2}-\eps\right)\overline{X}_1 + \left(\frac{1}{2}-\eps\right)^2 \overline{X}_2\right]\\
    \ge~&\frac{1}{4\eps^2}\left[\left(\frac{1}{2}+\eps\right)^2 \left(\widetilde{X}_0-\left(\frac{8\beta^2}{\eps^{16}\delta^6\eta^8}+\frac{ \eps^{2}\delta^2\eta^2 m}{32} \right)\right) - \left(\frac{1}{2}+\eps\right)\left(\frac{1}{2}-\eps\right)\left(\widetilde{X}_1+\frac{4\beta^2}{\eps^{16}\delta^6\eta^8}\right) \right.\\
    &+ \left.\left(\frac{1}{2}-\eps\right)^2 \left(\widetilde{X}_2-\left(\frac{8\beta^2}{\eps^{16}\delta^6\eta^8}+\frac{\eps^{2}\delta^2\eta^2 m}{32} \right)\right)\right]\\
    =~&\widetilde{k}_{\widetilde{L}} - \frac{1}{4\eps^2}\left[\left(\frac{1}{2}+\eps\right)^2 \left(\frac{8\beta^2}{\eps^{16}\delta^6\eta^8}+\frac{\eps^2\delta^2\eta^2 m}{32} \right) + \left(\frac{1}{2}+\eps\right)\left(\frac{1}{2}-\eps\right)\frac{4\beta^2}{\eps^{16}\delta^6\eta^8} \right.\\
    &+ \left.\left(\frac{1}{2}-\eps\right)^2 \left(\frac{8\beta^2}{\eps^{16}\delta^6\eta^8}+\frac{\eps^2\delta^2\eta^2 m}{32} \right)\right]\\
    =~&\widetilde{k}_{\widetilde{L}} - \left(\frac{5\beta^{2}}{4\eps^{18}\delta^{6}\eta^8} + \frac{3\beta^{2}}{\eps^{16}\delta^{6}\eta^8} + \frac{\delta^2 \eta^2 m}{256} + \frac{\eps^{2}\delta^2 \eta^2m}{64}\right)\\
    \ge~&\widetilde{k}_{\widetilde{L}} - \frac{\eta m}{32},
\end{align*}
where the last step follows since $m = \Omega(\eps^{-26}\delta^{-10}\eta^{-14})$. On the other hand,
\begin{align*}
\overline{k}_{\widetilde{L}} =~&\frac{1}{4\eps^2}\left[\left(\frac{1}{2}+\eps\right)^2 \overline{X}_0 - \left(\frac{1}{2}+\eps\right)\left(\frac{1}{2}-\eps\right)\overline{X}_1 + \left(\frac{1}{2}-\eps\right)^2 \overline{X}_2\right]\\
\le~&\frac{1}{4\eps^2}\left[\left(\frac{1}{2}+\eps\right)^2 \left(\widetilde{X}_0+\frac{4\beta^2}{\eps^{16}\delta^6 \eta^8}\right) - \left(\frac{1}{2}+\eps\right)\left(\frac{1}{2}-\eps\right)\left(\widetilde{X}_1-\left(\frac{16\beta^2}{\eps^{16}\delta^6\eta^8}+\frac{\eps^{2}\delta^2\eta^2 m}{16} \right)\right) \right.\\
&+ \left.\left(\frac{1}{2}-\eps\right)^2 \left(\widetilde{X}_2+\frac{4\beta^2}{\eps^{16}\delta^6 \eta^8}\right)\right]\\
=~&\widetilde{k}_{\widetilde{L}} +\frac{1}{4\eps^2}\left[\left(\frac{1}{2}+\eps\right)^2 \frac{4\beta^2}{\eps^{16}\delta^6\eta^8} + \left(\frac{1}{2}+\eps\right)\left(\frac{1}{2}-\eps\right)\left(\frac{16\beta^2}{\eps^{16}\delta^6\eta^8}+\frac{\eps^{2}\delta^2\eta^2 m}{16} \right) + \left(\frac{1}{2}-\eps\right)^2 \frac{4\beta^2}{\eps^{16}\delta^6\eta^8}\right]\\
=~&\widetilde{k}_{\widetilde{L}} + \left(\frac{3\beta^{2}}{2\eps^{18}\delta^{6}\eta^8} - \frac{2\beta^{2}}{\eps^{16}\delta^{6}\eta^8}+\frac{ \delta^{2}\eta^2m}{256}-\frac{\eps^{2}\delta^2\eta^2m}{64}\right)\\
\le~&\widetilde{k}_{\widetilde{L}} + \frac{\eta m}{64},
\end{align*}
where the last step follows since $m = \Omega(\eps^{-26}\delta^{-10}\eta^{-14})$.

For each $v\in \widetilde{H}$, let $\widetilde{k}_{S^*,v} := (\frac{1}{2}+\frac{1}{4\eps})\cdot e(\widetilde{L}^+,v) + (\frac{1}{2}-\frac{1}{4\eps})\cdot e(\widetilde{L}^-,v)$ and $\widetilde{k}_{v,T^*} := (\frac{1}{2}-\frac{1}{4\eps})\cdot e(v,\widetilde{L}^+) + (\frac{1}{2}+\frac{1}{4\eps})\cdot e(v,\widetilde{L}^-)$.
We have
\begin{align*}
    \overline{k}_{S^*,v} &= \left(\frac{1}{2}+\frac{1}{4\eps}\right)\cdot \widetilde{e}(V^+,v) + \left(\frac{1}{2}-\frac{1}{4\eps}\right)\cdot \widetilde{e}(V^-,v)\\
    &\ge \left(\frac{1}{2}+\frac{1}{4\eps}\right)\cdot e(\widetilde{L}^+,v) + \left(\frac{1}{2}-\frac{1}{4\eps}\right) \left( e(\widetilde{L}^-,v) + \frac{2\beta}{\eps^{8}\delta^3\eta^4}+  \frac{\eps^{10}\delta^5\eta^{6} m}{128\beta} \right)\\
    &= \widetilde{k}_{S^*,v} - \left(\frac{1}{4\eps}-\frac{1}{2}\right) \left(\frac{2\beta}{\eps^{8}\delta^3\eta^4}+  \frac{\eps^{10}\delta^5\eta^{6} m}{128\beta} \right)\\
    &= \widetilde{k}_{S^*,v} -\left(\frac{\beta}{2\eps^{9}\delta^3\eta^4} - \frac{\beta}{\eps^{8}\delta^3\eta^4} + \frac{\eps^{9}\delta^5\eta^{6} m}{512\beta} - \frac{\eps^{10}\delta^5\eta^{6} m}{256\beta}\right)\\
    &\ge \widetilde{k}_{S^*,v} - \frac{\eps^8 \delta^3 \eta^5 m}{128\beta},
\end{align*}
where the last step follows since $m = \Omega(\eps^{-26}\delta^{-10}\eta^{-14})$. On the other hand,
\begin{align*}
    \overline{k}_{S^*,v} &= \left(\frac{1}{2}+\frac{1}{4\eps}\right)\cdot \widetilde{e}(V^+,v) + \left(\frac{1}{2}-\frac{1}{4\eps}\right)\cdot \widetilde{e}(V^-,v)\\
    &\le \left(\frac{1}{2}+\frac{1}{4\eps}\right)\left( e(\widetilde{L}^+,v) + \frac{2\beta}{\eps^{8}\delta^3\eta^4}+  \frac{\eps^{10}\delta^5\eta^{6} m}{128\beta} \right) + \left(\frac{1}{2}-\frac{1}{4\eps}\right)\cdot e(\widetilde{L}^-,v)\\
    &= \widetilde{k}_{S^*,v} + \left(\frac{1}{2}+\frac{1}{4\eps}\right)\left(\frac{2\beta}{\eps^{8}\delta^3\eta^4}+  \frac{\eps^{10}\delta^5\eta^{6} m}{128\beta} \right)\\
    &= \widetilde{k}_{S^*,v} +\left(\frac{\beta}{2\eps^{9}\delta^3\eta^4} + \frac{\beta}{\eps^{8}\delta^3\eta^4} + \frac{\eps^{9}\delta^5 \eta^6 m}{512\beta} + \frac{\eps^{10}\delta^5\eta^6 m}{256\beta}\right)\\
    &\le \widetilde{k}_{S^*,v} +\frac{\eps^8 \delta^3 \eta^5 m}{128\beta},
\end{align*}
where the last step follows since $m = \Omega(\eps^{-26}\delta^{-10}\eta^{-14})$. Therefore, $|\overline{k}_{S^*,v} - \widetilde{k}_{S^*,v}| \le \frac{\eps^8 \delta^3 \eta^5 m}{128\beta}$.
Similarly, $|\overline{k}_{v,T^*} - \widetilde{k}_{v,T^*}| \le \frac{\eps^8 \delta^3 \eta^5 m}{128\beta}$.

Let $\widetilde{\mathrm{ALG}}:= \widetilde{k}_{\widetilde{L}} + \sum_{v\in \widetilde{H}}\max\{\widetilde{k}_{S^*,v}, \widetilde{k}_{v,T^*}\}  - \frac{\eta m}{32}$ denote the output of the algorithm obtained by running \cref{alg:general-offline} with $L$ and $H$ replaced by $\widetilde L$ and $\widetilde H$, respectively.
Analogous to the proof of \cref{thm:general-offline}, we can show that $(1-\frac{\eta}{2})\cdot \opt\le \widetilde{\mathrm{ALG}} \le \opt$ with probability at least $1-\frac{3\delta}{4}$. The argument relies primarily on two properties of the partition $(\widetilde{L},\widetilde{H})$: (i) every vertex in $\widetilde{L}$ has degree at most $\theta m$ for $\theta:=\frac{\eps^4 \delta \eta^2}{4096}$ (since $\widetilde{L} \subseteq L$), and (ii) $|\widetilde{H}|\leq \frac{2\beta}{\eps^8 \delta^3 \eta^4}$.

Let $\mathrm{ALG}$ denote the output of \cref{alg:dynamic-max-dicut}. By applying the union bound on Lemma~\ref{lem:ell-0-sampling-dicut}, Lemma~\ref{lem:high-deg-reservoir-dynamic-dicut} and Lemma~\ref{lem:error-bound-dynamic-dicut}, with probability at least $1-\frac{\delta}{16}-\frac{\delta}{16}-\frac{\delta}{8}-\frac{3\delta}{4} =1-\delta$, we have
\begin{align*}
    \mathrm{ALG} &= \overline{k}_{\widetilde{L}} + \sum_{v\in \widetilde{H}}\max\{\overline{k}_{S^*,v}, \overline{k}_{v,T^*}\} -  \frac{\eta m}{16}\\
    &\le \widetilde{k}_{\widetilde{L}} + \frac{\eta m}{64} + \sum_{v\in \widetilde{H}}\max\{\widetilde{k}_{S^*,v}, \widetilde{k}_{v,T^*}\} + |\widetilde{H}|\cdot \frac{\eps^8 \delta^3 \eta^5 m}{128\beta} - \frac{\eta m}{16}\\
    &\le \widetilde{\mathrm{ALG}} + \frac{\eta m}{32}+\frac{\eta m}{64}+\frac{\eta m}{64} - \frac{\eta m}{16} &&\rhd |\widetilde{H}|\le \frac{2\beta}{\eps^{8}\delta^3\eta^4}\\
    &\le \opt,
\end{align*}
and
\begin{align*}
    \mathrm{ALG} &= \overline{k}_{\widetilde{L}} + \sum_{v\in \widetilde{H}}\max\{\overline{k}_{S^*,v}, \overline{k}_{v,T^*}\} - \frac{\eta m}{16}\\
    &\ge \widetilde{k}_{\widetilde{L}} - \frac{\eta m}{32} + \sum_{v\in \widetilde{H}}\max\{\widetilde{k}_{S^*,v}, \widetilde{k}_{v,T^*}\} - |\widetilde{H}|\cdot \frac{\eps^8 \delta^3 \eta^5 m}{128\beta} - \frac{\eta m}{16} \\
    &\ge \widetilde{\mathrm{ALG}} - \frac{\eta m}{64} - \frac{\eta m}{16} &&\rhd |\widetilde{H}|\le \frac{2\beta}{\eps^{8}\delta^3\eta^4}\\
    &\ge \widetilde{\mathrm{ALG}} - \frac{\eta m}{8} \\
    &\ge \qty(1-\frac{\eta}{2})\cdot \opt - \frac{\eta}{2} \cdot\opt && \rhd m \le 4\opt \\
    &= (1-\eta)\cdot \opt.
\end{align*}
\end{proof}

\section{Proof of \cref{thm:max-2and-dynamic}}
\label{sec:proof-thm:max-2and-dynamic}
\begin{lemma}
\label{lem:ell-0-sampling-2and}
Let $\eps \in (0,\frac{1}{2}]$ and $\delta,\eta \in (0,1)$. With probability at least $1-\frac{\delta}{16}$, Line~\ref{line:ell-0-sampling-2and} of \cref{alg:max-2and-dynamic} returns a set $F$ which contains $\frac{\beta}{\eps^8\delta^3\eta^4}$ different constraints. 
\end{lemma}
\begin{proof}
    The error is introduced by the $\ell_0$-samplers and the probability that sampling $\frac{\beta}{\eps^8\delta^3\eta^4}$ uniform random constraints in parallel does not yield $\frac{\beta}{\eps^8\delta^3\eta^4}$ different constraints.

    Since we maintain $\frac{\beta}{\eps^8\delta^3\eta^4}$ $\ell_0$-samplers, by \cref{thm:l0-sampler} and applying the union bound, the failure probability introduced by these $\ell_0$-samplers is at most $\frac{\beta}{\eps^8\delta^3\eta^4} \cdot \frac{\eps^8 \delta^4 \eta^4}{32\beta}=\frac{\delta}{32}$.

    Suppose that all the $\ell_0$-samplers succeed, i.e., $|F| = \frac{\beta}{\eps^8\delta^3\eta^4}$.
    Now we bound the probability that there exist two identical constraints in $F$. Consider two fixed constraints $C_1,C_2 \in F$. Then $\Pr[C_1 = C_2] = \frac{1}{m}$. By union bound, the probability that there exists at least one pair of identical constraints is at most $\binom{\beta / \eps^8\delta^3\eta^4}{2}\cdot \frac{1}{m} \leq \frac{\beta^2}{\eps^{16}\delta^6 \eta^8 m}$.

    Therefore, the probability that $F$ contains $\frac{\beta}{\eps^8\delta^3\eta^4}$ different constraints is at least $1-\frac{\delta}{32} - \frac{\beta^2}{\eps^{16}\delta^6 \eta^8 m} \geq 1-\frac{\delta}{32}-\frac{\delta}{32}= 1-\frac{\delta}{16}$ (since we can assume $m=\Omega(\eps^{-26}\delta^{-10}\eta^{-14})$). 
\end{proof}

\begin{lemma}
\label{lem:high-deg-reservoir-max-2and-dynamic}
Let $\delta \in (0,1)$.
Then
 $\widetilde{H}$ contains all high-degree variables $H$ in $\Psi$ (i.e., $H \subseteq \widetilde{H}$) with probability at least $1-\frac{\delta}{16}$.
\end{lemma}
\begin{proof}
    Let $\theta:=\frac{\eps^{4}\delta\eta^2}{4096}$. Consider any $i \in H$. Then, by definition of $H$, we have $\deg(i) \geq \theta m$. 
    For each constraint $C_j$ involving variable $x_i$, define an indicator random variable 
    \begin{equation*}
            X_{j}= \begin{cases}1 & \text { if } C_j\in F, \\ 0 & \text { otherwise. }\end{cases}
    \end{equation*}
    So, $\E[X_j] = \Pr[X_j=1] = 1-(1-\frac{1}{m})^{\beta/\eps^8\delta^3\eta^4}\le \frac{\beta}{\eps^{8}\delta^3\eta^4 m}$. 
    On the other hand, $\E[X_j] = 1-(1-~\frac{1}{m})^{\beta/\eps^8\delta^3\eta^4}\ge 1-\exp(-\frac{\beta}{\eps^{8}\delta^3\eta^4 m})\ge \frac{\beta}{\eps^{8}\delta^3\eta^4 m} - \frac{\beta^2}{2\eps^{16}\delta^6\eta^8 m^2}\ge \frac{\beta}{5\eps^{8}\delta^3\eta^4 m}$ since $m=\Omega(\eps^{-26}\delta^{-10}\eta^{-14})$.
    We define $X:= \sum_{j\in [\deg(i)]} X_j$ to denote the number of constraints involving variable $x_i$ that are sampled by the end of the stream. Then, $\E[X] = \sum_{j\in [\deg(i)]} \E[X_j] \in [\frac{\beta\deg(i)}{5\eps^{8}\delta^3 \eta^4m},\frac{\beta\deg(i)}{\eps^{8}\delta^3 \eta^4m}]$. 

    Next, we bound $\Var[X]$. For any $j\in [\deg(i)]$, $\Var[X_j] = \E[X^2_j] - (\E[X_j])^2 = \E[X_j] - (\E[X_j])^2 \le \frac{\beta}{\eps^{8} \delta^3\eta^4 m}$.
    Since $(X_j)_{j\in [\deg(i)]}$ are negatively correlated, $\Var[X] \le \sum_{j\in [\deg(i)]}\Var[X_j] \le \deg(i)\cdot \frac{\beta}{\eps^{8} \delta^3\eta^4 m} = \frac{\beta\deg(i)}{\eps^{8} \delta^3\eta^4 m}$.
    Therefore, by Chebyshev's inequality, 
    \begin{align*}
        \Pr[i\notin \widetilde{H}] = \Pr[X\leq 0] 
        &\le \Pr\left[X \leq \E[X]- \frac{\beta\deg(i)}{5\eps^{8} \delta^3\eta^4 m}\right]\\
        &\leq \frac{\Var[X]\cdot 25\eps^{16} \delta^6\eta^8 m^2}{\beta^2\deg^2(i)}\\
        &\leq \frac{25\eps^{8} \delta^3\eta^4 m}{\beta \deg(i)} &&\rhd \Var[X] \leq \frac{\beta\deg(i)}{\eps^{8} \delta^3 \eta^4m}\\
        &\leq \frac{25\cdot 4096\eps^{4}\delta^2\eta^2}{\beta}. &&\rhd \deg(i)\geq \theta m
    \end{align*}
    Recall that $|H| \le \frac{2m}{\theta m}=\frac{2}{\theta} = \frac{8192}{\eps^{4} \delta\eta^2}$.
    By union bound, with probability at least $1-\frac{8192}{\eps^{4} \delta\eta^2}\cdot \frac{25\cdot 4096\eps^{4}\delta^2\eta^2}{\beta}  = 1-\frac{\delta}{16}$ (since $\beta = 2^{29}\cdot 25$), $\widetilde{H}$ contains all high-degree variables $H$ in $\Psi$.
\end{proof}

For each $i\in [n]$, define 
\begin{enumerate}
    \item $X_{i,+}^{\pos}:= |\{((i,j),(b_i,b_j))\in \Psi: j\in [n], b_i=1, b_j Y_j = 1\}|$ to be the number of constraints involving $x_i$ such that the literal of $x_i$ appears positively (i.e., $b_i=1$) and the other literal is satisfied by the prediction $\mathbf{Y}$ (i.e., $b_j Y_j=1$). 
    The other quantities are defined analogously:
    \item $X_{i,-}^{\pos}:= |\{((i,j),(b_i,b_j))\in \Psi: j\in [n], b_i=1, b_jY_j = -1\}|$.
    \item $X_{i,+}^{\neg}:= |\{((i,j),(b_i,b_j))\in \Psi: j\in [n], b_i=-1, b_jY_j = 1\}|$.
    \item $X_{i,-}^{\neg}:= |\{((i,j),(b_i,b_j))\in \Psi: j\in [n], b_i=-1, b_jY_j = -1\}|$.
\end{enumerate}

\begin{lemma}
\label{lem:error-bound-insertion-max-2and-dynamic}
    Let $\eps \in (0,\frac{1}{2}]$ and $\delta,\eta \in (0,1)$.
    For all $i\in \widetilde{H}$, we have
    \begin{enumerate}
        \item[(i)] $\widetilde{X}_{i,+}^{\pos} \geq X_{i,+}^{\pos},\  \widetilde{X}_{i,-}^{\pos} \geq X_{i,-}^{\pos},\  \widetilde{X}_{i,+}^{\neg} \geq X_{i,+}^{\neg},\  \widetilde{X}_{i,-}^{\neg} \geq X_{i,-}^{\neg}$;
        \item[(ii)] With probability at least $1-\frac{\delta}{8}$, all of the following inequalities hold: $\widetilde{X}_{i,+}^{\pos} \leq X_{i,+}^{\pos} + \frac{\eps^{10}\delta^5\eta^6}{128\beta} m,\  \widetilde{X}_{i,-}^{\pos} \leq X_{i,-}^{\pos}+ \frac{\eps^{10}\delta^5\eta^6}{128\beta} m,\  \widetilde{X}_{i,+}^{\neg} \leq X_{i,+}^{\neg}+ \frac{\eps^{10}\delta^5\eta^6}{128\beta} m,\  \widetilde{X}_{i,-}^{\neg} \leq X_{i,-}^{\neg}+\frac{\eps^{10}\delta^5\eta^6}{128\beta} m $. 
    \end{enumerate} 
\end{lemma}
\begin{proof}
    (i) By \cref{thm:CM-v2}, for all $i\in [n]$, we have $\widetilde{X}_{i,+}^{\pos} \geq X_{i,+}^{\pos},\  \widetilde{X}_{i,-}^{\pos} \geq X_{i,-}^{\pos},\  \widetilde{X}_{i,+}^{\neg} \geq X_{i,+}^{\neg}$, and $\widetilde{X}_{i,-}^{\neg} \geq X_{i,-}^{\neg}$.

    (ii) By \cref{thm:CM-v2}, for any $i\in [n]$, with probability at least $1-\frac{\eps^{8}\delta^4\eta^4}{64\beta}$, each of the following four inequalities holds: $\widetilde{X}_{i,+}^{\pos} \leq X_{i,+}^{\pos} + \frac{\eps^{10}\delta^5\eta^6}{128\beta} m,\  \widetilde{X}_{i,-}^{\pos} \leq X_{i,-}^{\pos}+ \frac{\eps^{10}\delta^5\eta^6}{128\beta} m,\  \widetilde{X}_{i,+}^{\neg} \leq X_{i,+}^{\neg}+ \frac{\eps^{10}\delta^5\eta^6}{128\beta} m$, and $\widetilde{X}_{i,-}^{\neg} \leq X_{i,-}^{\neg}+\frac{\eps^{10}\delta^5\eta^6}{128\beta} m $.
    Since $|\widetilde{H}|\le \frac{2\beta}{\eps^{8}\delta^3\eta^4}$, by applying the union bound we have that, with probability at least $1-\frac{\eps^{8}\delta^4\eta^4}{64\beta} \cdot \frac{2\beta}{\eps^{8}\delta^3\eta^4} \cdot 4 \ge 1-\frac{\delta}{8}$, for all $i\in \widetilde{H}$, all of the above four inequalities hold.
\end{proof}

\paragraph{Proof of \cref{thm:max-2and-dynamic}.}
Given the above lemmas, we are ready to prove \cref{thm:max-2and-dynamic}. 
\begin{proof}
We can assume that $m =\Omega(\eps^{-26}\delta^{-10}\eta^{-14})$, as otherwise we can store the whole instance in $O(\poly(1/\eps, 1/\delta, 1/\eta))$ space.

In \cref{alg:max-2and-dynamic}, we maintain four CountMin sketches with accuracy parameter $\frac{\eps^{10}\delta^5\eta^6}{128\beta}$ and failure probability $\frac{\eps^{8}\delta^4\eta^4}{64\beta}$, where $\beta=2^{29}\cdot 25$. By \cref{thm:CM-v2}, this takes $O(\frac{1}{\eps^{10}\delta^5\eta^6}\log\frac{1}{\eps^{8} \delta^4\eta^4}) = O( \frac{1}{\eps^{10}\delta^5\eta^6}\log\frac{1}{\eps \delta\eta} )$ words of space. 
    Also, we maintain $\frac{\beta}{\eps^8\delta^3\eta^4}$ $\ell_0$-samplers, each with failure probability $\frac{\eps^8 \delta^4 \eta^4}{32\beta}$. By \cref{thm:l0-sampler}, this takes $O(\frac{\log^2 n}{\eps^8\delta^3\eta^4}\log \frac{1}{\eps^{8}\delta^4\eta^4})=O(\frac{\log^2 n}{\eps^8\delta^3\eta^4}\log \frac{1}{\eps\delta\eta})$ bits of space. 
    Therefore, the total space complexity of \cref{alg:max-2and-dynamic} is $O( \frac{\log^2 n}{\eps^{10}\delta^5\eta^{6}}\log\frac{1}{\eps \delta\eta} )$ bits.

    Next, we analyze the approximation ratio.
Since $[n] = \widetilde{H}\cup \widetilde{L}$, we have
\begin{enumerate}
    \item $X_{i,+}^{\pos} = X_{i,\widetilde{H}(+)}^{\pos} + X_{i,\widetilde{L}(+)}^{\pos}$, where $X_{i,\widetilde{H}(+)}^{\pos} := |\{((i,j),(b_i,b_j))\in \Psi: j\in \widetilde{H}, b_i=1, b_j Y_j = 1\}|$ and $X_{i,\widetilde{L}(+)}^{\pos} := |\{((i,j),(b_i,b_j))\in \Psi: j\in \widetilde{L}, b_i=1, b_j Y_j = 1\}|$. 
    \item $X_{i,-}^{\pos} = X_{i,\widetilde{H}(-)}^{\pos} + X_{i,\widetilde{L}(-)}^{\pos}$, where $X_{i,\widetilde{H}(-)}^{\pos} := |\{((i,j),(b_i,b_j))\in \Psi: j\in \widetilde{H}, b_i=1, b_j Y_j = -1\}|$ and $X_{i,\widetilde{L}(-)}^{\pos} := |\{((i,j),(b_i,b_j))\in \Psi: j\in \widetilde{L}, b_i=1, b_j Y_j = -1\}|$.
    \item $X_{i,+}^{\neg} = X_{i,\widetilde{H}(+)}^{\neg} + X_{i,\widetilde{L}(+)}^{\neg}$, where $X_{i,\widetilde{H}(+)}^{\neg} := |\{((i,j),(b_i,b_j))\in \Psi: j\in \widetilde{H}, b_i=-1, b_j Y_j = 1\}|$ and $X_{i,\widetilde{L}(+)}^{\neg} := |\{((i,j),(b_i,b_j))\in \Psi: j\in \widetilde{L}, b_i=-1, b_j Y_j = 1\}|$.
    \item $X_{i,-}^{\neg} = X_{i,\widetilde{H}(-)}^{\neg} + X_{i,\widetilde{L}(-)}^{\neg}$, where $X_{i,\widetilde{H}(-)}^{\neg} := |\{((i,j),(b_i,b_j))\in \Psi: j\in \widetilde{H}, b_i=-1, b_j Y_j = -1\}|$ and $X_{i,\widetilde{L}(-)}^{\neg} := |\{((i,j),(b_i,b_j))\in \Psi: j\in \widetilde{L}, b_i=-1, b_j Y_j = -1\}|$.
    
\end{enumerate}
By Lemma~\ref{lem:error-bound-insertion-max-2and-dynamic}, with probability at least $1-\frac{\delta}{8}$, for all $i\in \widetilde{H}$, all of the following inequalities hold: 
\begin{align*}
    X_{i,\widetilde{H}(+)}^{\pos} + X_{i,\widetilde{L}(+)}^{\pos} \leq \widetilde{X}_{i,+}^{\pos} \leq X_{i,\widetilde{H}(+)}^{\pos} + X_{i,\widetilde{L}(+)}^{\pos} + \frac{\eps^{10}\delta^5\eta^6}{128\beta} m,\\
    X_{i,\widetilde{H}(-)}^{\pos} + X_{i,\widetilde{L}(-)}^{\pos} \leq \widetilde{X}_{i,-}^{\pos} \leq X_{i,\widetilde{H}(-)}^{\pos} + X_{i,\widetilde{L}(-)}^{\pos} + \frac{\eps^{10}\delta^5\eta^6}{128\beta} m,\\
    X_{i,\widetilde{H}(+)}^{\neg} + X_{i,\widetilde{L}(+)}^{\neg} \leq \widetilde{X}_{i,+}^{\neg} \leq X_{i,\widetilde{H}(+)}^{\neg} + X_{i,\widetilde{L}(+)}^{\neg} + \frac{\eps^{10}\delta^5\eta^6}{128\beta} m,\\
    X_{i,\widetilde{H}(-)}^{\neg} + X_{i,\widetilde{L}(-)}^{\neg} \leq \widetilde{X}_{i,-}^{\neg} \leq X_{i,\widetilde{H}(-)}^{\neg} + X_{i,\widetilde{L}(-)}^{\neg} + \frac{\eps^{10}\delta^5\eta^6}{128\beta} m.
\end{align*}
Since $|F|= \frac{\beta}{\eps^{8}\delta^3\eta^4}$ and $\widetilde{H}$ is the set of associated variables of constraints in $F$, we have that $|\widetilde{H}|\le \frac{2\beta}{\eps^{8}\delta^3\eta^4}$. Thus, it holds that $0 \le X_{i,\widetilde{H}(+)}^{\pos}, X_{i,\widetilde{H}(-)}^{\pos}, X_{i,\widetilde{H}(+)}^{\neg}, X_{i,\widetilde{H}(-)}^{\neg}\le \frac{2\beta}{\eps^{8}\delta^3\eta^4}$. Therefore,
\begin{align}
      0\leq \widetilde{X}_{i,+}^{\pos} -  X_{i,\widetilde{L}(+)}^{\pos}\leq \frac{2\beta}{\eps^{8}\delta^3\eta^4}+  \frac{\eps^{10}\delta^5\eta^6}{128\beta} m,\label{ineq:dynamic-cm_pos+}\\ 
      0 \leq \widetilde{X}_{i,-}^{\pos} - X_{i,\widetilde{L}(-)}^{\pos} \leq \frac{2\beta}{\eps^{8}\delta^3\eta^4} +  \frac{\eps^{10}\delta^5\eta^6}{128\beta} m,\label{ineq:dynamic-cm_pos-}\\
      0 \leq \widetilde{X}_{i,+}^{\neg} - X_{i,\widetilde{L}(+)}^{\neg} \leq \frac{2\beta}{\eps^{8}\delta^3\eta^4}+  \frac{\eps^{10}\delta^5\eta^6}{128\beta} m,\label{ineq:dynamic-cm_neg+}\\
      0 \leq \widetilde{X}_{i,-}^{\neg} - X_{i,\widetilde{L}(-)}^{\neg} \leq \frac{2\beta}{\eps^{8}\delta^3\eta^4} + \frac{\eps^{10}\delta^5\eta^6}{128\beta} m. \label{ineq:dynamic-cm_neg-}
\end{align}

We define
\begin{enumerate}
    \item $\widetilde{X}_0 := |\{((j_1,j_2),(b_1,b_2))\in \Psi: j_1, j_2\in \widetilde{L},\ b_1 Y_{j_1} = b_2 Y_{j_2}=1\}|$ to be the number of constraints in $\Psi$ whose associated variables are both in $\widetilde{L}$ and for which both literals are satisfied by the predictions $\mathbf{Y}$.
    \item $\widetilde{X}_1 := |\{((j_1,j_2),(b_1,b_2))\in \Psi: j_1, j_2\in \widetilde{L},\ b_1 Y_{j_1} \neq b_2 Y_{j_2}\}|$ to be the number of constraints in $\Psi$ whose associated variables are both in $\widetilde{L}$ and for which exactly one literal is satisfied by $\mathbf{Y}$.
    \item $\widetilde{X}_2 := |\{((j_1,j_2),(b_1,b_2))\in \Psi: j_1, j_2\in \widetilde{L},\ b_1 Y_{j_1} = b_2 Y_{j_2}=-1\}|$ to be the number of constraints in $\Psi$ whose associated variables are both in $\widetilde{L}$ and for which neither literal is satisfied by  $\mathbf{Y}$.
\end{enumerate}
Similarly, we define
\begin{enumerate}
    \item $\widetilde{X}_{0,\widetilde{H}} := |\{((j_1,j_2),(b_1,b_2))\in \Psi: j_1, j_2\in \widetilde{H},\ b_1 Y_{j_1} = b_2 Y_{j_2}=1\}|$ to be the number of constraints in $\Psi$ whose associated variables are both in $\widetilde{H}$ and for which both literals are satisfied by $\mathbf{Y}$.
    \item $\widetilde{X}_{1,\widetilde{H}} := |\{((j_1,j_2),(b_1,b_2))\in \Psi: j_1, j_2\in \widetilde{H},\ b_1 Y_{j_1} \neq b_2 Y_{j_2}\}|$ to be the number of constraints in $\Psi$ whose associated variables are both in $\widetilde{H}$ and for which exactly one literal is satisfied by $\mathbf{Y}$.
    \item $\widetilde{X}_{2,\widetilde{H}} := |\{((j_1,j_2),(b_1,b_2))\in \Psi: j_1, j_2\in \widetilde{H},\ b_1 Y_{j_1} = b_2 Y_{j_2}=-1\}|$ to be the number of constraints in $\Psi$ whose associated variables are both in $\widetilde{H}$ and for which neither literal is satisfied by  $\mathbf{Y}$.
\end{enumerate}

We have

\begin{align*}
    \widetilde{X}_0 &=  X_{+,+} - \widetilde{X}_{0,\widetilde{H}} - \sum_{i\in \widetilde{H}^+}X_{i,\widetilde{L}(+)}^{\pos}-\sum_{i\in \widetilde{H}^-}X_{i,\widetilde{L}(+)}^{\neg},\\
    \widetilde{X}_1 &= 
     X_{+,-} - \widetilde{X}_{1,\widetilde{H}} - \sum_{i\in \widetilde{H}^+}(X_{i,\widetilde{L}(-)}^{\pos} + X_{i,\widetilde{L}(+)}^{\neg}) - \sum_{i\in \widetilde{H}^-}(X_{i,\widetilde{L}(+)}^{\pos} + X_{i,\widetilde{L}(-)}^{\neg}), \\
    \widetilde{X}_2 &=  X_{-,-} - \widetilde{X}_{2,\widetilde{H}} -\sum_{i\in \widetilde{H}^+}X_{i,\widetilde{L}(-)}^{\neg}-\sum_{i\in \widetilde{H}^-}X_{i,\widetilde{L}(-)}^{\pos}.
\end{align*}

Recall that
\begin{align*}
\overline{X}_0 &= X_{+,+} - \sum_{i\in \widetilde{H}^+} \widetilde{X}^{\pos}_{i,+}-\sum_{i\in \widetilde{H}^-} \widetilde{X}^{\neg}_{i,+},\\
\overline{X}_1 &=X_{+,-} - \sum_{i\in \widetilde{H}^+}(\widetilde{X}^{\pos}_{i,-}+\widetilde{X}^{\neg}_{i,+}) - \sum_{i\in \widetilde{H}^-}(\widetilde{X}^{\pos}_{i,+}+\widetilde{X}^{\neg}_{i,-}),\\
\overline{X}_2 &= X_{-,-} - \sum_{i\in \widetilde{H}^+} \widetilde{X}^{\neg}_{i,-}-\sum_{i\in \widetilde{H}^-} \widetilde{X}^{\pos}_{i,-}.
\end{align*}
Therefore,
\begin{align*}
    \overline{X}_0 - \widetilde{X}_0 &= \widetilde{X}_{0,\widetilde{H}} + \sum_{i\in \widetilde{H}^+}(X_{i,\widetilde{L}(+)}^{\pos} - \widetilde{X}^{\pos}_{i,+}) + \sum_{i\in \widetilde{H}^-} (X_{i,\widetilde{L}(+)}^{\neg} - \widetilde{X}^{\neg}_{i,+}),\\
    \overline{X}_1 - \widetilde{X}_1 &= \widetilde{X}_{1,\widetilde{H}} + \sum_{i\in \widetilde{H}^+} \left[(X_{i,\widetilde{L}(-)}^{\pos} - \widetilde{X}^{\pos}_{i,-}) +  (X_{i,\widetilde{L}(+)}^{\neg} - \widetilde{X}^{\neg}_{i,+})\right]
     + \sum_{i\in \widetilde{H}^-} \left[(X_{i,\widetilde{L}(+)}^{\pos} - \widetilde{X}^{\pos}_{i,+}) +  (X_{i,\widetilde{L}(-)}^{\neg} - \widetilde{X}^{\neg}_{i,-})\right],\\
    \overline{X}_2 - \widetilde{X}_2 &=\widetilde{X}_{2,\widetilde{H}} + \sum_{i\in \widetilde{H}^+}(X_{i,\widetilde{L}(-)}^{\neg} - \widetilde{X}^{\neg}_{i,-}) + \sum_{i\in \widetilde{H}^-} (X_{i,\widetilde{L}(-)}^{\pos} - \widetilde{X}^{\pos}_{i,-}).
\end{align*}
Since $|\widetilde{H}|\le \frac{2\beta}{\eps^{8}\delta^3\eta^4}$, we have $|\widetilde{H}^+|, |\widetilde{H}^-|\le \frac{2\beta}{\eps^{8}\delta^3\eta^4}$, and $\widetilde{X}_{0,\widetilde{H}},\widetilde{X}_{1,\widetilde{H}},\widetilde{X}_{2,\widetilde{H}}\le \binom{|\widetilde{H}|}{2}\cdot 2^2 \le \frac{8\beta^2}{\eps^{16}\delta^6\eta^8}$.
Moreover, combining inequalities (\ref{ineq:dynamic-cm_pos+})--(\ref{ineq:dynamic-cm_neg-}), we have
\begin{align*}
    \widetilde{X}_0-\left(\frac{16\beta^2}{\eps^{16}\delta^6\eta^8}+\frac{\eps^{2}\delta^2\eta^2 m}{32} \right)\le \overline{X}_0  \le \widetilde{X}_0+\frac{8\beta^2}{\eps^{16}\delta^6\eta^8},\\
    \widetilde{X}_1-\left(\frac{32\beta^2}{\eps^{16}\delta^6\eta^8}+\frac{\eps^{2}\delta^2\eta^2 m}{16} \right) \le \overline{X}_1  \le \widetilde{X}_1+\frac{8\beta^2}{\eps^{16}\delta^6\eta^8},\\
    \widetilde{X}_2-\left(\frac{16\beta^2}{\eps^{16}\delta^6\eta^8}+\frac{\eps^{2}\delta^2\eta^2 m}{32} \right) \le \overline{X}_2   \le \widetilde{X}_2+\frac{8\beta^2}{\eps^{16}\delta^6\eta^8}.
\end{align*}

Let $\widetilde{k}_{\widetilde{L}} := \frac{1}{4\eps^2}[(\frac{1}{2}+\eps)^2 \widetilde{X}_0 - (\frac{1}{2}+\eps)(\frac{1}{2}-\eps)\widetilde{X}_1 + (\frac{1}{2}-\eps)^2 \widetilde{X}_2]$. 
Now we are ready to show that $\overline{k}_{\widetilde{L}}$ is close to $\widetilde{k}_{\widetilde{L}}$ up to an additive error of $\poly(\eta)\cdot m$ with probability at least $1-\frac{\delta}{4}$.
We have
\begin{align*}
    \overline{k}_{\widetilde{L}} 
    =~&\frac{1}{4\eps^2}\left[\left(\frac{1}{2}+\eps\right)^2 \overline{X}_0 - \left(\frac{1}{2}+\eps\right)\left(\frac{1}{2}-\eps\right)\overline{X}_1 + \left(\frac{1}{2}-\eps\right)^2 \overline{X}_2\right]\\
    \ge~&\frac{1}{4\eps^2}\left[\left(\frac{1}{2}+\eps\right)^2 \left(\widetilde{X}_0-\left(\frac{16\beta^2}{\eps^{16}\delta^6\eta^8}+\frac{\eps^{2}\delta^2\eta^2 m}{32} \right)\right) - \left(\frac{1}{2}+\eps\right)\left(\frac{1}{2}-\eps\right)\left(\widetilde{X}_1+\frac{8\beta^2}{\eps^{16}\delta^6\eta^8}\right) \right.\\
    &+\left. \left(\frac{1}{2}-\eps\right)^2 \left(\widetilde{X}_2-\left(\frac{16\beta^2}{\eps^{16}\delta^6\eta^8}+\frac{\eps^{2}\delta^2\eta^2 m}{32} \right)\right)\right]\\
    =~&\widetilde{k}_{\widetilde{L}} - \frac{1}{4\eps^2}\left[\left(\frac{1}{2}+\eps\right)^2 \left(\frac{16\beta^2}{\eps^{16}\delta^6\eta^8}+\frac{\eps^{2}\delta^2\eta^2 m}{32} \right) + \left(\frac{1}{2}+\eps\right)\left(\frac{1}{2}-\eps\right)\frac{8\beta^2}{\eps^{16}\delta^6\eta^8} \right.\\
    &+ \left.\left(\frac{1}{2}-\eps\right)^2 \left(\frac{16\beta^2}{\eps^{16}\delta^6\eta^8}+\frac{\eps^{2}\delta^2\eta^2 m}{32} \right)\right]\\
    =~&\widetilde{k}_{\widetilde{L}} - \left(\frac{5\beta^{2}}{2\eps^{18}\delta^{6}\eta^8} + \frac{6\beta^{2}}{\eps^{16}\delta^{6}\eta^8} + \frac{\delta^{2}\eta^2m}{256} + \frac{\eps^{2}\delta^2\eta^2 m}{64}\right)\\
    \ge~&\widetilde{k}_{\widetilde{L}} -  \frac{\eta m}{32},
\end{align*}
where the last step follows since $m = \Omega(\eps^{-26}\delta^{-10}\eta^{-14})$. On the other hand,
\begin{align*}
\overline{k}_{\widetilde{L}} =~&\frac{1}{4\eps^2}\left[\left(\frac{1}{2}+\eps\right)^2 \overline{X}_0 - \left(\frac{1}{2}+\eps\right)\left(\frac{1}{2}-\eps\right)\overline{X}_1 + \left(\frac{1}{2}-\eps\right)^2 \overline{X}_2\right]\\
\le~&\frac{1}{4\eps^2}\left[\left(\frac{1}{2}+\eps\right)^2 \left(\widetilde{X}_0+\frac{8\beta^2}{\eps^{16}\delta^6\eta^8}\right) - \left(\frac{1}{2}+\eps\right)\left(\frac{1}{2}-\eps\right)\left(\widetilde{X}_1-\left(\frac{32\beta^2}{\eps^{16}\delta^6\eta^8}+\frac{\eps^{2}\delta^2\eta^2 m}{16} \right)\right)\right.\\
&+\left. \left(\frac{1}{2}-\eps\right)^2 \left(\widetilde{X}_2+\frac{8\beta^2}{\eps^{16}\delta^6\eta^8}\right)\right]\\
=~&\widetilde{k}_{\widetilde{L}} +\frac{1}{4\eps^2}\left[\left(\frac{1}{2}+\eps\right)^2 \frac{8\beta^2}{\eps^{16}\delta^6\eta^8} + \left(\frac{1}{2}+\eps\right)\left(\frac{1}{2}-\eps\right)\left(\frac{32\beta^2}{\eps^{16}\delta^6\eta^8}+\frac{\eps^{2}\delta^2\eta^2 m}{16} \right) + \left(\frac{1}{2}-\eps\right)^2 \frac{8\beta^2}{\eps^{16}\delta^6\eta^8}\right]\\
=~&\widetilde{k}_{\widetilde{L}} + \left(\frac{3\beta^{2}}{\eps^{18}\delta^{6}\eta^8} - \frac{4\beta^{2}}{\eps^{16}\delta^{6}\eta^8}+\frac{ \delta^{2}\eta^2m}{256}-\frac{\eps^{2}\delta^2\eta^2m}{64}\right)\\
\le~&\widetilde{k}_{\widetilde{L}} + \frac{\eta m}{64},
\end{align*}
where the last step follows since $m = \Omega(\eps^{-26}\delta^{-10}\eta^{-14})$.

For each $i\in \widetilde{H}$, let $\widetilde{k}_{i,\widetilde{L}(+)}^{\pos} := (\frac{1}{2}+\frac{1}{4\eps})\cdot X_{i,\widetilde{L}(+)}^{\pos} + (\frac{1}{2}-\frac{1}{4\eps})\cdot X_{i,\widetilde{L}(-)}^{\pos}$ and $\widetilde{k}_{i,\widetilde{L}(+)}^{\neg} := (\frac{1}{2}+\frac{1}{4\eps})\cdot X_{i,\widetilde{L}(+)}^{\neg} + (\frac{1}{2}-\frac{1}{4\eps})\cdot X_{i,\widetilde{L}(-)}^{\neg}$.
We have
\begin{align*}
    \overline{k}_{i,\widetilde{L}(+)}^{\pos} &= \left(\frac{1}{2}+\frac{1}{4\eps}\right)\cdot \widetilde{X}_{i,+}^{\pos} + \left(\frac{1}{2}-\frac{1}{4\eps}\right)\cdot \widetilde{X}_{i,-}^{\pos}\\
    &\ge \left(\frac{1}{2}+\frac{1}{4\eps}\right)\cdot X_{i,\widetilde{L}(+)}^{\pos} + \left(\frac{1}{2}-\frac{1}{4\eps}\right) \left( X_{i,\widetilde{L}(-)}^{\pos} + \frac{2\beta}{\eps^{8}\delta^3\eta^4} + \frac{\eps^{10}\delta^5 \eta^6m}{128\beta} \right)\\
    &= \widetilde{k}_{i,\widetilde{L}(+)}^{\pos} - \left(\frac{1}{4\eps}-\frac{1}{2}\right) \left(\frac{2\beta}{\eps^{8}\delta^3\eta^4} + \frac{\eps^{10}\delta^5 \eta^6m}{128\beta} \right)\\
    &= \widetilde{k}_{i,\widetilde{L}(+)}^{\pos} -\left(\frac{\beta}{2\eps^{9}\delta^3\eta^4} - \frac{\beta}{\eps^{8}\delta^3\eta^4} + \frac{\eps^{9}\delta^5\eta^6 m}{512\beta} - \frac{\eps^{10}\delta^5\eta^6 m}{256\beta}\right)\\
    &\ge \widetilde{k}_{i,\widetilde{L}(+)}^{\pos} - \frac{\eps^8 \delta^3 \eta^5 m}{128\beta},
\end{align*}
where the last step follows since $m = \Omega(\eps^{-26}\delta^{-10}\eta^{-14})$. On the other hand,
\begin{align*}
    \overline{k}_{i,\widetilde{L}(+)}^{\pos} &= \left(\frac{1}{2}+\frac{1}{4\eps}\right)\cdot \widetilde{X}_{i,+}^{\pos} + \left(\frac{1}{2}-\frac{1}{4\eps}\right)\cdot \widetilde{X}_{i,-}^{\pos}\\
    &\le \left(\frac{1}{2}+\frac{1}{4\eps}\right)\left( X_{i,\widetilde{L}(+)}^{\pos} + \frac{2\beta}{\eps^{8}\delta^3\eta^4} + \frac{\eps^{10}\delta^5\eta^6 m}{128\beta} \right) + \left(\frac{1}{2}-\frac{1}{4\eps}\right)\cdot X_{i,\widetilde{L}(-)}^{\pos}\\
    &= \widetilde{k}_{i,\widetilde{L}(+)}^{\pos} + \left(\frac{1}{2}+\frac{1}{4\eps}\right)\left(\frac{2\beta}{\eps^{8}\delta^3\eta^4} + \frac{\eps^{10}\delta^5\eta^6 m}{128\beta} \right)\\
    &= \widetilde{k}_{i,\widetilde{L}(+)}^{\pos} +\left(\frac{\beta}{2\eps^{9}\delta^3\eta^4} + \frac{\beta}{\eps^{8}\delta^3\eta^4} + \frac{\eps^{9}\delta^5\eta^6 m}{512\beta} + \frac{\eps^{10}\delta^5\eta^6 m}{256\beta}\right)\\
    &\le \widetilde{k}_{i,\widetilde{L}(+)}^{\pos} +\frac{\eps^8 \delta^3 \eta^5 m}{128\beta},
\end{align*}
where the last step follows since $m = \Omega(\eps^{-26}\delta^{-10}\eta^{-14})$. Therefore, $|\overline{k}_{i,\widetilde{L}(+)}^{\pos} - \widetilde{k}_{i,\widetilde{L}(+)}^{\pos}| \le \frac{\eps^8 \delta^3 \eta^5 m}{128\beta}$.
Similarly, $|\overline{k}_{i,\widetilde{L}(+)}^{\neg} - \widetilde{k}_{i,\widetilde{L}(+)}^{\neg}| \le \frac{\eps^8 \delta^3 \eta^5 m}{128\beta}$.

Let $\widetilde{\mathrm{ALG}}:= \widetilde{k}_{\widetilde{L}} + \sum_{i\in \widetilde{H}}\max\{\widetilde{k}_{i,\widetilde{L}(+)}^{\pos}, \widetilde{k}_{i,\widetilde{L}(+)}^{\neg}\}  - \frac{\eta m}{32}$ denote the output of the algorithm obtained by running \cref{alg:max-2and-offline} with $L$ and $H$ replaced by $\widetilde L$ and $\widetilde H$, respectively.
Analogous to the proof of \cref{thm:max-2and-offline}, we can show that $(1-\frac{\eta}{2})\cdot \opt\le \widetilde{\mathrm{ALG}} \le \opt$ with probability at least $1-\frac{3\delta}{4}$. The argument relies primarily on two properties of the partition $(\widetilde{L},\widetilde{H})$: (i) every variable in $\widetilde{L}$ has degree at most $\theta m$ for $\theta:=\frac{\eps^4 \delta \eta^2}{4096}$ (since $\widetilde{L} \subseteq L$), and (ii) $|\widetilde{H}|\leq \frac{2\beta}{\eps^8 \delta^3 \eta^4}$.

Let $\mathrm{ALG}$ denote the output of \cref{alg:max-2and-dynamic}. By applying the union bound on Lemma~\ref{lem:ell-0-sampling-2and}, Lemma~\ref{lem:high-deg-reservoir-max-2and-dynamic}, and Lemma~\ref{lem:error-bound-insertion-max-2and-dynamic}, with probability at least $1-\frac{\delta}{16}-\frac{\delta}{16}-\frac{\delta}{8}-\frac{3\delta}{4}=1-\delta$, we have
\begin{align*}
    \mathrm{ALG} &= \overline{k}_{\widetilde{L}} + \sum_{i\in \widetilde{H}}\max\{\overline{k}_{i,\widetilde{L}(+)}^{\pos}, \overline{k}_{i,\widetilde{L}(+)}^{\neg}\} - \frac{\eta m}{16}\\
    &\le \widetilde{k}_{\widetilde{L}} + \frac{\eta m}{64} + \sum_{i\in \widetilde{H}}\max\{\widetilde{k}_{i,\widetilde{L}(+)}^{\pos}, \widetilde{k}_{i,\widetilde{L}(+)}^{\neg}\} + |\widetilde{H}|\cdot \frac{\eps^8 \delta^3 \eta^5 m}{128\beta} - \frac{\eta m}{16}\\
    &\le \widetilde{\mathrm{ALG}} + \frac{\eta m}{32}+\frac{\eta m}{64}+\frac{\eta m}{64} - \frac{\eta m}{16} &&\rhd |\widetilde{H}|\le \frac{2\beta}{\eps^8 \delta^3 \eta^4}\\
    &\le \opt,
\end{align*}
and
\begin{align*}
    \mathrm{ALG} &= \overline{k}_{\widetilde{L}} + \sum_{i\in \widetilde{H}}\max\{\overline{k}_{i,\widetilde{L}(+)}^{\pos}, \overline{k}_{i,\widetilde{L}(+)}^{\neg}\} - \frac{\eta m}{16}\\
    &\ge \widetilde{k}_{\widetilde{L}} - \frac{\eta m}{32} + \sum_{i\in \widetilde{H}}\max\{\widetilde{k}_{i,\widetilde{L}(+)}^{\pos}, \widetilde{k}_{i,\widetilde{L}(+)}^{\neg}\} - |\widetilde{H}|\cdot \frac{\eps^8 \delta^3 \eta^5 m}{128\beta} - \frac{\eta m}{16}\\
    &\ge \widetilde{\mathrm{ALG}} - \frac{\eta m}{64} - \frac{\eta m}{16} &&\rhd |\widetilde{H}|\le \frac{2\beta}{\eps^8\delta^3\eta^4}\\
    &\ge \widetilde{\mathrm{ALG}} - \frac{\eta m}{8} \\
    &\ge \qty(1-\frac{\eta}{2})\cdot \opt - \frac{\eta}{2} \cdot\opt &&\rhd m\le 4\opt \\
    &= (1-\eta)\cdot \opt.
\end{align*}
\end{proof}

\section{The algorithm with an estimated prediction accuracy parameter}
\label{sec:estimated-eps}
In this section, we show that it suffices to know a highly accurate estimate of the prediction accuracy parameter $\eps$ for our algorithms to work. Specifically, it is enough to have an estimate $\eps' \in (0,1/2]$ such that $\eps' = (1\pm O(\eta))\cdot \eps$ for \textsc{Max-$2$CSP}s; for \textsc{Max-$k$CSP}s, the hidden constant may depend on $k$.
The detailed proof below is for the offline \textsc{Max-Dicut} algorithm; the corresponding \textsc{Max-$k$AND} statement is recorded afterwards.

\subsection{Offline algorithm for \textsc{Max-Dicut} estimation with estimated $\eps$}
\label{sec:estimated-eps-max-dicut}
The algorithm is very similar to \cref{alg:general-offline}. The only changes are as follows. First, we define the high-degree threshold in Line~\ref{line:threshold-estimated-eps} using $\eps'$ instead of $\eps$, namely, $\theta m$ where $\theta := \frac{\eps'^4 \delta \eta^2}{3072}$. 
Compared with \cref{alg:general-offline}, the denominator constant is increased from $768$ to $3072$ to accommodate the additional error caused by replacing $\eps$ by $\eps'$.
Second, we replace every occurrence of $\eps$  in the decoding estimators by $\eps'$ (see Line~\ref{line:tildekL-estimated-eps} and Line~\ref{line:tilde-v-S*-estimated-eps}).
The resulting algorithm is formally given in \cref{alg:general-offline-estimated-eps}.

Next, we analyze the approximation guarantee of \cref{alg:general-offline-estimated-eps}.
\begin{theorem}
    Let $\eps \in (0,\frac{1}{2}]$, $\delta, \eta \in (0,1)$, and $G$ be a directed graph with $m=\Omega(\eps^{-26}\delta^{-10}\eta^{-14})$. Given an  oracle $\mathcal{O}$ that is $\eps$-accurate for the \textsc{Max-Dicut} of $G$, and an estimate $\eps' \in (0,\frac{1}{2}]$ satisfying $|\eps' - \eps|\le \frac{\eta}{256}\eps$,
    \cref{alg:general-offline-estimated-eps} provides a $(1-\eta)$-approximation for the \textsc{Max-Dicut} value of $G$ with probability at least $1-\delta$.
    \label{thm:general-offline-estimated-eps}
\end{theorem}

\begin{algorithm}[t!]
\begin{algorithmic}[1]
\Require Directed graph $G=(V,E)$, an oracle $\mathcal{O} \rightarrow \{-1, 1\}$ that is $\eps$-accurate for some \emph{unknown} $\eps \in (0,\frac{1}{2}]$, an estimate $\eps'\in (0,\frac{1}{2}]$ satisfying $|\eps' - \eps| \le \frac{\eta}{256}\eps$, and $\delta,\eta \in (0,1)$
\Ensure Estimate of the \textsc{Max-Dicut} value of $G$

\State \phantomsection \label{line:threshold-estimated-eps} $m \leftarrow |E|$, $\theta \leftarrow \frac{\eps'^{4} \delta \eta^2}{3072}$,  $L\leftarrow \{v\in V: \deg(v) < \theta m\}, H\leftarrow V \setminus L$
\State {\bf initialize} $X_0\leftarrow 0, X_1 \leftarrow 0, X_2 \leftarrow 0$.
\State {\bf initialize} $X_{v,L^+} \leftarrow 0, X_{v,L^-} \leftarrow 0, X_{L^+,v} \leftarrow 0, X_{L^-,v}\leftarrow 0$ for each $v\in H$.
\ForEach{directed edge $(u,v)\in E$}
    \State $Y_u = \mathcal{O}(u), Y_v = \mathcal{O}(v)$
    \If{$u\in L$ and $v \in L$}
        \If{$Y_u = 1$ and $Y_v = -1$}
            $X_0 \leftarrow X_0 +1$
        \EndIf
        \If{$Y_u = Y_v$}
            $X_1 \leftarrow X_1 +1$
        \EndIf
        \If{$Y_u = -1$ and $Y_v = 1$}
            $X_2 \leftarrow X_2 +1$
        \EndIf
    \EndIf
    
    \If{$u\in H$ and $v \in L$}
        \State \textbf{if} $Y_v = 1$ \textbf{then} $X_{u,L^+} \leftarrow X_{u,L^+} + 1$ \textbf{else} $X_{u,L^-} \leftarrow X_{u,L^-} + 1$        
    \EndIf
    
    \If{$u\in L$ and $v \in H$}
        \State \textbf{if} $Y_u = 1$ \textbf{then} $X_{L^+,v} \leftarrow X_{L^+,v} + 1$ \textbf{else} $X_{L^-,v} \leftarrow X_{L^-,v} + 1$ 
    \EndIf 
\EndFor
\State \phantomsection \label{line:tildekL-estimated-eps} $\widetilde{k}_{L} \leftarrow \frac{1}{4\eps'^2}[(\frac{1}{2}+\eps')^2 X_0 - (\frac{1}{2}+\eps')(\frac{1}{2}-\eps')X_1 + (\frac{1}{2}-\eps')^2 X_2]$ 
\ForEach{$v\in H$} 
\State $\widetilde{k}_{S^*,v} \leftarrow (\frac{1}{2}+\frac{1}{4\eps'})X_{L^+,v} + (\frac{1}{2}-\frac{1}{4\eps'})X_{L^-,v}$, $\widetilde{k}_{v,T^*} \leftarrow (\frac{1}{2}-\frac{1}{4\eps'})X_{v,L^+} + (\frac{1}{2}+\frac{1}{4\eps'})X_{v,L^-}$  \phantomsection \label{line:tilde-v-S*-estimated-eps}
\label{line:num-in-edges-estimated-eps}

\EndFor
\State \textbf{return} $\widetilde{k}_{L} + \sum_{v\in H}\max\{\widetilde{k}_{S^*,v}, \widetilde{k}_{v,T^*}\}  - \frac{\eta m}{16}$ 
\end{algorithmic}
\caption{A learning-augmented \emph{offline} algorithm for \textsc{Max-Dicut} estimation with estimated $\eps$}
\label{alg:general-offline-estimated-eps}
\end{algorithm}

\subsubsection{Analysis of \cref{alg:general-offline-estimated-eps}}

We first reprove concentration bounds for $X_0,X_1,X_2$ with the new high-degree threshold.

\begin{lemma}
\label{lem:concentration-low-estimated-eps}
Let $\eps' \in (0,\frac{1}{2}]$ and $\delta,\eta \in (0,1)$.
For each $i\in \{0,1,2\}$, we have $|X_i - \E[X_i]| \le \frac{\eps'^{2}\eta m}{16}$ with probability at least $1-\frac{\delta}{6}$.
\end{lemma}
\begin{proof}
    Let $E_L:=\{(u,v)\in E: u,v\in L\}$ denote the set of directed edges with both endpoints in $L$.
    Let $m_L:=|E_L|$.
    For each edge $(u,v)\in E_L$, define the following indicator random variables:
    \begin{align*}
    X_{0,uv} &=
    \begin{cases}
    1 & \text{if } Y_u=1 \wedge Y_v=-1,\\
    0 & \text{otherwise},
    \end{cases}
    &
    X_{1,uv} &=
    \begin{cases}
    1 & \text{if } Y_u=Y_v,\\
    0 & \text{otherwise},
    \end{cases}
    &
    X_{2,uv} &=
    \begin{cases}
    1 & \text{if } Y_u=-1 \wedge Y_v=1,\\
    0 & \text{otherwise}.
    \end{cases}
    \end{align*}
    
    Note that for $i\in \{0,1,2\}$, we have $\E[X_{i,uv}] = \Pr[X_{i,uv} = 1]$ and $X_i=\sum_{(u,v)\in E_L}X_{i,uv}$. Therefore,
    \begin{align*}
        \E\qty[X_i^2] &= \E\qty[ \qty(\sum_{(u,v)\in E_L} X_{i,uv})^2]\\
        &= \E\qty[\sum_{(u,v)\in E_L}\sum_{(w,z)\in E_L}X_{i,uv} X_{i,wz}]\\
        &=\sum_{(u,v)\in E_L}\sum_{(w,z)\in E_L}\E[X_{i,uv}X_{i,wz}]\\ 
        &=\sum_{(u,v)\in E_L}\sum_{(w,z)\in E_L} \Pr[X_{i,uv} = 1\wedge X_{i,wz} = 1]\\
        &= \sum_{\substack{(u,v), (w,z)\in E_L: \\ \{u,v\}\cap \{w,z\} = \emptyset}} \Pr[X_{i,uv} = 1\wedge X_{i,wz} = 1] + \sum_{\substack{(u,v), (w,z)\in E_L: \\\{u,v\}\cap \{w,z\} \neq \emptyset}} \Pr[X_{i,uv} = 1\wedge X_{i,wz} = 1]\\
        &\leq \sum_{\substack{(u,v), (w,z)\in E_L: \\\{u,v\}\cap \{w,z\} = \emptyset}} \Pr[X_{i,uv} = 1] \Pr[X_{i,wz} = 1] + \sum_{\substack{(u,v), (w,z)\in E_L: \\\{u,v\}\cap \{w,z\} \neq \emptyset}} 1\\
        &\leq \sum_{(u,v)\in E_L}\sum_{(w,z)\in E_L}\E[X_{i,uv}]\E[X_{i,wz}] + m_L\cdot 2\cdot \theta m \\
        &\leq (\E[\sum_{(u,v)\in E_L}X_{i,uv}])^2 + 2\theta m^2\\
        &= (\E[X_i])^2 + 2\theta m^2,
    \end{align*}
    where the third-to-last step follows from $\deg(v) < \theta m$ for $v\in L$, and the second-to-last step follows from $m_L \le m$.
    
    Then we have $\Var[X_i] = \E[X_i^2] - (\E[X_i])^2 \le 2\theta m^2$.
    By Chebyshev's inequality, 
    \begin{align*}
        \Pr[|X_i - \E[X_i]| > \frac{\eps'^{2}\eta m}{16}] 
        \le \frac{\Var[X_i]}{(\eps'^{2}\eta m/{16})^2}
        &\le \frac{512\theta m^2}{\eps'^{4}\eta^2 m^2} &&\rhd \Var[X_i] \le 2\theta m^2\\
        &= \frac{512\theta}{\eps'^{4}\eta^2}
        < \frac{\delta}{6}.  &&\rhd \theta < \frac{\eps'^{4} \delta \eta^2}{3072}
    \end{align*}
\end{proof}

Define
    $k_{L}:= |\{(u,v)\in E: u\in S^*\cap L, v\in T^*\cap L\}|$
to be the number of directed edges with both endpoints in $L$ (i.e., both low-degree) that contribute to the \textsc{Max-Dicut} value.
The next lemma shows that the modified estimator for $k_L$ remains accurate  when $\eps$ is replaced by $\eps'$.

\begin{lemma}
\label{lem:concentration-k-L-estimated-eps}
Let $\delta,\eta \in (0,1)$.
We have $|\widetilde{k}_L - k_L| \le \frac{\eta m}{32}$ with probability at least $1-\frac{\delta}{2}$.
\end{lemma}
\begin{proof}
    Let $k_0$ denote the number of edges from $S^* \cap L$ to $T^* \cap L$. Let $k_1$ denote the number of edges whose endpoints both lie in $S^*\cap L$ or both lie in $T^*\cap L$. Let $k_2$ denote the number of edges from $T^* \cap L$ to $S^* \cap L$. Note that $k_0 = k_L$.
We have
\begin{align}
    \E[X_0] &= \left(\frac{1}{2}+\eps \right)^2 \cdot k_0 + \left(\frac{1}{2}+\eps \right)\left(\frac{1}{2}-\eps \right)\cdot k_1 + \left(\frac{1}{2}-\eps \right)^2 \cdot k_2, \label{eq:exp-X0}\\
    \E[X_1] &= 2\left(\frac{1}{2}+\eps \right)\left(\frac{1}{2}-\eps \right)\cdot k_0 + \left[ \left(\frac{1}{2}+\eps \right)^2 + \left(\frac{1}{2}-\eps \right)^2 \right]\cdot k_1 + 2\left(\frac{1}{2}+\eps \right)\left(\frac{1}{2}-\eps \right)\cdot k_2, \label{eq:exp-X1}\\
    \E[X_2] &= \left(\frac{1}{2}-\eps \right)^2 \cdot k_0 + \left(\frac{1}{2}+\eps \right)\left(\frac{1}{2}-\eps \right)\cdot k_1 + \left(\frac{1}{2}+\eps \right)^2 \cdot k_2. \label{eq:exp-X2}
\end{align}

To bound $|\widetilde{k}_L - k_L|$, we apply  the triangle inequality: $|\widetilde{k}_L - k_L| \le |\widetilde{k}_L - \E[\widetilde{k}_L]| + |\E[\widetilde{k}_L] - k_L|$.

We first bound the second term. By Line~\ref{line:tildekL-estimated-eps} of \cref{alg:general-offline-estimated-eps}, we have $\widetilde{k}_{L} := \frac{1}{4\eps'^2}[(\frac{1}{2}+\eps')^2 X_0 - (\frac{1}{2}+\eps')(\frac{1}{2}-\eps')X_1 + (\frac{1}{2}-\eps')^2 X_2]$. Therefore,
\begin{align*}
    \E[\widetilde{k}_{L}] &= \frac{1}{4\eps'^2}\left[\left(\frac{1}{2}+\eps'\right)^2 \E[X_0] - \left(\frac{1}{2}+\eps'\right)\left(\frac{1}{2}-\eps'\right)\E[X_1] + \left(\frac{1}{2}-\eps'\right)^2 \E[X_2]\right]\\
    &= \frac{(\eps+\eps')^2}{4\eps'^2}\cdot k_0 + \frac{\eps'^2 - \eps^2}{4\eps'^2}\cdot k_1 + \frac{(\eps-\eps')^2}{4\eps'^2}\cdot k_2,
\end{align*}
where the second equality follows by substituting Equations~(\ref{eq:exp-X0})--(\ref{eq:exp-X2}).
Then we have
\begin{align*}
    |\E[\widetilde{k}_{L}] - k_L| &= \left|\left(\frac{(\eps+\eps')^2}{4\eps'^2}-1\right)\cdot k_0 + \frac{\eps'^2 - \eps^2}{4\eps'^2}\cdot k_1 + \frac{(\eps-\eps')^2}{4\eps'^2}\cdot k_2\right|\\
    &\le \left|\frac{(\eps+\eps')^2}{4\eps'^2}-1\right|\cdot k_0 + \left|\frac{\eps'^2 - \eps^2}{4\eps'^2}\right|\cdot k_1 + \left|\frac{(\eps-\eps')^2}{4\eps'^2}\right|\cdot k_2.
\end{align*}

Let $q:= \frac{|\eps - \eps'|}{\eps'}$. We have $q \le \frac{\frac{\eta}{256}\eps}{(1-\frac{\eta}{256})\eps} = \frac{\eta}{256-\eta} < \frac{\eta}{255}$. Then
\begin{align*}
    \left|\frac{(\eps+\eps')^2}{4\eps'^2}-1\right| &= \left| \frac{\eps - \eps'}{\eps'} + \frac{(\eps - \eps')^2}{4\eps'^2} \right| \le \frac{|\eps - \eps'|}{\eps'} + \frac{(\eps - \eps')^2}{4\eps'^2} = q + \frac{q^2}{4} < 2q,\\
    \left|\frac{\eps'^2 - \eps^2}{4\eps'^2}\right| &= \frac{|\eps - \eps'|\cdot |\eps + \eps'|}{4\eps'^2} \le \frac{q\eps' \cdot (|\eps - \eps'|+2\eps')}{4\eps'^2} = \frac{q(2+q)}{4}<2q,\\
    \left|\frac{(\eps-\eps')^2}{4\eps'^2}\right| &= \frac{(\eps-\eps')^2}{4\eps'^2} = \frac{q^2}{4}<2q.
\end{align*}
Therefore, $|\E[\widetilde{k}_{L}] - k_L| < 2q(k_0+k_1+k_2)\le 2qm < \frac{2\eta m}{255} < \frac{\eta m}{64}$.

Next, we bound $|\widetilde{k}_L - \E[\widetilde{k}_L]|$. By applying union bound on Lemma~\ref{lem:concentration-low-estimated-eps}, we have that with probability at least $1-\frac{\delta}{6}\cdot 3 = 1-\frac{\delta}{2}$, $|X_i - \E[X_i]| \le \frac{\eps'^{2}\eta m}{16}$ for all $i\in \{0,1,2\}$. Therefore,
\begin{align*}
    &|\widetilde{k}_L - \E[\widetilde{k}_L]| \\ \le~&\frac{1}{4\eps'^2}\left[\left(\frac{1}{2}+\eps'\right)^2 |X_0 - \E[X_0]| + \left(\frac{1}{2}+\eps'\right)\left(\frac{1}{2}-\eps'\right)|X_1-\E[X_1]| + 
     \left(\frac{1}{2}-\eps'\right)^2 |X_2-\E[X_2]|\right]\\
    \le~&\frac{1}{4\eps'^2}\left[\left(\frac{1}{2}+\eps'\right)^2 + \left(\frac{1}{2}+\eps'\right)\left(\frac{1}{2}-\eps'\right)
    + \left(\frac{1}{2}-\eps'\right)^2 \right] \cdot \frac{\eps'^{2}\eta m}{16}\\
    =~& \frac{1}{4} \cdot \left( \frac{3}{4}+\eps'^2\right) \cdot \frac{\eta m}{16}
    \le\frac{\eta m}{64}.
\end{align*}

Therefore, we have $|\widetilde{k}_L - k_L| \le |\widetilde{k}_L - \E[\widetilde{k}_L]| + |\E[\widetilde{k}_L] - k_L| \le \frac{\eta m}{64} + \frac{\eta m}{64}= \frac{\eta m}{32}$ with probability at least $1-\frac{\delta}{2}$.
\end{proof}

For each high-degree vertex $v\in H$, define 
\begin{enumerate}
    \item $k_{S^*,v}:= |\{(u,v)\in E: u\in S^*\cap L\}|$ to be the number of directed edges from $S^*\cap L$ to $v$.
    \item $k_{T^*,v}:= |\{(u,v)\in E: u\in T^*\cap L\}|$ to be the number of directed edges from $T^*\cap L$ to $v$.
    \item $k_{v,S^*}:= |\{(v,u)\in E:  u\in S^*\cap L\}|$ to be the number of directed edges from $v$ to $S^*\cap L$.
    \item $k_{v,T^*}:= |\{(v,u)\in E: u\in T^*\cap L\}|$ to be the number of directed edges from $v$ to $T^*\cap L$.
\end{enumerate}

\begin{lemma}
\label{lem:concentration-general-estimated-eps}
Let $\eps' \in (0,\frac{1}{2}]$ and $\delta,\eta \in (0,1)$. For each vertex $v\in H$,
each of the following inequalities holds with probability at least $1-\frac{\eps'^{4}\delta^2 \eta^2}{24576}$:
\begin{enumerate}
    \item[(1)] $|\widetilde{k}_{S^*,v} - k_{S^*,v}| \le \frac{\eps'^{4}\delta \eta^3 m}{393216}+\frac{q}{2}\cdot e(L,v)$,
    \item[(2)] $|\widetilde{k}_{v,T^*} - k_{v,T^*}| \le \frac{\eps'^{4}\delta \eta^3 m}{393216}+\frac{q}{2}\cdot e(v,L)$,
\end{enumerate}
where $q:= \frac{|\eps - \eps'|}{\eps'}$.
\end{lemma}
\begin{proof}
    We prove the first inequality below; the second inequality follows by the same argument.

    We have
    \begin{align}
    \E[X_{L^+,v}] &= \left(\frac{1}{2}+\eps \right)\cdot k_{S^*,v} + \left(\frac{1}{2}-\eps \right) \cdot k_{T^*,v}, \label{eq:exp-XL+v}\\
    \E[X_{L^-,v}] &= \left(\frac{1}{2}-\eps \right)\cdot k_{S^*,v} + \left(\frac{1}{2}+\eps \right) \cdot k_{T^*,v}.\label{eq:exp-XL-v}
\end{align}

To bound $|\widetilde{k}_{S^*,v} - k_{S^*,v}|$, we apply  the triangle inequality: $|\widetilde{k}_{S^*,v} - k_{S^*,v}| \le |\widetilde{k}_{S^*,v} - \E[\widetilde{k}_{S^*,v}]| + |\E[\widetilde{k}_{S^*,v}] - k_{S^*,v}|$.

We first bound the second term. By Line~\ref{line:tilde-v-S*-estimated-eps} of \cref{alg:general-offline-estimated-eps}, we have $\widetilde{k}_{S^*,v} := (\frac{1}{2}+\frac{1}{4\eps'})X_{L^+,v} + (\frac{1}{2}-\frac{1}{4\eps'})X_{L^-,v}$. Hence,
\begin{align*}
    \E[\widetilde{k}_{S^*,v}] &= \left(\frac{1}{2}+\frac{1}{4\eps'}\right)\E[X_{L^+,v}] + \left(\frac{1}{2}-\frac{1}{4\eps'}\right)\E[X_{L^-,v}]\\
    &=\frac{\eps + \eps'}{2\eps'}\cdot k_{S^*,v} + \frac{\eps'-\eps}{2\eps'}\cdot k_{T^*,v}\\
    &= k_{S^*,v}+\frac{\eps-\eps'}{2\eps'}\cdot(k_{S^*,v}-k_{T^*,v}).
\end{align*}
Therefore, we have
\begin{align*}
    |\E[\widetilde{k}_{S^*,v}]-k_{S^*,v}| 
    = \left|\frac{\eps-\eps'}{2\eps'}\cdot(k_{S^*,v}-k_{T^*,v})\right|
    = \left|\frac{\eps-\eps'}{2\eps'}\right|\cdot|k_{S^*,v}-k_{T^*,v}|
    \le \frac{q}{2}\cdot e(L,v).
\end{align*}

Next, we bound $|\widetilde{k}_{S^*,v} - \E[\widetilde{k}_{S^*,v}]|$. Note that 
    \begin{align*}
        \widetilde{k}_{S^*,v} &= \left(\frac{1}{2}+\frac{1}{4\eps'}\right)X_{L^+,v} + \left(\frac{1}{2}-\frac{1}{4\eps'}\right)X_{L^-,v} \\
        &= \left(\frac{1}{2}+\frac{1}{4\eps'}\right)X_{L^+,v} + \left(\frac{1}{2}-\frac{1}{4\eps'}\right)(e(L,v) - X_{L^+,v}) \\
        &= \frac{1}{2\eps'}X_{L^+,v} + \left(\frac{1}{2}-\frac{1}{4\eps'}\right)e(L,v).
    \end{align*}
Therefore, $\Var[\widetilde{k}_{S^*,v}] = \frac{1}{4\eps'^2}\Var[X_{L^+,v}]$.

Given a high-degree vertex $v\in H$, for each incoming edge $(u,v)\in E$ with $u\in L$, define an indicator random variable $X_{uv}$  as follows:
    \begin{align*}
        X_{uv} =
    \begin{cases}
    1 & \text{if } Y_u=1,\\
    0 & \text{otherwise}.
    \end{cases}
    \end{align*}
    Note that $X_{L^+,v}=\sum_{(u,v)\in E: u\in L}X_{uv}$ and $\E[X_{uv}] = \Pr[X_{uv} = 1]$.
    Then $\Var[X_{uv}] = \E[X^2_{uv}] - (\E[X_{uv}])^2 = \E[X_{uv}]\cdot (1-\E[X_{uv}]) =(\frac{1}{2}+\eps)(\frac{1}{2}-\eps)=\frac{1}{4}-\eps^2 \le \frac{1}{4}$. 
    Since $X_{uv}$'s are independent, we have
    \begin{align*}
        \Var[X_{L^+,v}] = \sum_{(u,v)\in E: u\in L}\Var[X_{uv}]\le \frac{m}{4}.
    \end{align*}

    So, $\Var[\widetilde{k}_{S^*,v}]=\frac{1}{4\eps'^2}\Var[X_{L^+,v}] \le \frac{m}{16\eps'^2}$.
    By Chebyshev's inequality,
    \begin{align*}
        \Pr[|\widetilde{k}_{S^*,v} - \E[\widetilde{k}_{S^*,v}]|
        \ge \frac{\eps'^{4}\delta \eta^3 m}{393216}] 
        \le \frac{\Var[\widetilde{k}_{S^*,v}]}{(\eps'^{4}\delta \eta^3 m/393216)^2}
        \le \frac{393216^2 / 16}{\eps'^{10}\delta^2\eta^6 m} \le \frac{\eps'^{4}\delta^2 \eta^2}{24576},
    \end{align*}
    where the last step follows from our assumption that $m =\Omega(\eps^{-26}\delta^{-10}\eta^{-14})$ and $\eps' = \Theta(\eps)$.

    Therefore, we have $|\widetilde{k}_{S^*,v} - k_{S^*,v}| \le |\widetilde{k}_{S^*,v} - \E[\widetilde{k}_{S^*,v}]| + |\E[\widetilde{k}_{S^*,v}] - k_{S^*,v}|\le \frac{\eps'^{4}\delta \eta^3 m}{393216}+\frac{q}{2}\cdot e(L,v)$ with probability at least $1-\frac{\eps'^{4}\delta^2 \eta^2}{24576}$.
\end{proof}

Define 
    $k_{H}:= |\{(u,v)\in E: u\in S^*\cap H, v\in T^*\cap H\}|$
to be the number of directed edges with both endpoints in $H$ (i.e., both high-degree) that contribute to the \textsc{Max-Dicut} value.
Note that $|H| \le \frac{2m}{\theta m}=\frac{2}{\theta} = \frac{6144}{\eps'^{4} \delta\eta^2}$. Thus, we have the following observation. 

\begin{fact}
\label{fact:size-H-estimated-eps}
It holds that $k_{H} < |H|^2 \le \frac{6144^2}{\eps'^{8} \delta^2\eta^4}$.
\end{fact}

\subsubsection{Proof of \cref{thm:general-offline-estimated-eps}}
\begin{proof}
    Recall that $k_L$ (resp. $k_H$) is the number of directed edges with both endpoints in $L$ (resp. $H$) that contribute to the \textsc{Max-Dicut} value.

Furthermore, define 
\begin{align*}
    k_{L,H}:=|\{(u,v)\in E: u\in S^*\cap L, v\in T^*\cap H\}|
\end{align*}
to be the number of directed edges from $L$ to $H$ that contribute to the \textsc{Max-Dicut} value, and 
\begin{align*}
    k_{H,L}:=|\{(u,v)\in E: u\in S^*\cap H, v\in T^*\cap L\}|
\end{align*}
to be the number of directed edges from $H$ to $L$ that contribute to the \textsc{Max-Dicut} value.

Therefore, we have $\opt = k_L+k_H+ k_{L,H} + k_{H,L}$. 

For each $v\in H$, define $c_v := \frac{\eps'^{4}\delta \eta^3 m}{393216} + \frac{q}{2}\cdot (e(L,v)+e(v,L))$. By Lemma~\ref{lem:concentration-general-estimated-eps}, we have that for each vertex $v\in H$,
each of the following inequalities holds with probability at least $1-\frac{\eps'^{4}\delta^2 \eta^2}{24576}$:
\begin{align*}
    |\widetilde{k}_{S^*,v} - k_{S^*,v}| &\le \frac{\eps'^{4}\delta \eta^3 m}{393216}+\frac{q}{2}\cdot e(L,v)\le c_v,\\
    |\widetilde{k}_{v,T^*} - k_{v,T^*}| &\le \frac{\eps'^{4}\delta \eta^3 m}{393216}+\frac{q}{2}\cdot e(v,L) \le c_v.
\end{align*}
Hence,
    $|\max\{\widetilde{k}_{S^*,v}, \widetilde{k}_{v,T^*}\} - \max\{{k}_{S^*,v}, {k}_{v,T^*}\}|\le c_v$.
Therefore, with probability at least $1-\frac{\eps'^{4}\delta^2 \eta^2}{24576}\cdot 2 \cdot |H| \ge 1-\frac{\delta}{2}$, we have
\begin{align}
    &\left| \sum_{v\in H}\max\{\widetilde{k}_{S^*,v}, \widetilde{k}_{v,T^*}\} - \sum_{v\in H}\max\{{k}_{S^*,v}, {k}_{v,T^*}\}\right| \notag \\
    \le~&\sum_{v\in H} |\max\{\widetilde{k}_{S^*,v}, \widetilde{k}_{v,T^*}\} - \max\{{k}_{S^*,v}, {k}_{v,T^*}\}| \notag\\
    \le~&\sum_{v\in H} c_v \notag\\
    =~&|H|\cdot \frac{\eps'^{4}\delta \eta^3 m}{393216} + \frac{q}{2}\sum_{v\in H} (e(L,v)+e(v,L))\notag\\
    \le~&\frac{\eta m}{64} + \frac{qm}{2} &&\rhd \sum_{v\in H} (e(L,v)+e(v,L))\le m\notag\\
    <~&\frac{\eta m}{64} + \frac{\eta m}{510}&&\rhd q:= \frac{|\eps - \eps'|}{\eps'}= \frac{\eta}{256-\eta} <\frac{\eta }{255}\notag\\
    <~& \frac{\eta m}{32}.\label{ineq:high-part-sum-error}
\end{align}

Let $\mathrm{ALG}$ denote the output of \cref{alg:general-offline-estimated-eps}.
By union bound, with probability at least $1-\frac{\delta}{2}-\frac{\delta}{2} = 1-\delta$, we have
\begin{align*}
        \mathrm{ALG} &= \widetilde{k}_{L} + \sum_{v\in H}\max\{\widetilde{k}_{S^*,v}, \widetilde{k}_{v,T^*}\} - \frac{\eta m}{16}\\
        &\le k_L + \frac{\eta m}{32} + \sum_{v\in H}\max\{{k}_{S^*,v}, {k}_{v,T^*}\} + \frac{\eta m}{32} - \frac{\eta m}{16} &&\rhd \text{Lemma~\ref{lem:concentration-k-L-estimated-eps} and Inequality~(\ref{ineq:high-part-sum-error})}\\
        &= k_L + \sum_{v\in H}\max\{{k}_{S^*,v}, {k}_{v,T^*}\} \\
        &\le \opt,
    \end{align*}
where the last step follows since $k_L + \sum_{v\in H}\max\{{k}_{S^*,v}, {k}_{v,T^*}\}$ is a lower bound on the value of a feasible dicut of $G$.
Specifically, we fix the restriction on $L$  to be $(S^*\cap L, T^*\cap L)$, and then assign each vertex $v \in H$ to the side that maximizes the resulting dicut value.

Furthermore, 
    \begin{align*}
        \mathrm{ALG} &= \widetilde{k}_{L} + \sum_{v\in H}\max\{\widetilde{k}_{S^*,v}, \widetilde{k}_{v,T^*}\} - \frac{\eta m}{16}\\
        &\ge k_L - \frac{\eta m}{32} + \sum_{v\in H}\max\{{k}_{S^*,v}, {k}_{v,T^*}\} -  \frac{\eta m}{32}- \frac{\eta m}{16} &&\rhd \text{Lemma~\ref{lem:concentration-k-L-estimated-eps} and Inequality~(\ref{ineq:high-part-sum-error})}\\
        &\ge k_L + k_{L,H} + k_{H,L} -\frac{\eta m}{8} &&\rhd \sum_{v\in H}\max\{{k}_{S^*,v}, {k}_{v,T^*}\} \ge k_{L,H} + k_{H,L}\\
        &= \opt - k_H  -\frac{\eta m}{8} &&\rhd \opt = k_L+k_H+ k_{L,H} + k_{H,L}\\
        &\ge 
        \opt - \frac{6144^2}{\eps'^{8} \delta^2\eta^4} - \frac{\eta}{2} \cdot \opt &&\rhd \text{Fact~\ref{fact:size-H-estimated-eps} and } m\le 4\opt\\
        &\ge (1-\eta)\cdot \opt, &&\rhd \frac{6144^2}{\eps'^{8} \delta^2\eta^4} \le \frac{\eta m}{8}\le \frac{\eta}{2}\cdot \opt 
    \end{align*}
    where the third step follows from 
    \begin{align*}
        \sum_{v\in H}\max\{{k}_{S^*,v}, {k}_{v,T^*}\} &= \sum_{v\in H\cap S^*}\max\{{k}_{S^*,v}, {k}_{v,T^*}\}+\sum_{v\in H\cap T^*}\max\{{k}_{S^*,v}, {k}_{v,T^*}\} \\
        &\ge \sum_{v\in H\cap S^*} {k}_{v,T^*} + \sum_{v\in H\cap T^*} {k}_{S^*,v}\\
        &= k_{H,L} + k_{L,H},
    \end{align*}
    and the last step follows from our assumption that $m =\Omega(\eps^{-26}\delta^{-10}\eta^{-14})$ and $\eps' = \Theta(\eps)$.

\end{proof}

\subsection{Estimated $\eps$ for \textsc{Max-$k$AND}}
\label{sec:estimated-eps-max-kand}
We also use the same estimated-$\eps$ modification for the \textsc{Max-$k$AND} algorithms.  Namely, replace $\eps$ by $\eps'$ in every decoder coefficient and in every threshold or sketch parameter, run the corresponding exact-$\eps$ algorithm with accuracy parameter $\eta/8$, and use final subtractive slack $\eta m/2^{k+2}$.  This applies to:
\begin{enumerate}
    \item the bounded-degree algorithm  \cref{alg:bounded-max-kand}, with the bounded-degree assumption stated using $\eps'$ and $\eta/8$;
    \item the offline unbounded-degree algorithm  \cref{alg:offline-unbounded-max-kand}, with its adaptive threshold stated using $\eps'$ and $\eta/8$;
    \item the insertion-only and dynamic unbounded-degree algorithms in \cref{thm:streaming-unbounded-max-kand,thm:dynamic-streaming-unbounded-max-kand}, with all sketch parameters stated using $\eps'$ and $\eta/8$.
\end{enumerate}
The required accuracy of $\eps'$ is as follows.

\begin{corollary}[Estimated $\eps$ for \textsc{Max-$k$AND}]
\label{cor:estimated-eps-max-kand}
Let $\eps\in(0,\frac12]$, $\delta,\eta\in(0,1)$, and $k\ge 3$. If the oracle is $\eps$-accurate and
\[
\left|\frac{\eps'-\eps}{2\eps'}\right|\le \frac{\eta}{2^{k+5}k},
\]
then each modified algorithm above remains a $(1-\eta)$-approximation algorithm under the same assumptions as its exact-$\eps$ version.  The running-time and space bounds are unchanged up to constants, with $\eps'$ in place of $\eps$. Equivalently, it suffices to have $\eps'=(1\pm O(\eta/(k2^k)))\eps$.
\end{corollary}

\begin{proof}
Fix one of the modified algorithms and fix the high set used by it; for
\cref{alg:bounded-max-kand} we take $H=\emptyset$ and $L=[n]$. Fix an
assignment $\boldsymbol{\tau}_H\in\{\pm1\}^H$, and let
$\boldsymbol{\tau}$ be the assignment that agrees with
$\boldsymbol{\tau}_H$ on $H$ and with $\boldsymbol{\sigma}^*$ on $L$. We compare the modified decoded value with
$\mathrm{val}_{\Psi}(\boldsymbol{\tau})$.

We first record the one-literal calculation. Since
$\E[Y_i]=2\eps\sigma_i^*$,
\begin{align*}
\E\left[\frac12+\frac{bY_i}{4\eps'}\right]
=
\frac12+\frac{b\E[Y_i]}{4\eps'}
=
\frac12+\frac{b\sigma_i^*\eps}{2\eps'} 
=
\begin{cases}
1-\frac{\eps'-\eps}{2\eps'}, & \text { if } b\sigma_i^*=1,\\[1mm]
\frac{\eps'-\eps}{2\eps'}, & \text { if } b\sigma_i^*=-1.
\end{cases}
\end{align*}
The bounded-degree counter decoder is exactly the product of these one-literal
decoders. Indeed, if a constraint has exactly $i$ literals not satisfied by
the predicted assignment, then
\begin{align*}
\prod_{t=1}^k\left(\frac12+\frac{b_tY_{j_t}}{4\eps'}\right)
=
\left(\frac12+\frac{1}{4\eps'}\right)^{k-i}
\left(\frac12-\frac{1}{4\eps'}\right)^i 
=
\frac{1}{(2\eps')^k}
\left(\frac12+\eps'\right)^{k-i}
\left(-\frac12+\eps'\right)^i,
\end{align*}
which is the coefficient used by \cref{alg:bounded-max-kand} after replacing
$\eps$ by $\eps'$.

Now consider a constraint
$C=((j_1,\ldots,j_k),(b_1,\ldots,b_k))$. If some literal with $j_t\in H$ is
false under $\boldsymbol{\tau}_H$, then the contribution of $C$ to both the
decoded expression and $\mathrm{val}_{\Psi}(\boldsymbol{\tau})$ is zero.
Otherwise, suppose that there are $\ell$ low variables in $C$, and exactly
$s$ of their literals are false under $\boldsymbol{\sigma}^*$. Since the
coordinates of $C$ are distinct, the corresponding oracle labels are
independent, and the expected decoded contribution of $C$ is
\[
\left(1-\frac{\eps'-\eps}{2\eps'}\right)^{\ell-s}
\left(\frac{\eps'-\eps}{2\eps'}\right)^s .
\]
The true contribution of $C$ to $\mathrm{val}_{\Psi}(\boldsymbol{\tau})$ is
$1$ if $s=0$ and $0$ otherwise. Under the assumption
$\left|\frac{\eps'-\eps}{2\eps'}\right|\le 1/(2k)$, if $s=0$ then
\begin{align*}
\left|
\left(1-\frac{\eps'-\eps}{2\eps'}\right)^\ell-1
\right|
\le
\ell\left|\frac{\eps'-\eps}{2\eps'}\right|
\left(1+\left|\frac{\eps'-\eps}{2\eps'}\right|\right)^{\ell-1} 
\le
2k\left|\frac{\eps'-\eps}{2\eps'}\right|,
\end{align*}
and if $s\ge1$ then
\begin{align*}
\left|
\left(1-\frac{\eps'-\eps}{2\eps'}\right)^{\ell-s}
\left(\frac{\eps'-\eps}{2\eps'}\right)^s
\right|
\le
\left(1+\left|\frac{\eps'-\eps}{2\eps'}\right|\right)^{k-1}
\left|\frac{\eps'-\eps}{2\eps'}\right| 
\le
2\left|\frac{\eps'-\eps}{2\eps'}\right|
\le
2k\left|\frac{\eps'-\eps}{2\eps'}\right|.
\end{align*}
Thus, for every fixed $\boldsymbol{\tau}_H$,
\begin{align}
\left|
\E[\widetilde{\mathrm{val}}_{\eps'}(\boldsymbol{\tau}_H)]
-
\mathrm{val}_{\Psi}(\boldsymbol{\tau})
\right|
&\le
2km\left|\frac{\eps'-\eps}{2\eps'}\right|
\le
\frac{\eta m}{2^{k+4}}.\label{eq:max-kand-eps-bias}
\end{align}
Here $\widetilde{\mathrm{val}}_{\eps'}$ denotes the decoded value after
replacing $\eps$ by $\eps'$.

Next we bound the centered error around the new expectation. For the
bounded-degree algorithm, applying \cref{lem:concentration-max-kand} with
$\eps'$ and internal accuracy parameter $\eta/8$ gives, after a union bound,
\[
|X_i-\E[X_i]|\le \frac{(\eps')^k\eta m}{16}
\quad\text{for all }i\in\{0,\ldots,k\}.
\]
Therefore the centered error of the decoded counter value is at most
\begin{align}
\frac{(\eps')^k\eta m}{16(2\eps')^k}
\sum_{i=0}^k
\left(\frac12+\eps'\right)^{k-i}
\left(\frac12-\eps'\right)^i
&\le \frac{\eta m}{2^{k+4}},\label{eq:max-kand-bounded-centered}
\end{align}
where we used
$\sum_{i=0}^k(\frac12+\eps')^{k-i}(\frac12-\eps')^i\le 1$.

For the unbounded-degree offline algorithm, the same bounded-differences
calculation as in \cref{lem:unbounded-max-kand-concentration} applies with
$\eps'$ in place of $\eps$. Let $r=|H|$. By the stopping condition, every
$i\in L$ satisfies $\deg(i)\le \theta m/(r+1)$. For a fixed high assignment,
changing the low-variable oracle label $Y_i$ changes one factor by
$1/(2\eps')$, and all other decoded factors have absolute value at most
$(2\eps')^{-1}$. Hence changing $Y_i$ changes the decoded value by at most
\[
c_i=(2\eps')^{-k}\deg(i),
\]
and
\begin{align*}
\sum_{i\in L} c_i^2
=(2\eps')^{-2k}\sum_{i\in L}\deg(i)^2
\le
(2\eps')^{-2k}\frac{\theta m}{r+1}\sum_{i\in L}\deg(i) 
\le
\frac{k\theta m^2}{(2\eps')^{2k}(r+1)}.
\end{align*}
With internal accuracy $\eta/8$, the threshold is chosen so that the
bounded-differences inequality and a union bound over the $2^r$ high
assignments give
\begin{align}
\Pr\left[
\exists \boldsymbol{\tau}_H:
\left|
\widetilde{\mathrm{val}}_{\eps'}(\boldsymbol{\tau}_H)
-
\E[\widetilde{\mathrm{val}}_{\eps'}(\boldsymbol{\tau}_H)]
\right|
>
\frac{\eta m}{2^{k+4}}
\right]
&\le \delta.\label{eq:max-kand-unbounded-centered}
\end{align}
In the streaming algorithms, the reconstruction proofs are algebraic and use
only the coefficient bound $(2\eps')^{-k}$; after the same replacement of
parameters and with internal accuracy $\eta/8$, the additional sketching error
is also at most $\eta m/2^{k+4}$ uniformly over all candidates.

Combining the bounds in (\ref{eq:max-kand-eps-bias})--(\ref{eq:max-kand-unbounded-centered})
and the streaming reconstruction bound, every decoded candidate value differs
from the value of its corresponding feasible assignment by at most
$\eta m/2^{k+2}$. The maximum of these feasible-assignment values is $\opt$:
choosing $\boldsymbol{\tau}_H=\boldsymbol{\sigma}^*|_H$ gives $\opt$, while
every other choice is still a feasible assignment. Hence, if $M$ denotes the
maximum decoded candidate value, then
\[
|M-\opt|\le \frac{\eta m}{2^{k+2}}.
\]
The modified algorithms output
$\mathrm{ALG}=M-\eta m/2^{k+2}$. Therefore
\[
\mathrm{ALG}\le \opt,
\]
and, by taking $\boldsymbol{\tau}_H=\boldsymbol{\sigma}^*|_H$ in the maximum,
\begin{align*}
\mathrm{ALG}
\ge
\opt-\frac{\eta m}{2^{k+2}}-\frac{\eta m}{2^{k+2}} 
\ge
\opt-\frac{\eta m}{2^k}
\ge
(1-\eta)\opt,
\end{align*}
where the last step uses $\opt\ge m/2^k$ for \textsc{Max-$k$AND}. The proof
uses independence only among the oracle labels of distinct variables in this
single oracle vector, exactly as in the exact-$\eps$ proofs; it does not use
independence between repeated executions of the algorithm.
\end{proof}

\end{document}